\documentclass{article}

\input{preamble.sty}
\usepackage{lmodern}

\newcommand{\Q}{\mathbb{Q}}

\newcommand{\Exp}{\mathrm{Exp}}

\newcommand{\Err}{\mathrm{Err}}

\newcommand{\convweak}{\overset{d}{\to}}
\newcommand{\vech}{{\rm vech}}
\newcommand{\Trunc}{{\rm Trunc}}
\newcommand{\PlugIn}{{\rm PlugIn}}
\newcommand{\TG}{{\rm TG}}
\newcommand{\Miscov}{{\rm PCMR}}
\newcommand{\sel}{{\rm sel}}
\newcommand{\pair}{{\rm pair}}

\newcommand{\STG}{{\rm STG}}
\newcommand{\carve}{{\rm carve}}
\newcommand{\PCER}{{\rm PCER}}

\newcommand{\appropto}{\mathrel{\vcenter{
			\offinterlineskip\halign{\hfil$##$\cr
				\propto\cr\noalign{\kern2pt}\sim\cr\noalign{\kern-2pt}}}}}

\title{ {\bf Inference for location and height of peaks of a standardized field after selection} }

\author{
	Alden Green$^{1}$, Jonathan Taylor$^{2}$\\[0.5em]
	$^{1}$Department of Applied Mathematics and Statistics, Johns Hopkins University\\
	$^{2}$Department of Statistics, Stanford University 
}

\begin{document}
	
	\maketitle
	\RaggedRight
	
	\begin{abstract}
		Peak inference concerns the use of local maxima (``peaks'') of a noisy random field to detect and localize regions where underlying signal is present. We propose a peak inference method that first subjects observed peaks to a significance test of the null hypothesis that no signal is present, and then uses the peaks that are declared significant to construct post-selectively valid confidence regions for the location and height of nearby true peaks. We analyze the performance of this method in a smooth signal plus constant variance noise model under a high-curvature asymptotic assumption, and prove that it asymptotically controls both the number of false discoveries, and the number of confidence regions that do not contain a true peak, relative to the number of points at which inference is conducted. An important intermediate theoretical result uses the Kac-Rice formula to derive a novel approximation to the intensity function of a point process that counts local maxima, which is second-order accurate under the alternative, nearby high-curvature true peaks. 
	\end{abstract}
	
	\section{Introduction}
	
	Peak inference~\citep{friston1991comparing,worsley1992three,chumbley2010topological,schwartzman2011multiple,cheng2017multiple} refers to a class
	of statistical methods designed for multiple testing problems with a spatial or temporal component, in which the goal is to detect and localize regions where an
	underlying signal is non-zero. This is a fundamental goal in neuroimaging analyses, where such regions correspond to areas of the brain that activate in response to an external stimulus. It is also of interest in scientific applications such as astroimaging~\citep{perone2004false} and climate modeling~\citep{sommerfeld2018confidence}.
	
	It has been argued that the classical approach of conducting
	inference simultaneously across all points is ill-suited for identifying regions that contain signal~\citep{chumbley2009false, chumbley2010topological}. In contrast, peak inference methods treat the data as discrete observations of an
	underlying smooth process, and conduct inference only at
	local maxima (``peaks'') of this process, rather than at every
	observed location or time point. Intuitively, the location of peaks identify promising areas where signal may be present,  while peak height estimates the strength of the signal which in many problems is a useful measure of effect size. The task of the statistician is to conduct formal inference based on the location and height of peaks. 
	
	Previous work on peak inference has primarily dealt with \emph{detection}: does a given peak, observed in the data,  indicate the presence of underlying signal? This detection problem is well-studied
	in the neuroimaging literature, with previous work showing how to calibrate tests conducted at
	peaks to control the FWER~\citep{friston1991comparing,worsley1992three,taylor2007detecting} or FDR~\citep{chumbley2010topological,schwartzman2011multiple,cheng2017multiple} under a null hypothesis that no true signal is present at an observed peak. The statistical foundation of this work is Random Field Theory (RFT)~\citep{adler2010applications},
	which provides a mathematical framework to quantify the probability
	of observing peaks of a certain height, or clusters of a certain size,
	under the null. These methods use the Kac-Rice formula
	\citep{rice1945mathematical,adler2007random} applied to critical points of centered smooth
	Gaussian and related processes (which of course include peaks, i.e. local maxima).
	
	Having using peaks to identify regions likely to contain underlying signal, an obvious next question is \emph{localization}: where exactly is the signal? In this article, we consider this problem in a smooth signal plus constant variance noise
	model in which there is a well-defined notion of ground truth: the
	local maxima of the signal, which we will call \emph{true
		peaks}. While one can imagine different instantiations of the localization task, we will take the perspective that these true peaks are (perhaps implicitly) the targets of a peak inference method. This is the same perspective as adopted by~\citet{davenport2022confidence}, but they consider a setting where all observed peaks basically correspond to true peaks, and there is no need for selection. We consider the more realistic setting where peaks must be screened to control the number of false discoveries. Concretely, we 
		look to use only those peaks declared significant by a peak detection method  to estimate and construct confidence regions for the
		location and height of nearby true peaks.

	
	\subsection{Summary of contributions}
	
	We propose a straightforward, three-step procedure for peak detection followed by
	localization. The first two steps constitute a peak
	detection algorithm in the spirit of existing RFT methods~\citep{cheng2017multiple}. First, the
	random field is thresholded at a preliminary level to identify a
	candidate set of interesting peaks. Second, a formal significance
	test, calibrated using an asymptotic RFT approximation, is conducted
	at each peak larger than this threshold, to control for the probability of false
	discovery. 
	
	The third and novel step is the localization step: for each
	observed peak at which the null is rejected, we produce a confidence ellipse for
	the location and a confidence interval for the height of a nearby
	true peak. Of course, for these confidence regions to be valid one must account for the fact that inference is conducted only at peaks that pass the detection step. We account for this by using the conditional distribution of peaks given selection to calibrate our inferences, so that this localization step is an example of post-selection inference~\citep{taylor2015statistical},
	
	Our primary contribution
	is to provide rigorous theoretical guarantees for each step in this
	procedure, and for the procedure as a whole. We work within a smooth signal plus constant variance noise model and derive
	our results in an asymptotic regime where the signal's curvature
	around its maxima grows large. Our main results can be summarized as follows:
	\begin{itemize}
		\item \emph{Null discovery control}. The peak detection portion of our method controls the asymptotic probability of falsely discovering a peak at which no signal is present.
		\item \emph{Consistency of point estimates}. Most peaks at which inference is conducted are either rejected by our peak detection significance test, or in fact asymptotically consistently estimate a true peak.
		\item \emph{Conditional coverage}. Each confidence region produced
		by our method attains asymptotically nominal
		coverage, conditional on the event of selection. 
		\item \emph{Marginal coverage}.  Finally, we introduce and analyze an overall miscoverage
		rate -- the expected number of regions that fail to cover the location or height of a nearby true peak, relative to the total number of points at which inference is conducted -- and show that our method controls miscoverage
		at the nominal level, providing a guarantee on the marginal
		performance of our procedure.
	\end{itemize}
	Figure~\ref{fig:illustration} illustrates our method, and the notion of overall miscoverage it asymptotically controls, in a cartoon example. 
	
		\begin{figure}[htbp]
		\centering
		\begin{subfigure}[t]{0.44\linewidth}
			\centering
			\includegraphics[width=\linewidth]{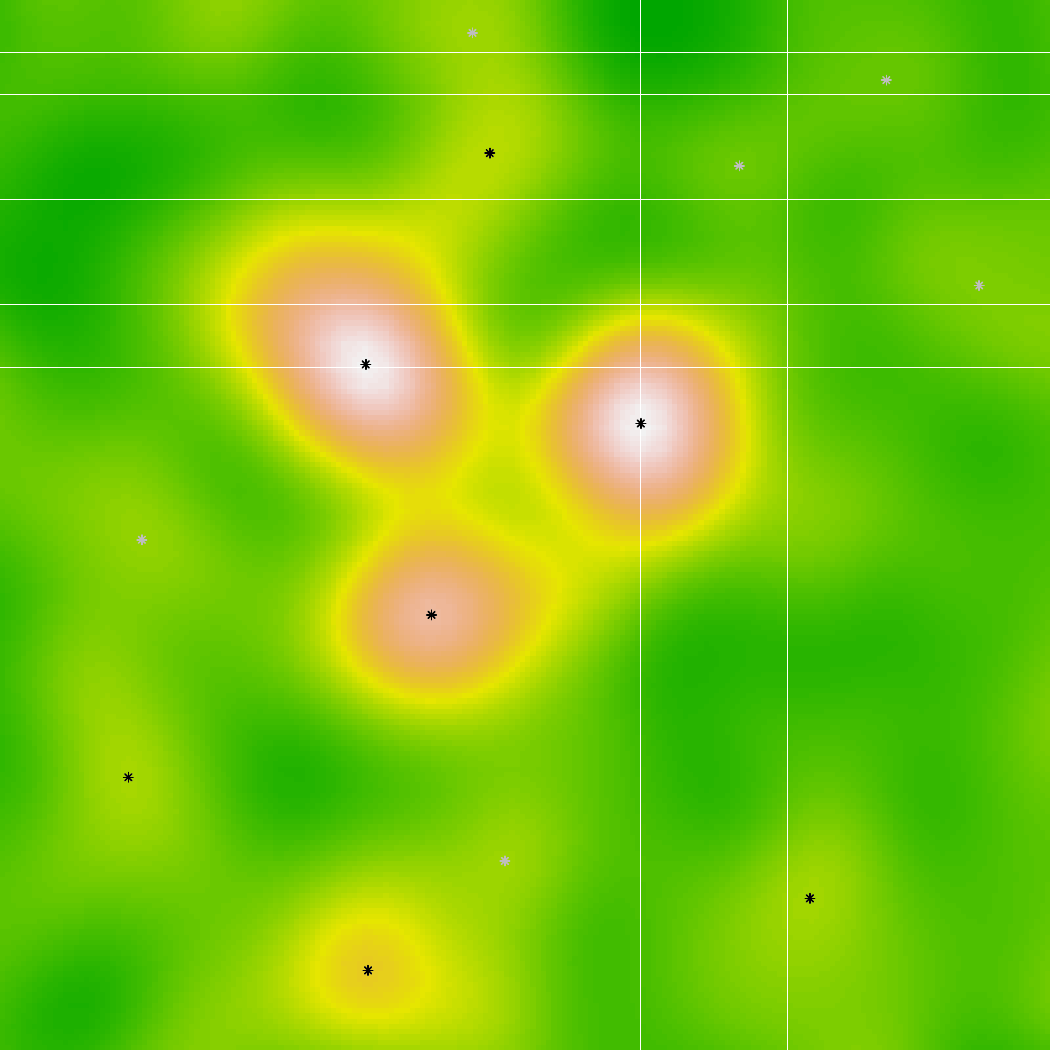}
			\caption{Field and pre-thresholded peaks.}
			\label{fig:illustration-1}
		\end{subfigure}\hfill
		\begin{subfigure}[t]{0.44\linewidth}
			\centering
			\includegraphics[width=\linewidth]{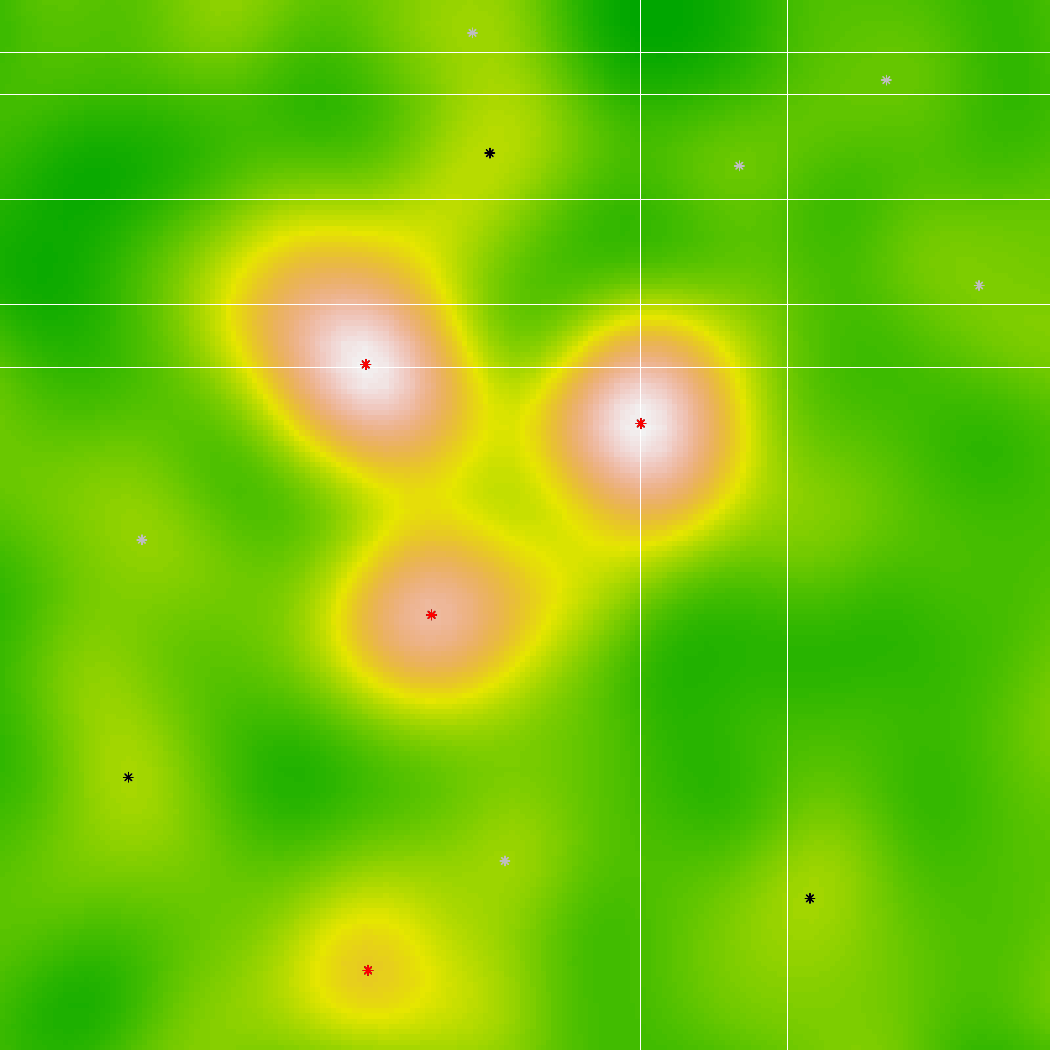}
			\caption{Peaks discovered by TG test.}
			\label{fig:illustration-2}
		\end{subfigure}\hfill
		
		\begin{subfigure}[t]{0.44\linewidth}
			\centering
			\includegraphics[width=\linewidth]{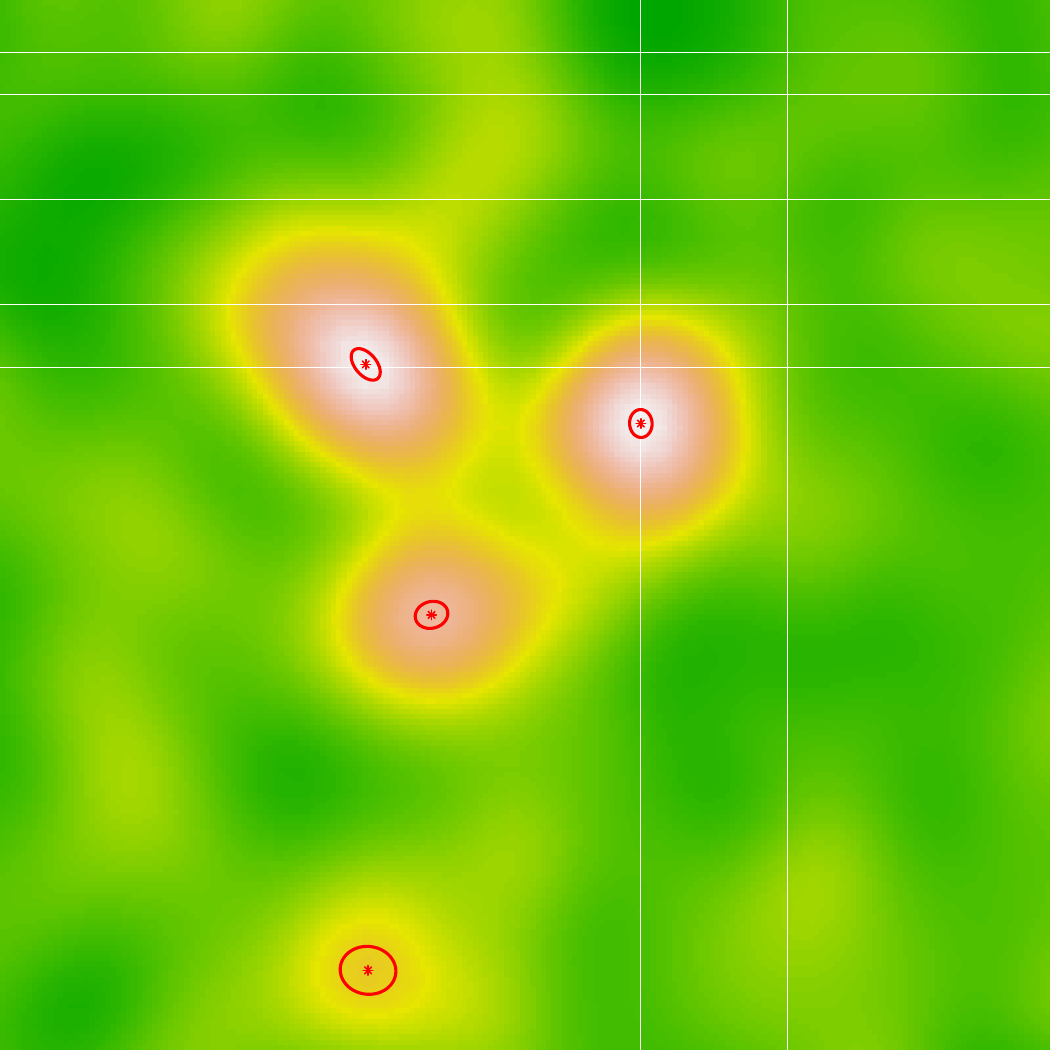}
			\caption{Confidence regions for location.}
			\label{fig:illustration-4}
		\end{subfigure}\hfill
		\begin{subfigure}[t]{0.44\linewidth}
			\centering
			\includegraphics[width=\linewidth]{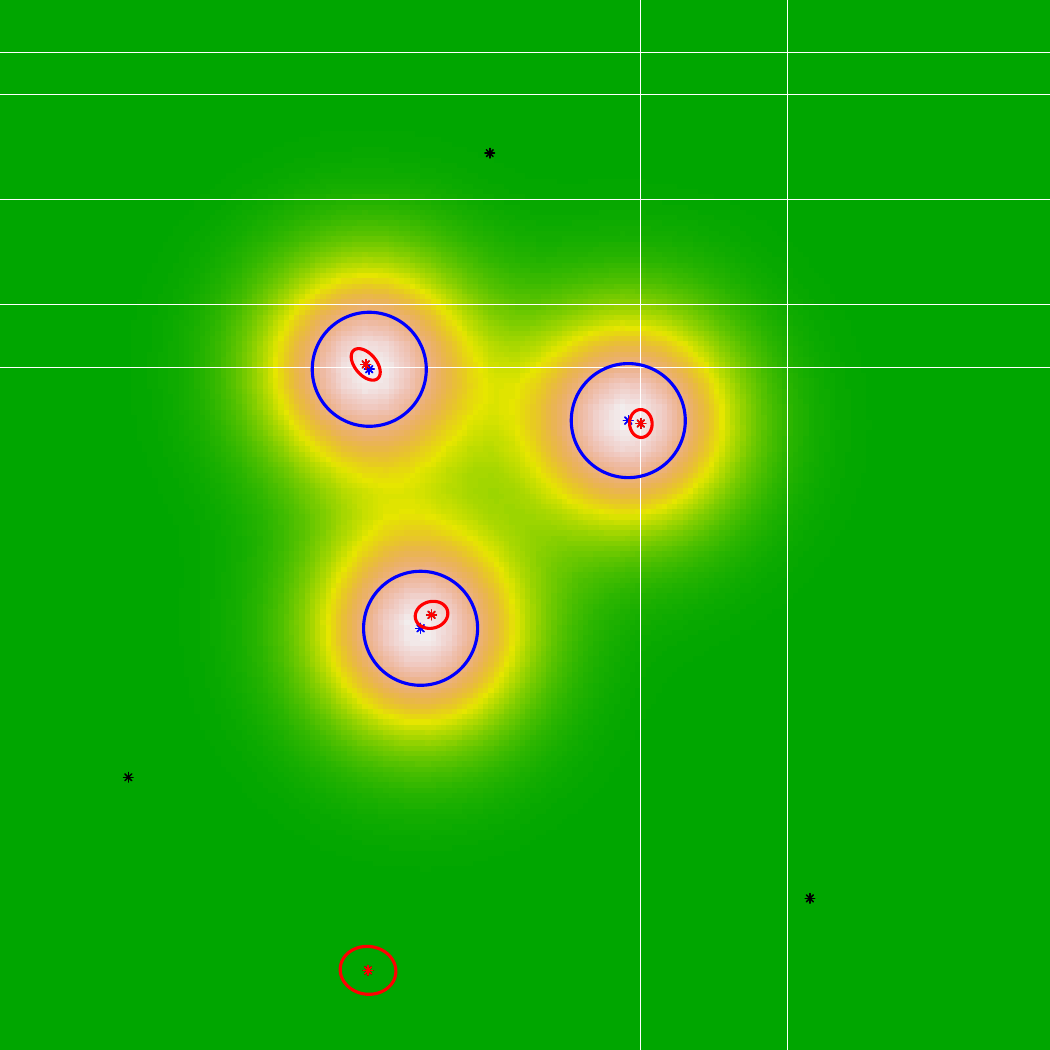}
			\caption{Confidence regions superimposed on the signal, with $\epsilon_n$ balls around true peaks.}
			\label{fig:illustration-5}
		\end{subfigure}
		
		\caption{ {\small Illustration of our method for peak detection and localization. The top left, top right and bottom left panels show Steps 1-3 of our method, respectively. The bottom right panel superimposes confidence ellipsoids for location over the true signal, which is sparse with 3 true peaks. Out of $13$ observed peaks (colored grey), $7$ survive the initial pre-thresholding step (black), of which $4$ are declared significant (red). One of these is a false discovery, in the sense of not being within distance $\varepsilon_n$ (defined in Section~\ref{subsec:peak-estimation}) of any true peak. Out of the remaining $3$ true discoveries, one of the corresponding confidence ellipsoids fails to cover the nearby true peak. Thus the miscoverage proportion for location -- the number of confidence regions that do not cover a true peak, out of the total number of points at which inference is conducted -- is $2/7$.} }
		\label{fig:illustration}
	\end{figure}
	
	These theoretical guarantees do not require that peaks be detected with asymptotic power one, and so they apply even in the challenging (but realistic) case where a true peak may or may not be discovered from data. In this regime, peaks that are discovered will be subject to selection pressure.  For instance, the height of a peak which survives a preliminary thresholding step will be a biased estimate for the value of the signal at that point~\citep{davenport2020selective}. Indeed, our significance test rejects the null only if the height is above some larger cutoff, introducing further bias. We will show that selection also affects the distribuition of the location, though in a more subtle way. Overall, our analysis precisely quantifies the various effects of selection, and our methods carefully construct confidence regions that \emph{account for selection pressure} in order to achieve the desired conditional and marginal (mis)coverage.
	
	In neuroimaging applications effect sizes are often weak and inferential methods derived from limiting approximations may have poor behavior in practice~(\cite{davenport2022confidence}, see also the discussion at the end of \cite{worsley2007RFT}). To address this, previous work on signal detection using RFT has derived corrections to asymptotic distributions of maxima that achieve higher-order accuracy under the null~\citep{taylor2005validity,taylor2007detecting,cheng2015distribution,cheng2017multiple}. In this work, we develop \emph{second-order accurate} approximations to post-selection peak distributions, under a smooth, high-curvature alternative. We demonstrate via simulation that the corrections used to achieve asymptotic second-order accuracy lead to improvements in empirical performance.
		
	Finally, although our theoretical guarantees account for selection pressure, when the probability of selection is small and selection pressure is strong, the post-selection peak inference problem becomes inherently challenging. This leads to a breakdown in coverage for peak location, and very wide intervals for peak height. To address this, we propose \emph{randomized peak inference}, which applies peak thresholding and detection to a randomized version of the data, and then uses the full un-randomized data for post-selection inference. Inspired by data splitting~\citep{cox1975note,kriegeskorte2009circular,kriegeskorte2010everything}, this randomization ensures that some information is always left for post-selection inference of peaks. We demonstrate empirically that under strong selection pressure, randomization improves coverage for a true peak's location, and leads to smaller intervals for the true peak's height.
	
	\subsection{Previous work}
	
	There is a long line of existing literature on signal
	detection using maxima, with tests calibrated via RFT~(e.g. \citet{friston1991comparing,worsley1992three,worsley1996unified,worsley2004unified,taylor2007detecting}), but this work does not typically consider what happens under the alternative, nearby true peaks. Some notable exceptions include~\citet{schwartzman2011multiple,cheng2017multiple} who consider
	FDR analysis of the {detection problem} described above, and~\citet{davenport2022confidence} who construct confidence regions for the location of true peaks
	in a regime where detection is essentially deterministic. Our work 
	addresses what lies between these two scenarios: \emph{how should
		we localize signals when detection is non-trivial?}
		
	We properly handle this intermediate case by conditioning on the event that a peak is detected, thus placing our problem squarely within the modern tradition of post-selection, or conditional, inference (e.g. \citet{lockhart2014significance,taylor2015statistical,lee2016exact,tibshirani2016exact}). Our problem is different from the ones considered in these previous works in a number of ways, but a particularly important difference is that our smooth signal plus noise model is not an exponential family, and as a result, the effect of selection on the distribution of peaks is in certain ways quite different than the selection effect in more well-studied post-selection inference problems. 
	
	To elaborate on this point slightly, consider the file drawer problem which served
	as a canonical post-selection inference problem in~\citet{fithian2014optimal}. In the file drawer problem, researchers decide to
	report a confidence interval for a population mean $\theta$ only if a sample mean $\bar{X}$ exceeds some threshold $t$. This problem has a close feel to inference for the height of a true peak $\mu_{t^*}$, given that a nearby observed peak $\hat{t}$ has height $Y_{\hat{t}}$ exceeding a significance cutoff $u$. (This notation will be formally defined in Section~\ref{sec:peak-detection-and-inference}.) However there no obvious analogy to inference for peak location
	in the file drawer problem. Indeed, we will see that the limiting distribution of the
	location is Gaussian, in contrast to the usual Truncated Gaussian distributions
	that appear in the file-drawer problem, and (in one form or another) in most previous work on post-selection inference. Even for the height, the distribution of $Y_{\hat{t}}$ is only asymptotically Truncated Gaussian, and higher-order asymptotic accuracy requires corrections that have no counterpart in the file drawer problem.

	Our approach adapts the well-studied Kac-Rice formula~\blue{(Rice 1945; Adler and Taylor 2007)} for the intensity of  local
	maxima of null smooth Gaussian processes to include the possibility of a growing signal.  Concretely, we develop a local expansion of this intensity nearby true peaks, and subsequently use this to obtain a second-order accurate approximation to the conditional distribution of observed peaks. Many connections can be drawn between these results and the classical higher-order theory of maximum likelihood estimation. In particular, our asymptotic expansions will be reminiscent of a number of classical higher-order accurate formulae for the limiting distribution of a maximum likelihood estimator~\citep{barndorff1979edgeworth,cox1980local,barndorff1983formula,skovgaard1985second,skovgaard1990density}. Moreover, the limiting variance of peak location in our problem depends on the observed curvature near the peak, and the difference between this observed curvature and its marginal expectation depends on the strength of the local selection
	effect. The relevance of observed curvature suggests connections to
	the literature on 2nd order efficiency and statistical curvature~\citep{efron1975defining,efron1978assessing}. We comment on these connections in more detail in Section~\ref{sec:peak-intensity-local-expansion}. 
	
	Finally, there have been some recent proposals for signal localization via formal statistical inference that do not target peaks or use asymptotic approximations to peak distributions.~\citet{sommerfeld2018confidence,bowring2019spatial,bowring2021confidence,maullin-sapey2024spatial} use RFT to conduct asymptotic inference for upper-level sets rather than peaks, while~\citet{spector2025controlled,gablenz2025catch} consider the same general problem of signal localization, but propose different methods that are not based on RFT.

	\subsection{Outline}
	In Section~\ref{sec:peak-detection-and-inference} we propose a method for peak detection and post-selection inference derived using tools from RFT, and summarize the theoretical properties of the method. Sections~\ref{sec:peak-intensity-local-expansion}-\ref{sec:peak-distribution} contain our main results. In Section~\ref{sec:peak-intensity-local-expansion} we give an asymptotic expansion of the distribution of local maxima in the neighborhood of true peaks. In Section~\ref{sec:peak-estimation} this result is used to establish that local maxima selected by our peak detection tend to successfully localize signal, by consistently estimating a true peak. In Section~\ref{sec:peak-distribution} we use the asymptotic distribution of local maxima to derive confidence regions for the true height and location of a nearby peak. Section~\ref{sec:randomized-peak-inference} proposes the randomized variant of our method, which has better asymptotic coverage for location and narrower confidence intervals for height. Section~\ref{sec:experiments} contains experiments with simulated data, and we conclude with some discussion in Section~\ref{sec:discussion}.
	
	\section{Peak detection and post-selection inference via Random Field Theory}
	\label{sec:peak-detection-and-inference}
	We begin with a concrete proposal for peak detection and localization via post-selection inference in a smooth signal plus constant variance noise model. In explaining this proposal we also provide a roadmap for the theory that follows, and give informal statements of our main theoretical results along with pointers to their formal counterparts.
	
	\subsection{Signal-plus-noise model}
	\label{subsec:model}
	Consider a sequence $(Y_n)_{n \in \mathbb{N}}$ of random fields, each distributed according to 
	\begin{equation}
		\label{eqn:model}
		Y_{n,t} = \mu_{n,t} + \epsilon_{n,t}, \quad t \in \R^d,
	\end{equation}
	where the signal $\mu_n \in C^2(\Rd)$, and the noise $\epsilon_n$ is a Gaussian process with mean zero and covariance kernel $\Cov[\epsilon_{n,t}, \epsilon_{n,s}] = K(t,s)$ not depending on $n$. We observe the restriction of $Y_{n,t}$ to a compact domain $\mc{T}_n \subset \Rd$, which is assumed to have non-empty interior. The \emph{true peaks} of the signal are strict interior local maxima of $\mu_n$,
	\begin{equation*}
		T^*_n = \{t \in \mathrm{int}(\mc{T}_n): \nabla \mu_{n,t} = 0,\nabla^2 \mu_{n,t} \prec 0\}.
	\end{equation*}
	Our targets of inference are the locations $T_n^*$ and heights $\mu_{T_{n}^*} = \{\mu_{n,t^*}: t^* \in T_n^*\}$ of true peaks. Throughout, we will assume that $K(\cdot,\cdot)$ is known and that the field has been standardized, meaning $K(t,t) = 1$ for all $t \in \Rd$. We will also assume that $K(\cdot,\cdot) \in C^5(\Rd \times \Rd)$, which -- combined with the smoothness assumptions placed on the signal --  implies that the field $Y_n$ is almost surely $C^2(\mc{T}_n)$. Furthermore, to avoid pathological situations we will always assume that
	\begin{equation}
		\label{eqn:non-degenerate}
		\begin{aligned}
			& \sigma_{1}^2 := \inf_{t \in \Rd} \lambda_{\min}(\Lambda_t) > 0, \quad \textrm{where} \quad \Lambda_t := \Cov[\nabla \epsilon_t], \quad \textrm{and} \\
			& \tau_{2}^2 :=  \inf_{t \in \Rd} \inf_{x \in \S^{d - 1}} \Var[x'\nabla^2 \epsilon_t x|\epsilon_t,\nabla \epsilon_t] > 0.
		\end{aligned}
	\end{equation}
	As a result $\epsilon_n$ is almost surely nowhere flat, and so almost surely there are a finite number of observed peaks within $\mc{T}_n$.  Moreover, the Hessian $\nabla^2 \epsilon_{n,t}$ is not a deterministic function of $\epsilon_{n,t},\nabla \epsilon_{n,t}$. Throughout we will drop notational dependence on $n$ whenever convenient.
	
	Our theoretical results are asymptotic in nature, and rely on the curvature of $\mu$ growing around true peaks $t^* \in T^*$ as $n \to \infty$. One measure of the curvature of $\mu$ at a peak $t^* \in T^*$ is the minimum eigenvalue of the negative Hessian $\lambda_{t^*} := \lambda_{\min}(-\nabla^2\mu_{t^* })$. We denote $\delta_{t^*} := 1/\lambda_{t^*}$ and assume throughout that $\delta_n := \sup_{t^* \in T^*} \delta_{t^*} \to 0$. We allow for the domain $\mc{T}_n$ to grow with $n$, since this is sometimes a natural asymptotic assumption in applications of RFT, but all of our asymptotic results hold even if $\mc{T}_n$ is fixed.
	
	It is worth making several comments on this model and asymptotic setup. Models such as the smooth signal plus noise model are typical in both theory and application of peak inference~\citep{schwartzman2011multiple,cheng2017multiple,davenport2022confidence}. For example, in neuroimaging applications it is accepted that neighboring voxels are spatially autocorrelated~\citep{worsley1992three,friston1996detecting}, and that signal is spatially smooth~\citep{chumbley2009false} and moreover additional smoothness is often introduced by a pre-processing step that smooths the data in order to increase SNR~\citep{worsley1992three,smith2009threshold}. When the data are pre-processed by smoothing, the mean function $\mu$ -- and hence the targets of inference $T^*$ and $\mu_{T^*}$ --  can depend on the smoothing algorithm. Henceforth we simply accept this, and do not attempt to infer on some underlying ground-truth; for some work that addresses how smoothing affects the type I error of peak detection methods, see~\citet{schwartzman2011multiple}.  
	
	Model~\eqref{eqn:model} restricts the noise distribution to be marginally stationary, since the variance of $\epsilon_{t}$ does not depend on $t$. Constant error variance is a common assumption in applications of RFT, where inferential procedures are often applied to standardized or ``test statistic'' fields~\citep{worsley1992three,siegmund1995testing,worsley1996unified,schwartzman2011multiple}. Most of the calculations that follow can be generalized to allow for non-constant error variance. However the resulting expressions are more complex and many of the implications for peak estimation and inference will change, and so we defer this to future work. 
	
	 There are two important concrete examples where the high-curvature asymptotic Assumption~\ref{asmp:curvature-asymptotics} is satisfied. The first is \emph{i.i.d asymptotics}, in which $Y_n = \frac{1}{\sqrt{n}}\sum_{i = 1}^{n}Y^{(i)}$ is $\sqrt{n}$ times the sample mean of $n$ i.i.d random fields $Y^{(i)} \sim N(\mu^{(1)},C)$ defined on a fixed domain $\mc{T}$, so that if $\min_{t^* \in T^*} \lambda_{\min}(-\nabla^2\mu^{(1)}_{t^*}) > c > 0$ then $\delta_n \to 0$ at rate $1/\sqrt{n}$. The second is \emph{strong signal asymptotics}, where $\mu_n = \sum_{j} a_{j,n} \phi_{t_j^*}(t)$ is modeled as a superposition of unimodal functions, each  $\phi_{t_j^*}(\cdot) \in C^2(\Rd)$ being compactly supported with a unique local maximum at $t_j^{*}$, and $\inf_{j} a_{j,n} \to \infty$~\citep{schwartzman2011multiple,cheng2017multiple}. In this case $\delta_n \to 0$ at rate $1/(\inf_{j} a_{j,n})$. 
	
	\subsection{Peak detection via the Truncated Gaussian test}
	\label{subsec:peak-detection-via-TG-test}
	The first high-level objective of peak inference is \emph{detection}, which can be formalized as a test of the null hypothesis $$H_{0,t}: \mu_s = 0 \; \textrm{at all} \; s \in B(t,r) \; \textrm{for some} \;r > 0.$$ 
	\citet{schwartzman2011multiple,cheng2017multiple} propose tests of $H_{0,t}$ that use RFT to derive p-values which properly account for selection of peaks $t \in \wh{T} = \{t \in T: \nabla Y_t = 0, \nabla^2 Y_t \prec 0\}$. We consider a slight variant of the method of~\citet{cheng2017multiple} that we call \emph{peak detection via the Truncated Gaussian (TG) test}. Our method first screens peaks via a thresholding step, restricting attention to
	\begin{equation}
		\label{eqn:selected-peaks}
		\wh{T}_{v} = \{t \in \mc{T}: Y_t > v, \nabla Y_t = 0,\nabla^2 Y_t \prec 0\}.
	\end{equation}
	Pre-thresholding in this manner is common in neuroimaging analyses~\citep{chumbley2009false,zhang2009cluster,chumbley2010topological,cheng2017multiple}, in part because it helps identify only those signals of practical, rather than purely statistical, significance.
	
	The next step of the method is to perform a hypothesis test at each pre-thresholded peak, rejecting $H_{0,t}$ if the value of the field $Y_t > u$ for some significance threshold $u > v$. This test is calibrated using RFT, which provides an exact formula -- given later in~\eqref{eqn:density-height-null} -- for the conditional density of $Y_t|t \in \wh{T}_v$ under $H_{0,t}$. Asymptotic analysis of this formula shows that as the pre-threshold $v \to \infty$, this conditional distribution approaches a Truncated Gaussian limit with a shifted mean, having asymptotic survival function\footnote{More precisely, if $H_{0,t}$ is correct then $|\P(Y_{t} > u|t \in \wh{T_v}) - \S_v(u,d/v)| = O(1/v^2)$; see Lemma~\ref{lem:survival-height-null}.}
	\begin{equation}
		\label{eqn:tg-null}
		\S_v(u,d/v) := \frac{\Psi(u - d/v)}{\Psi(v - d/v)}.
	\end{equation}
	Thus we calibrate the TG test by setting the threshold $u$ equal to the $(1 - \alpha)$th quantile of this distribution: that is, we set $u = u_{\TG}(\alpha,v)$ where $\S_v(u_{\TG}(\alpha,v),d/v) = \alpha$. Our first result states that this procedure asymptotically controls the \emph{null per-comparison error rate (null-PCER)}\footnote{Throughout the notation $N(S)$ refers to the cardinality of a set $S$.}
	\begin{equation}
		\label{eqn:null-pcer}
		\PCER_0(u,v) := \frac{\E[N(t \in \wh{T}_{u}: \textrm{$H_{0,t}$ is true} )]}{\E[N(\wh{T}_{v})]}.
	\end{equation}
	In words, this is the expected number of peaks at which the null is falsely rejected, over the expected number of peaks at which a hypothesis test is conducted.
	\begin{result}[Proposition~\ref{prop:null-false-positive-rate}, informal]
		\label{result:null-fpr}
		Fix $\alpha \in (0,1)$ and let $v_n \to \infty$. Under the regularity conditions of Section~\ref{subsec:assumptions},
		\begin{equation*}
			\lim_{n \to \infty}{\rm PCER}_0(u_{\TG}(\alpha,v_n),v_n) \leq \alpha.
		\end{equation*}
	\end{result}
	In fact the rate of convergence is second-order: $\PCER_0(u_{\TG}(\alpha,v_n),v_n) \leq \alpha + O(v_n^{-2})$. 
	
	To be clear,~\citet{schwartzman2011multiple,cheng2017multiple} have already shown that peak detection via RFT
	can successfully control type I error in testing $H_{0,t}$, building (in the latter case) on~\citet{cheng2015distribution}, who study approximations to the normalized overshoot distribution of $Y_t|t \in \wh{T}_v$ under high-threshold asymptotics. These approximations are both more accurate and more complicated than~\eqref{eqn:tg-null}.
	The novelty in Result~\ref{result:null-fpr} comes from recognizing that second-order accuracy can be obtained by applying a simple mean-shift (the $d/v$ term in~\eqref{eqn:tg-null}) in the Truncated Gaussian reference distribution. 
	
	PCER is less stringent than FWER or FDR, which are the other typical criteria used in calibrating peak detection methods. We believe that a suitable choice of significance cutoff $u$ would control FWER or FDR, and that this could be shown using analysis similar to that of~\citet{schwartzman2011multiple,cheng2017multiple}, but do not pursue this direction further. 
	
	\subsection{Peak estimation and rates of convergence}
	\label{subsec:peak-estimation}
	The output of a peak detection method is a set of \emph{discoveries} $\wh{T}_u = \{t \in \wh{T}_v: Y_t > u\}$. Discoveries can be used to estimate the location of true peaks, most basically by viewing $\wh{T}_u$ themselves as point estimates of true peaks. To measure the accuracy of these estimates, we introduce the
	\emph{$\varepsilon$-per-comparison error rate}
	($\varepsilon$-PCER): 
	\begin{equation}
		\label{eqn:epsilon-pcer}
		{\rm PCER}(\varepsilon,u,v) := \frac{\E[N(t \in \wh{T}_u: {\rm dist}(t,T^*) > \varepsilon)]}{\E[N(\wh{T}_v)]}.
	\end{equation}  
	In words, this is the expected number of points
	declared significant despite being distance $\varepsilon$
	from any true peak, over the expected number of points at which inference is conducted. The $\varepsilon$-PCER is a qualitatively different criterion than null-PCER that is stricter when $\varepsilon$ is small. Nevertheless, our second result shows that if the reciprocal curvature $\delta_n \to 0$, then peak detection via the TG test controls $\varepsilon$-PCER at rate $\varepsilon_n = \wt{O}(\delta_n)$.
	
	\begin{result}[Theorem~\ref{thm:epsilon-fpr}, informal]
		\label{result:epsilon-fpr}
		Fix $\alpha \in (0,1)$, let $v_n \to \infty$ and define $\varepsilon_n = \delta_n\sqrt{6 \sigma_1^2 \log \delta_n^{-1}}$. Under the signal-plus-noise model of Section~\ref{subsec:model} with high-curvature asymptotics $\delta_n \to 0$, and assuming the regularity conditions of Section~\ref{subsec:assumptions} and~\ref{sec:peak-estimation}, 
		\begin{equation*}
			\lim_{n \to \infty} {\rm PCER}(\varepsilon_n,u_{\TG}(\alpha,v_n),v_n) \leq \alpha.
		\end{equation*}
	\end{result}
	
	In other words, most peaks selected by pre-thresholding are either rejected by the TG test, or in fact consistently estimate a true peak $t^* \in T^*$ at an $\varepsilon_n$-rate of convergence. We will refer to such discoveries as \emph{$\varepsilon_n$-consistent discoveries}, and say that a true peak $t^*$ is \emph{$\varepsilon_n$-consistently discovered} by an observed peak $\hat{t}$ (or more colloquially, that $t^*$ is ``nearby'' $\hat{t}$) if the distance between them is at most $\varepsilon_n$.
	
	\subsection{Post-selection inference for height and location}
	\label{sec:informal:inference}
	The final objective of a peak inference procedure is to \emph{localize} the signal by producing confidence regions for the location of a nearby true peak, if one happens to exist. To achieve this, we derive the asymptotic distribution of $\hat{t} \in \wh{T}_u$, given that it is the unique $\varepsilon_n$-consistent discovery of a particular $t^* \in T^*$. Since we are also interested in constructing confidence intervals for the height $\mu_{t^*}$, we also examine the distribution of $\wh{Y} := Y_{\hat{t}}$ under the same event. In this analysis, RFT again plays a central role: we calculate these limiting distributions through asymptotic expansion of an exact Kac-Rice formula~\citep{rice1945mathematical,adler2007random} for the peak intensity function.

    The informal results stated below rely on some key quantities used throughout. Let $H_{t|y}, G_{t|y}$ denote the symmetric, positive semi-definite matrices
 	\begin{equation}
 		\label{eqn:deterministic-hessian}
 		H_{t|y} :=  -\nabla^2 \mu_t + K_{21}(t,t)(\Lambda_t^{-1}\nabla\mu_t) + (y - \mu_t)\Lambda_t , \quad G_{t|y} := (-\nabla^2 \mu_{t})\Lambda_{t}^{-1}(-\nabla^2\mu_{t}) + (y - \mu_{t})(-\nabla^2\mu_{t})
 	\end{equation} 
        with $K_{21}(t,t)$ denoting the covariance\footnote{For readers familiar with
        RFT in the non-stationary setting as in~\citet{adler2007random}, these terms describes the Levi-Civita
        connection of the metric induced by the random field $\epsilon_t$. Thus, the first two terms in $H_{t|y}$
        are nothing but the Riemannian Hessian of $\mu$ at $t$. We have decided to work in fixed coordinates
        here to ease exposition somewhat.} between $\nabla^2  Y_t$ and $\nabla Y_t$. Notice that at points $t^* \in T^*$ the gradient $\nabla \mu_{t^*}  = 0$ and so $H_{t^*|y} = -\nabla^2 \mu_{t^*} + (y - \mu_{t^*}) \Lambda_{t^*}$. 	Finally, we define $\bar{u}_{t^*} := \max\{u,\mu_t^*\}, \bar{H}_{t^*} = H_{t^*|\bar{u}_{t^*}}$ and $\bar{G}_{t^*} := G_{t^*|\bar{u}_{t^*}}$.
 
		The quantities $u_{t^*}$ and $\bar{H}_{t^*}$ have natural interpretations: roughly speaking, $\bar{u}_{t^*}$ is the asymptotic equivalent of $\wh{Y}$ conditional on $\hat{t} \in \wh{T}_u$ consistently discovering $t^*$;  while $\bar{H}_{t^*}$ is the asymptotic equivalent of the Hessian $-\nabla^2 Y_t$ conditional on the same event. (See Section~\ref{subsec:pf-pivot-hessian} for a more precise claim to this effect.)

	\begin{result}[Theorem~\ref{thm:approximate-joint-distribution}, informal]
		\label{result:approximate-joint-distribution}

		Under the smooth signal-plus-noise model of Section~\ref{subsec:model} with high-curvature asymptotics $\delta_n \to 0$, and the regularity conditions of Section~\ref{subsec:assumptions}, given that $\hat{t} \in \wh{T}_u$ is the unique $\varepsilon_n$-consistent discovery of $t^* \in T^*$, the following approximations hold as $n \to \infty$:
		\begin{itemize}
			\item The conditional density of the height $\wh{Y}$ is about
			\begin{equation*}
				p(y) \approx \1(y > u) \cdot \frac{1}{\sqrt{2\pi}} \frac{\exp\big(-\frac{1}{2}(y - \mu_{t^*} - \frac{1}{2}\tr(\bar{H}_{t^*}^{-1}\Lambda_{t^*}))^2\big) }{\Psi\big(u - \mu_{t^*} - \frac{1}{2}\tr(\bar{H}_{t^*}^{-1}\Lambda_{t^*})\big)}.
			\end{equation*}
			\item The conditional density of the location $\hat{t}$ given $\wh{Y} = y$ is about
			\begin{equation*}
				p(t|y) \approx \frac{\sqrt{\det(G_{t^*|y})}}{\sqrt{(2\pi)^d}} \cdot \exp\Big(-\frac{1}{2}(t - t)'G_{t^*|y} (t - t^*)\Big).
			\end{equation*}
		\end{itemize}
	\end{result}
	The densities $p(y)$ and $p(t|y)$ are examples of post-selection distributions, since they condition (among other things) on the event that a peak $\hat{t} \in \wh{T}$ nearby $t^*$ is also a discovery, i.e. that it has height $\wh{Y} > u$. In each case, the approximations are accurate up to second-order relative error. We defer a more detailed interpretation of these post-selection distributions to Section~\ref{subsec:selection-effect}, where we compare them to the marginal distribution of the height and location of a peak $\hat{t} \in \wh{T}$ nearby $t^*$, which does not condition on selection. This comparison reveals several effects of selection, some of which (to the best of our knowledge) have not previously been observed in work on post-selection inference. 
	
	On a more practical level, we can construct asymptotically pivotal quantites and valid confidence regions by taking the limiting distributions in Result~\ref{result:approximate-joint-distribution} and plugging in estimates for nuisance parameters. Concretely, to infer the height $\mu_{t^*}$ we propose to use the \emph{Truncated Gaussian (TG)} pivot
	\begin{equation}
		\label{eqn:tg-pivot}
		\wh{\S}_{\mu_{t^*}}(\wh{Y}) := \frac{\Psi(\wh{Y} - \mu_{t^*} - \frac{1}{2}\tr(\wh{H}^{-1}\Lambda_{\hat{t}}))}{\Psi(u - \mu_{t^*} - \frac{1}{2}\tr(\wh{H}^{-1}\Lambda_{\hat{t}}))},
	\end{equation}
	and to infer the location $t^*$ we propose to use the (studentized)~\emph{Wald} pivot 
	\begin{equation}
		\label{eqn:wald-pivot}
		\wh{W}_{t^*}(\hat{t}) = (\hat{t} - t^*)'\wh{H} \Lambda_{\hat{t}}^{-1} \wh{H}(\hat{t} - t^*),
	\end{equation}
	where we have written $\wh{H} = -\nabla^2 Y_{\hat{t}}$ for the Hessian of the field at $\hat{t}$. The TG pivot is asymptotically pivotal, and the Wald pivot is also asymptotically pivotal provided that $u - \mu_{t^*}$ -- which controls the probability of discovery -- not be too large.
	\begin{result}[Theorem~\ref{thm:pivot}, informal]
		\label{result:pivot}
		Under the conditions of Result~\ref{result:approximate-joint-distribution},
		$
		\wh{\S}_{\mu}(\wh{Y}) \convweak {\rm Unif}(0,1).
		$
		If additionally $(u - \mu_{t^*})_{+}\delta_n \to 0$, then
		$
		\wh{W}_{t^*}(\hat{t}) \convweak \chi_d^2.
		$
	\end{result}
	As an immediate consequence, the confidence regions
	\begin{equation}
		\label{eqn:confidence-regions}
		I_{\hat{t}} := \big\{\mu \in \R: \frac{\alpha}{2} \leq \wh{\S}_{\mu}(\wh{Y}) \leq 1 - \frac{\alpha}{2}\big\}, \quad C_{\hat{t}} := \big\{t \in \R^d: \wh{W}_{t}(\hat{t}) \leq q_{\chi_d^2}(1 - \alpha)\big\},
	\end{equation}
	have asymptotically nominal conditional coverage: given that $\hat{t}$ is the unique $\varepsilon_n$-consistent discovery of $t^*$, $C_{\hat{t}}$ will contain $t^*$, and $I_{\hat{t}}$ will contain $\mu_{t^*}$, each with probability approaching $(1 - \alpha)$. Notice that achieving conditional coverage is conceptually distinct from controlling either null- or $\varepsilon$-PCER, and that conditional coverage can be achieved even when $u \neq u_{\TG}$.

	\subsection{Overall coverage guarantees}
	In addition to providing conditional coverage, our overall method for peak detection and localization by post-selection inference -- summarized for the reader's convenience in Algorithm~\ref{alg:post-selective-inference-after-peak-detection-via-TG-test} -- also controls a marginal notion of miscoverage, which we call the \emph{per-comparison miscoverage rate} (PCMR). To define PCMR let $t^*(t) := \argmin_{t^* \in T^*} \|t - t^*\|$ denote the true peak closest to $t \in \mc{T}$: then the PCMR for location and height respectively is
	\begin{equation}
		\label{eqn:miscoverage}
		{\rm PCMR}_{T^*}(u,v) := \frac{\E[N(t \in \wh{T}_{u}: t^*(t) \not\in C_t)]}{\E[N(\wh{T}_v)]}, \quad\textrm{and}\quad {\rm PCMR}_{\mu_{T^*}}(u,v) := \frac{\E[N(t \in \wh{T}_{u}: \mu_{t^*(t)} \not\in I_t)]}{\E[N(\wh{T}_v)]}.
	\end{equation} 
	In other words, ${\rm PCMR}_{T^*}$ is the expected number of confidence regions which fail to contain the location of a nearby true peak, over the expected number of points at which inference is conducte, and an analogous interpretation holds for ${\rm PCMR}_{\mu_{T^*}}(u,v)$.
	
	\begin{result}[Theorem~\ref{thm:miscoverage}, informal]
		\label{result:miscoverage}
		Under the conditions of Result~\ref{result:epsilon-fpr}, 
		\begin{equation}
			\lim_{n \to \infty }{\rm PCMR}_{\mu_{T^*}}(u_{\TG}(\alpha,v_n),v_n) \leq \alpha. 
		\end{equation}
		If additionally $\sup_{t^* \in T^*}(u - \mu_{t^*})_{+}\delta_n \to 0$, then
		\begin{equation}
			\lim_{n \to \infty }{\rm PCMR}_{{T^*}}(u_{\TG}(\alpha,v_n),v_n) \leq \alpha. 
		\end{equation}
	\end{result}
	These overall bounds on marginal miscoverage are implied by the various results described above. The logic is simple. Among {null selections}, only an $100\alpha\%$ (in expectation) will be declared significant. Among $\varepsilon_n$-consistent discoveries, only a $100\alpha\%$ (in expectation) will produce confidence regions that do not cover the truth. And under suitable regularity conditions the contribution of all other discoveries is negligible. Thus, the overall PCMR is at most $\alpha$.
	
	\begin{algorithm}[t]
		\caption{Peak detection and localization via post-selection inference}
		\label{alg:post-selective-inference-after-peak-detection-via-TG-test}
		\begin{algorithmic}[1]
			\REQUIRE Field $Y$, gradient covariance $\Cov[\nabla \epsilon_t]$, pre-threshold $v \in \mathbb{R}$, significance level $\alpha \in (0,1)$.
			\STATE Select peaks $\wh{T}_v$ by pre-thresholding.
			\STATE At each peak $t \in \wh{T}_v$, test the null hypothesis $H_{0,t}$ using the TG test, making discoveries $\wh{T}_{u_{\TG}(\alpha,v)}$.
			\STATE For each significant peak $\hat{t} \in \wh{T}_{u_{\TG}(\alpha,v)}$, return a confidence interval $I_{{\hat{t}}}$ for the height, and a confidence ellipsoid $C_{\hat{t}}$ for the location, of a nearby true peak.
		\end{algorithmic}
	\end{algorithm}
	
	\section{Local expansion of intensity via the Kac-Rice formula}
	\label{sec:peak-intensity-local-expansion}
	Our various criteria for error are stated in terms of expectations involving the point process
	\begin{equation*}
		N_{\mc{I}}(\mc{B}) = N(t \in \wh{T} \cap \mc{B}: Y_t \in \mc{I}),
	\end{equation*}
	which counts the number of local maxima in $\mc{B} \subseteq \mc{T}$ with height in $\mc{I} \subseteq \R$. The main result of this section, Theorem~\ref{thm:approximate-joint-intensity}, establishes an approximation $\bar{\rho}(t,y)$ of the intensity $\rho(t,y)$ of this point process that is locally accurate around true peaks $t^* \in T^*$ up to second-order relative error. This approximation is subsequently used to derive both the asymptotics of $\varepsilon_n$-PCER (in Section~\ref{sec:peak-estimation}), and the asymptotic post-selection distributions summarized in Result~\ref{result:approximate-joint-distribution} (in Section~\ref{sec:peak-distribution}). Thus our guarantees on consistent peak estimation, and conditional and marginal asymptotic miscoverage, all ultimately flow from Theorem~\ref{thm:approximate-joint-intensity}.
	
	\subsection{Preliminaries}
	\label{subsec:assumptions}
	
	\paragraph{Regularity conditions and asymptotics.}
	In addition to the basic signal plus noise model of Section~\ref{subsec:model}, we will work under certain regularity conditions. These conditions assert that the signal is smooth and well-conditioned in some fixed radius ball around each true peak, and that all peaks belong strictly to the interior of $\mc{T}$. 
	\begin{assumption}
		\label{asmp:signal-holder}
		Define $\mc{T}_{r}^* := \bigcup_{t^* \in T^*} B_d(t^*,r)$. There exists a constant $c_1 > 0$ such that the mean $\mu \in C^{4}(\mc{T}_{c_1}^*)$ for all $n \in \mathbb{N}$.
	\end{assumption} 
	\begin{assumption}
		\label{asmp:covariance-holder}
		The covariance $K(\cdot,\cdot) \in C^{5}(\Rd \times \Rd)$. Moreover, there exists a constant $C_1 < \infty$ such that $\sup_{t \in \Rd} \|K(t,\cdot)\|_{C_5(\Rd)} < C_1$. 
	\end{assumption}
	\begin{assumption}
		\label{asmp:well-conditioned}
		There exists a constant $c_2 > 0$ and a constant $C_2 < \infty$ such that for all $n \in \mathbb{N}$,
		\begin{equation}
			\label{eqn:well-conditioned}
			\sup_{t^* \in T^*} \lambda_{\max}(-\nabla^2 \mu_{t^*}) \cdot \delta_{t^*} \leq C_2, \quad \sup_{t^* \in T^*} \sup_{t \in B(t^*,c_2)} (\|\nabla^3 \mu_{t}\|, \|\nabla^4\mu_t\|) \cdot \delta_{t^*}  \leq C_2.
		\end{equation}
	\end{assumption}
	\begin{assumption}
		\label{asmp:signal-boundary}
		There exists a constant $c_3 > 0$ such that ${\rm dist}(T^{\ast},\partial \mc{T}) > c_3$ for all $n \in \mathbb{N}$.
	\end{assumption}
	Additionally, as a reminder, our local expansions will be asymptotically accurate as the curvature of $\mu$ grows around peaks $t^*$.
	\begin{assumption}
		\label{asmp:curvature-asymptotics}
		Let $\delta_n := \sup_{t^* \in T^*} \delta_{t^*}$. Assume $\delta_n \to 0$.
	\end{assumption}
	Assumptions~\ref{asmp:signal-holder}-\ref{asmp:curvature-asymptotics}
	are comparable to classical assumptions made in the asymptotic
	analysis of M-estimators~\citep{vandervaart2000asymptotic}, but they allow
	for the signal to have multiple peaks and for the domain
	$\mc{T}$ to change with $n$. The amount of smoothness assumed
	on $\mu$ and $C$ in Assumptions~\ref{asmp:signal-holder}
	and~\ref{asmp:covariance-holder} is used to obtain
	second-order accuracy. Assumption~\ref{asmp:well-conditioned}
	ensures that $\mu$ is well-conditioned in a neighborhood of
	each peak. Note that an implication of
	Assumption~\ref{asmp:well-conditioned} is that all true peaks
	are well-separated, meaning there exists some constant $c >
	0$ such that $\inf_{t^* \neq s^*: t^*, s^* \in T^*} \|t^* - s^*\| > c$.
	
	Finally, in order to more concisely state our upper bound on the relative error between $\bar{\rho}(t,y)$ and $\rho(t, y)$, we make an additional technical assumption.
	\begin{assumption}
		\label{asmp:hessian-curvatures}
		There exists a constant $C_3 < \infty$ such that for all $n \in \mathbb{N}$, 
		\begin{equation}
			\label{eqn:hessian-curvatures}
			\sup_{t^* \in T^*} \lambda_{\max}\big(\bar{H}_{t^*}\big)\delta_n \leq C_3.
		\end{equation}
	\end{assumption}
	This assumption is made strictly for ease of exposition, since it means that the single asymptotic parameter $\delta_n$ will ultimately govern asymptotic error. Note that~\eqref{eqn:hessian-curvatures} is violated if $\sup_{t^* \in T^*} (u - \mu_{t^*})\delta_n \to \infty$ i.e. if the threshold is too large relative to the signal; however the TG test would be extremely unlikely to detect such a peak in the first place. We emphasize that all subsequent local expansions are correct even if Assumption~\ref{asmp:hessian-curvatures} does not hold, the upper bounds on error simply become more complicated. 
	\paragraph{Notation.}
	It will be useful to introduce some notation to more compactly write the local expansion $\bar{\rho}(t,y)$. This local expansion will depend in part on higher-order derivatives of $K(s,t)$. These are encoded in the tensors $K_{ij}(s,t) := \Cov[\nabla^i\epsilon_s,\nabla^j\epsilon_t]$, which (written in Euclidean coordinates) are $(i + j)$-arrays, having entries $[K_{ij}(s,t)]_{ab} = \Cov[D^aY_s,D^bY_t]$ for multi-indices $a,b$.  We use the following compact notation for array-vector composition: if $A$ is a $k$-array and $v \in \Rd$ is a vector, then $A(v)$ will denote the $(k - 1)$-array with entries
	\begin{equation*}
		A(v)_{a_1,\ldots,a_{k - 1}} = \sum_{j = 1}^{d} A_{a_1,\ldots,a_{k - 1},j} v_j,
	\end{equation*}
	Additionally, $A(v,v) = (A(v))(v)$, $A(v,v,v) = ((A(v))(v))(v)$, and so forth. We let $\dot{A}_t$ denote the derivative of an array-valued mapping $A_t: \mc{T} \to \R^{d \times d \times \cdots \times d}$, which has entries $[\dot{A}_t]_{i_1i_2\dots i_k m} = [D^m A_t]_{i_1i_2\dots i_k}$. 
        
We will use $C$ to represent a constant that may change from line to line, and that may depend on the covariance $K(t,s)$ and the dimension $d$ (and sometimes, on the level $\alpha$ used in our inferential procedures), but that does not depend on $n$. We will use the phrase ``for all $n$ sufficiently large'' to mean ``for all $n \geq C$ where $C$ is constant that may depend on $K(t,s), d$ and $\alpha$.'' Finally, for convenience we will let $c_0 = \min\{c_1,c_2,c_3\}$ and $C_0 = \max\{C_1,C_2,C_3,C_4\}$. 
	
	\subsection{Kac-Rice formula}
	In developing an approximation to the intensity of $N_u(\mc{B})$, our starting place will be the Kac-Rice theorem as stated in~\citet{adler2007random}, which is valid under the smoothness conditions of Section~\ref{subsec:model} and gives the exact formula 
	\begin{equation}
		\label{eqn:kac-rice}
		\E[N_u(\mc{B})] = \int_{\mc{B}} \E[\det(-\nabla^2 Y_t) \cdot \1(Y_t > u) \cdot \1(\nabla^2 Y_t \prec 0)|\nabla Y_t = 0,Y_t = y] \cdot f_{\nabla Y_t}(0) \,dt.
	\end{equation}
	Here we have adopted the notational shorthand $N_u(\mc{B}) = N_{(u,\infty)}(\mc{B})$, and write $f_{\nabla Y_t}(\cdot)$ for the Gaussian density of $\nabla Y_t$. Since we are interested in the joint distribution of location and height, it is useful to rewrite this as $\E[N_u(\mc{B})] = \int_{\mc{B}} \int_{u}^{\infty} \rho(t,y) \,dy \,dt$ where the intensity function
	\begin{equation}
		\begin{aligned}
			\label{eqn:kac-rice-intensity}
			\rho(t,y) 
			& =  \E[\det(-\nabla^2 Y_t) \cdot \1(\nabla^2 Y_t \prec 0)|\nabla Y_t = 0,Y_t = y] \cdot f_{\nabla Y_t,Y_t}(0,y)  \\
			& =  \E[\det(-\nabla^2 Y_t) \cdot \1(\nabla^2 Y_t \prec 0)|\nabla Y_t = 0,Y_t = y] \cdot f_{\nabla Y_t}(0) \cdot f_{Y_t}(y),
		\end{aligned}
	\end{equation}
	and in the second line we have used the assumption that $\epsilon_t$ has constant variance -- and thus $\Cov[\nabla \epsilon_t,\epsilon_t] = 0$ and $\nabla \epsilon_t$ and $\epsilon_t$ are independent -- to factorize $f_{\nabla Y_t, Y_t}(0,y) = f_{\nabla Y_t}(0) \cdot f_{Y_t}(y)$. 
	\subsection{Local expansion of intensity}
	Our approximation $\bar{\rho}(t,y)$ to $\rho(t,y)$ is a second-order accurate local expansion of~\eqref{eqn:kac-rice-intensity} around a true peak $t = t^*$, and around the limiting height $y = \bar{u}_{t^*}$.
	To define the distance at which this expansion is valid, we introduce the notation
	\begin{equation}
		\label{eqn:local-parameters}
		\varepsilon_n = \delta_n\sqrt{6\sigma_1^2\log(\lambda_n)}, \; \Delta_n := \sqrt{6\log(\lambda_n)},
	\end{equation}
	where we recall $\lambda_n = 1/\delta_n$. (The choice of constant $6$ in the definitions of $\varepsilon_n$ and $\Delta_n$ is arbitrary, and the results should continue to hold for any constant $>4$.) The first-order terms in the local expansion are
	\begin{equation}
		\label{eqn:intensity-first-order}
		\begin{aligned}
			T_{01}^{\rho}(y) & := T_{01}^{\det}(y) = (y - \bar{u}_{t^*})\tr\big(\bar{H}_{t^*}^{-1}\Lambda_{t^*}) \\
			T_{10}^{\rho}(h) & := T_{10}^{\det}(h) + T_{10}^{\nabla Y}(h)= \tr\big(\bar{H}_{t^*}^{-1} \dot{H}_{t^*|\bar{u}_{t^*}}(h)\big) - \frac{1}{2}\tr\big(\Lambda_{t^*}^{-1}\dot{\Lambda}_{t^*}(h)\big)  \\
			T_{30}^{\rho}(h) & := T_{30}^{\nabla Y}(h)= h' \nabla^2 \mu_{t^*} \Lambda_{t^*}^{-1} \{\nabla^3\mu_{t^*}(h,h)\} - h'\nabla^2 \mu_{t^*} \Lambda_{t^*}^{-1}\{\dot{\Lambda}_{t^*}(\Lambda_{t^*}^{-1}h)\} h\\ 
			T_{21}^{\rho}(h,y) & := T_{21}^{Y}(h,y) = \frac{1}{2}(y - \bar{u}_{t^*})h'\nabla^2\mu_{t^*}h,
		\end{aligned}
	\end{equation}
	where $\dot{H}_{t^*|\bar{u}_{t^*}}(h) = -\nabla^3\mu_{t^*}(h) + (\bar{u} - \mu_{t^*})\dot{\Lambda}_{t^*}(h) - \Gamma_{t^*}(\nabla^2\mu_{t^*}(h))$ denotes the derivative of the mapping $t \mapsto H_{t|\bar{u}_{t^*}}$ evaluated at $t = t^*$. The multi-index notation in~\eqref{eqn:intensity-first-order} suggests terms in a Taylor expansion, while the superscripts indicate correspondence to terms in~\eqref{eqn:kac-rice}. The relative error incurred by the approximation is determined by
	\begin{equation*}
		\begin{aligned}
			\Err_{Y}(h,y) & := C\Big(\big(|y - \bar{u}_{t^*}| + |\bar{u}_{t^*} - \mu_{t^*}|\big) \lambda_n \|h\|^3 + \lambda_n^2\|h\|^4 + |y - \bar{u}_{t^*}|^2 \lambda_n^2\|h\|^4\Big) \\
			\Err_{\det}(h,y) & := C\Big(\big(|y - \bar{u}_{t^*}|^2 + 1\big)\delta_{n}^{2} + \|h\|^2\Big) \\
			\Err_{\nabla Y}(h) & := C\Big(\|h\|^2 + \|h\|^4 \lambda_n^2\Big).
		\end{aligned}
	\end{equation*}
	\begin{theorem}
		\label{thm:approximate-joint-intensity}
		Under Assumptions~\ref{asmp:signal-holder}-\ref{asmp:hessian-curvatures}, for all $n \in \mathbb{N}$ sufficiently large the following statement holds: at all $t^* \in T^*$, $h \in B_{d}(0,\varepsilon_n), y \in \bar{u}_{t^*} \pm \Delta_n$,
		\begin{align*}
			\frac{\Big|\rho(t^* + h,y) - \bar{\rho}(t^* + h,y)\Big|}{\bar{\rho}(t^* + h,y)} \leq C\Big(\Err_{\det}(h,y) + \Err_{Y}(h,y) + \Err_{\nabla Y}(h)\Big) := \Err_{\rho}(h,y).
		\end{align*}
		where
		\begin{equation}
			\begin{aligned}
				\label{eqn:approximate-joint-intensity}
				\bar{\rho}(t^* + h,y) & := \\
				\1(y > u) & \cdot  \frac{\det(\bar{H}_{t^*})\big(1 + T_{01}^{\rho}(y) + T_{10}^{\rho}(h) - \frac{1}{2}T_{30}^{\rho}(h) + T_{21}^{\rho}(h,y)\big)}{\sqrt{(2\pi)^{d + 1}\det(\Lambda_{t^*})}} \exp\Big(-\frac{(y - \mu_{t^*})^2}{2}\Big) \cdot \exp\bigg(-\frac{h' \bar{G}_{t^*}h}{2}\bigg).
			\end{aligned}
		\end{equation}	
	\end{theorem}
	Theorem~\ref{thm:approximate-joint-intensity} is proved in Section~\ref{sec:pfs-peak-intensity-local-expansion}. Roughly speaking, the approximation $\bar{\rho}(t,y)$ is derived by local expansion of each of the three terms in~\eqref{eqn:kac-rice} -- the expectation of the Hessian determinant, the gradient density $f_{\nabla Y_t}(0)$ and the height density $f_{Y_t}(y)$ -- about $t = t^*, y = \bar{u}_{t^*}$, followed by a careful analysis of the relative magnitude of all terms in the resulting expansions. Several remarks are in order.
	
	\begin{remark}The terms in~\eqref{eqn:intensity-first-order} are labeled ``first-order'' because they are each $\wt{O}(\delta_n)$, whereas the approximation is accurate up to second-order relative error, meaning\footnote{A more careful analysis removes the $\log(\lambda_n)$ factor, see Theorem~\ref{thm:approximate-joint-distribution}.}
	$$
	\sup_{h \in B_d(0,\varepsilon_n), y \in \bar{u}_{t^*} \pm \Delta_n} \Err_{\rho}(h,y) = \wt{O}\Big( \delta_n^2 + (\bar{u}_{t^*} - \mu_{t^*}) \delta_{n}^2\Big).
 	$$
	More precisely, under negligible or moderate selection pressure (meaning $u \ll \mu_{t^*}$ or $u = \mu_{t^*} + O(1)$ respectively) the error is $\wt{O}(\delta_n^2)$ and the approximation is second-order accurate; under strong selection pressure ($u \gg \mu_{t^*}$ but $u/\mu_{t^*} \to 1$) the error is $\wt{O}((\bar{u}_{t^*} - \mu_{t^*}) \delta_{n}^2)$ and the approximation is between first- and second-order accurate; and under very strong selection pressure ($u/\mu_{t^*} \to C > 1$) the error is in fact also $\wt{O}(\delta_n)$. In a slight abuse of terminology we will sometimes simply summarize these different cases by saying that $\bar{\rho}(t,y)$ is ``second-order accurate.''
	\end{remark}
	
	\begin{remark}
		There exist many second-order accurate approximations to the distribution of a maximum likelihood estimator, either unconditionally or conditional on an ancillary statistic (see~\citep{barndorff1979edgeworth,cox1980local,barndorff1983formula,skovgaard1985second,skovgaard1990density} among many others). To a certain extent $\bar{\rho}(t,y)$ resembles some of these approximations, though it is not a direct consequence of any of them. Some distinctive features of $\bar{\rho}(t,y)$ are: (i) it is correct even when $Y$ is not a well-specified likelihood; (ii) it approximates the intensity of local maxima of height at least $u$, and thus accounts for selection and the possibility of multiple peaks; (iii) it is accurate up to second-order relative error, and valid even when $u = C\mu_{t^*}$ for $C > 1$ which corresponds to a large-deviations event, (iv) it is derived via expansion of a Kac-Rice formula. As far as we know, the Kac-Rice formula has not seen much use in higher-order parametric statistical theory, with the notable exception of~\citep{skovgaard1990density} who rederives the Kac-Rice formula and elucidates a connection to Barndorff-Nielsen's $p^*$-formula~\citep{barndorff1983formula}. Indeed we believe that many classical results in higher-order parametric statistics could be recovered from the Kac-Rice perspective, but this is out of the scope of the present work.
	\end{remark}
	
	\begin{remark}
		Inspection of the proof of Theorem~\ref{thm:approximate-joint-intensity} shows that the leading-order asymptotic behavior of $\rho(t,y)$ is entirely determined by the density term $f_{\nabla Y_t}(0) \cdot f_{Y_t}(y)$ in~\eqref{eqn:kac-rice}. This is because under high-curvature asymptotics the observed Hessian $-\nabla^2 Y_t$ is close to deterministic, and the ``determinant term'' $\E[(\det(-\nabla^2Y_t))\cdot\1(\nabla^2Y_t\prec0)|Y_t = y,\nabla Y_t = 0]$ is close to $\det(\bar{H}_{t^*})$. The fluctuations of the Hessian about its deterministic limit contribute to $\bar{\rho}(t,y)$ only through $T_{01}^{\rho}(y) = T_{01}^{\det}(y)$, which is the first-order term in a Taylor expansion of the determinant term about $\det(\bar{H}_{t^*})$. 
	\end{remark}
	
	
	\subsection{Distributional effect of selection: insights and connections}
	\label{subsec:selection-effect}

	Looking ahead, later in Section~\ref{sec:peak-distribution} we will show that up to a constant of proportionality $\bar{\rho}(t,y)$ gives the asymptotic post-selection \emph{density} of the height and location of a peak nearby $t^*$. This conditional density is defined on the event that (i) there is a unique $\hat{t} \in \wh{T}$ within distance $\varepsilon_n$ of $t^*$, (ii) the peak $\hat{t}$ is discovered, meaning $\wh{Y}$ survives the threshold $u$, and (iii) $\wh{Y} \in [\bar{u}_{t^*} - \Delta_n, \bar{u}_{t^*} + \Delta_n]$. It is instructive to compare this to the density conditional only (i); indeed with probability tending to one there is a unique $\hat{t} \in \wh{T}$ within distance $\varepsilon_n$ of $t^*$ with, so that it is only very slightly incorrect to think of this as the asymptotic \emph{marginal density} of $(\hat{t},\wh{Y})$. We now make a detailed comparison between conditional and marginal densities, revealing various effects of selection on the distribution of $(\hat{t},\wh{Y})$, and drawing connections where possible to related previous work on RFT and conditional inference. 
	
	\paragraph{Leading-order effect of selection: truncation and Goldilocks precision.}
	By ignoring the first-order terms in~\eqref{eqn:approximate-joint-intensity}, we obtain a simpler approximation to the intensity that is only first-order accurate, but nonetheless asymptotically valid, and which, when truncated to $y \in (u,\infty)$ reads
	\begin{equation*}
		\rho(t^* + h,y) \1(y > u) \propto (1 + \wt{O}(\delta_n)) \cdot \1(y > u) \exp\Big(-\frac{(y - \mu_{t^*})^2}{2}\Big) \cdot \exp\bigg(-\frac{h' G_{t^*|\bar{u}_{t^*}}h}{2}\bigg).
	\end{equation*}
	The aforementioned marginal density is obtained by setting $u = -\infty$ in the above expression, while taking $u > -\infty$ gives the density that conditions on selection. We draw the following conclusions.
	\begin{itemize}
		\item To a leading order $\wh{Y}$ and $\hat{t}$ are independent, both marginally and conditionally on selection.
		\item The leading-order effect of selection on the distribution of $\wh{Y}$ is to truncate its limiting Gaussian distribution to $(u,\infty)$. This is intuitive and will be familiar to readers with a background in post-selection inference, as it is exactly analogous to the selection effect in the ``file drawer'' problem of~\citet{fithian2014optimal}. 
		\item The leading-order effect of selection on the distribution of the location $\hat{t}$ is different: conditioning on selection preserves asymptotic Normality and unbiasedness of $\hat{t}$ but changes its asymptotic precision, from 
		\begin{equation*}
			(-\nabla^2 \mu_{t^*})\Lambda_{t^*}^{-1}(-\nabla^2\mu_{t^*}) \quad \textrm{to} \quad \bar{G}_{t^*} = G_{t^*|\bar{u}_{t^*}} = (-\nabla^2\mu_{t^*})\Lambda_{t^*}^{-1}(-\nabla^2\mu_{t^*}) + (\bar{u}_{t^*} - \mu_{t^*})(-\nabla^2\mu_{t^*}).
		\end{equation*}
		In particular we can see that conditioning on selection has the (surprising and welcome) effect of \emph{increasing} the precision with which $\hat{t}$ estimates $t^*$. Intuitively, this is because selection biases the observed Hessian $\wh{H}$ upwards, to be nearer to $\bar{H}_{t^*} = \E[-\nabla^2 Y_{t^*}|Y_{t^*} = \bar{u}_{t^*}]$ than to its marginal limit $-\nabla^2\mu_{t^*}$, and this upwards bias in curvature increases asymptotic precision. 
	\end{itemize} 
	Less happily, the bias in $\wh{H}$ also poses challenges for inference. When selection is guaranteed, the asymptotic precision of the location is $(-\nabla^2 \mu_{t^*})\Lambda_{t^*}^{-1}(-\nabla^2\mu_{t^*})$~(as shown originally in \citet{amemiya1985advanced,davenport2022confidence}), and moreover this is consistently estimated by the sandwich estimator $\wh{H}\Lambda_{\hat{t}}^{-1}\wh{H}$. However, under selection pressure this sandwich estimator has the wrong asymptotic limit: $\bar{H}_{t^*}\Lambda_{t^*}^{-1}\bar{H}_{t^*}$ rather than $\bar{G}_{t^*}$. We refer to the correct asymptotic precision $\bar{G}_{t^*}$ as the ``Goldilocks'' precision, since routine algebra establishes 
	\begin{equation*}
		\bar{G}_{t^*} =  \bar{H}_{t^*}\Lambda_{t^*}^{-1}(-\nabla^2\mu_{t^*}),
	\end{equation*}
	revealing that $\bar{G}_{t^*}$ is a perfect compromise between the two incorrect sandwich forms of precision, i.e.
        \begin{equation*}
          (-\nabla^2 \mu_{t^*}) \Lambda_{t^*}^{-1} (-\nabla^2 \mu_{t^*}) \preceq \bar{G}_{t^*} \preceq \bar{H}_{t^*} \Lambda_{t^*}^{-1} \bar{H}_{t^*}.
          \end{equation*}
	
	\begin{remark} Conditional distributions where conditioning preserves asymptotic Normality and unbiasedness, but changes the limiting precision, are unusual in the study of either post-selection conditional inference or RFT. However, there are close parallels between these phenomena and those seen in a distinct conditional inference problem: maximum likelihood inference in curved exponential families~\citep{efron1975defining}. In this latter context, a classical line of work~\citep{fisher1925theory,efron1978assessing} recommends conditioning on an (approximately) ancillary statistic, and shows that while the MLE is conditionally asymptotically Normal and unbiased, the expected Fisher information is not the right measure of conditional precision. We condition on $\wh{Y} > u$ for an entirely different purpose, but the effect is similar: after conditioning, the expected Hessian $-\nabla^2\mu_{t^*}$ alone does not determine the precision of $\hat{t}$. 
	\end{remark}
	
	\paragraph{First-order effect of selection: mean-shift and conditional Goldilocks.} We now turn to interpreting the first-order terms in~\eqref{eqn:intensity-first-order}. 
	\begin{itemize}
		\item Both $T_{10}^{\rho}(h)$ and $T_{30}^{\rho}(h)$ are first-order corrections to the distribution of $\hat{t}$. While these terms must be accounted for to produce a second-order accurate approximation to the distribution of $\hat{t}$, a fortunate symmetry means these terms have no effect on the asymptotic distribution of the Wald pivot $\wh{W}_{t^*}(\hat{t})$, and thus, no effect on the accuracy of our asymptotic inferences. 
		\item The term $T_{10}^{\rho}(y) = T_{10}^{\det}(y)$ is a first-order correction to the distribution of $\wh{Y}$, that originates from the determinant term in~\eqref{eqn:kac-rice}, and which roughly speaking reflects upwards bias due to evaluating at a local maximum. This kind of bias will be familiar to readers familiar with RFT, since it is present even under the null $\mu \equiv 0$. This bias manifests itself in a mean shift of size $\tr(\bar{H}_{t^*}^{-1} \Lambda_{t^*})$ under both null and alternative (with $\bar{H}_{t^*} = v \Lambda_{t^*}$ under the null when thresholding at $v$).
		\item Finally, the term $T_{21}^{\rho}(h,y) = T_{21}^{Y}(h,y)$ term is a first-order correction to the distribution of both $\hat{t}$ and $\wh{Y}$ that originates from the height density $f_{Y_t}(y)$ in~\eqref{eqn:kac-rice}. Roughly speaking, it reflects downwards bias in $\wh{Y}$ due to evaluating $Y$ at a point $\hat{t} \neq t^*$ where $\mu_{\hat{t}} < \mu_{t^*}$. This kind of bias is not present under the null $\mu \equiv 0$, and manifests as an additional shift in the mean of $Y_{\hat{t}}$ near points $t^* \in T^*$. It also reflects a change in the conditional variance of $\hat{t}|\wh{Y}$, which can be positive or negative depending on whether $\wh{Y} > \bar{u}_{t^*}$ or $\wh{Y} < \bar{u}_{t^*}$.  
	\end{itemize}
	In obtaining local expansions to the post-selection densities $p(y)$ and $p(t|y)$ we show precisely that 
  at $t^* \in T^*$ the total effect of these  first-order terms is to shift the (pre-selection) mean of $\wh{Y}$ from 
	$$\mu_{t^*} \quad \textrm{to} \quad \mu_{t^*} + \frac{1}{2}\tr(\bar{H}_{t^*}^{-1}\Lambda_{t^*})$$
	and to change the conditional precision of $\hat{t}|\wh{Y} = y$ from
	$$
	\bar{G}_{t^*} \quad \textrm{to} \quad G_{t^*|y} = \bar{G}_{t^*} + (y - \bar{u}_{t^*}) (-\nabla^2\mu_{t^*}).
	$$
	The change in the conditional precision of $\hat{t}|\wh{Y} = y$ reflects the fact that $\hat{t}$ and $\wh{Y}$ are not second-order independent, due to the $T_{21}^{\rho}(h,y)$ term. 
	
	\section{Peak detection and estimation}
	\label{sec:peak-estimation}
	In this section, we combine the results of Section~\ref{sec:peak-intensity-local-expansion} with a more classical RFT analysis under the null $H_{0,t}$ to show that the peak detection part of our method -- i.e application of the TG test to peaks selected by pre-thresholding -- controls $\varepsilon_n$-PCER. We copy the definition of $\varepsilon$-PCER here for convenience:
	\begin{equation*}
		{\rm PCER}(\varepsilon,u,v) = \frac{\E[N(t \in \wh{T}_u: {\rm dist}(t,T^*) > \varepsilon)]}{\E[N(\wh{T}_v)]}.
	\end{equation*}
	To prove control of $\varepsilon_n$-PCER, we first show that the null distribution of $Y_t|t \in \wh{Y}_v$ is asymptotically Truncated Gaussian, and so the TG test controls the expected number of null discoveries. On the other hand, the local expansion $\bar{\rho}(t,y)$ of Section~\ref{sec:peak-intensity-local-expansion}  can be used to derive a lower bound on the expected number of $\varepsilon_n$-consistent discoveries. Finally, under the following assumptions the expected number of all other discoveries -- peaks that are neither null nor $\varepsilon_n$-consistent -- is negligible in comparison.
	
	\begin{assumption}
		\label{asmp:well-separated}
		Define the \emph{null region} $\mc{T}_0 = \{t: \exists r > 0 \; \textrm{such that} \; \mu_s = 0 \; \forall s \in B_d(t,r)\}$, and let 
		$$g_n(\varepsilon) := \inf_{t \in \mc{T}: d(t,T^*) > \varepsilon, \; d(t,\mc{T}_0) > 0} \|\nabla \mu_{t}\|.$$ 
		Assume that for any fixed $\varepsilon > 0$, $g_n(\varepsilon) \to \infty$, $\log \lambda_n/\{g_n(\varepsilon)\}^2 \to 0$, and $\log|\mc{T}_n|/\{g_n(\varepsilon)\}^2 \to 0$.  
	\end{assumption} 
	\begin{assumption}
		\label{asmp:bounded-hessian}
		There exists a constant $C_4$ such that $\sup_{t \in \mc{T}} \lambda_{\max}(-\nabla^2\mu_t) \leq C_4 \lambda_n$ for all $n \in \mathbb{N}$. 
	\end{assumption}	
	\begin{assumption}
		\label{asmp:interior-global-max}
		The global maximum of $\mu$ is achieved at an interior critical point: $\sup_{t \in \mc{T}} \mu_t = \max_{t^* \in T^*}\mu_{t^{*}}$ for all $n \in \mathbb{N}$.  
	\end{assumption}
	
	\begin{theorem}
		\label{thm:epsilon-fpr}
		Fix $\alpha \in (0,1)$, and suppose $v_n \to \infty$. Under Assumptions~\ref{asmp:signal-holder}-\ref{asmp:interior-global-max}, for all $n \in \mathbb{N}$ sufficiently large, 
		\begin{equation*}
			{\rm PCER}(\varepsilon_n,u_{\TG}(\alpha,v_n),v_n) \leq \alpha + C\Big(\frac{1}{v_n^2} + \delta_n^2 + \sup_{t^* \in T^*}(\bar{u}_{t^*} - \mu_{t^*}) \delta_n^2\Big).
		\end{equation*}
	\end{theorem}
	To make the argument outlined above rigorous, we partition $\mc{T}$ into three regions:
	\begin{enumerate}
		\item the \emph{null region} $\mc{T}_0$ defined in Assumption~\ref{asmp:well-separated},
		\item the \emph{$\varepsilon$-consistent discovery region} $\mc{T}_{\varepsilon_n}^{*} = \bigcup_{t^* \in T^*} B(t^*,\varepsilon_n)$, and 
		\item the \emph{high-gradient region} $\mc{G}_{\varepsilon_n} = \mc{T}\setminus(\mc{T}_0 \cup \mc{T}_{\varepsilon_n}^{*})$,\footnote{We call $\mc{G}_{\varepsilon_n}$ the high-gradient region because under Assumptions~\ref{asmp:well-conditioned} and~\ref{asmp:well-separated}, $g_n(\varepsilon_n) \to \infty$.}
	\end{enumerate}
	and write $\varepsilon_n$-PCER in terms of the expected number of thresholded peaks, and discoveries, in each region:
	\begin{equation*}
		{\rm PCER}(\varepsilon_n,u,v) := \frac{\E[N_{u}(\mc{T})] - \E[N_{u}(\mc{T}_{\varepsilon_n}^{*})]}{\E[N_{v}(\mc{T})]} = \frac{\E[N_{u}(\mc{G}_{\varepsilon_n})] + \E[N_{u}(\mc{T}_0)]}{\E[N_{v}(\mc{T}_{\varepsilon_n}^{*})] + \E[N_{v}(\mc{G}_{\varepsilon_n})] + \E[N_{v}(\mc{T}_0)]}.
	\end{equation*}
	In the subsequent sections we derive bounds on each expectation in the expression above, finally establishing Theorem~\ref{thm:epsilon-fpr} in Section~\ref{subsec:pf-thm-epsilon-pcer}. We comment more on Assumptions~\ref{asmp:well-separated}-\ref{asmp:interior-global-max} in Section~\ref{subsec:high-gradient-discoveries}.

	\subsection{Null per-comparison error rate}
	\label{subsec:null-pcer}
	Recall that $\PCER_0(u,v)$ is the expected number of null discoveries over the expected number of peaks selected by pre-thresholding. Peak detection via the TG test asymptotically controls the null-PCER as the pre-threshold $v_n \to \infty$. 
	\begin{proposition}
		\label{prop:null-false-positive-rate}
		Fix $\alpha \in (0,1)$, and suppose $v_n \to \infty$. For all $n \in \mathbb{N}$ sufficiently large,
		\begin{equation*}
			\PCER_0(u_{\TG}(\alpha,v_n),v_n) \leq \alpha + \frac{C}{v_n^2}.
		\end{equation*}
	\end{proposition}
	Proposition~\ref{prop:null-false-positive-rate} is proved in Sections~\ref{subsec:pf-null-false-positive-rate}-\ref{subsec:pf-jacobian-null}. The proof begins from the exact distribution of $Y_t|t \in \wh{T}_v$; formally speaking this is a \emph{Palm distribution} with density~\citep{adler2010applications,cheng2017multiple}
	\begin{equation}
		\label{eqn:density-height-null}
		p_v(y|t) = \frac{\E[\det(-\nabla^2 Y_t) \cdot \1(\nabla^2 Y_t \prec 0)|\nabla Y_t = 0, Y_t = y]}{\E[\det(-\nabla^2 Y_t) \cdot \1(Y_t > v, \nabla^2 Y_t \prec 0)| \nabla Y_t = 0]} \cdot f_{Y_t}(y) \cdot \1(y > v).
	\end{equation}
	Careful analysis of the determinant term above shows that the survival function of this Palm distribution agrees with the mean-shifted TG survival function $\S_v(u,d/v)$ up to second-order relative error, at which point integrating over $t \in \mc{T}_0$ yields the claim of Proposition~\ref{prop:null-false-positive-rate}. 
	
	Previous work~\citep{schwartzman2011multiple,taylor2016inference,cheng2017multiple} has used either the exact distribution~\eqref{eqn:density-height-null}, or exponentially accurate approximations, to calibrate tests for $H_{0,t}$ or the global null $H_0: \mu \equiv 0$. Indeed our mean-shifted TG approximation can be obtained by extracting the first few terms from the asymptotic approximation of~\citet{cheng2015distribution}. The accuracy of the resulting approximation is a middle ground between~\citet{cheng2015distribution}, and the classic $\Exp(1)$ approximation~\citep{adler2010geometry} of the normalized overshoot distribution of $Y_t$. The insight here is that second-order accuracy can be obtained by a simple mean shift in the reference distribution. One reason this insight is valuable is that it also applies -- locally, near high-curvature peaks -- under the alternative, where using the exact distribution~\eqref{eqn:density-height-null} for inference appears more challenging.
	
	\subsection{$\varepsilon_n$-consistent discoveries, and power}
	The results of Section~\ref{sec:peak-intensity-local-expansion} can be used to lower bound $\E[N_u(\mc{T}_{\varepsilon_n}^{*})]$, which is the total expected number of $\varepsilon_n$-consistent discoveries. In fact, we prove a slightly stronger result. Integrating $\bar{\rho}(t,y)$ over $t \in \mc{B}_{t^*} := B_d(t^*,\varepsilon_n)$ and $y \in \mc{I}_{t^*} := (\bar{u}_{t^*} \pm \Delta_n) \cap (u,\infty)$ gives a second-order accurate asymptotic approximation of $\E[N_{\mc{I}_{t^*}}(\mc{B}_{t^*})]$, which is the expected number of $\varepsilon_n$-consistent discoveries of $t^*$ that also have height falling within $\Delta_n$ of $\bar{u}_{t^*}$.  Of course $\E[N_u(\mc{B}_{t^*})] \geq \E[N_{\mc{I}_{t^*}}(\mc{B}_{t^*})]$, and some additional work shows that in fact the difference between the two is asymptotically negligible. To lighten notation, we let $\mc{S}_{t^*} := \mc{B}_{t^*} \times \mc{I}_{t^*}$, $N(\mc{S}_{t^*}) := N_{\mc{I}_{t^*}}(\mc{B}_{t^*})$, and 
	\begin{align*}
		\Err_{\rho}(\mc{S}_{t^*}) & := C \Big( |\bar{u}_{t^*} - \mu_{t^*}| \delta_n^2 + \delta_n^2\Big). \\
		\wb{\Trunc}_{\rho}(\mc{I}_{t^*}) & := C\delta_n^2.
	\end{align*}
	\begin{proposition}
		\label{prop:expectation-counting-process}
		Under the assumptions of Theorem~\ref{thm:approximate-joint-intensity}, for all $n \in \mathbb{N}$ sufficiently large:
		\begin{equation}
			\begin{aligned}
				\label{eqn:expectation-counting-process}
				\frac{|\E[N(\mc{S}_{t^*})] - \bar{\E}[N(\mc{S}_{t^*})]|}{\bar{\E}[N(\mc{S}_{t^*})]} 
				& \leq \Err_{\rho}(\mc{S}_{t^*}) + \wb{\Trunc}_{\rho}(\mc{I}_{t^*}) := \Err_{\E[N(\mc{S}_{t^*})]}, \quad \textrm{and} \\
				\frac{|\E[N_u(\mc{B}_{t^*}) - \E[N(\mc{S}_{t^*})]|}{\bar{\E}[N(\mc{S}_{t^*})]} 
				& \leq \wb{\Trunc}_{\rho}(\mc{I}_{t^*}), 
			\end{aligned}
		\end{equation}
		where
		\begin{equation}
			\label{eqn:approximate-number-true-discoveries}
			\bar{\E}[N(\mc{S}_{t^*})] := \sqrt{\frac{\det(\bar{H}_{t^*})}{\det(-\nabla^2\mu_{t^*})}}\exp\Big(-\frac{1}{2}(\bar{u}_{t^*} - \mu_{t^*})\tr(\bar{H}_{t^*}^{-1}\Lambda_{t^*})\Big)\Psi\Big(u - \mu_{t^*} - \frac{1}{2}\tr(\bar{H}_{t^*}^{-1}\Lambda_{t^*})\Big).
		\end{equation}
		Thus, $\E[N_{u}(\mc{T}_{\varepsilon_n}^*)] \geq (1 - C \delta_n^2 - C |\sup_{t^*} \bar{u}_{t^*} - \mu_{t^*}| \delta_n^2) \cdot \sum_{t^* \in T^*} \bar{\E}[N(\mc{S}_{t^*})]$.
	\end{proposition}
	The proof of Proposition~\ref{prop:expectation-counting-process} is given in Section~\ref{subsec:pf-expectation-counting-process}. 
	\begin{remark} As mentioned previously, under the high-curvature asymptotic assumptions of Theorem~\ref{thm:approximate-joint-intensity}, with probability tending to one there will be a unique $\hat{t} \in \wh{T}$ that $\varepsilon_n$-consistently estimates a given $t^*$. Equation~\eqref{eqn:approximate-number-true-discoveries} can thus be interpreted as the asymptotic probability that the TG test, applied at $\hat{t}$, correctly rejects the null $H_{0,t}$; in other words, as an asymptotic approximation to power. Some routine calculations show that if $(\bar{u}_{t^*} - \mu_{t^*})\delta_n \to 0$ then
		\begin{equation}
			\label{eqn:approximate-number-true-discoveries-2}
			\bar{\E}[N(\mc{S}_{t^*})] = 
			\big(1 + O(|\bar{u}_{t^*} - \mu_{t^*}|^2 \delta_n^2)\big) \cdot \Psi\Big(u - \mu_{t^*} - \frac{1}{2} \tr(\bar{H}_{t^*}^{-1}\Lambda_{t^*}\big)\Big).
		\end{equation}
		Once again, we see that an up to second-order accurate approximation -- this time to asymptotic power -- is given by applying a  first-order mean-shift to a reference distribution, this time to $\Psi(u - \mu_{t^*}) = \P(Y_{t^*} > u)$.  
	\end{remark}

	\subsection{High-gradient discoveries}
	\label{subsec:high-gradient-discoveries}
	Finally, the expected number of high-gradient discoveries is small relative to the expected number of $\varepsilon_n$-discoveries, under the regularity conditions of Section~\ref{subsec:assumptions} and Assumptions~\ref{asmp:well-separated}-\ref{asmp:interior-global-max}.
	
	\begin{proposition}
		\label{prop:high-gradient-peaks}
		Under Assumptions~\ref{asmp:signal-holder}-\ref{asmp:interior-global-max}, for all $n \in \mathbb{N}$ sufficiently large:
		\begin{equation}
			\label{eqn:high-gradient-peaks}
			\frac{\E[N_{u}(\mc{G}_{\varepsilon_n})]}{\sum_{t^* \in T^*}\bar{\E}[N(\mc{S}_{t^*})] } \leq C \delta_n^2.
		\end{equation}
	\end{proposition}
	
	Assumption~\ref{asmp:well-separated} resembles a typical condition made in the analysis of M-estimators (see e.g. \citet{vandervaart2000asymptotic} Theorem 5.9), but is weaker in that it allows for the possibility of multiple peaks and a non-empty null region $\mc{T}_0$. Assumption~\ref{asmp:bounded-hessian} prevents the signal from being dramatically curvier away from a true peak, implying a uniform upper bound on the determinant term in the Kac-Rice formula.
	Assumption~\ref{asmp:interior-global-max} ensures that the global maximum of $\mu_t$ is achieved by an interior local maximum and not a boundary point with large gradient.  Notice that these assumptions do not rule out the global null hypothesis $\mu \equiv 0$; they simply assert that any peaks that are present are asymptotically well-separated from the null.
	
	\subsection{Proof of Theorem~\ref{thm:epsilon-fpr}}
	\label{subsec:pf-thm-epsilon-pcer}
	Propositions~\ref{prop:null-false-positive-rate}-\ref{prop:high-gradient-peaks} immediately imply an upper bound on the $\varepsilon_n$-PCER: letting $u_{\TG} \equiv u_{\TG}(\alpha,v_n)$,
	\begin{equation}
		\begin{aligned}
			\frac{\E[N_{u_{\TG}}(\mc{G}_{\varepsilon_n})] + \E[N_{u_{\TG}}(\mc{T}_0)]}{\E[N_{u_{\TG}}(\mc{T}_{\varepsilon_n}^{*})] + \E[N_{v_n}(\mc{T}_0)]} 
			& \leq C\frac{\sum_{t^* \in T^*} \bar{\E}[N(\mc{S}_{t^*})]\delta_n^2 + (1 + C/v_n^2) \alpha \cdot \E[N_{v_n}(\mc{T}_0)]}{(1 - C \delta_n^2 - C (\sup_{t^* \in T^*} \bar{u}_{t^*} - \mu_{t^*}) \delta_n^2) \cdot \sum_{t^* \in T^*} \bar{\E}[N(\mc{S}_{t^*})] + \E[N_{v_n}(\mc{T}_0)]} \\
			& \leq \alpha + C\Big(\frac{1}{v_n^2} + \delta_n^2 + \sup_{t^* \in T^*}(\bar{u}_{t^*} - \mu_{t^*}) \delta_n^2\Big),
		\end{aligned}
	\end{equation}
	for all $n \in \mathbb{N}$ large enough that $C(\delta_n^2 + \sup_{t^* \in T^*}(\bar{u}_{t^*} - \mu_{t^*}) \delta_n) \leq 1/2$. This is the claim of Theorem~\ref{thm:epsilon-fpr}.
	
	
	\section{Post-selection inference with $\varepsilon_n$-consistent discoveries}
	\label{sec:peak-distribution}

	We now turn to localizing peaks via post-selection inference; more explicitly, inference that is valid conditional on $\hat{t} \in \wh{T}_u$ being the unique $\varepsilon_n$-consistent discovery of a true peak $t^* \in T^*$, with height $\wh{Y}$ within $\pm \Delta_n$ of $\bar{u}_{t^*}$. The conditional distribution of such a peak is given by
	\begin{equation}
		\label{eqn:distribution-after-selection}
		\Q^{\mc{S}_{t^*}}(\mc{A}) := \P\Big((\hat{t},\wh{Y}) \in \mc{A}\big|N(\mc{S}_{t^*}) = 1\Big) = \frac{\P(N(\mc{A}) = 1, N(\mc{S}_{t^*}) = 1)}{\P(N(\mc{S}_{t^*}) = 1)}, \quad \mc{A} \subseteq \mc{S}_{t^*}.
	\end{equation}
	The smoothness assumptions of~Section~\ref{subsec:model} suffice to guarantee that $(\hat{t},\wh{Y})$ are continuous under $\Q^{\mc{S}_{t^*}}$. The joint density of $(\hat{t},\wh{Y})$, marginal density of $\wh{Y}$, and conditional density of $\hat{t}|\wh{Y} = y$ are respectively (writing $\nu_d$ for the Lebesgue measure of $B_d(0,1)$)
	\begin{equation}
		\label{eqn:density-after-selection}
		p(t,y) := \lim_{r \to 0}\frac{1}{2|r|^{d + 1} \nu_d}\Q^{\mc{S}_{t^*}}\Big(B_d(t,r) \times [y - r, y + r]\Big), \quad p(y) = \int_{\mc{B}_{t^*}} p(t,y) \,dt, \quad p(t|y) = \frac{p(t,y)}{p(y)}, \quad (t,y) \in \mc{S}_{t^*}
	\end{equation}
	\begin{theorem}
		\label{thm:approximate-joint-distribution}
		Under the assumptions of Theorem~\ref{thm:approximate-joint-intensity}, for all $n \in \mathbb{N}$ sufficiently large the following statements hold at all $t^* \in T^*, h \in B_d(0,\varepsilon_n), y \in \bar{u}_{t^*} \pm \Delta_n$:
		\begin{itemize}
			\item The joint density of $(\hat{t},\wh{Y})$ given $N(\mc{S}_{t^*}) = 1$ satisfies
			\begin{equation}
				\label{eqn:joint-density-error}
				\begin{aligned}
					\frac{\Big|p(t^* + h,y) - \bar{p}(t^* + h,y)\Big|}{\bar{p}(t^* + h,y)} 
					& \leq C\Big(\Err_{\rho}(h,y) + \Err_{\E[N(\mc{S}_{t^*})]}\Big) := \Err_{p}(h,y),
				\end{aligned}
			\end{equation}
			where
			\begin{equation}
				\label{eqn:approximate-joint-density}
				\begin{aligned}
					\bar{p}(t^* + h,y) & := \frac{\bar{\rho}(t^* + h,y)}{\bar{\E}[N(\mc{S}_{t^*})]}.
				\end{aligned}
			\end{equation}
			\item The marginal density of $\wh{Y}$ given $N(\mc{S}_{t^*}) = 1$ satisfies
			\begin{equation}
				\label{eqn:density-height-error}
				\begin{aligned}
					\frac{|p(y) - \bar{p}(y)|}{p(y)} & \leq C\Big(\Err_{\rho}(y) + \Err_{\E[N(\mc{S}_{t^*})]}\Big) := \Err_{p}(y),
				\end{aligned}
			\end{equation}
			where $\Err_{\rho}(y)$ is defined in~\eqref{eqn:height-intensity}, and
			\begin{equation}
				\label{eqn:approximate-marginal-density-height}
				\begin{aligned}
					\bar{p}(y) & := \1(y > u) \cdot \frac{1}{\sqrt{2\pi}}\frac{\exp\big(-\frac{1}{2}(y - \mu_{t^*} - \frac{1}{2}\tr(\bar{H}_{t^*}^{-1}\Lambda_{t^*}))^2\big) }{\Psi\big(u - \mu_{t^*} - \frac{1}{2}\tr(\bar{H}_{t^*}^{-1}\Lambda_{t^*})\big)}.
				\end{aligned}
			\end{equation}
			\item The conditional density of $\hat{t}$ given $N(\mc{S}_{t^*}) = 1$ and $\wh{Y} = y$ satisfies
			\begin{equation}
				\label{eqn:density-location-error}
				\frac{\Big|p(t^* + h|y) - \bar{p}(t^* + h|y)\Big|}{\bar{p}(t^* + h|y)} \leq C\Big(\Err_{\rho}(h,y) + \Err_{\rho}(y)\Big) := \Err_{p}(h|y),
			\end{equation}
			where
			\begin{equation}
				\label{eqn:approximate-conditional-density-location}
				\begin{aligned}
					\bar{p}(t^* + h|y) & := \Big(1 + T_{10}^{\rho}(h) - \frac{1}{2}T_{30}^{\rho}(h)\Big) \cdot \frac{\sqrt{\det(G_{t^*|y})}}{\sqrt{(2\pi)^d}} \cdot \exp\Big(-\frac{1}{2}h'G_{t^*|y} h\Big).
				\end{aligned}
			\end{equation}
		\end{itemize}
	\end{theorem}
	The proof of Theorem~\ref{thm:approximate-joint-distribution} is given in Section~\ref{subsec:pf-approximate-joint-distribution}. Intuitively, in the high-curvature limit $\delta_n \to 0$, we are very likely to observe exactly one peak in $\mc{B}_{t^*}$ (see Proposition~\ref{prop:no-more-than-one-process-peak-per-signal-peak}). Consequently, the post-selection density $p(t,y)$ of this unique peak is exponentially well approximated  by $\rho(t,y)/\E[N(\mc{S}_{t^*})]$, which is the intensity of the expected number of such peaks, normalized to integrate to one. Thus the local expansions of Theorem~\ref{thm:approximate-joint-distribution} follow from Theorem~\ref{thm:approximate-joint-intensity} and Proposition~\ref{prop:expectation-counting-process}. We refer back to Section~\ref{subsec:selection-effect} for more explanation and interpretation of these asymptotic formulae.
	
	\begin{remark}
		Calculations similar to those used to prove Theorem~\ref{thm:approximate-joint-distribution} also lead to a second-order accurate local expansion of the post-selection density of $\hat{t}$ that does not condition on $\wh{Y}$:
		\begin{equation*}
			\bar{p}(t^* + h) := \Big(1 + T_{10}^{\rho}(h) - \frac{1}{2}T_{30}^{\rho}(h)\Big) \cdot \frac{\sqrt{\det(\bar{G}_{t^*})}}{\sqrt{(2\pi)^d}} \cdot \exp\Big(-\frac{1}{2}h'\bar{G}_{t^*} h\Big).
		\end{equation*}
		Comparing this to~\eqref{eqn:approximate-conditional-density-location} shows that the effect of conditioning on $\wh{Y} = y$ is to change the post-selection precision from $\bar{G}_{t^*} = G_{t^*|\bar{u}_{t^*}}$ to $G_{t^*|y} = G_{t^*|\bar{u}_{t^*}} + (y  - \bar{u}_{t^*}) (-\nabla^2 \mu_{t^*})$, which is a first-order change in the precision. However this will not affect our approach to subsequent inference: both $\bar{G}_{t^*}$ and $G_{t^*|y}$ depend on nuisance parameters, so for inference we must plug in an estimate of precision, and for either $\bar{G}_{t^*}$ or $G_{t^*|y}$ we will use the estimated precision $\wh{G} = \wh{H}\Lambda_{\hat{t}}^{-1}\wh{H}$. Henceforth we will deal only with the conditional local expansion $\bar{p}(t^* + h|y)$. 
	\end{remark}

	\subsection{Pivotal quantities for height and location}
	Theorem~\ref{thm:approximate-joint-distribution} implies that under $\Q^{\mc{S}_{t^*}}$, 
	\begin{equation*}
		\bar{\S}_{\mu_{t^*}}(\wh{Y}) := \frac{\Psi\big(\wh{Y} - \mu_{t^*} - \frac{1}{2}\tr(\bar{H}_{t^*}^{-1}\Lambda_{t^*})\big)}{\Psi\big(u - \mu_{t^*} - \frac{1}{2}\tr(\bar{H}_{t^*}^{-1}\Lambda_{t^*})\big)} \overset{d}{\to} \mathrm{Unif}(0,1),
	\end{equation*}
	while additionally conditional on $\wh{Y}  = y$,
	\begin{equation*}
		\bar{W}_{t^*}(y,\hat{t}) := (\hat{t} - t^*)'G_{t^*|y}(\hat{t} - t^*) \overset{d}{\to} \chi_d^2,
	\end{equation*}
	In each case the rate of convergence is up to second-order. However, neither $\bar{\S}_{\mu_{t^*}}$ nor $\bar{W}_{t^*}$ can be directly used for inference as they depend on unknown nuisance parameters. To conduct inference for the height, we will plug in the estimates $\wh{H} = -\nabla^2 Y_{\hat{t}}$ for $\bar{H}_{t^*}$ and $\Lambda_{\hat{t}}$ for $\Lambda_{t^*}$, resulting in the asymptotic \emph{Truncated Gaussian (TG)} pivot
	\begin{equation*}
		\wh{\S}_{\mu_{t^*}}(\wh{Y}) = \wh{\S}_{\mu_{t^*}}(\wh{Y},\hat{t},\wh{H}) := \frac{\Psi(\wh{Y} - \mu_{t^*} - \frac{1}{2}\tr(\wh{H}^{-1}\Lambda_{\hat{t}}))}{\Psi(u - \mu_{t^*} - \frac{1}{2}\tr(\wh{H}^{-1}\Lambda_{\hat{t}}))}.
	\end{equation*}
	To conduct inference for the location we will plug in $\wh{H}$ for both $\bar{H}_{t^*}$ and $-\nabla^2\mu_{t^*}$, resulting in the asymptotic \emph{Wald} pivot
	\begin{equation*}
		\wh{W}_{t^*}(\hat{t}) = {W}_{t^*}(\hat{t},\wh{H}) := (\hat{t} - t^*)'\wh{G}(\hat{t} - t^*), \quad \wh{G} := \wh{H} \Lambda_{\hat{t}}^{-1} \wh{H}.
	\end{equation*}
	\begin{theorem}
		\label{thm:pivot}
		Fix $\alpha \in (0,1)$. Under the assumptions of Theorem~\ref{thm:approximate-joint-distribution}, for all $n \in \mathbb{N}$ sufficiently large: 
		\begin{equation}
			\label{eqn:height-pivot}
			\Big|\Q^{\mc{S}_{t^*}}\big(\wh{\S}_{\mu_{t^*}}(\wh{Y}) \leq \alpha\big) - \alpha\Big| \leq C\big(|\bar{u}_{t^*} - \mu_{t^*}| + \Delta_n + 1\big)\varepsilon_n \delta_n := \Err_{\S}.
		\end{equation}
		If additionally $(\bar{u}_{t^*} - \mu_{t^*})\delta_n \to 0$, then
		\begin{equation}
			\label{eqn:pivot-location}
			\begin{aligned}
			& \Big|\Q^{\mc{S}_{t^*}}\Big(\wh{W}_{t^*}(\hat{t}) \leq q_{\chi_d^2(\alpha)}|\wh{Y} = y\Big)  - \alpha + A\cdot\frac{(y - \bar{u}_{t^*}) + (\bar{u}_{t^*} - \mu_{t^*})}{2}\tr(\bar{H}_{t^*}^{-1}\Lambda_{t^*})\Big| \\
			& \leq C \Big(\delta_n^2(1 + |y - \bar{u}_{t^*}| ^2+ |\bar{u}_{t^*} - \mu_{t^*}|^2)\Big) := \Err_W(y),
			\end{aligned}
		\end{equation}
		where $q_{\chi_d^2}(\alpha)$ is the $\alpha$th quantile of $\chi_d^2$ distribution, and $A := \alpha - \frac{1}{d}\E[\|Z\|^2 \cdot \1\{\|Z\| \leq q_{\chi_d^2}(\alpha)\}], Z \sim N(0,1)$ is a positive constant.
	\end{theorem}
	
	The proof of Theorem~\ref{thm:pivot} is given in Sections~\ref{subsec:pf-pivot-hessian}-\ref{subsec:pf-approximate-density-studentized-peak}. Theorem~\ref{thm:pivot} implies that $\wh{\S}_{\mu_{t^*}}(\wh{Y})$ is up to nearly \emph{second-order pivotal} for $\mu_{t^*}$. On the other hand, the Wald-pivot $\wh{W}_{t^*}(\hat{t})$ is at most first-order pivotal for $t^*$, due to the term
	\begin{equation*}
		A\cdot\frac{(y - \bar{u}_{t^*}) + (\bar{u}_{t^*} - \mu_{t^*})}{2}\tr(\bar{H}_{t^*}^{-1}\Lambda_{t^*}) = \wt{O}(|\bar{u} - \mu_{t^*}|\delta_n).
	\end{equation*} 
	The failure of $\wh{W}_{t^*}(\hat{t})$ to be second-order pivotal is caused by plugging in a biased estimate $\wh{H}$ for the nuisance parameter $-\nabla^2 \mu_{t^*}$ in estimating the precision $G_{t^*|y}$. 
	
	\subsection{Conditional coverage}
	\label{subsec:coverage}
	Confidence intervals for peak height, and confidence ellipsoids for peak location, are constructed at each discovery $\hat{t} \in \wh{T}_u$ by inverting tests based on the TG and Wald pivots: 
	\begin{equation}
		\label{eqn:confidence-regions}
		I_{\hat{t}} := \Big\{\mu \in \R: \frac{\alpha}{2} \leq \wh{\S}_{\mu}(\wh{Y}) \leq 1 - \frac{\alpha}{2}\Big\}, \quad C_{\hat{t}} := \Big\{t \in \Rd: \wh{W}_{t}(\hat{t}) \leq q_{\chi_d^2}(1 - \alpha)\Big\}.
	\end{equation}
	These confidence regions have asymptotic $(1 - \alpha)$ coverage under the conditions of Theorem~\ref{thm:pivot}.
	\begin{corollary}
		\label{cor:conditional-coverage}
		Under the conditions of Theorem~\ref{thm:pivot},
		\begin{equation*}
			\begin{aligned}
			\big|\Q^{\mc{S}_{t^*}}(\mu_{t^*} \in I_{\hat{t}}) - (1 - \alpha)\big| & \leq \Err_{\S}, \quad \textrm{and} \quad \\ \Big|\Q^{\mc{S}_{t^*}}(t^* \in C_{\hat{t}}) - (1 - \alpha)\Big| & \leq A\cdot \frac{(\E^{\mc{S}_{t^*}}[\wh{Y}]  - \mu_{t^*})}{2}\tr(\bar{H}_{t^*}^{-1}\Lambda_{t^*}) + C\delta_n^2\big(1 + (\bar{u}_{t^*} - \mu_{t^*})^2\big),
			\end{aligned}
		\end{equation*}
			where $\E^{\mc{S}_{t^*}}[\wh{Y}]$ denotes the expectation of $\wh{Y}$ under $\Q^{\mc{S}_{t^*}}$. 
	\end{corollary}
	Corollary~\ref{cor:conditional-coverage} implies that  $I_{\hat{t}}$ has up to nearly second-order accurate coverage, as $\Err_{\S} = \wt{O}(\delta_n^2 + (\bar{u}_{t^*} - \mu_{t^*})\delta_n^2)$.  Under negligible selection pressure ($\mu_{t^*} \gg u$) coverage of the location $t^*$ is also nearly second-order accurate: in this case $\E^{\mc{S}_{t^*}}[\wh{Y}] - \mu_{t^*} = O(\delta_n)$ and Corollary~\ref{cor:conditional-coverage} implies $\Q^{\mc{S}_{t^*}}(t^* \in C_{\hat{t}}) - (1 - \alpha) = \wt{O}(\delta_n^2)$. However if selection pressure is not negligible, then coverage of the location is at best nearly first-order accurate; and under very strong selection pressure, $C_{\hat{t}}$ does not achieve nominal coverage, even asymptotically.
	
	Clearly, an important and challenging part of conducting valid inference for peaks involves dealing with nuisance parameters, particularly $-\nabla^2\mu_{t^*}$. In more traditional post-selection inference problems, a common strategy~\citep{fithian2014optimal,lee2016exact} is to condition on a sufficient statistic for the nuisance parameters, so that the resulting conditional distributions do not depend on the nuisance. This approach is effective when the data generating process belongs to a well-specified exponential family. However in our setting this strategy fails: the minimal sufficient statistic for the nuisance $-\nabla^2\mu_{t^*}$ is the entire field $Y$ itself, and any inference conducted conditional on $Y$ is trivial. Instead, in Section~\ref{sec:randomized-peak-inference} we pursue an approach based on auxiliary randomization, which leaves behind enough information after selection to accurately estimate $-\nabla^2 \mu_{t^*}$.

	\subsection{Marginal coverage}
	We measure marginal (mis)coverage of the overall method using per-comparison miscoverage rate (PCMR), defined in~\eqref{eqn:miscoverage} and copied here for convenience: recalling that $t^*(t) := \argmin_{t^* \in T^*} \|t^* - t\|$ denotes the true peak closest to a point $t$,
	\begin{equation*}
		{\rm PCMR}_{T^*}(u,v) := \frac{\E[N(t \in \wh{T}_{u}: t^*(t) \not\in C_t)]}{\E[N(\wh{T}_v)]}, \quad\textrm{and}\quad {\rm PCMR}_{\mu_{T^*}}(u,v) := \frac{\E[N(t \in \wh{T}_{u}: \mu_{t^*(t)} \not\in I_t)]}{\E[N(\wh{T}_v)]}.
	\end{equation*} 
	In other words, PCMR measures the expected number of confidence regions that fail to cover the location or height of the nearest peak, compared to the expected number of peaks at which inference is conducted.
	
	The conditional coverage guarantees of Corollary~\ref{cor:conditional-coverage} hold for any threshold $u$ satisfying the conditions of Theorem~\ref{thm:approximate-joint-distribution}, and do not specifically require that $u = u_{\TG}$ be chosen to calibrate the TG test for significance. Taking $u = u_{\TG}$ results in asymptotic marginal coverage: Theorem~\ref{thm:epsilon-fpr} implies that at most $100\alpha\%$ of selected peaks are falsely declared significant without being unique $\varepsilon_n$-consistent estimates of some $t^* \in T^*$, while Corollary~\ref{cor:conditional-coverage} implies that at most $100\alpha\%$-percent of $\varepsilon_n$-consistent estimate have corresponding confidence regions that fail to cover the truth. Together, these imply asymptotic control of PCMR at level
	\begin{equation*}
		\wb{{\rm PCMR}}(\alpha,v_n) := \frac{\sum_{t^* \in T^*}\P(N(\mc{S}_{t^*}) = 1) \cdot \alpha + \alpha \cdot \E[N_{v_n}(\mc{T}_0)]}{\sum_{t^* \in T^*}\P(N_{v_n}(\mc{B}_{t^*}) = 1) + \E[N_{v_n}(\mc{T}_0)]} \leq \alpha.
	\end{equation*}
	
	\begin{theorem}
		\label{thm:miscoverage}
		Fix $\alpha \in (0,1)$, and suppose $v_n \to \infty$. Under Assumptions~\ref{asmp:signal-holder}-\ref{asmp:interior-global-max}, for all $n \in \mathbb{N}$ sufficiently large:
		\begin{equation}
			\label{eqn:miscoverage-height}
			{\rm PCMR}_{\mu_{T^*}} \leq \wb{{\rm PCMR}}(\alpha,v_n)+ C\Big(\frac{1}{v_n^2} + \big(\sup_{t^* \in T^*}(\bar{u}_{t^*} - \mu_{t^*}) + \Delta_n + 1\big)\varepsilon_n \delta_n\Big).
		\end{equation} 
		If additionally $\sup_{t^* \in T^*}(\bar{u}_{t^*} - \mu_{t^*})\delta_n \to 0$, then
		\begin{equation}
			\label{eqn:miscoverage-location}
			{\rm PCMR}_{T^*}(\alpha,v_n) \leq \wb{{\rm PCMR}}(\alpha,v_n) + C\Big(\frac{1}{v_n^2} + \delta_n^2 + A \cdot \sup_{t^* \in T^*} (\E^{\mc{S}_{t^*}}[\wh{Y}]  - \mu_{t^*}) \delta_n + \sup_{t^* \in T^*} (\bar{u}_{t^*} - \mu_{t^*})^2\delta_n^2\Big).
		\end{equation}
	\end{theorem}
	For the proof of Theorem~\ref{thm:miscoverage} see Section~\ref{subsec:pf-miscoverage}. When there is asymptotically negligible selection pressure -- i.e. $\min_{t^* \in T^*} \mu_{t^*} - u_{\TG}(\alpha,v_n) \to \infty$ --  then $\inf_{t^* \in T^*}\P(N(\mc{S}_{t^*}) = 1)/\P(N_{v}(\mc{B}_{t^*}) = 1) \to 1$ and $\wb{{\rm PCMR}}(\alpha,v_n)\to \alpha$. This corresponds to a limiting regime where there is exactly one selected peak $\hat{t} \in \wh{T}_v$ that consistently estimates each $t^* \in T^*$, and that is declared significant with probability tending to one; in this case limiting miscoverage will be exactly the nominal level $\alpha$. Otherwise $\wb{{\rm PCMR}}(\alpha,v_n) < \alpha$ and the overall procedure will be conservative, since some consistent selections $\hat{t} \in \wh{T}_v$ may not be declared significant by the TG test.
	
	\section{Peak inference with a randomized field}
	\label{sec:randomized-peak-inference}
	
	There are at least two issues with our non-randomized method for selective peak inference. First, the confidence intervals for the height are based on a Truncated Gaussian distribution, and may therefore be quite wide, particularly if there is strong selection pressure~\citep{kivaranovic2021length}. Second, the confidence regions for location may not have nominal coverage, even asymptotically. Both problems become more dramatic as the significance threshold $u$ increases. One way of understanding this is through the idea of left over (Fisher) information~\citep{fithian2014optimal}.  Intuitively, as $u$ increases and selection pressure grows, there is less information left over for $\mu_{t^*}$, and consequently less information for the Hessian $\nabla^2 \mu_{t^*}$ as well. Less information for $\mu_{t^*}$ leads to wider intervals for the height; less information for $\nabla^2\mu_{t^*}$ leads to less accurate estimates of nuisance parameters, a less pivotal Wald statistic, and ultimately confidence regions with less accurate coverage.
	
	A conceptually simple way of achieving valid inferences after selection is \emph{data splitting}~\citep{cox1975note,wasserman2009high,kriegeskorte2009circular,kriegeskorte2010everything}, which always sets aside some information for inference. However, in our generic problem setup there is only replicate which cannot further be split. A separate issue is that in post-model-selection inference it has been observed that data splitting can be ineffecient~\citep{fithian2014optimal}. We propose an alternative method for randomized peak inference that addresses both issues. Roughly speaking the method works as follows: first, synthetic randomization is injected in a way that is designed to mimic the effects of data splitting, similar to~\citet{tian2018selective,rasines2023splitting,leiner2025data}. The peak detection portion of Algorithm~\ref{alg:post-selective-inference-after-peak-detection-via-TG-test} -- i.e. pre-thresholding followed by significance testing using the TG test -- is then applied to this randomized field. Finally, post-selection inference for height and location is conducted using the resulting discoveries, in a way that is designed to use all of the information left over for inference, along the lines of \emph{data carving}~\citep{fithian2014optimal,tian2018selective}. 
	
	As in the non-randomized setting, our method for peak inference after randomized selection will be based on asymptotic expansions of the density of (randomized) peaks. We derive these expansions heuristically in Section~\ref{sec:randomized-peak-inference-analysis}, but do not provide explicit upper bounds on the relative error, nor theoretical guarantees on coverage. Instead we confirm experimentally in Section~\ref{sec:experiments} that the method achieves close to nominal coverage, both conditionally and marginally. 
	
	\subsection{Randomized peak detection}
	Our randomized method for peak detection works as follows. First we sample $\omega \sim N(0, C)$, independently of $Y$. We then use $\omega$ to ``split'' the field $Y$ into two parts,
	$$
	Y_t^{\sel} := Y_t + \sqrt{\gamma} \cdot \omega_t, \quad  Y_{t}^{\inf} := Y_t - \frac{1}{\sqrt{\gamma}}\cdot\omega_t,
	$$ 
	where $\gamma \in (0,\infty)$ is a user-determined parameter controlling the degree of randomization. Notice that by construction, the fields $Y_t^{\sel}$ and $Y_t^{\inf}$ are independent. We produce a set of candidate peaks for inference by first pre-thresholding the peaks of $Y_t^{\sel}$, 
	$$
	\wc{T}_v := \{t \in \mc{T}: Y_{t}^{\sel} > v, \nabla Y_t^{\sel} = 0, \nabla^2 Y_t^{\sel} \prec 0\},
	$$
	and then applying the TG test for peak detection to thresholded peaks $\check{t} \in \wc{T}_v$, resulting in discoveries 
	$$
	\wc{T}_{u} = \{t \in \mc{T}: Y_{t}^{\sel} > u, \nabla Y_t^{\sel} = 0, \nabla^2 Y_t^{\sel} \prec 0\}.
	$$
	Just as in the non-randomized setting, the randomized peaks detected by the TG test tend to consistently discover true peaks. Concretely, the results of Section~\ref{sec:peak-estimation} imply that applying the TG test to each $\check{t} \in \wc{T}_v$, with significance threshold $u_{\TG}(\alpha,v) \cdot \sigma_{\gamma}$ asymptotically controls the $\varepsilon_n^{\gamma}$-PCER, where $\sigma_{\gamma} := \sqrt{1 + \gamma}$, $\varepsilon_n^{\gamma} := \delta_n^{\gamma} \sqrt{6\sigma_1^2 \log(1/\delta_{n}^{\gamma})}$ and $\delta_n^{\gamma} = \delta_n\cdot \sigma_{\gamma}$. 
	
	Randomization guarantees that at least $1/\Var[Y_t^{\inf}] = \gamma/(1 + \gamma) := \pi$ fraction of the total information is left over for inference, regardless of the degree of selection pressure. But there is a tradeoff: more randomization increases the information left over for inference, but can also result in less power for discovering true peaks (at the same nominal level $\alpha$).
	
	\subsection{Post-selection peak inference via data carving}
	
	Our high-level goal for inference after randomized peak detection remains the same as in the non-randomized setting: we would like to construct confidence regions for the location and height of some $t^* \in T^*$ that is nearby a discovery $\check{t} \in \wc{T}_{u}$, given that such a true peak happens to exist. One way of achieving this, in the spirit of data splitting, is to only use $Y^{\inf}$ to form the regions. Since $Y^{\inf}$ and $Y^{\sel}$ are independent, the distribution of peaks of $Y^{\inf}$ is unaffected by selection, and (second-order accurate) confidence regions can be derived from the marginal distribution of peaks of $Y^{\inf}$. We give a concrete method along these lines in Section~\ref{subsec:split}, which we term selective peak inference via \emph{data splitting} (despite the fact that in our setup there is only one replicate $Y$.)
	
	We now propose an alternative that is more similar in spirit to \emph{data carving}~\citep{fithian2014optimal,tian2018selective}. As in data splitting, in data carving only $Y^{\sel}$ is used for peak detection. However, data carving forms confidence regions using the full data $Y$, incorporating information from both $Y^{\sel}$ and $Y^{\inf}$. Intuitively, by using all the information left after selection for conditional inference, data carving leads to more powerful inferences and tighter confidence regions than data splitting, while maintaining the same asympotic coverage guarantees. We demonstrate this experimentally in Section~\ref{sec:experiments}.
	
	\paragraph{Conditional distribution after randomized selection.}
	To construct confidence regions we examine the post-randomized selection distribution of a full-data peak $\hat{t} \in \wh{T}$ about a true peak $t^* \in T^*$. Roughly speaking, this is the conditional distribution given that there is a unique randomized peak $\check{t} \in \wc{T}$ that $\varepsilon_n^{\gamma}$-consistently discovers $t^*$.\footnote{More precisely, the distribution conditions on the event that (i) there is a unique randomized peak $\check{t} \in \wc{T}$ that $\varepsilon_n^{\gamma}$-consistently discovers $t^*$; (ii) there is a unique full-data peak $\hat{t} \in \wh{T}$ that $\varepsilon_n$-consistently discoveres $t^*$; (iii) $Y_{\check{t}}^{\sel} \in \bar{u}_{t^*}^{\gamma} \pm \Delta_n$, (iv) $\wh{Y} \in \bar{u}_{t^*}^{\gamma} \pm \Delta_n$ where $\bar{u}_{t^*}^{\gamma} = (1 - \pi) \bar{u}_{t^*} + \pi \mu_{t^*}$ . However (ii) - (iv) occur with high probability as $\delta_n \to 0$ and thus do not have a large effect on the conditional distribution. See Section~\ref{subsec:carve-densities} for more details.} Heuristic calculations carried out in Section~\ref{sec:randomized-peak-inference-analysis} suggest that a second-order accurate approximation to the post-randomized selection density of $\wh{Y}$ is
	\begin{equation}
		\label{eqn:density-height-carve}
		\bar{p}^{\carve}(y):\propto\Psi\Big(\frac{u - y - \frac{\gamma}{2}\tr(\bar{H}_{t^*}^{-1}\Lambda_{t^*})}{\sqrt{\gamma}}\Big) \cdot \frac{1}{\sqrt{2\pi}}\exp\Big(-\frac{1}{2}(y - \mu_{t^*} - \frac{1}{2}\tr(\bar{H}_{t^*}^{-1}\Lambda_{t^*}))^2\Big),
	\end{equation}
	while a second-order accurate approximaton to density of $\hat{t}|\wh{Y} = y$ is
	\begin{equation}
		\label{eqn:density-location-carve}
		\bar{p}^{\carve}(t^* + h|y) := (1 + T_{10}^{\carve}(h) + T_{30}^{\carve}(h)) \cdot \frac{\sqrt{\det(G_{t^*|y})}}{\sqrt{2\pi}} \exp\Big(-\frac{1}{2}h'G_{t^*|y}h\Big).
	\end{equation}
	(The first-order terms $T_{10}^{\carve}(h),T_{30}^{\carve}(h)$ are defined in Section~\ref{sec:randomized-peak-inference-analysis}.) We call these calculations heuristic as we do not provide formal upper bounds on the error; we expect that they could be made fully rigorous using similar techniques to those used to prove the results of Sections~\ref{sec:peak-intensity-local-expansion}-\ref{sec:peak-distribution}, but leave this to future work.
	
	Examining these asymptotic densities reveals the effect of randomization on the post-selection distribution of height and location.  Comparing~\eqref{eqn:density-height-carve} to~\eqref{eqn:approximate-marginal-density-height} shows that the effect of randomization on the post-selection density of $\wh{Y}$ is to replace the indicator $\1(y > u)$ by the ``soft-truncation'' Gaussian survival  function.\footnote{Here we are borrowing terminology from~\citet{panigrahi2023approximate}, who refer to likelihoods with a similar functional form to~\eqref{eqn:density-height-carve} as ``soft-truncated'' likelihoods.} This is a familiar consequence of randomization in post-selection inference problems~\citep{tian2018selective,panigrahi2023approximate}. We call the asymptotic post-randomized selection distribution of $\wh{Y}$ a \emph{soft Truncated Gaussian} (soft TG) distribution. To a leading order, the soft TG asymptotic limit matches the distribution of $Y_{t^*}|Y_{t^*}^{\sel} > u$, as might be expected. Our analysis suggests that once again second-order accuracy is achieved through applying a first-order mean-shift to this reference distribution.
	
	On the other hand~\eqref{eqn:density-location-carve} suggests that the post-randomized selection distribution of $\hat{t}|\wh{Y} = y$ has the same conditional Goldilocks precision matrix as under non-randomized selection. However under randomized selection $\wh{Y}$ has a different deterministic limit -- $\bar{u}_{t^*}^{\gamma} = \pi \bar{u}_{t^*} + (1 - \pi)\mu_{t^*}$ rather than $\bar{u}_{t^*}$ -- and thus under randomized selection the conditional Goldilocks matrix is closer to 
	\begin{equation*}
		G_{t^*|\bar{u}_{t^*}^{\gamma}} = (-\nabla^2\mu_{t^*}) \Lambda_{t^*}(-\nabla^2\mu_{t^*})+ (1 - \pi) (\bar{u}_{t^*} - \mu_{t^*}) \Lambda_{t^*},
	\end{equation*}
	than to $\bar{G}_{t^*}$. Since $(-\nabla^2\mu_{t^*}) \Lambda_{t^*}(-\nabla^2\mu_{t^*}) \prec G_{t^*|\bar{u}_{t^*}^{\gamma}} \prec \bar{G}_{t^*}$, we see that randomized selection increases precision for the location, but by a smaller amount than non-randomized selection. On the other hand, after randomized selection we can use $Y^{\inf}$ to construct an estimate of precision that is asymptotically unbiased even under strong selection pressure. 
	
	\paragraph{Pivots and confidence regions.}
	To construct approximate pivots, we use the limiting distributions suggested by~\eqref{eqn:density-height-carve} and~\eqref{eqn:density-location-carve}, and then plug in estimates for nuisance parameters. For the height, this leads to the \emph{soft TG} pivot:
	\begin{equation}
		\label{eqn:height-pivot-carve}
		\wh{\STG}_{\mu_{t^*}}(y) := \frac{\int_{-\infty}^{y} \Psi\Big(\frac{u - z - \frac{\gamma}{2}\tr(\wh{H}^{-1}\Lambda_{\hat{t}})}{\sqrt{\gamma}}\Big) \cdot \frac{1}{\sqrt{2\pi}}\exp\Big(-\frac{1}{2}(z - \mu_{t^*} - \frac{1}{2}\tr(\wh{H}^{-1}\Lambda_{\hat{t}}))^2 \Big) \,dz}{\int_{-\infty}^{\infty} \Psi\Big(\frac{u - z - \frac{\gamma}{2}\tr(\wh{H}^{-1}\Lambda_{\hat{t}})}{\sqrt{\gamma}}\Big) \cdot \frac{1}{\sqrt{2\pi}}\exp\Big(-\frac{1}{2}(z - \mu_{t^*} - \frac{1}{2}\tr(\wh{H}^{-1}\Lambda_{\hat{t}}))^2 \Big) \,dz}.
	\end{equation}
	(This is not available in closed-form but can be calculated numerically.) For the location, we again use a \emph{Wald}-type pivot; but now plug in the estimate $\wh{H}^{\inf} := -\nabla^2 Y_{\hat{t}}^{\inf}$ for the nuisance parameter $-\nabla^2\mu_{t^*}$:
	\begin{equation}
		\label{eqn:location-pivot-carve}
		\wh{W}_{t^*}^{\carve}(t) := (t - t^*)'\wh{H}\Lambda_{\hat{t}}^{-1}\wh{H}^{\inf}(t - t^*).
	\end{equation}
	As in the non-randomized setting, we can construct confidence regions by inverting tests based on the soft TG and Wald pivots, resulting in 
	\begin{equation}
		\label{eqn:confidence-regions-carve}
		\begin{aligned}
			I_{\hat{t}}^{{\rm carve}} & := \Big\{\mu: \frac{\alpha}{2} \leq \wh{\STG}_{\mu}(\wh{Y}) \leq 1 - \frac{\alpha}{2}\Big\}, \quad\textrm{and}\quad C_{\hat{t}}^{{\rm carve}} := \Big\{t: \wh{W}_{t}^{\carve}(\hat{t}) \leq q_{\chi_d^2}(1 - \alpha) \Big\}. \\
		\end{aligned}
	\end{equation}
	 For convenience, we summarize this method for peak inference via data carving in Algorithm~\ref{alg:carve}. 
	 
	We expect that randomization should improve post-selection inference for both height and location, but in different ways. For the height, we expect that $I_{\hat{t}}^{{\rm carve}}$ should have asymptotic up to second-order accurate coverage, which is the same order of accuracy achieved by the non-randomized $I_{\hat{t}}$. However we expect the post randomized-selection intervals to be significantly shorter under strong selection pressure, as is the case in more traditional post-selection inference problems~\citep{tian2018selective,panigrahi2023approximate,rasines2023splitting}. For the location, the distribution of $\wh{H}^{\inf}$ is only weakly affected by selection -- since $\wh{H}^{\inf}$ depends on $Y^{\sel}$ only through the location $\hat{t}$ at which it is evaluated -- and in particular $\wh{H}^{\inf}$ should be an asymptotically unbiased estimate of $-\nabla^2\mu_{t^*}$, with relative error converging to $0$ in probability.  As a result we expect that $C_{\hat{t}}^{\carve}$ should have asymptotically up to second-order accurate coverage, even under non-negligible selection pressure, which improves on the coverage of the non-randomized $C_{\hat{t}}$. In Section~\ref{sec:experiments} we verify both of these improvements experimentally.
	
	\begin{algorithm}[t]
		\caption{Selective peak inference via data carving}
		\label{alg:carve}
		\begin{algorithmic}[1]
			\REQUIRE Field $Y$, gradient covariance $\Cov[\nabla \epsilon_t]$, pre-threshold $v \in \mathbb{R}$, significance level $\alpha \in (0,1)$, randomization level $\gamma > 0$.
			\STATE Select peaks $\wc{T}_v$ by pre-thresholding $Y^{\sel}$.
			\STATE At each peak $t \in \wc{T}_v$, test the null hypothesis $\mu_{t} = 0$ by applying the TG test to $Y^{\sel}$ with significance threshold $\sigma_{\gamma} \cdot u_{\TG}(\alpha,v)$, making discoveries $\wc{T}_{u_{\TG}(\alpha,v)}$.
			\STATE For each discovery $\check{t} \in \wc{T}_{u_{\TG}(\alpha,v)}$, let $\hat{t} := \argmin_{t \in \wh{T}} \|t - \check{t}\|$ be the nearest peak in $Y$. Compute a confidence region $I_{\hat{t}}^{\carve}$ for the height, and a confidence ellipsoid $C_{\hat{t}}^{\carve}$ for the location, using~\eqref{eqn:height-pivot-carve}-\eqref{eqn:confidence-regions-carve}.
		\end{algorithmic}
	\end{algorithm}
	
	\section{Experiments}
	\label{sec:experiments}
	
	We conduct some experiments on simulated data to first empirically evaluate the finite-curvature accuracy of our asymptotic theory, and then to compare the performance of our various proposals for selective peak inference.
	
	\subsection{Validation of asymptotic theory}
	\label{subsec:experiment-1}
	To evaluate the finite-curvature accuracy of the asymptotic theory of Section~\ref{sec:peak-distribution}, we draw $50000$ independent replicates of a two-dimensional random field $Y \sim N(\mu,C)$ defined over domain $\mc{T} = [-1,1]^2$, with signal and covariance kernel 
	\begin{equation*}
		\mu_t = \mu_{0} \cdot K(0,t), \quad K(s,t) = \exp\Big(-\frac{\|s - t\|^2}{2 \cdot 0.15^2}\Big),
	\end{equation*}
	so that $T^* = \{0\}$ and $\mu_{t^*} = \mu_0$. We retain only those replicates where there is a single peak $\hat{t} \in \wh{T}_u$ within distance $\varepsilon_n$ of $t^* = 0$. This is repeated for $\mu_0 \in \{3,4,5,\cdots,11\}$ to study the effect of stronger signal/higher curvature,  and thresholds $u \in \{\mu_0 - 2,\mu_0,\mu_0 + 2\}$, to compare results under weak, moderate, and strong selection pressure. The results are displayed in Figure~\ref{fig:experiment-1-2d}.
	
	For the location, we compare how often the quantities
	\begin{align*}
		\chi_d^2\big((\hat{t} - t^*)'\bar{G}_{t^*}(\hat{t} - t^*)\big), \quad  \chi_d^2\big((\hat{t} - t^*)'\nabla^2\mu_{t^*}\Lambda_{t^*}^{-1}\nabla^2\mu_{t^*}(\hat{t} - t^*)\big), \quad \textrm{and} \quad \chi_d^2\big((\hat{t} - t^*)'\bar{H}_{t^*}\Lambda_{t^*}^{-1}\bar{H}_{t^*}(\hat{t} - t^*)\big)
	\end{align*}
	fall below the particular choice $\alpha = 0.2$, to evaluate how close empirical precision is to the asymptotically correct $\bar{G}_{t^*}$, as opposed to the sandwich precision with either the marginal Hessian $-\nabla^2\mu_{t^*}$ or the conditional Hessian $\bar{H}_{t^*}$. As predicted by the asymptotic theory, the precision is increasingly close to $\bar{G}_{t^*}$ as $\mu_0$ is increased, across different choices of $u$. In contrast, neither of the sandwich precision matrices are correct under moderate or strong selection pressure. 
	
	For the height, we compare the distribution of
	\begin{align*}
		\bar{\S}_{\mu_{t^*}}(\wh{Y}) = \frac{\Psi(\wh{Y} - \mu_{t^*} - \frac{1}{2}\tr(\bar{H}_{t^*}^{-1}\Lambda_{t^*}))}{\Psi(u  - \mu_{t^*} - \frac{1}{2}\tr(\bar{H}_{t^*}^{-1}\Lambda_{t^*}))}, \quad \frac{\Psi(\wh{Y} - \mu_{t^*})}{\Psi(u  - \mu_{t^*})}, \quad \textrm{and} \quad \Psi(\wh{Y} - \mu_{t^*}).
	\end{align*}
  	The asymptotically correct $\bar{\S}_{\mu_{t^*}}(\wh{Y})$ is well calibrated for all $\mu_0$, while the naive choice $\Psi(\wh{Y} - \mu_{t^*})$, which does not account for selection, is not calibrated under moderate or strong selection pressure. Additionally, at lower signal strength and under weak and moderate selection pressure, the first-order mean-shift correction $\frac{1}{2}\tr(\bar{H}_{t^*}^{-1}\Lambda_{t^*})$ noticeably improves calibration.
	
	Additional experiments presented in Section~\ref{sec:additional-experiments} demonstrate that these conclusions are robust to the choice of $\alpha$, and are qualitatively similar in a one-dimensional setting.
	
	\begin{figure}[htbp]
		\centering
		\begin{subfigure}[t]{0.32\linewidth}
			\centering
			\includegraphics[width=\linewidth]{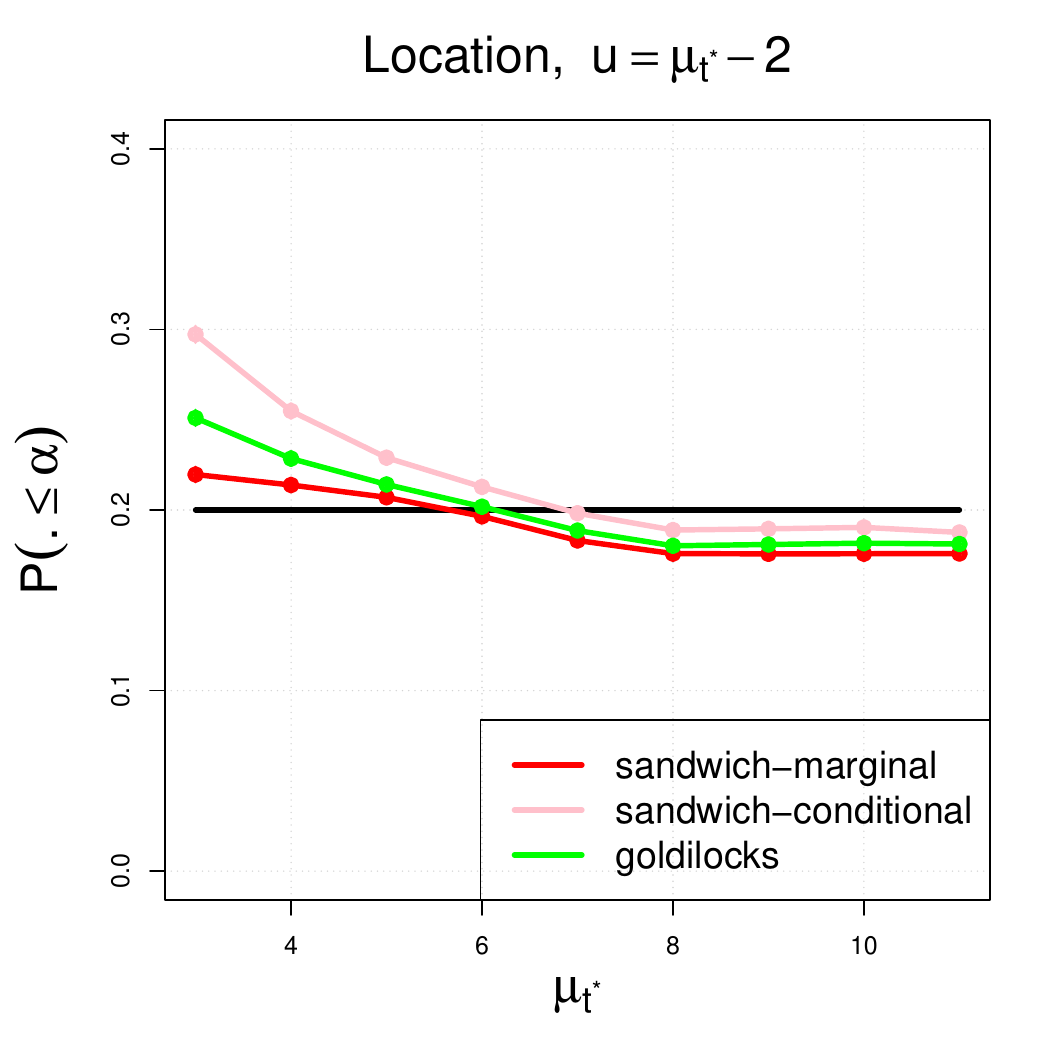}
			\label{fig:experiment-1-2d-4}
		\end{subfigure}\hfill
		\begin{subfigure}[t]{0.32\linewidth}
			\centering
			\includegraphics[width=\linewidth]{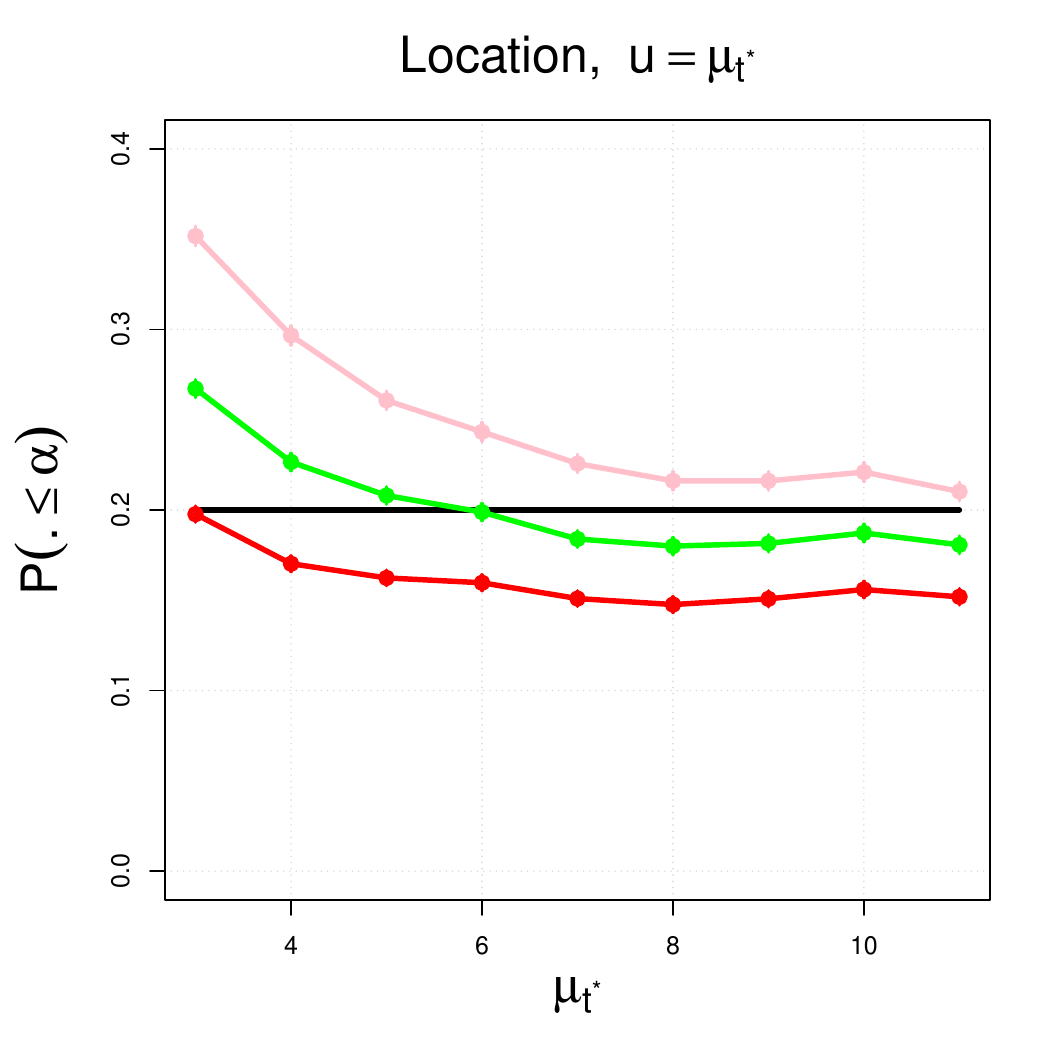}
			\label{fig:experiment-1-2d-5}
		\end{subfigure}\hfill
		\begin{subfigure}[t]{0.32\linewidth}
			\centering
			\includegraphics[width=\linewidth]{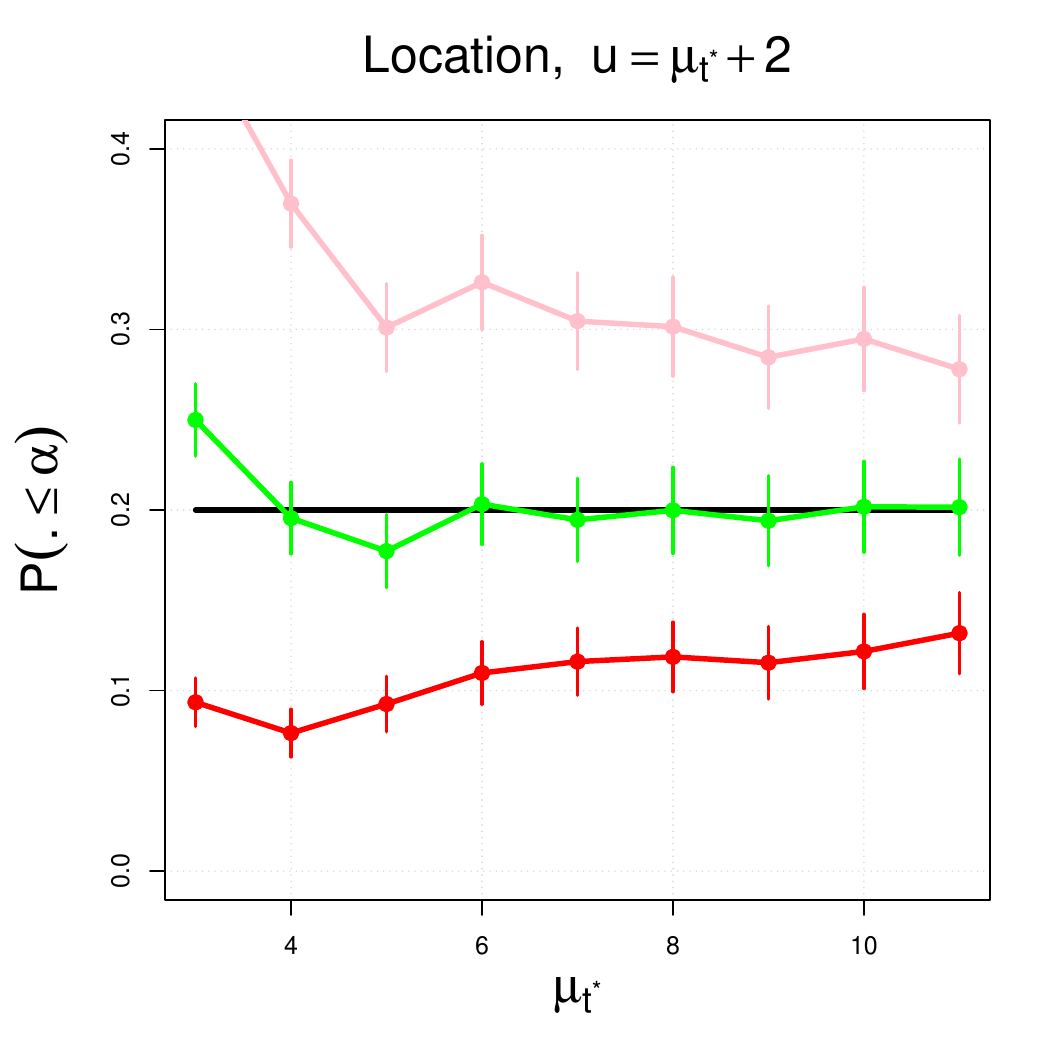}
			\label{fig:experiment-1-2d-6}
		\end{subfigure}\hfill
		\begin{subfigure}[t]{0.32\linewidth}
			\centering
			\includegraphics[width=\linewidth]{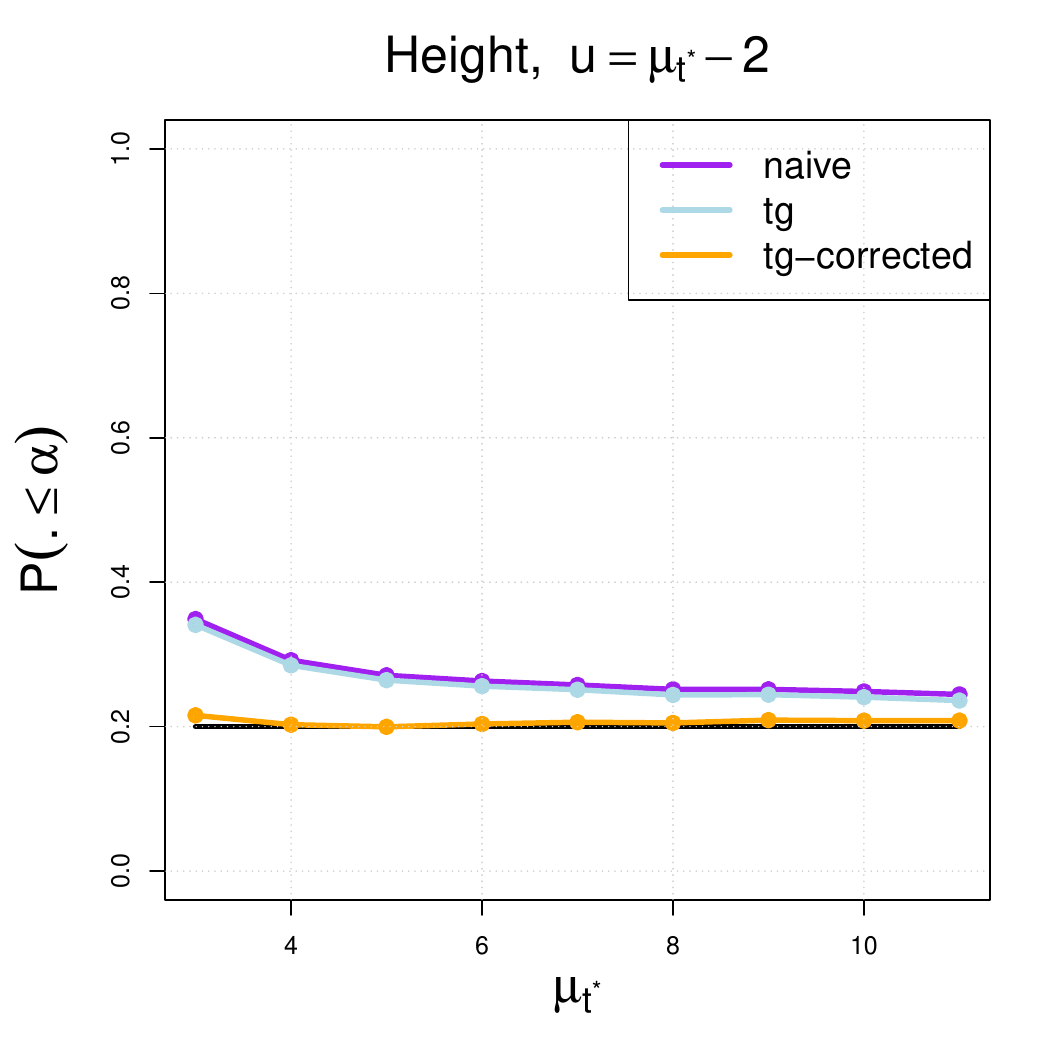}
			\label{fig:experiment-1-2d-1}
		\end{subfigure}\hfill
		\begin{subfigure}[t]{0.32\linewidth}
			\centering
			\includegraphics[width=\linewidth]{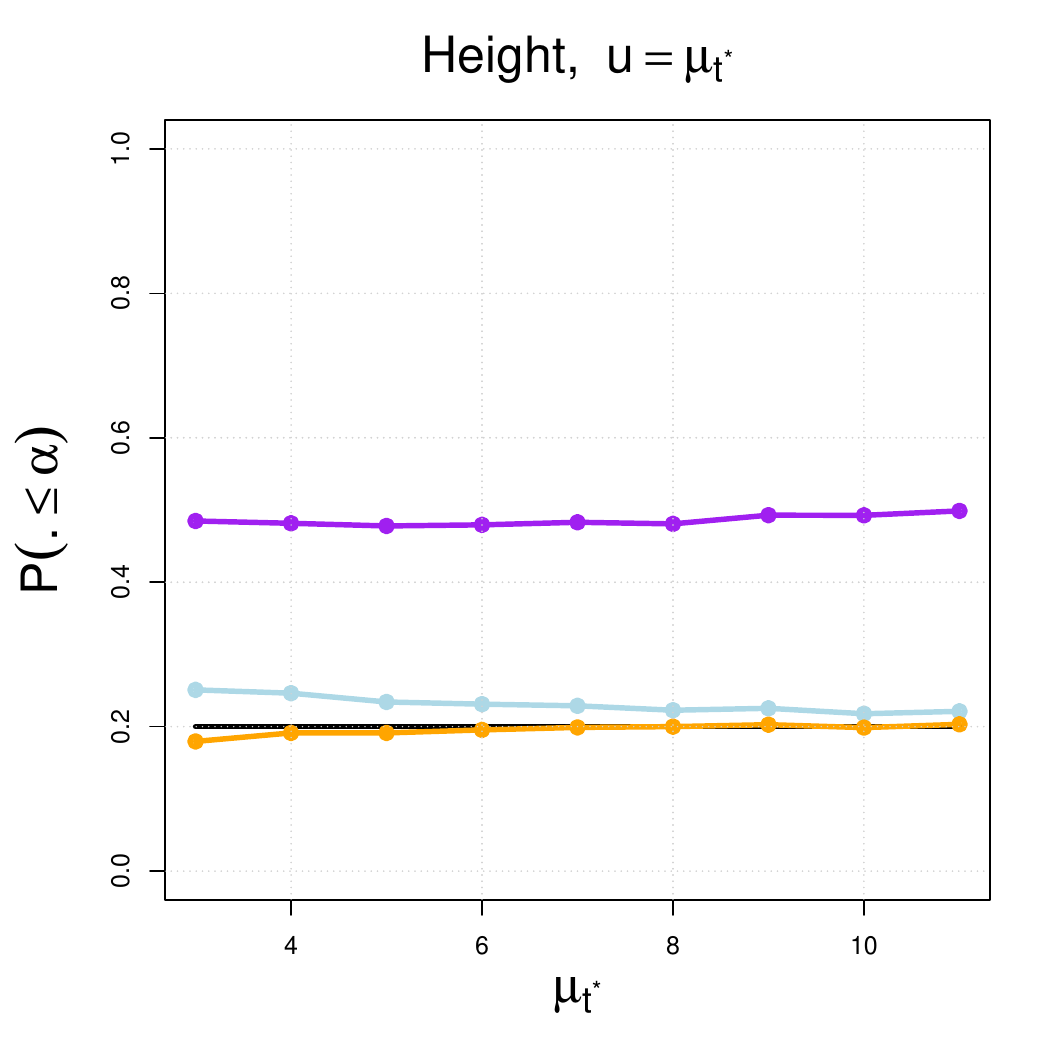}
			\label{fig:experiment-1-2d-2}
		\end{subfigure}\hfill
		\begin{subfigure}[t]{0.32\linewidth}
			\centering
			\includegraphics[width=\linewidth]{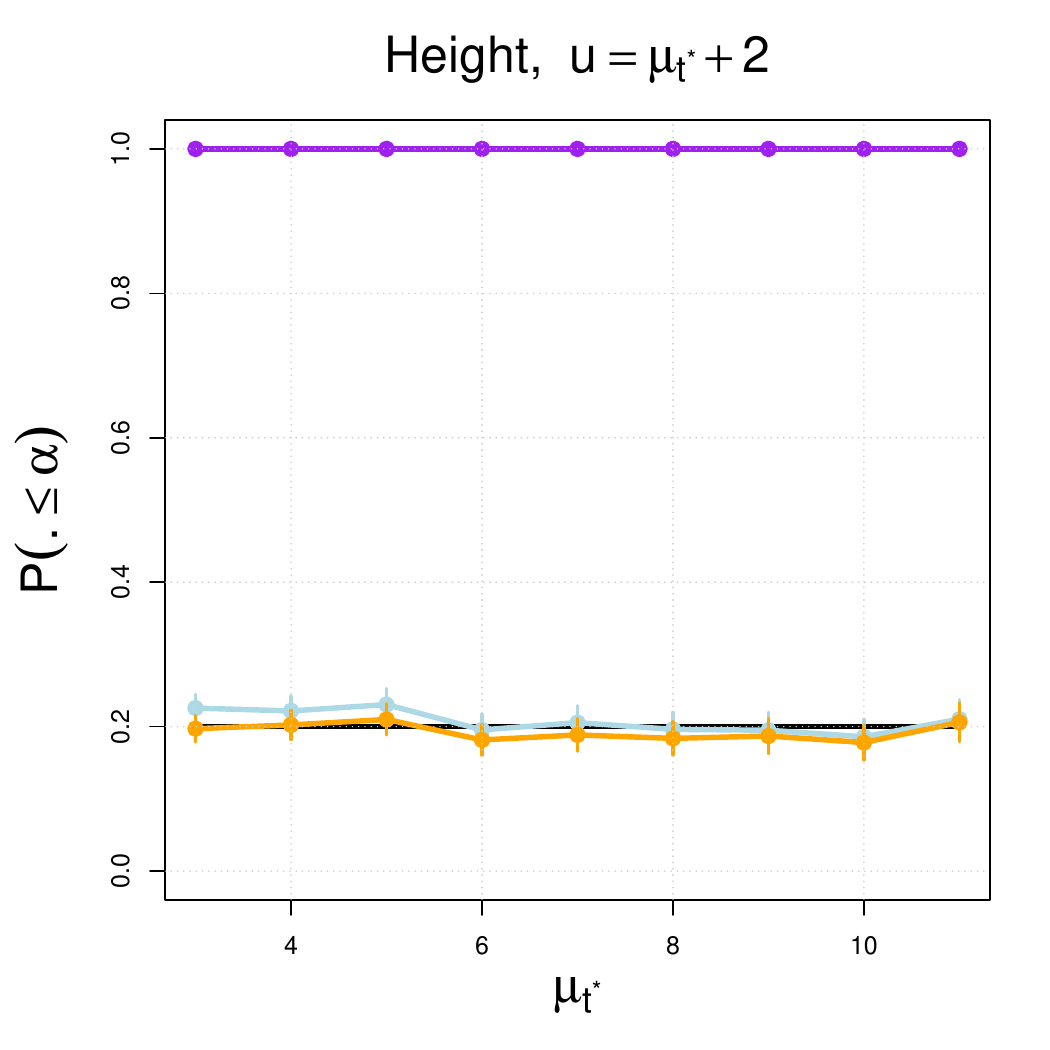}
			\label{fig:experiment-1-2d-3}
		\end{subfigure}
		
		\caption{Distribution of candidate quantities for location (top row) and height (bottom row). Different columns corresponds to different thresholds $u$. Details of experimental setup and takeaways are in the main text.}
		\label{fig:experiment-1-2d}
	\end{figure}
	
	\subsection{Conditional coverage}
	\label{subsec:experiment-2}
	Our second experiment compares the performance of our three proposals for post-selection peak inference: the non-randomized method outlined in Algorithm~\ref{alg:post-selective-inference-after-peak-detection-via-TG-test}, and the two randomized methods, carve~(Algorithm~\ref{alg:carve}) and split~(Algorithm~\ref{alg:split}). For the latter two methods the randomization tuning parameter is set to be $\gamma = 1$, which would correspond to a 50/50 split if $Y$ were constructed by averaging independent replicates. The distribution of $Y$ and choices of threshold are the same as in our first experiment, but we now examine coverage and size of confidence regions conditional on selection. (Notice that the conditional distributions underlying each method differ in what precisely is conditioned on; see~\eqref{eqn:distribution-after-selection} and~\eqref{eqn:distribution-after-randomized-selection} respectively.) The results are displayed in Figure~\ref{fig:experiment-2-2d}.  
	
	For the location: all methods achieve nominal coverage when the signal strength is large and selection pressure is weak. Split inference is noticeably worst at the smallest signal strengths: this is due to the fact that, after standardization, the curvture of the field $Y^{\inf}$ that split uses for inference is smaller than the curvature of the full data field $Y$. Under moderate or strong selection pressure, the non-randomized method also has inflated type I error, even when the signal is strong, as predicted by our theory. 
	
	For the height: all methods achieve very close to nominal coverage across different signal strengths and sample sizes, except for split and carve at the very smallest signal strengths. However, the split intervals are always wider than the carved intervals, and the standard intervals are much wider than either under strong selection pressure. 
	
	Additional experiments presented in Section~\ref{subsec:experiment-2-additional} demonstrate that these conclusions are qualitatively similar in a one-dimensional setting. They also evaluate performance in a setup where the signal is wider than the covariance kernel. In this last setting, the relative performance of the non-randomized, carve, and split methods is similar, but all methods perform worse at low signal strengths. This makes sense as the wider signal has much smaller curvature, making the problem more challenging. 
	
	\begin{figure}[htbp]
		\centering
		\begin{subfigure}[t]{0.32\linewidth}
			\centering
			\includegraphics[width=\linewidth]{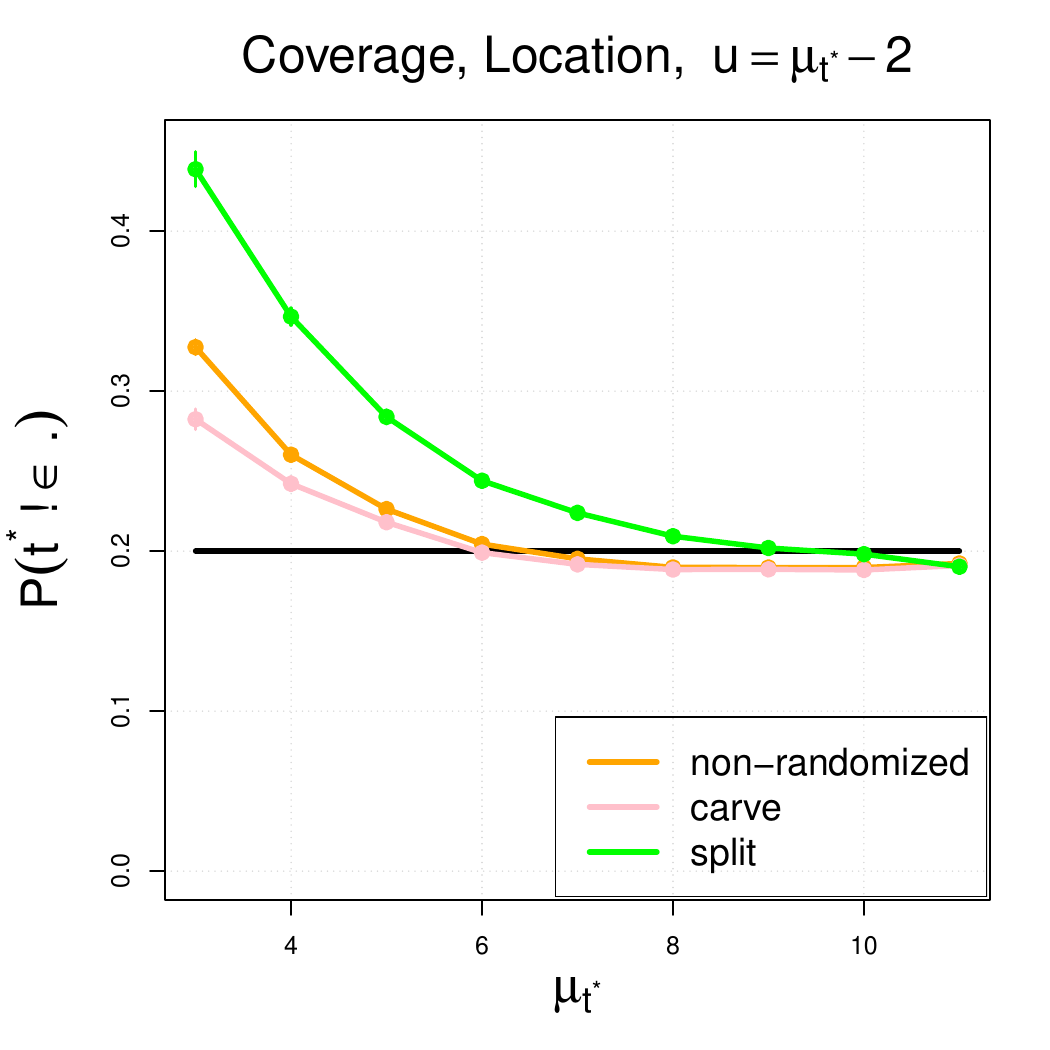}
		\end{subfigure}\hfill
		\begin{subfigure}[t]{0.32\linewidth}
			\centering
			\includegraphics[width=\linewidth]{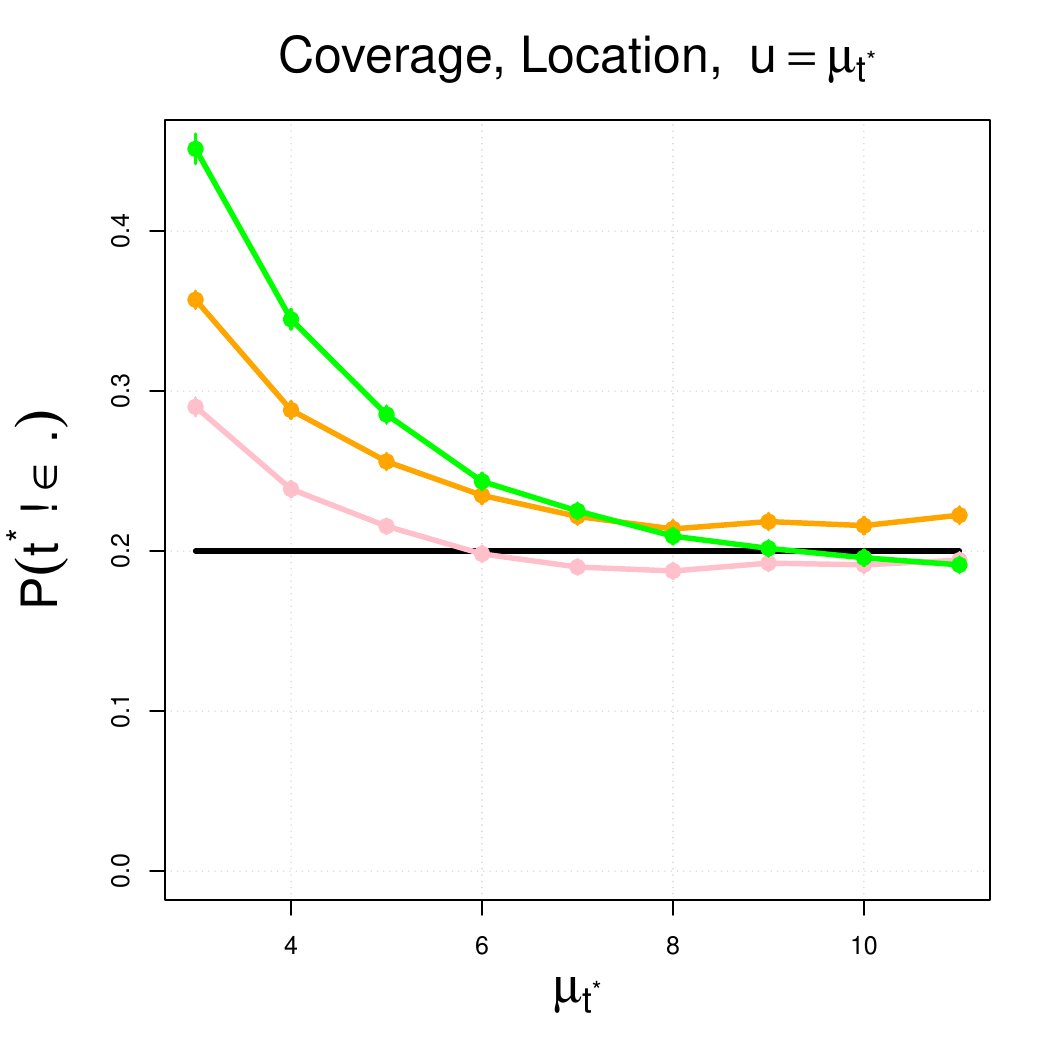}
		\end{subfigure}\hfill
		\begin{subfigure}[t]{0.32\linewidth}
			\centering
			\includegraphics[width=\linewidth]{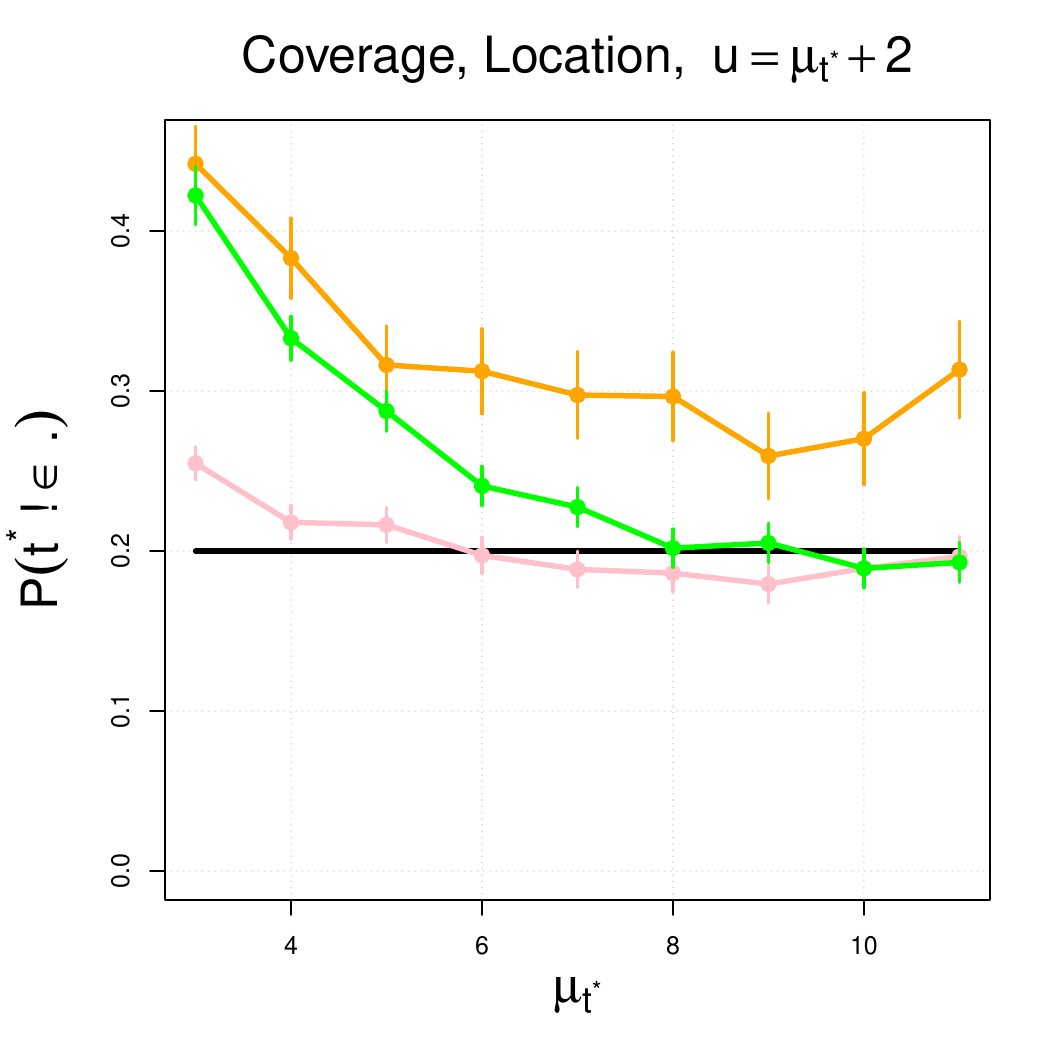}
		\end{subfigure}\hfill
		
		\begin{subfigure}[t]{0.32\linewidth}
			\centering
			\includegraphics[width=\linewidth]{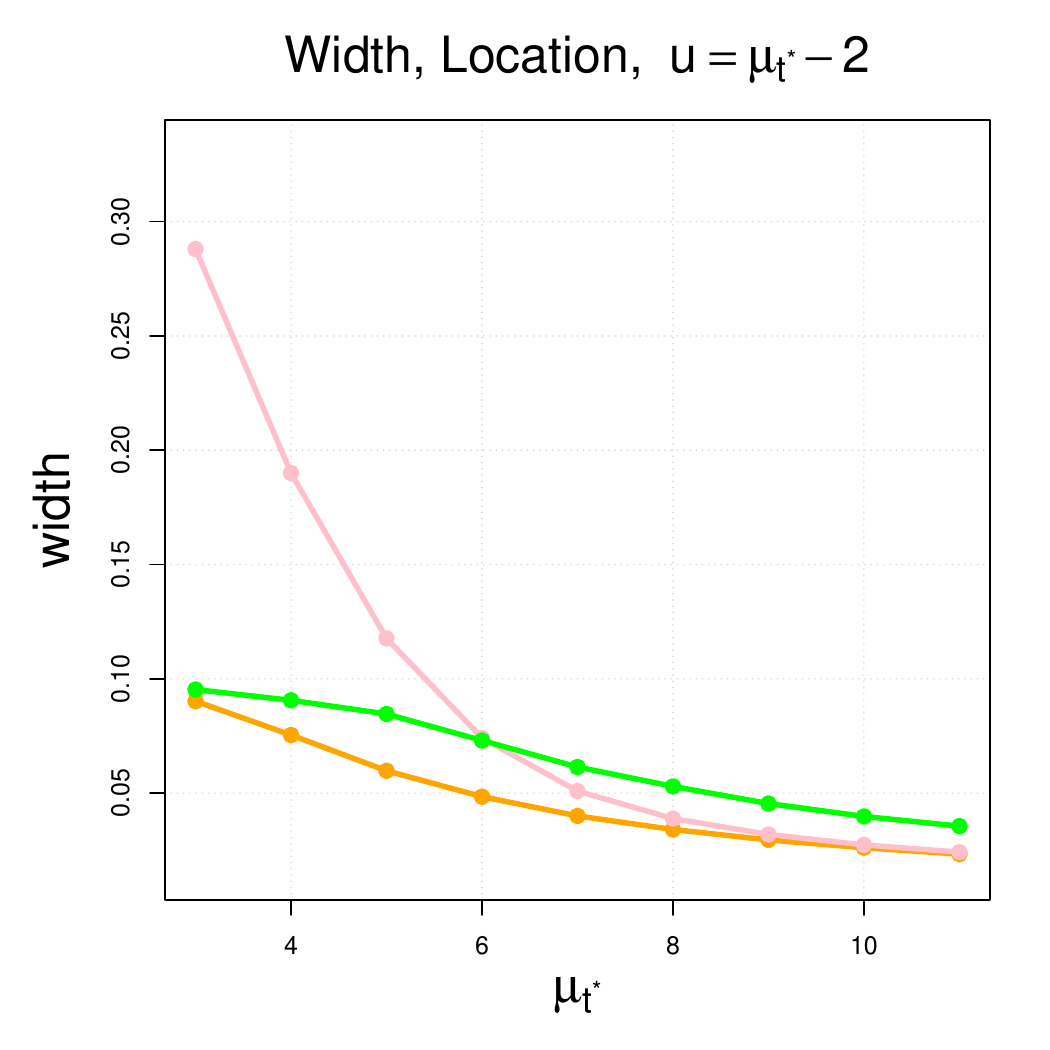}
		\end{subfigure}\hfill
		\begin{subfigure}[t]{0.32\linewidth}
			\centering
			\includegraphics[width=\linewidth]{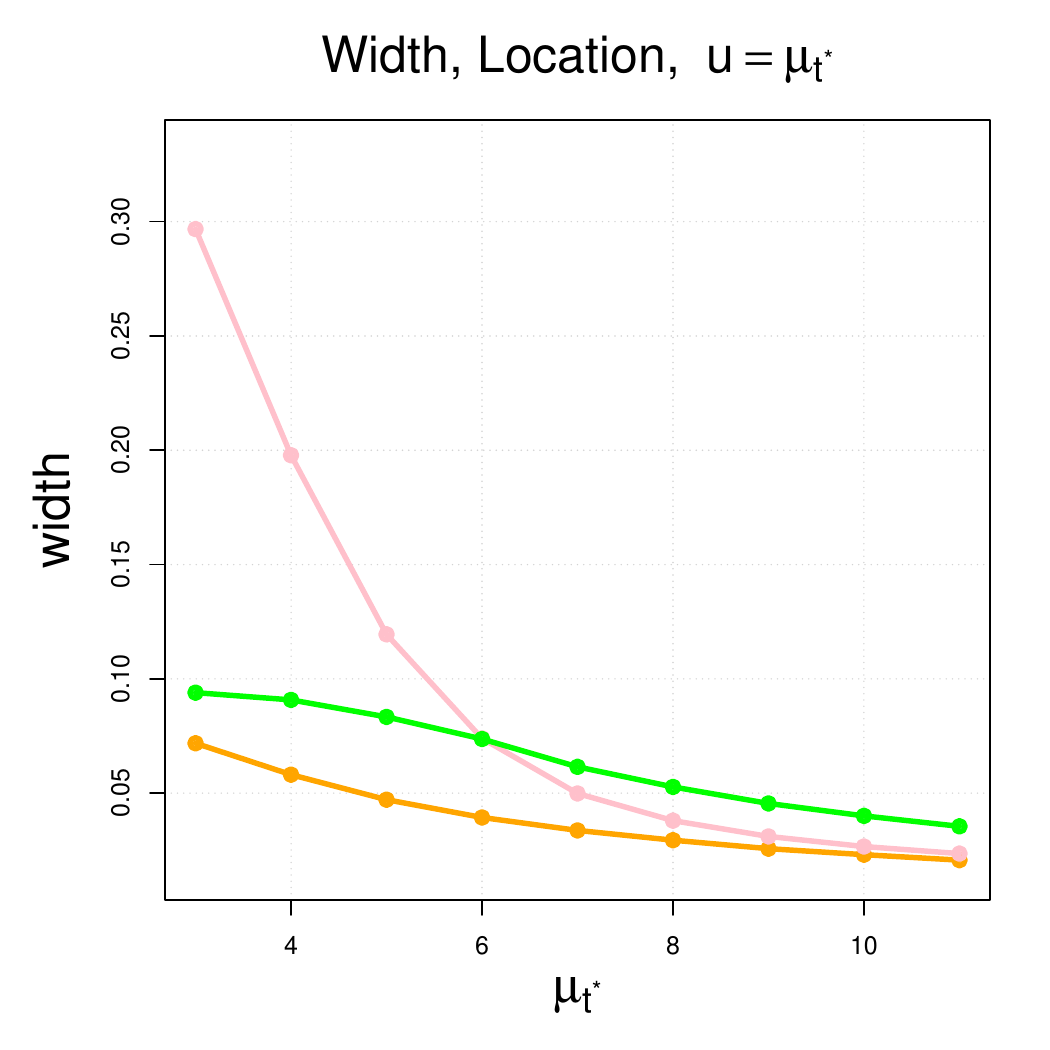}
		\end{subfigure}\hfill
		\begin{subfigure}[t]{0.32\linewidth}
			\centering
			\includegraphics[width=\linewidth]{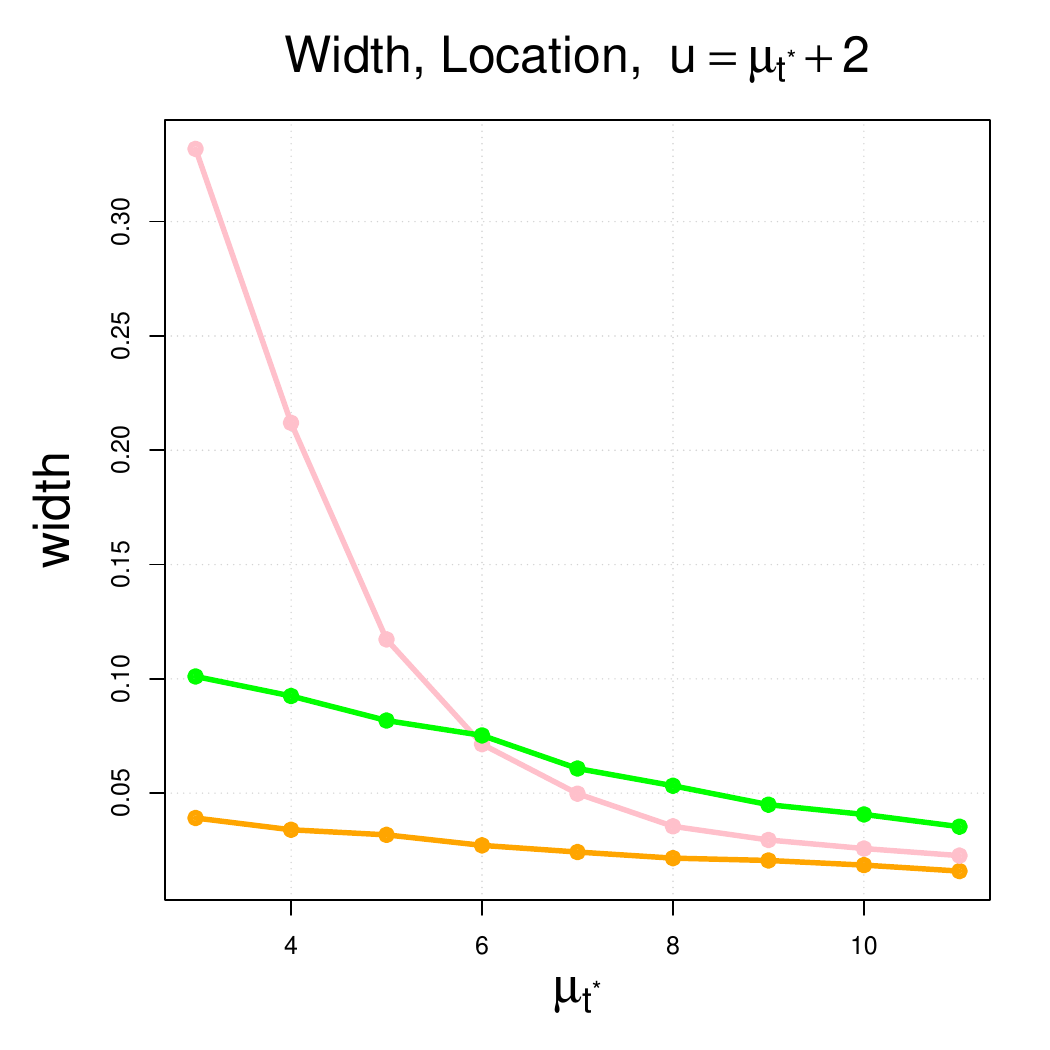}
		\end{subfigure}\hfill
		
		\begin{subfigure}[t]{0.32\linewidth}
			\centering
			\includegraphics[width=\linewidth]{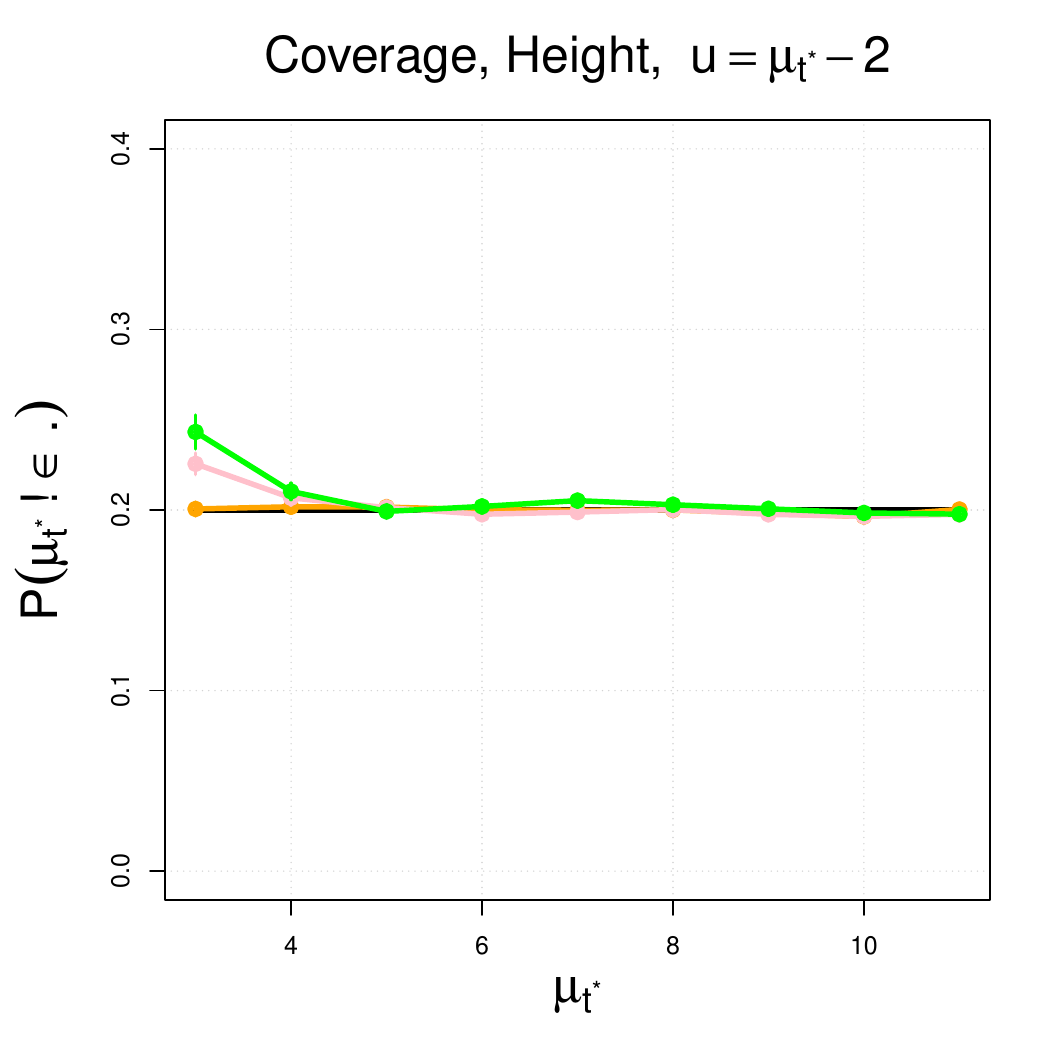}
		\end{subfigure}\hfill
		\begin{subfigure}[t]{0.32\linewidth}
			\centering
			\includegraphics[width=\linewidth]{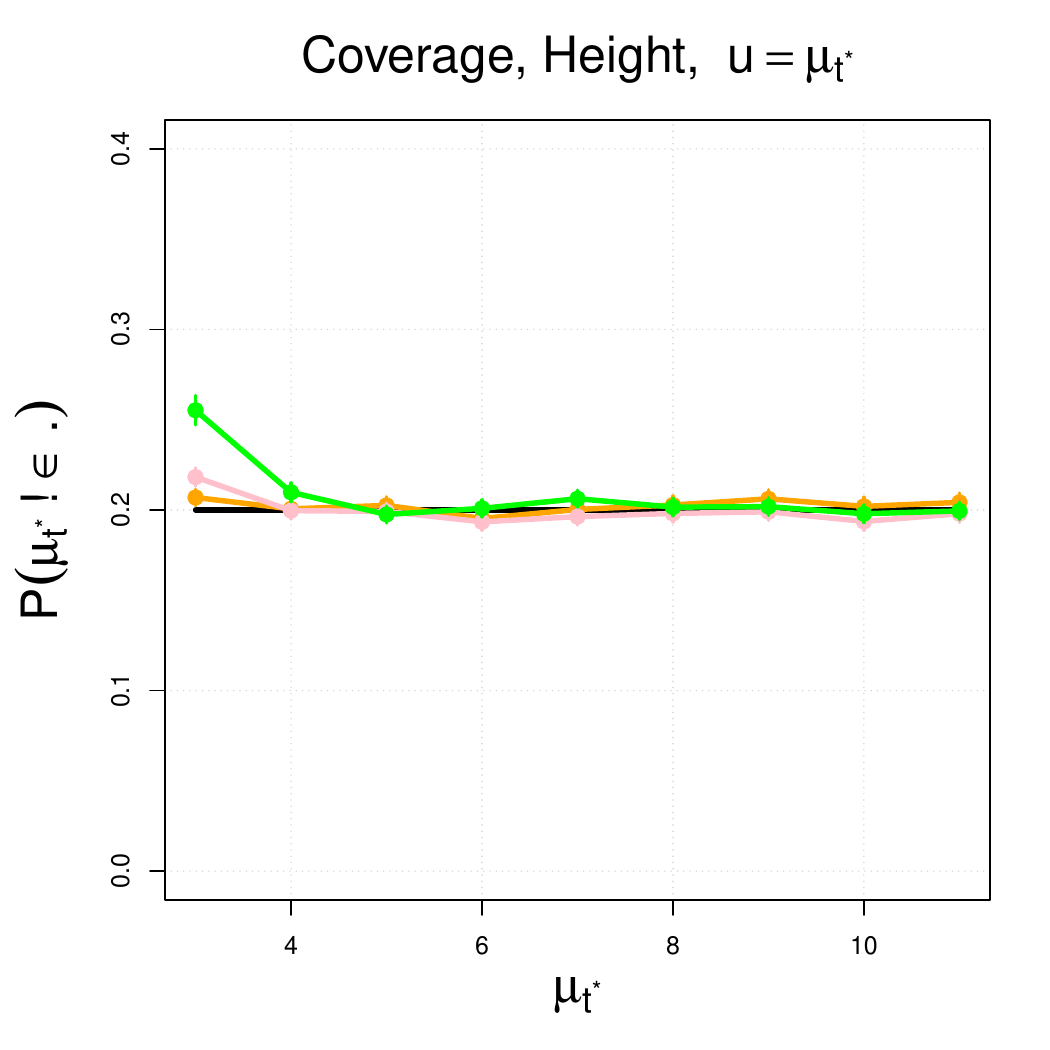}
		\end{subfigure}\hfill
		\begin{subfigure}[t]{0.32\linewidth}
			\centering
			\includegraphics[width=\linewidth]{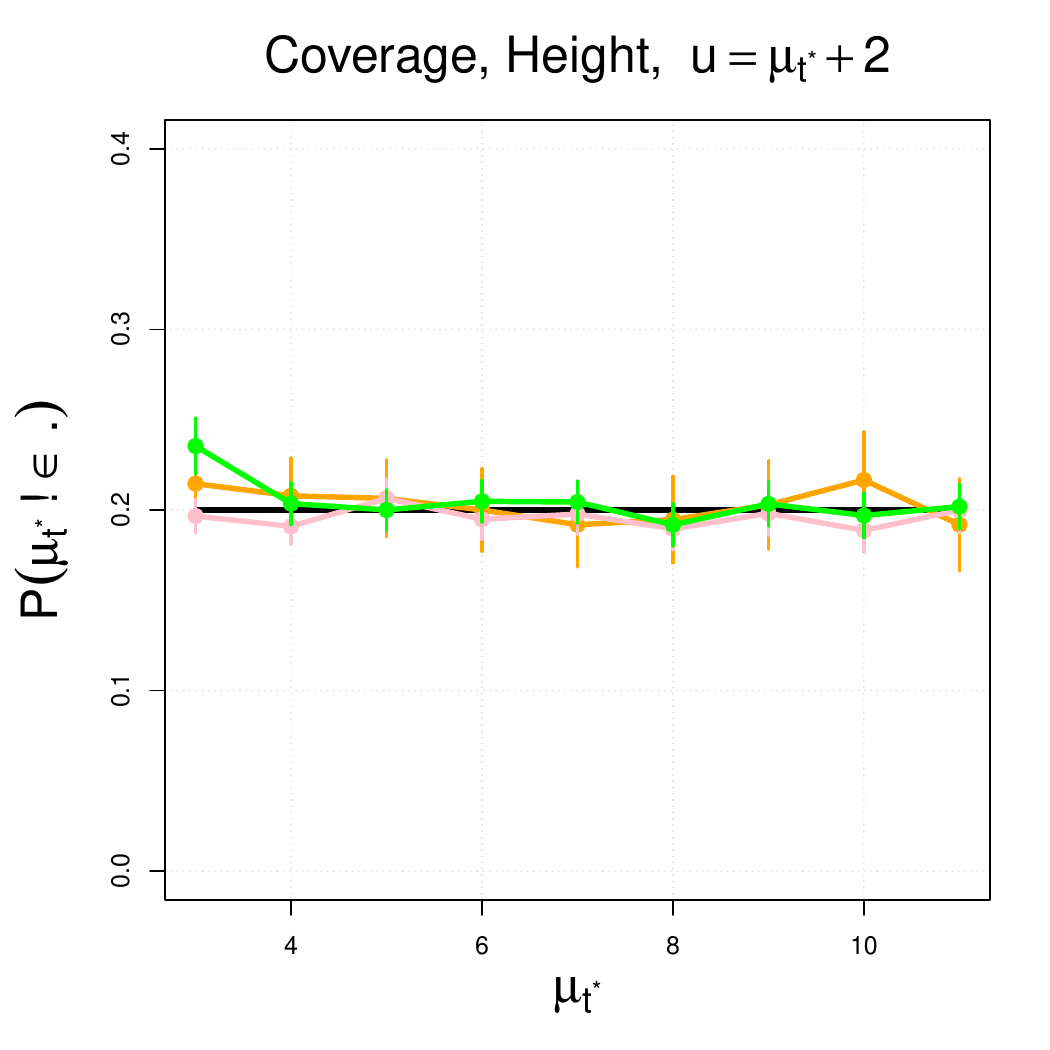}
		\end{subfigure}\hfill
		
		\begin{subfigure}[t]{0.32\linewidth}
			\centering
			\includegraphics[width=\linewidth]{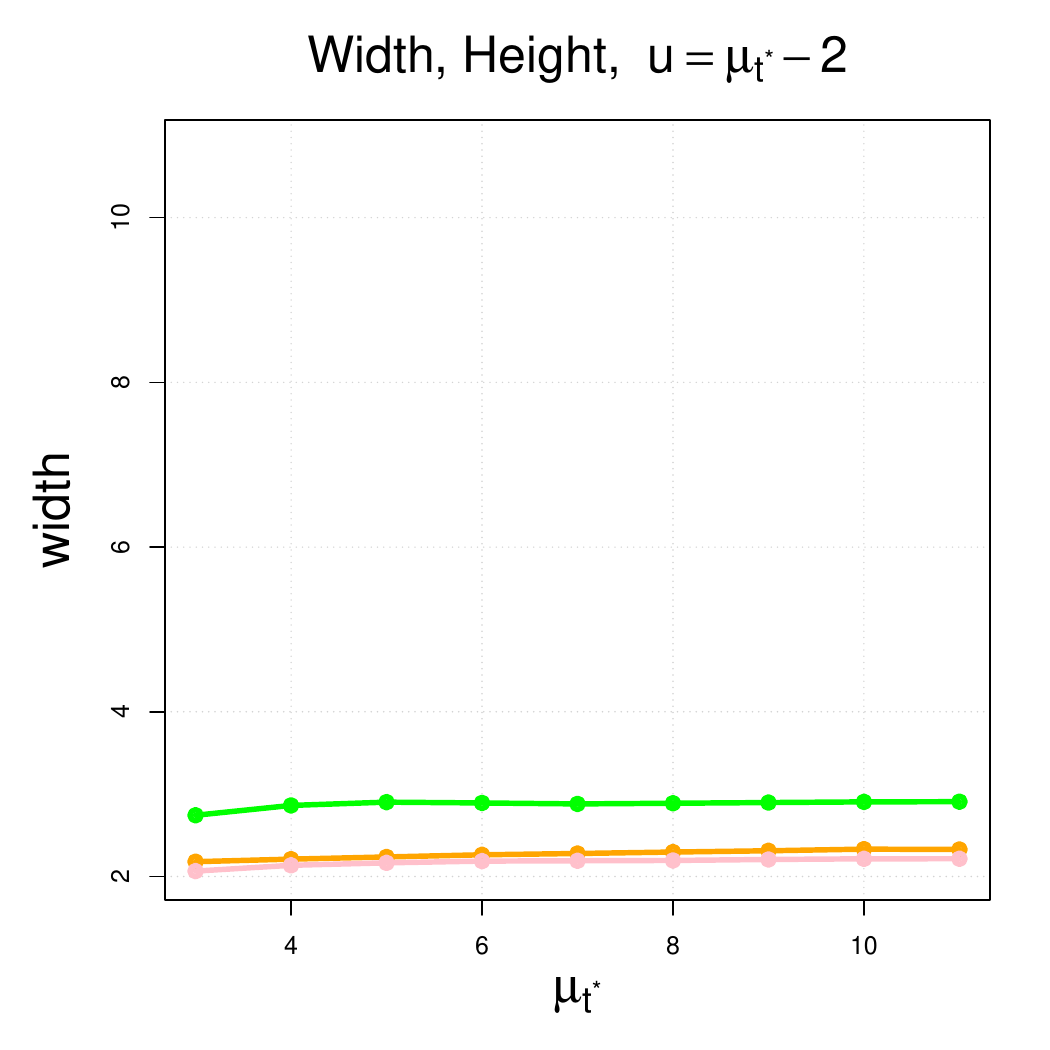}
		\end{subfigure}\hfill
		\begin{subfigure}[t]{0.32\linewidth}
			\centering
			\includegraphics[width=\linewidth]{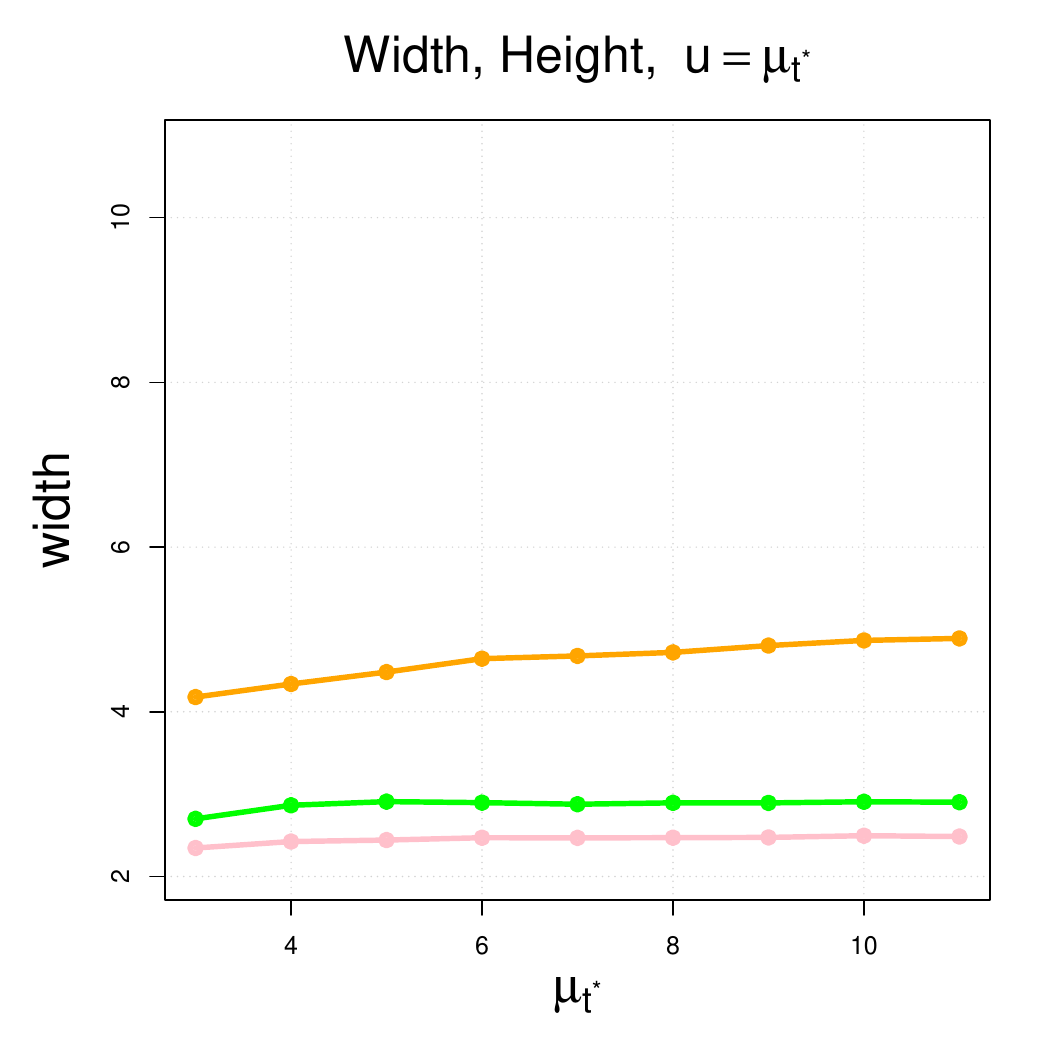}
		\end{subfigure}\hfill
		\begin{subfigure}[t]{0.32\linewidth}
			\centering
			\includegraphics[width=\linewidth]{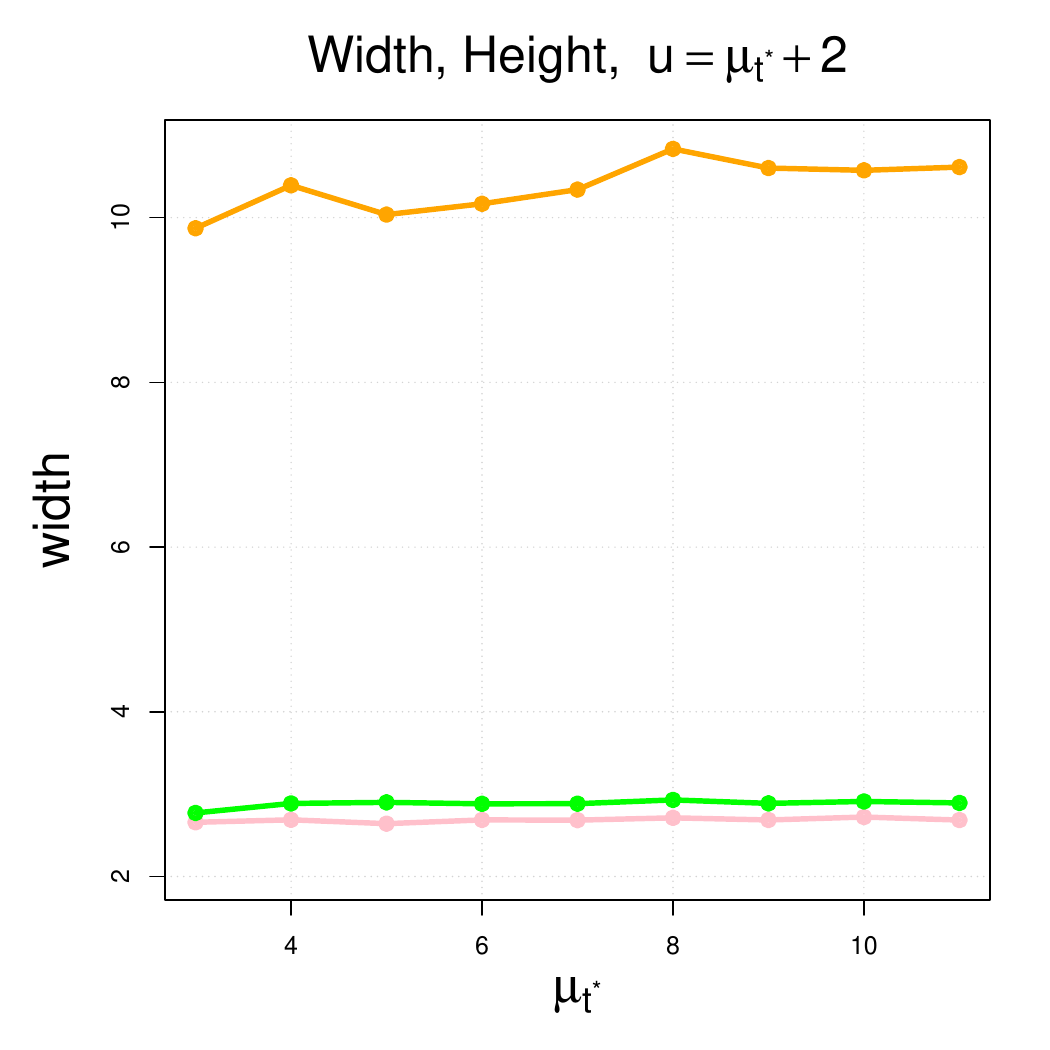}
		\end{subfigure}\hfill
		
		\caption{Comparing miscoverage and width of non-randomized, carve, and split methods for peak inference. Top two rows correspond to inference for location, bottom two rows to inference for height. ``Width'' of a confidence ellipse is the length of its largest semi-axis. Details of experimental setup and takeaways are in the main text.}
		\label{fig:experiment-2-2d}
	\end{figure}
	
	\subsection{Multiple peaks and overall performance}
	\label{subsec:experiment-3}
	Our third and final experiment compares the performance of the non-randomized, carve, and split methods in a more realistic experiment with multiple peaks. In this experiment the signal is a superposition of nine compactly supported unimodal functions of various peak heights. Additionally, in this experiment we do not set the threshold $u$ ad-hoc, but instead set $u = u_{\TG}(\alpha,v)$ based on theory, to control null-PCER at $\alpha = 0.1$. Figure~\ref{fig:experiment-3} shows both conditional coverage on a per-peak basis and overall miscoverage, for both height and location. 
	
	For the location: both non-randomized and split methods have inflated miscoverage on a per-peak, conditional basis, while carved is closer to nominal. Marginally, however, all methods have nominal coverage, with split and carve even being conservative. This conservatism is expected because this is a setting where not all true peaks are discovered with high probability; see the discussion after Theorem~\ref{thm:miscoverage}.
	
	For the height: all methods achieve very close to nominal coverage on a conditional, per-peak basis and are marginally conservative. However the carved intervals are smallest, particularly at the true peaks which have the lowest heights.
	
	\begin{figure}[htbp]
		\centering
		\begin{subfigure}[t]{0.32\linewidth}
			\centering
			\includegraphics[width=\linewidth]{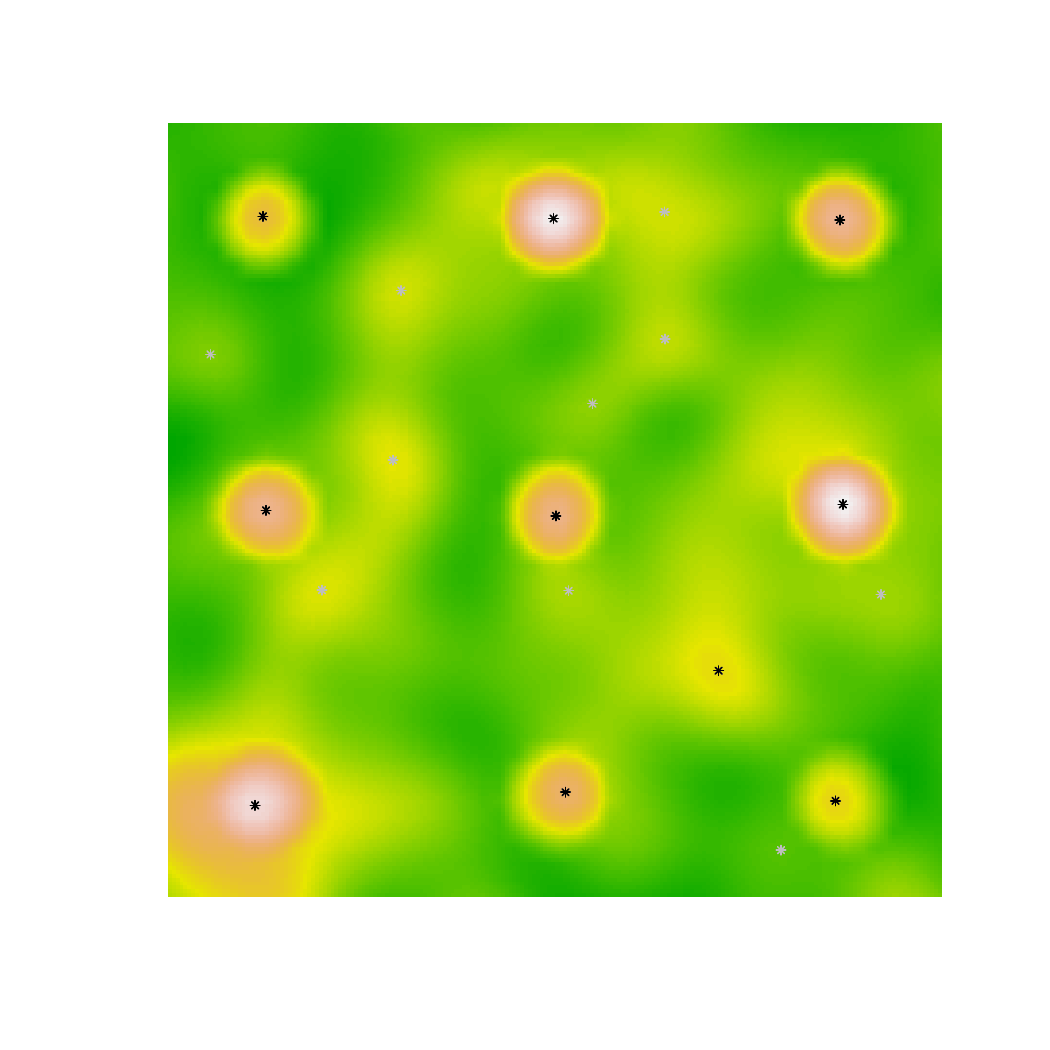}
		\end{subfigure}	
		\begin{subfigure}[t]{0.32\linewidth}
		\centering
		\includegraphics[width=\linewidth]{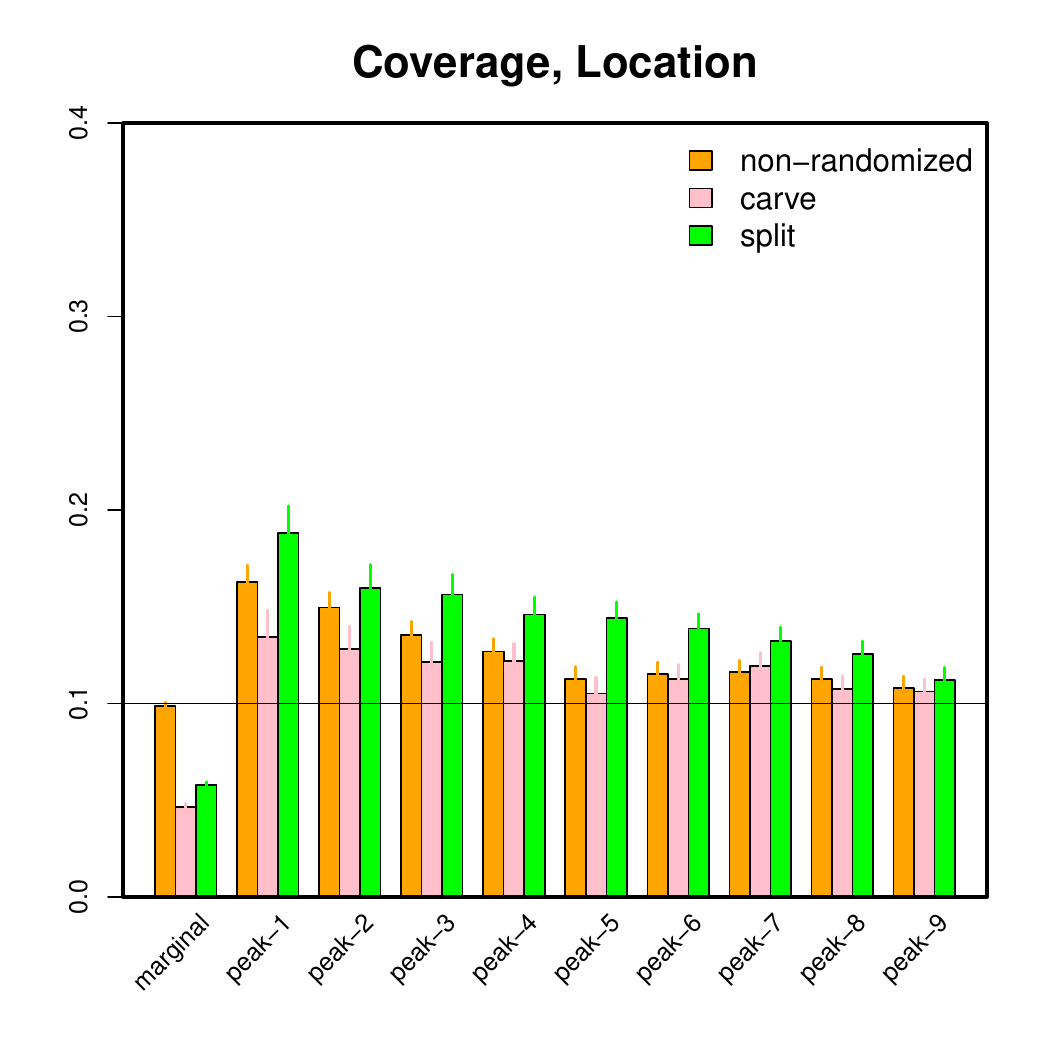}
		\end{subfigure}	
		\begin{subfigure}[t]{0.32\linewidth}
		\centering
		\includegraphics[width=\linewidth]{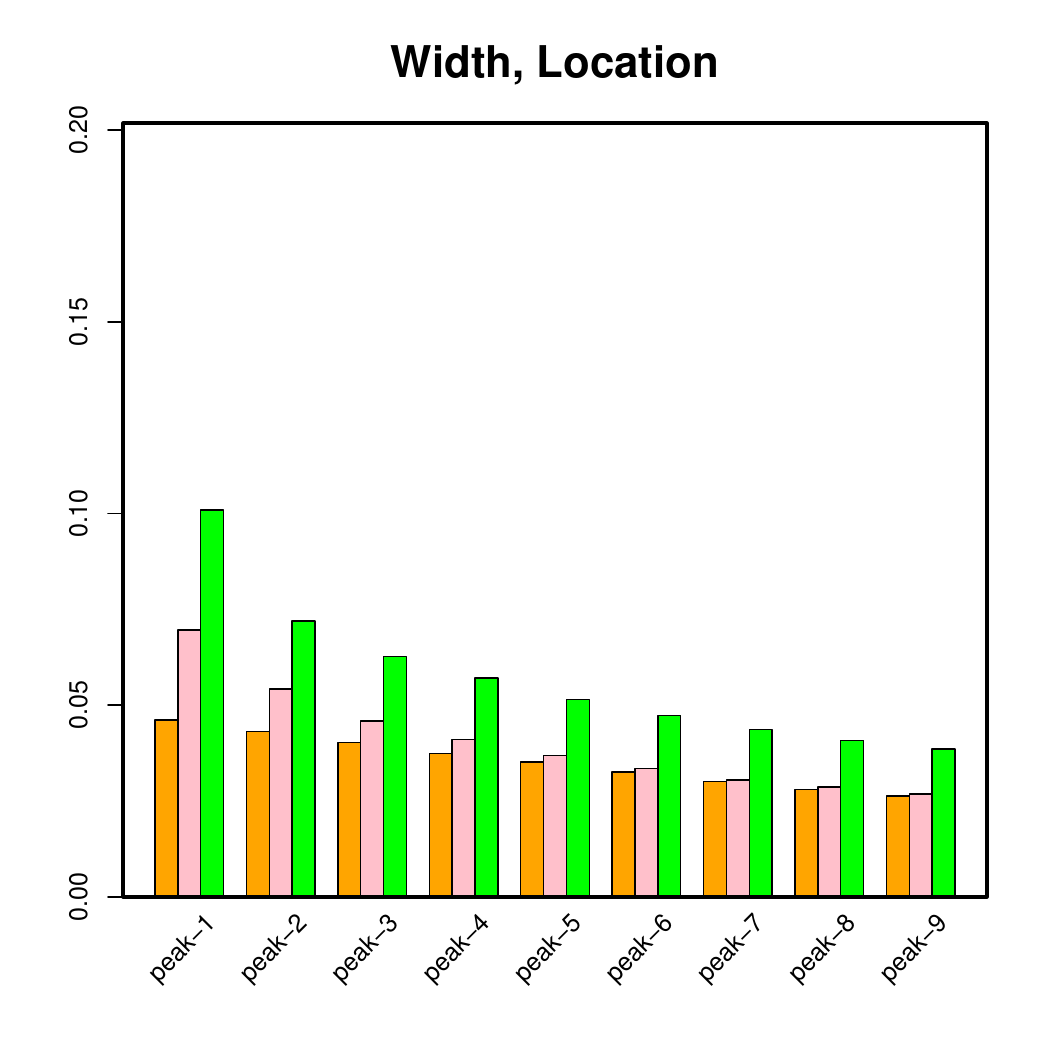}
		\end{subfigure} \hfill \\
		\begin{subfigure}[t]{0.32\linewidth}
			\centering
			\includegraphics[width=\linewidth]{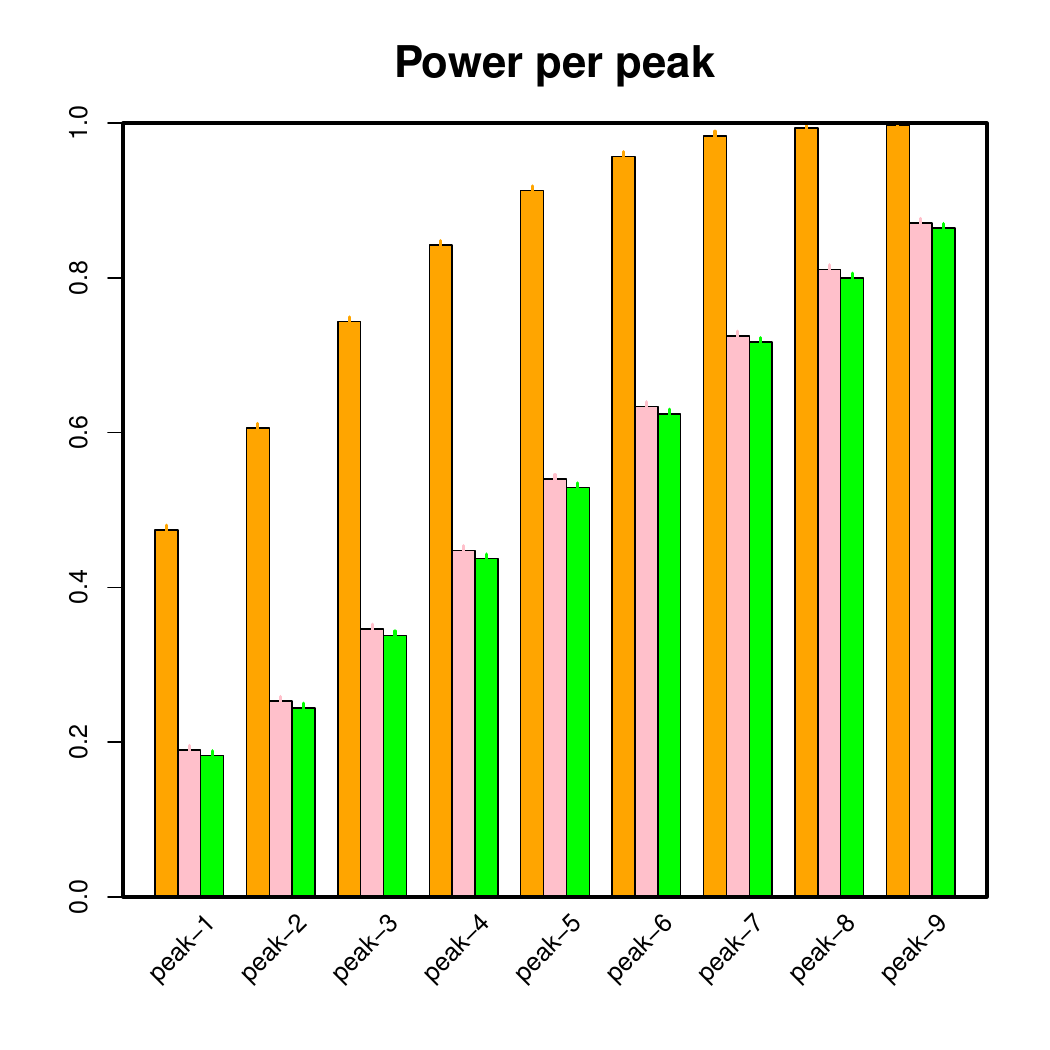}
		\end{subfigure}	
		\begin{subfigure}[t]{0.32\linewidth}
			\centering
			\includegraphics[width=\linewidth]{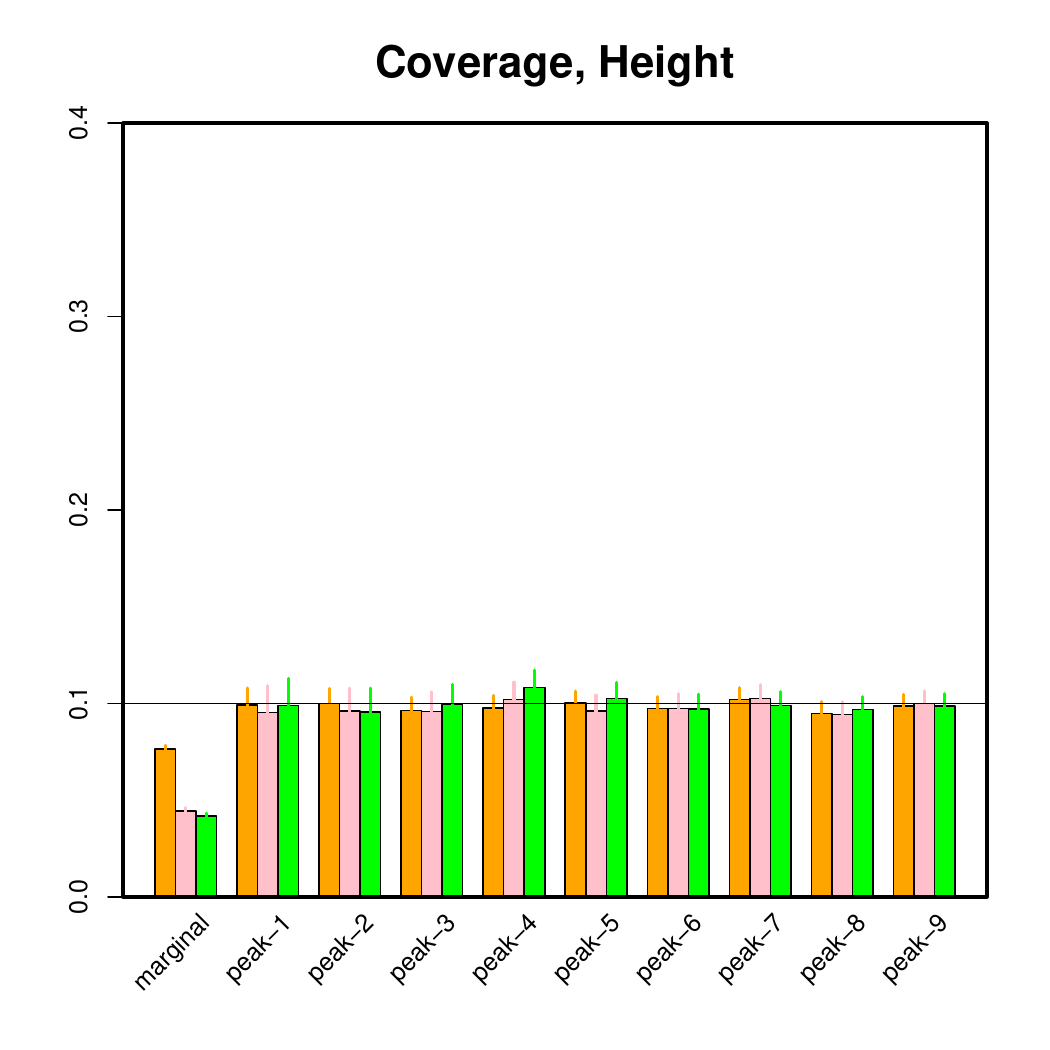}
		\end{subfigure}
		\begin{subfigure}[t]{0.32\linewidth}
			\centering
			\includegraphics[width=\linewidth]{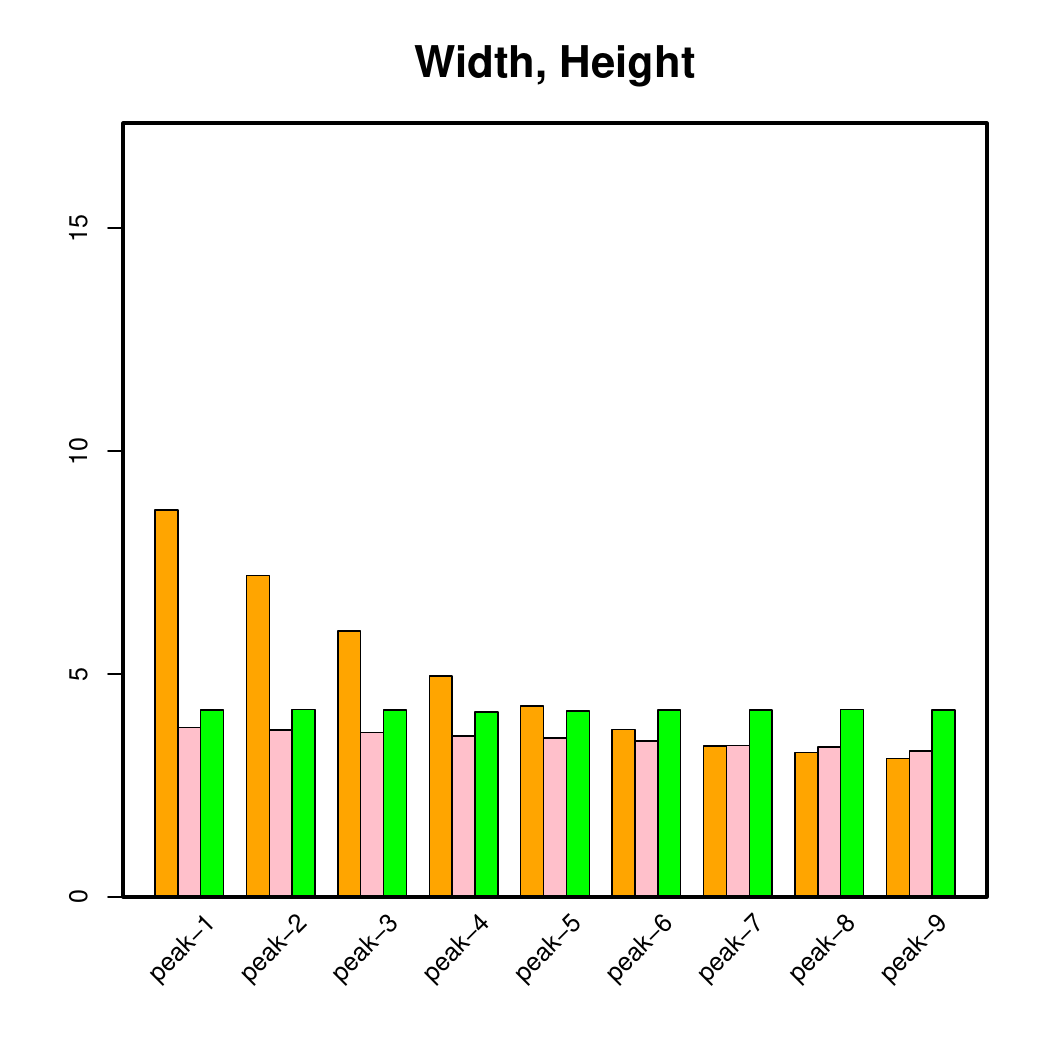}
		\end{subfigure}\hfill
		
		\caption{Comparing power, miscoverage and width of non-randomized, carve, and split methods for peak inference, in the multipeak experiment described in Section~\ref{subsec:experiment-3}. Peak heights are evenly spaced between $\mu_{t^*} = 3$ (peak-1) and $\mu_{t^*} = 6$ (peak-9). Top left: example field along with all peaks declared significant by TG test. }
		\label{fig:experiment-3}
	\end{figure}
	
	\section{Summary and discussion}
	\label{sec:discussion}
	
	This article lays out several procedures for peak detection followed by localization via formal statistical inference. These procedures account for the fact that peaks were subjected to thresholding and significance testing prior to inference. As a result, they are valid, both marginally and in a post-selective, conditional sense, even in regimes where some peaks will be falsely selected and discovery of true peaks is not guaranteed. 
	
	We conclude by mentioning a few interesting directions in which our theory and methodology might be extended. 
	First, our procedures assume that the field $Y$ is standardized, and assume knowledge of $\Lambda_t$ in constructing confidence regions. If the true covariance kernel $K$ is unknown, then the field must be standardized using an estimate of variance, and an estimate of $\Lambda$ must be used; it is not clear to us the effect this will have on subsequent inferences.  Second, the TG peak detection procedure is calibrated to control PCER rather than a more stringent criterion like FWER or FDR. We believe that a larger choice of significance threshold $u$, based on a Bonferroni/BH-like procedure, should result in control of either FWER or FDR. A likely downstream consequence of this would be that our overall method for inference -- with this larger choice of $u$ -- would control asymptotic false coverage rate rather than the focus of this paper, which is the per-comparison miscoverage rate. Finally, it would be interesting to consider post-selection peak inference using methods for identifying signal regions that are more sophisticated than simply thresholding peaks.

	\bibliographystyle{plainnat}   
	\bibliography{bibliography}         
	
	\section{Proofs for Section~\ref{sec:peak-intensity-local-expansion}}
	\label{sec:pfs-peak-intensity-local-expansion}
	
	In this section we build to the proof of Theorem~\ref{thm:approximate-joint-intensity}. Sections~\ref{subsec:asymptotics-signal} and~\ref{subsec:asymptotics-covariance} contain some preliminary estimates on the asymptotic behavior of the signal and covariance. In Section~\ref{subsec:hessian-asymptotics} we show that conditional on there being a discovery $\hat{t} \in \wh{T}_u$ that consistently estimates $t^*$, the relative error between the Hessian $-\nabla^2 Y_t$ and $\bar{H}_{t^*}$ converges to $0$ uniformly over $t$ in a neighborhood of $t^*$. In Section~\ref{subsec:pf-approximate-joint-intensity} we state two Lemmas, Lemma~\ref{lem:approximation-peak-intensity} and Lemma~\ref{lem:asymptotic-expansion-density}, that give local expansions of the determinant and density terms in the Kac-Rice formula for $\rho(t,y)$. Lemma~\ref{lem:approximation-peak-intensity} is proved in Section~\ref{subsec:pf-approximation-peak-intensity} and  Lemma~\ref{lem:asymptotic-expansion-density} is proved in Section~\ref{subsec:asymptotic-analysis-density}; both proofs rely on the results of Sections~\ref{subsec:asymptotics-signal}-\ref{subsec:hessian-asymptotics}. Combined these local expansions imply Theorem~\ref{thm:approximate-joint-intensity}, as shown in Section~\ref{subsec:pf-approximate-joint-intensity-2}. 
	
	Throughout, we will assume that $n \in \mathbb{N}$ is large enough that each of the following is true:~\eqref{eqn:well-conditioned} is satisfied, $\mc{B}_{t^*} \subseteq B(t^*,c_0)$, and $C\sqrt{\log(\lambda_n)} \leq \lambda_n$ where $C$ is some constant to be determined later. Notice that these are all true if $n \geq C$. In the future, we will not be as explicit about what lower bounds $n$ must satisfy, and will often simply revert to stating that $n$ is ``sufficiently large.''
	
	\subsection{Asymptotics of signal}
	\label{subsec:asymptotics-signal}
	We consider Taylor expansion of $\mu_t,\nabla \mu_t, \nabla^2 \mu_t$ around $t = t^*$ for some $t^* \in T^*$. By Assumption~\ref{asmp:well-conditioned}, the remainder terms in these expansions are bounded as follows: 
	\begin{equation}
		\begin{aligned}
			\label{eqn:signal-taylor-expansion}
			\sup_{h \in B(0,c_2)}\Big|\mu_{t^* + h} - \mu_{t^*} -  \frac{1}{2}h'\nabla^2\mu_{t^*}h\Big| 
			& \leq C\|h\|^3 \lambda_{t^*}, \\
			\sup_{h \in B(0,c_2)} \|\nabla \mu_{t^* + h} - \nabla^2 \mu_{t^*}h - \frac{1}{2}\nabla^3\mu_{t^*}(h,h)\| 
			& \leq C \|h\|^3 \lambda_{t^*}, \\
			\sup_{h \in B(0,c_2)} \|\nabla^2 \mu_{t^* + h} - \nabla^2 \mu_{t^*} - \nabla^3\mu_{t^*}(h)\| 
			& \leq C \|h\|^2 \lambda_{t^*}.
		\end{aligned}
	\end{equation}
	
	We will use the following consequences of \eqref{eqn:signal-taylor-expansion} and Assumption~\ref{asmp:well-conditioned}: there exists a constant $c > 0$ such that for all $n \in \mathbb{N}$ sufficiently large,
	\begin{equation}
		\begin{aligned}
			\label{eqn:signal-1st-and-2nd-derivative-control}
			\sup_{h \in B(0,c)} \|\nabla^2 \mu_{t^* + h}\| 
			& \leq C\lambda_{t^*}, \\
			\sup_{h \in B(0,c)} \Big\{\frac{\|\nabla \mu_{t^* + h}\|}{\|h\|}\Big\}
			& \leq C\lambda_{t^*}, \\
			\sup_{h \in B(0,c)} \Big\{ \frac{|\mu_{t^* + h} - \mu_{t^*}|}{\|h\|^{2}} \Big\}
			& \leq
			C\lambda_{t^*}.
		\end{aligned}
	\end{equation}
	
	\subsection{Asymptotics of covariance}
	\label{subsec:asymptotics-covariance}
	
	We begin by recalling some notation from the main text, and introducing some new notation as well. Recall that $\Lambda_t = \Cov[\nabla \epsilon_t] \in \R^{d \times d}$ and that we write $K_{ij}(t,s) := \Cov[\nabla^i\epsilon_t,\nabla^j\epsilon_s]$; from now on we abbreviate $K_{ij}(t) := K_{ij}(t,t)$. Let $\Gamma_t := \Cov[\nabla^2 \epsilon_t,\nabla \epsilon_t]\Lambda_t^{-1} = K_{21}(t)\Lambda_t^{-1}$; this is a $d \times d \times d$ array with entries $(\Gamma_t)_{ijk} = \sum_{m = 1}^{d} [K_{21}(t)]_{ijm} [\Lambda_{t}^{-1}]_{mk}$. 
	
	Assumption~\ref{asmp:covariance-holder} (and in the case of $\Gamma_t$, Equation~\eqref{eqn:non-degenerate}) imply that $\Lambda_t,\Gamma_t$ and the first and second derivatives of these array-valued mappings are bounded above by constants: 
	\begin{equation}
		\begin{aligned}
			\label{eqn:variance-control}
			\sup_{t \in \Rd } \|\Lambda_t\| < C, \sup_{t \in \Rd} \|\dot{\Lambda}_t\| < C, \sup_{t \in \Rd}\|\ddot{\Lambda}_t\| < C, \\
			\sup_{t \in \Rd} \|\Gamma_t\| < C, \sup_{t \in \Rd} \|\dot{\Gamma}_t\| < C, \sup_{t \in \Rd} \|\ddot{\Gamma}_t\| < C.
		\end{aligned}
	\end{equation}
	We also note that for any points $t,t'$,
	\begin{equation*}
		\begin{aligned}
			\|K_{20}(t,t') - \Lambda_{t}\| \leq C\sup_{\tilde{t} \in B(t,\|t' - t\|)} \|K_{21}(t,\tilde{t})\| \cdot \|t' - t\| \leq C \|t' - t\|, \\
			\|K_{21}(t,t') - K_{21}(t)\| \leq C\sup_{\tilde{t} \in B(t,\|t' - t\|)} \|K_{22}(t,\tilde{t})\| \cdot \|t' - t\| \leq C \|t' - t\|.
		\end{aligned}
	\end{equation*}
	
	\subsection{Asymptotics of Hessian}
	\label{subsec:hessian-asymptotics}
	In the main text we have observed that conditional on the selection event $N(\mc{S}_{t^*}) = 1$, the observed Hessian $\wh{H}$ is a biased estimate for $-\nabla^2\mu_{t^*}$ and instead has deterministic limit $\bar{H}_{t^*} = \E[-\nabla^2 Y_{t^*}|Y_{t^*} = \bar{u}_{t^*}]$. In this section, we prove the related claim: conditional on $\{Y_t = y,\nabla Y_t = 0\}$ at some $t \in B_d(t^*,\varepsilon_n), y \in \bar{u}_{t^*} \pm \Delta_n$, the relative error between $-\nabla^2 Y_s$ and $\bar{H}_{t^*}$ converges to $0$ in $\P$-probability, uniformly over $s \in B_d(t^*,\varepsilon_n)$. In fact we derive explicit bounds on the rate of convergence, see~\eqref{eqn:hessian-uniform-concentration}. Later in Section~\ref{subsec:pf-pivot-hessian} we show that this implies $\wh{H}$ has deterministic limit $\bar{H}_{t^*} = \E[-\nabla^2 Y_{t^*}|Y_{t^*} = \bar{u}_{t^*}]$ under $\Q^{\mc{S}_{t^*}}$. 
	
	\paragraph{Projection/residual decomposition.}
	To describe the asymptotic behavior of the process $-\nabla^2Y_s$ nearby a selected peak $t \in \wh{T}_u$, we make use of the following decomposition of the negative Hessian of the noise:
	\begin{equation*}
		\begin{aligned}
			-\nabla^2 \epsilon_s =  
			-K_{20}(s,t) \epsilon_t  - K_{21}(s,t)(\Lambda_t^{-1} \nabla \epsilon_t) + R_s^{t}.
		\end{aligned}
	\end{equation*}
	By construction the residual $R_s^t$ is a symmetric, mean-zero Gaussian random matrix that is independent of $(\epsilon_t,\nabla \epsilon_t)$. The covariance between $R_s^{t},R_{s'}^t$ -- either marginally or conditionally on $(\epsilon_t,\nabla \epsilon_t)$ -- is an array $K_{22}^{t}(s,s') \in \R^{d \times d \times d \times d}$, having elements
	\begin{equation}
		\label{eqn:hessian-covariance}
		[K_{22}^{t}(s,s')]_{ijkl} := \Cov\big([R_s^t]_{ij},[R_{s'}^t]_{kl}\big) = [K_{22}(s,s')]_{ijkl} - [K_{20}(s,t)]_{ij} [K_{20}(s',t)]_{kl} - \sum_{m,o} [K_{21}(s,t)]_{klm} [\Lambda_t^{-1}]_{mo} [K_{12}(t,s')]_{oij}.
	\end{equation}
	Conditional on $\{\epsilon_t = y - \mu_t, \nabla \epsilon_t = -\nabla \mu_t\}$, we have that
	\begin{equation*}
		-\nabla^2 \epsilon_s \overset{d}{=} -K_{20}(s,t)(y - \mu_t) + K_{21}(s,t)(\Lambda_t^{-1} \nabla \mu_t) + R_s^{t},
	\end{equation*}
	and therefore conditional $\{Y_t = y, \nabla Y_t = 0\}$ we have that
	\begin{equation}
		\label{eqn:hessian-decomposition}
		-\nabla^2 Y_s \overset{d}{=} H_s^t(y) + R_s^{t}, \quad H_s^t(y) := -\nabla^2 \mu_s - K_{20}(s,t)(y - \mu_t) + K_{21}(s,t)(\Lambda_t^{-1} \nabla \mu_t), \quad R_s^{t} \sim N_{d \times d}(0,K_{22}^t(s,s)).
	\end{equation}
	We note that the pinned mean $H_t^t(y)$ is equal to the deterministic Hessian $H_{t|y}$ defined in~\eqref{eqn:deterministic-hessian}. Notationally, hereafter when evaluating the pinned mean $H_s^t(y)$ and residual $R_s^t$ at $t = s$, we will drop the superscript, writing $H_{t|y}$ for $H_t^{t}(y)$ and $R_t$ for $R_t^t$. 
	
	\paragraph{Asymptotic deterministic equivalent: signal region.}
	We now consider the Hessian $-\nabla^2 Y_s$ in the neighborhood of a true peak $t^*$. Recall the notation from the main text $\mc{B}_{t^*} := B_d(t^*,\varepsilon_n)$ and $\mc{I}_{t^*} := u_{t^*} \pm \Delta_n \cap (u,\infty)$. We will show that conditional on $\{Y_t = y,\nabla Y_t = 0\}$ for any $t \in \mc{B}_{t^*}, y \in \mc{I}_{t^*}$, with high probability the process $-\nabla^2 Y_{s}$ is uniformly close to $\bar{H}_{t^*}$ over all $s \in \mc{B}_{t^*}$. To show this we proceed from the projection/residual decomposition in~\eqref{eqn:hessian-decomposition}. We will show first that the relative error between the pinned mean $H_{s}^{t}(y)$ and the deterministic Hessian $\bar{H}_{t^*}$ is uniformly small. Then we will show that with high probability, uniformly over $s \in \mc{B}_{t^*}$, the residual Hessian $R_s^{t} \in \mc{R}_{t^*} := B_{d \times d}(0,\xi_n) \cap \R_{sym}^{d \times d}$, where $B_{d \times d}(0,r)$ denotes an $r$-ball in the set of symmetric matrices $\R_{sym}^{d \times d} = \{A \in \R^{d \times d}: A = A'\}$, the radius 
	\begin{equation*}
		\xi_n := \sigma_{22} \cdot \Delta_n \sqrt{d},
	\end{equation*}
	and the variance term is
	\begin{equation}
		\label{eqn:hessian-variance}
		\sigma_{22}^2 := \sup_{s,t \in \Rd} \sup_{x \in \S^{d - 1}} \Var[x'R_s^t x] < \infty.
	\end{equation}
	We begin with the pinned mean. Observe that $H_s^t(y)$ is Lipschitz in its lower argument; to bound its Lipschitz constant, we combine the following:
	\begin{align*}
		\sup_{s,t \in \mc{B}_{t^*}} \|\nabla^2 \mu_s - \nabla^2 \mu_t\| & \leq C\lambda_{t^*}\|s - t\| \\
		\sup_{s,t \in \mc{T}} \|K_{20}(s,t) - K_{20}(t,t) \| & \leq C \|s - t\| \\
		\sup_{s,t \in \mc{T}} \|K_{21}(s,t) - K_{21}(t,t)\| & \leq C \|s - t\|, 
	\end{align*}
	to conclude via the triangle inequality that
	\begin{equation*}
		\begin{aligned}
			& \sup_{s,t \in \mc{B}_{t^*}}\|H_s^t(y) - H_{t|y}\| \\
			& = \sup_{s,t \in \mc{B}_{t^*}} \|\nabla^2 \mu_t - \nabla^2 \mu_s + (y - \mu_t) (K_{20}(t,t) - K_{20}(s,t)) + (K_{21}(s,t) - K_{21}(t,t))\Lambda_t^{-1}\nabla \mu_t\| \\
			& \leq \sup_{s,t \in \mc{B}_{t^*}} C\Big(\lambda_{t^*}\|s - t\| + |y - \mu_t| \cdot \|s - t\| + \|\nabla \mu_t\| \cdot \|s - t\|\Big) \\
			& \overset{(i)}{\leq} \sup_{s,t \in \mc{B}_{t^*}} C\Big(\lambda_{t^*}\|s - t\| + \big(|y - \bar{u}_{t^*}| + |\bar{u}_{t^*} - \mu_{t^*}| + \lambda_{t^*}\|t - t^*\|\big)\|s - t\| + \lambda_{t^*} \cdot \|t - t^*\|\cdot\|s - t\|\Big) \\
			& \overset{(ii)}{\leq} C\Big(\lambda_{t^*} \varepsilon_n + \Delta_n \varepsilon_n + |\bar{u}_{t^*} - \mu_{t^*}| \varepsilon_n\Big) \\
			& \overset{(iii)}{\leq} C\Delta_n.
		\end{aligned}
	\end{equation*}
	Above $(i)$ uses the upper bounds on $|\mu_t - \mu_{t^*}|$ and $\|\nabla \mu_t\|$ stated in~\eqref{eqn:signal-1st-and-2nd-derivative-control}, in $(ii)$ and $(iii)$ we have absorbed asymptotically negligible terms into the constant $C$, and in $(iii)$ we observe that $\lambda_{t^*} \varepsilon_n \leq C \Delta_n$. Essentially the same analysis bounds the difference between $H_{t|y}$ and $\bar{H}_{t^*}$ uniformly over $t \in \mc{B}_{t^*}, y \in \mc{I}_{t^*}$:
	\begin{equation}
		\label{eqn:deterministic-hessian-relative-error}
		\sup_{t \in \mc{B}_{t^*}, y \in \mc{I}_{t^*}} \|H_{t|y} - \bar{H}_{t^*}\| \leq \sup_{t \in \mc{B}_{t^*}, y \in \mc{I}_{t^*}} C\Big(\lambda_{t^*}\|t - t^*\| + (|y - \bar{u}_{t^*}| + |\bar{u}_{t^*} - \mu_{t^*}|) \cdot \|t - t^*\| + \lambda_{t^*} \|t - t^*\|^2\Big) \leq C\Big(\lambda_{t^*} \varepsilon_n \Big) \leq C \Delta_n.
	\end{equation}
	To show that the residual $R_{s}^t$ is relatively small with high probability, we will apply the Borell-TIS inequality as stated in~\eqref{eqn:borell-tis-consequence}, which applies as $K_{22}^{t}(s,\cdot) \in C^1(\Rd)$ for all $s  \in \Rd$ by Assumption~\ref{asmp:covariance-holder}, and implies that for any constant $\eta > 0$, there exists a constant $C(\eta)$ such that 
	\begin{equation}
		\label{eqn:hessian-residual-tail-behavior}
		\begin{aligned}
			\P\Big(\sup_{s \in \mc{B}_{t^*}} \lambda_{\max}(R_s^t) \geq \xi_n\Big) 
			& \leq \P\Big(\sup_{s \in B(t^*,1)} \sup_{x \in \S^{d - 1}} x'R_s^t x \geq \xi_n\Big) \\
			& \leq C(\eta) \exp\Big(\eta \xi_n^2 - \frac{\xi_n^2}{2\sigma_{22}^2}\Big).
		\end{aligned}
	\end{equation}
	Combining our analyses of the pinned mean and pinned residual, we conclude the following: conditional on $\{Y_t = y,\nabla Y_t = 0\}$ for any $t \in \mc{B}_{t^*},y \in \mc{I}_{t^*}, \eta > 0$, 
	\begin{equation}
		\label{eqn:hessian-uniform-concentration}
		\sup_{s \in \mc{B}_{t^*}} \frac{\|-\nabla^2 Y_s - \bar{H}_{t^*}\|}{\|\bar{H}_{t^*}\|} \leq C\delta_{t^*}(\Delta_n + \xi_n) \leq C \varepsilon_n,
	\end{equation}
	with probability at least $1 - C(\eta) \exp(\eta \xi_n^2 - \xi_n^2/(2\sigma_{22}^2))$.
	
	\paragraph{First-order Taylor expansion of Hessian.}
	So far we have derived an upper bound on the difference between the observed Hessian $-\nabla^2 Y_t$ and the deterministic Hessian $\bar{H}_{t^*}$. To obtain second-order accurate approximations of the density $p(t,y)$, we will need to compute a first-order Taylor expansion of the observed Hessian about $\bar{H}_{t^*}$. Conditional on $\{Y_t = y,\nabla Y_t = 0,R_t = R\}$, the observed Hessian is equal to
	\begin{equation*}
		-\nabla^2 Y_t = H_{t|y} + R.
	\end{equation*}
	The first-order Taylor expansion of the right hand side about $(t = t^*, y = \bar{u}_{t^*},R = 0)$ is
	\begin{equation}
		\begin{aligned}
			\label{eqn:deterministic-hessian-taylor-expansion}
			\bar{H}_t(y,R) & := \bar{H}_{t^*} + T_{10}^{H}(t - t^*) + T_{01}^{H}(y - \bar{u}_{t^*}) + R, \\
			T_{10}^{H}(h) & = -\nabla^3 \mu_{t^*}(h) + (\bar{u}_{t^*} - \mu_{t^*})\dot{\Lambda}_{t^*}(h) - \Gamma_{t^*}(\nabla^2\mu_{t^*} (h)), \quad T_{01}^{H}(y - \bar{u}_{t^*}) = (y - \bar{u}_{t^*})\Lambda_{t^*}.
		\end{aligned}
	\end{equation}
	To upper bound the error incurred by Taylor expansion, we combine the following:
	\begin{equation*}
		\begin{aligned}
			\|\nabla^2 \mu_{t^* + h} - \nabla^2 \mu_{t^*} - \nabla^3 \mu_{t^*}(h)\| 
			& \leq C \lambda_{t^*} \|h\|^2, \\
			\|\Gamma_{t^* + h}(\nabla \mu_{t^* + h}) - \Gamma_{t^*}(\nabla^2\mu_{t^*} h)\| 
			& \leq C\lambda_{t^*}\|h\|^2, \\
			\|(y - \mu_{t^*})\Lambda_{t^*} + (\bar{u}_{t^*} - \mu_{t^*})\dot{\Lambda}_{t^*}(h) - (y - \mu_{t^* + h})\Lambda_{t^* + h}\| 
			& \leq C(|y - \bar{u}_{t^*}|\|h\| + |\bar{u}_{t^*} - \mu_{t^*}|\|h\|^2), 
		\end{aligned}
	\end{equation*}
	to conclude that the Taylor expansion remainder term is upper bounded by 
	\begin{equation}
		\label{eqn:deterministic-hessian-taylor-expansion-remainder}
		C\Big(\lambda_{t^*}\|t - t^*\|^2 + |y - \bar{u}_{t^*}|\|t - t^*\| + |\bar{u}_{t^*} - \mu_{t^*}|\|t - t^*\|^2\Big) \leq C\Big(\lambda_n\|t - t^*\|^2 + |y - \bar{u}_{t^*}|\|t - t^*\|\Big).
	\end{equation}
	The second inequality above follows from the upper bound $|\bar{u}_{t^*} - \mu_{t^*}| \leq C\lambda_n$ assumed in Assumption~\ref{asmp:hessian-curvatures}. Finally, notice that the first-order terms in the Taylor expansion have magnitude at most
	\begin{equation}
		\label{eqn:deterministic-hessian-taylor-expansion-first-order}
		\|T_{10}^{H}(h)\| \leq C \lambda_n \|h\| \leq C \sqrt{\log \lambda_n},   \quad \|T_{01}^{H}(y)\| \leq |y - \bar{u}_{t^*}| \leq \Delta_n \leq C \sqrt{\log \lambda_n}, \quad \|R\| \leq C \sqrt{\log \lambda_n},
	\end{equation}
	and so these terms are asymptotically negligible compared to the leading order term $\bar{H}_{t^*} \asymp \lambda_n$. Thus, for all $n \in \mathbb{N}$ sufficiently large that $C \sqrt{\log \lambda_n} < \lambda_n/4$, it follows that $\|\bar{H}_{t}(y,R) - \bar{H}_{t^*}\| \leq \lambda_n/2$ and so $\|\bar{H}_{t}(y,R)\| \geq \lambda_n/2$. We conclude that for all such $n \in \mathbb{N}$, the relative error in a first-order Taylor expansion of the observed Hessian is at most
	\begin{equation}
		\frac{\|H_{t|y} + R - \bar{H}_t(y,R)\|}{\|\bar{H}_t(y,R)\|} \leq C\big(\|h\|^2 + \delta_n |y - \bar{u}_{t^*}| \|h\|\big) := \Err_{H}(h,y).
	\end{equation}

	\paragraph{Local expansion of determinant.}
	Applying the results of the preceding paragraph leads to a local expansion of the determinant of the observed Hessian $\det(H_{t|y} + R)$ about $(t = t^*,y = \bar{u}_{t^*}, R = 0)$. Let $E_t(y) = H_{t|y}  +  R - \bar{H}_t(y,R)$, and suppose $n \in \mathbb{N}$ is large enough that $C(\delta_n^2 \log \lambda_n) \leq 1/4$, where $C$ is the constant from~\eqref{eqn:deterministic-hessian-taylor-expansion-remainder}, and $C\sqrt{\log \lambda_n} \leq \lambda_n/4$. It follows that $\|E_{t}(y)\|/\|\bar{H}_t(y,R)\| \leq 1/2$. As a result we can apply~\eqref{eqn:matrix-det-taylor-expansion} to conclude that
	\begin{equation*}
		\begin{aligned}
			\Big|\det(H_{t|y} + R) - \det(\bar{H}_{t}(y,R))\Big| \leq C \det(\bar{H}_{t}(y,R)) \cdot \frac{\|E_{t}(y)\|}{\lambda_{\min}(\bar{H}_{t}(y,R))} \leq C \det(\bar{H}_{t^*})  \cdot \Err_H(h,y).
		\end{aligned}
	\end{equation*}
	Now let's express $\bar{H}_t(y,R) = \bar{H}_{t^*} + \bar{E}(h,y,R)$, where $\bar{E}(h,y,R) = T_{10}^{H}(h) + T_{01}^{H}(y) + R$. As stated earlier if $n$ is large enough so that $C\sqrt{\log \lambda_n} \leq \lambda_n/4$, then $\|\bar{E}(h,y,R)\| \leq \|H_{t^*}\|/2$. We may therefore apply~\eqref{eqn:matrix-det-taylor-expansion} again, concluding that 
	\begin{equation*}
		\begin{aligned}
			\Big|\det(\bar{H}_{t}(y,R)) - \det\big(\bar{H}_{t^*})(1 + \tr(\bar{H}_{t^*}^{-1}\bar{E}(h,y,R)\big)\Big|
			& \leq C \det(\bar{H}_{t^*}) \cdot \frac{\|\bar{E}(h,y,R)\|^2}{\lambda_{\min}(\bar{H}_{t^*})^2} \\
			& \leq C  \det(\bar{H}_{t^*}) (\|h\|^2 + |y - \bar{u}_{t^*}|^2\delta_n^2 + \|R\|^2 \delta_n^2).
		\end{aligned}
	\end{equation*}
	In the latter inequality we have inserted the upper bounds in~\eqref{eqn:deterministic-hessian-taylor-expansion-first-order}. Additionally, recognize that $\tr(\bar{H}_{t^*}^{-1}\bar{E}(h,y,R)) = T_{10}^{\det}(h) + T_{01}^{\det}(y) + \tr(\bar{H}_{t^*}^{-1}R)$.  In summary, we have shown that
	\begin{equation}
		\label{eqn:hessian-determinant-local-expansion}
		\begin{aligned}
			& \Big|\det(H_{t^* + h|y} + R) - \det(\bar{H}_{t^*})\big(1 + T_{10}^{\det}(h) + T_{01}^{\det}(y) + \tr(\bar{H}_{t^*}^{-1}R)\big)\Big| \\
			& \quad \leq C\det(\bar{H}_{t^*})\Big(\|h\|^2 + |y - \bar{u}_{t^*}|^2 \delta_n^2 + \|R\|^2\delta_n^2\Big) := \Err_{\det}(h,y,R).
		\end{aligned}
	\end{equation}
	We will use this local expansion in the proof of Lemma~\ref{lem:approximation-peak-intensity}.
	
	\subsection{Proof of Theorem~\ref{thm:approximate-joint-intensity}: Expansions of determinant and density terms}
	\label{subsec:pf-approximate-joint-intensity}
	We are now ready to develop the local expansions of the determinant and density terms in the Kac-Rice formula for $\rho(t,y)$, which we copy here for convenience:
	\begin{equation*}
		\rho(t,y) = \E\Big[\det(-\nabla^2 Y_t) \cdot \1(\nabla^2 Y_t \prec 0)|Y_t = y,\nabla Y_t = 0\Big] \cdot f_{Y_t,\nabla Y_t}(y,0)
	\end{equation*}
	\paragraph{Determinant term.}
	The expectation of the determinant of a Gaussian random matrix, multiplied by the indicator that the matrix is negative definite, is in general difficult to compute. However, in our high-curvature asymptotic setup, at points $t$ nearby true peaks $t^*$ the probability that $\nabla^2 Y_t$ is negative definite is exponentially close to $1$. We are thus left with the conditional expectation of a determinant of a Gaussian matrix which can be effectively approximated by local expansion. To concisely write the first-order terms in this expansion, we recall notation from the main text:
	\begin{equation*}
		T_{10}^{\det}(h) = \tr\big(\bar{H}_{t^*}^{-1} \dot{H}_{t^*|\bar{u}_{t^*}}(h)\big), \quad T_{01}^{\det}(y) = (y - \bar{u}_{t^*})\tr\big(\bar{H}_{t^*}^{-1}\Lambda_{t^*}),
	\end{equation*}
	where $\dot{H}_{t^*|\bar{u}_{t^*}} \in \R^{d \times d \times d}$ is the derivative of the deterministic Hessian, 
	\begin{equation}
		\label{eqn:hessian-deterministic-equivalent-derivative}
		\dot{H}_{t^*|\bar{u}_{t^*}}(h) = -\nabla^3 \mu_{t^*}(h) + (\bar{u}_{t^*} - \mu_{t^*})\cdot \dot{\Lambda}_{t^*}(h) + \Gamma_{t^*}(\nabla^2 \mu_{t^*}h).
	\end{equation}
	\begin{lemma}
		\label{lem:approximation-peak-intensity}
		Under the assumptions of Theorem~\ref{thm:approximate-joint-intensity}, at points $t \in \mc{B}_{t^*}, y \in \mc{I}_{t^*}$, writing $h = t - t^*$:
		\begin{equation}
			\begin{aligned}
				\label{eqn:approximate-intensity-1}
				& \frac{\Big|\E[\det(-\nabla^2 Y_t) \cdot \1(\nabla^2 Y_t \prec 0)|Y_t = y,\nabla Y_t = 0] - \det(\bar{H}_{t^*})\big(1 + T_{10}^{\det}(h) + T_{01}^{\det}(y)\big)\Big|}{\det(\bar{H}_{t^*})} \leq \Err_{\det}(h,y),
			\end{aligned}
		\end{equation}
		where
		\begin{equation*}
			\Err_{\det}(h,y) := C\Big(\big(|y - \bar{u}_{t^*}|^2 + 1\big)\delta_n^2 + \|h\|^2\Big).
		\end{equation*}
	\end{lemma}
	The proof of Lemma~\ref{lem:approximation-peak-intensity} is given in Section~\ref{subsec:pf-approximation-peak-intensity}. The linear terms and error term have respective magnitude
	\begin{equation}
		\label{eqn:determinant-first-order}
		\sup_{t \in \mc{B}_{t^*}} T_{10}^{\det}(h) \leq C \varepsilon_n, 
		\quad \sup_{y \in \mc{I}_{t^*}} T_{01}^{\det}(y) \leq C\varepsilon_n, 
		\quad \sup_{t \in \mc{B}_{t^*}, y \in \mc{I}_{t^*}} \Err_{\det}(h,y) \leq C\varepsilon_n^2.
	\end{equation}
	This is why we refer to $T_{10}^{\det}(h), T_{01}^{\det}(y)$ as first-order terms, while the error is only second-order.
	
	\paragraph{Density term.}
	Standard asymptotic analysis (carried out in Section~\ref{subsec:asymptotic-analysis-density}) shows that under Assumption~\ref{asmp:well-conditioned}, at points $t \in \mc{B}_{t^*}, y \in \mc{I}_{t^*}$, the density of both the height and gradient are locally well-approximated by Taylor expansion around $t = t^*, y = \bar{u}_{t^*}$. We again recall the notation introduced in the main text for the first-order terms in this expansion:
	\begin{equation*}
		\begin{aligned}
			T_{21}^{Y}(h,y) 
			& := 
			\frac{1}{2}(y - \bar{u}_{t^*})h'\nabla^2\mu_{t^*}h \\
			T_{10}^{\nabla Y}(h)
			& := -\frac{1}{2}\tr\big(\Lambda_{t^*}^{-1}\dot{\Lambda}_{t^*}(h)\big), \quad
			\\
			T_{30}^{\nabla Y}(h)
			& := 
			h' \nabla^2 \mu_{t^*} \Lambda_{t^*}^{-1} \{\nabla^3\mu_{t^*}(h,h)\} -  h'\nabla^2 \mu_{t^*} \Lambda_{t^*}^{-1}\{\dot{\Lambda}_{t^*}(h)\}\Lambda_{t^*}^{-1} h
		\end{aligned}
	\end{equation*}
	\begin{lemma}
		\label{lem:asymptotic-expansion-density}
		Under the assumptions of Theorem~\ref{thm:approximate-joint-intensity}, at points $t \in \mc{B}_{t^*}, y \in \mc{I}_{t^*}$, writing $h = t - t^*$:
		\begin{equation}
			\begin{aligned}
				\frac{\Big|f_{Y_t}(y) - \bar{f}_{Y_t}(y)\Big|}{\bar{f}_{Y_t}(y)} & \leq C\Big(\big(|y - \bar{u}_{t^*}| + |\bar{u}_{t^*} - \mu_{t^*}|\big) \lambda_n \|h\|^3 + \lambda_n^2\|h\|^4 + |y - \bar{u}_{t^*}|^2 \lambda_n^2 \|h\|^4\Big) := \Err_{Y}(h,y),
			\end{aligned}
		\end{equation}
		and
		\begin{equation}
			\label{eqn:approximation-density-gradient}
			\begin{aligned}
				\frac{\Big|f_{\nabla Y_t}(0) - \bar{f}_{\nabla Y_t}(0)\Big|}{\bar{f}_{\nabla Y_t}(0)} & \leq C\Big(\|h\|^2 + \|h\|^4 \lambda_n^2\Big) := \Err_{f_{\nabla Y}}(h),
			\end{aligned}
		\end{equation}
		where
		\begin{equation}
			\label{eqn:asymptotic-density-height}
			\bar{f}_{Y_t}(y) := \frac{(1 + T_{21}^{Y}(h,y))}{\sqrt{2\pi}}\exp\Big(-\frac{1}{2}\Big\{(y - \mu_{t^*})^2 - (\bar{u}_{t^*} - \mu_{t^*})h'\nabla^2 \mu_{t^*}h\Big\} \Big),
		\end{equation}
		and
		\begin{equation}
			\label{eqn:asymptotic-density-gradient}
			\bar{f}_{\nabla Y_t}(0) := \Big(1 - \frac{T_{30}^{\nabla Y}(h)}{2} + T_{10}^{\nabla Y}(h)\Big)\frac{1}{\sqrt{(2\pi)^d\det(\Lambda_{t^*})}}\exp\bigg(-\frac{1}{2}h' \nabla^2 \mu_{t^*} \Lambda_{t^*}^{-1}  \nabla^2 \mu_{t^*} h\bigg), 
		\end{equation}
	\end{lemma}
	The proof of Lemma~\ref{lem:asymptotic-expansion-density} is given in Section~\ref{subsec:asymptotic-analysis-density}. Theorem~\ref{thm:approximate-joint-intensity} is essentially an immediate consequence of Lemmas~\ref{lem:approximation-peak-intensity} and~\ref{lem:asymptotic-expansion-density}; for completeness, we include a proof in Section~\ref{subsec:pf-approximate-joint-intensity}.
	

	\subsection{Proof of Lemma~\ref{lem:approximation-peak-intensity}}
	\label{subsec:pf-approximation-peak-intensity}
	We begin by rewriting the determinant term using the decomposition in~\eqref{eqn:hessian-decomposition}:
	\begin{equation}
		\begin{aligned}
			\E\Big[\det(-\nabla^2 Y_t) \cdot \1(\nabla^2 Y_t \prec 0)|Y_t = y, \nabla Y_t = 0\Big] 
			& = \int_{\R_{d \times d}^{sym}} \det(H_{t|y} + R) \cdot \1(H_{t|y} + R \succ 0) \cdot f_{R_t}(R) \,dR,
		\end{aligned}
	\end{equation}
	where $f_{R_t}(R)$ is the Gaussian density of the pinned residual $R_t$ defined in~\eqref{eqn:hessian-covariance}. We derive a local expansion of this integral via an approach that we will use repeatedly in the proofs to come: truncate the integral to a range at which a local expansion of the integrand is valid, apply this local expansion, then ``undo'' the truncation. This yields 
	\begin{equation*}
		\begin{aligned}
			& \int \det(H_{t|y} + R) \cdot \1(H_{t|y} + R \succ 0) \cdot f_{R_t}(R) \,dR \\
			& \quad = \det(\bar{H}_{t^*})(1 + T_{10}^{\det}(h) + T_{01}^{\det}(y)) \\
			& \quad - \int_{\mc{R}_{t^*}^c}  \det(\bar{H}_{t^*})(1 + T_{10}^{\det}(h) + T_{01}^{\det}(y)) \cdot f_{R_t}(R) \,dR \\
			& \quad + \int_{\mc{R}_{t^*}} \Big(\det(H_{t|y} + R)  \1(H_{t|y} + R \succ 0) -\det(\bar{H}_{t^*})\big(1 + T_{10}^{\det}(h) + T_{01}^{\det}(y)\big) \Big) \cdot f_{R_t}(R) \,dR  \\
			& \quad + \int_{\mc{R}_{t^*}^c} \det(H_{t|y} + R) \cdot \1(H_{t|y} + R \succ 0) \cdot f_{R_t}(R) \,dR.
		\end{aligned}
	\end{equation*}
	The first term on the right hand side of the equality above is the claimed local expansion, and so what remains is to bound the various sources of error, due to local expansion, truncation, and undoing the truncation.
	
	\paragraph{Error due to local expansion.}
	For all $n$ large enough such that $C\sqrt{\log \lambda_n} < \lambda_n$, the positive definite requirement $H_{t|y} + R \succ 0$ is automatically satisfied for all $t \in \mc{B}_{t^*},y \in \mc{I}_{t^*}, R \in \mc{R}_{t^*}$. (Recall that $\mc{R}_{t^*} = B_{d \times d}(0,\xi_n) \cap \R_{sym}^{d \times d}$). Moreover, $\mc{R}_{t^*}$ is symmetric about $0$ and $\tr(\bar{H}_{t^*}^{-1}R)$ is an odd function of $R$, so $\int_{\mc{R}_{t^*}} \tr(\bar{H}_{t^*}^{-1}R) f_{R_t}(R) \,dR = 0$. Thus the error due to local expansion is 
	\begin{equation}
		\begin{aligned}
			& \int_{\mc{R}_{t^*}} \Big(\det(H_{t|y} + R)  \1(H_{t|y} + R \succ 0) -\det(\bar{H}_{t^*})(1 + T_{10}^{\det}(h) + T_{01}^{\det}(y)\big) \Big) \cdot f_{R_t}(R) \,dR \\
			&\quad = \int_{\mc{R}_{t^*}} \Big(\det(H_{t|y} + R) -\det(\bar{H}_{t^*})\big(1 + T_{10}^{\det}(h) + T_{01}^{\det}(y) + \tr(\bar{H}_{t^*}^{-1}R)\big) \Big) \cdot f_{R_t}(R) \,dR \\
			& \quad \leq \int_{\mc{R}_{t^*}} \Err_{\det}(t,y,R) f_{R_t}(R) \,dR \\
			& \quad = \int_{\mc{R}_{t^*}} C\det(\bar{H}_{t^*})\Big(\|h\|^2 + |y - \bar{u}_{t^*}|^2\delta_n^2 + \delta_n^2\|R\|^2\Big) f_{R_t}(R) \,dR \\
			& \quad \leq C\det(\bar{H}_{t^*})\Big(\|h\|^2 + |y - \bar{u}_{t^*}|^2\delta_n^2 + \delta_n^2\Big) := \Err_{\det}(h,y,\mc{R}_{t^*}).
		\end{aligned}
	\end{equation}
	The final inequality follows since $f_{R_t}(R)$ is a probability density of a mean-zero random matrix, with variance $\Var([R_t]_{ij}) \leq \sigma_{22}^2$ for all $i,j \in \{1,\ldots,d\}$. 
	
	\paragraph{Truncation error.}
	The error due to truncation and the error due to undoing the truncation can both be bounded using the moderate-deviation inequality~\eqref{eqn:hessian-residual-tail-behavior}. The error due to undoing the truncation is at most
	\begin{align*}
		& \int_{\mc{R}_{t^*}^c}  \det(\bar{H}_{t^*})(1 + T_{10}^{\det}(h) + T_{01}^{\det}(y)) \cdot f_{R_t}(R) \,dR \\
		& \quad = \det(\bar{H}_{t^*})(1 + T_{10}^{\det}(h) + T_{01}^{\det}(y)) \cdot \P\Big(\|R_t\| \geq \xi_n\Big) \\
		& \quad \overset{(i)}{\leq} C \det(\bar{H}_{t^*}) \cdot \P\Big(\|R_t\| \geq \xi_n\Big) \\
		& \quad \leq C \det(\bar{H}_{t^*}) \cdot \P\Big(\lambda_{\max}(R_t) \geq \frac{1}{d}\xi_n\Big) \\
		& \quad \overset{(ii)}{\leq} C \det(\bar{H}_{t^*}) \cdot \exp\Big(-\frac{\xi_n^2}{3d\sigma_{22}^2}\Big) \\
		& \quad = C \det(\bar{H}_{t^*}) \delta_n^2 := \wb{\Trunc}_{\det},
	\end{align*}
	where~$(i)$ follows from the bounds on $T_{10}^{\det}(h), T_{01}^{\det}(y)$ set out in~\eqref{eqn:determinant-first-order}, and $(ii)$ invokes~\eqref{eqn:hessian-residual-tail-behavior} with a proper choice of constant $\eta$. A similar analysis takes care of the first truncation error term, with the slight complication that we need to properly account for the (polynomial) growth of the determinant term:
	\begin{equation*}
		\begin{aligned}
			\int_{\mc{R}_{t^*}^c} \det(H_{t|y} + R) f_{R_t}(R) \,dR 
			& \overset{(i)}{\leq} \Big(\int_{\mc{R}_{t^*}^c} |\det(H_{t|y} + R)|^{q_1} f_{R_t}(R) \,dR\Big)^{1/q_1} \cdot \Big(\P(\|R_t\| \geq \xi_n)\Big)^{1/q_2} \\
			& \leq \Big(\int_{\R_{sym}^{d \times d}} |\det(H_{t|y} + R)|^{q_1} f_{R_t}(R) \,dR\Big)^{1/q_1} \cdot \Big(\P(\|R_t\| \geq \xi_n)\Big)^{1/q_2} \\
			& \leq \Big(\int_{\R_{sym}^{d \times d}} (\|H_{t|y}\|^{d} + \|R\|^{d})^{q_1} f_{R_t}(R) \,dR\Big)^{1/q_1} \cdot \Big(\P(\|R_t\| \geq \xi_n)\Big)^{1/q_2} \\
			& \overset{(ii)}{\leq} C\Big(\int_{\R_{sym}^{d \times d}} (\lambda_n^{d q_1} + \|R\|^{d q_1}) f_{R_t}(R) \,dR\Big)^{1/q_1} \cdot \Big(\P(\|R_t\| \geq \xi_n)\Big)^{1/q_2} \\
			& \overset{(iii)}{\leq} C\Big(\lambda_n^{dq_1}+ 1\Big)^{1/q_1} \cdot \Big(\P(\|R_t\| \geq \xi_n)\Big)^{1/q_2} \\
			& \leq C\det(\bar{H}_{t^*}) \cdot \Big(\P(\|R_t\| \geq \xi_n)\Big)^{1/q_2} \\
			& \overset{(iv)}{\leq} C \det(\bar{H}_{t^*}) \delta_n^2 := \Trunc_{\det}.
		\end{aligned}
	\end{equation*}
	Above $(i)$ is H\"{o}lder's inequality, and holds for any conjugate exponents $1/q_1 + 1/q_2 = 1$; $(ii)$ relies on~\eqref{eqn:deterministic-hessian-relative-error} to upper bound $\|H_{t|y}\|^p \leq C \|\bar{H}_{t^*}\|^p$ for all $n$ sufficiently large and $(t,y) \in \mc{S}_{t^*}$; $(iii)$ follows because $f_{R_t}(R)$ is the probability density of a mean-zero Gaussian random matrix with variance $\Var([R_t]_{ij}) \leq \sigma_{22}^2$ for all $i,j \in \{1,\ldots,d\}$; and $(iv)$ follows from a sufficiently small (but constant) choice of conjugate exponent $q_2$. 
	
	Summing up the three sources of error gives $\Err_{\det}(h,y)$, and completes the proof of Lemma~\ref{lem:approximation-peak-intensity}.
	
	\subsection{Proof of Lemma~\ref{lem:asymptotic-expansion-density}}
	\label{subsec:asymptotic-analysis-density}
	
	We separately derive~\eqref{eqn:asymptotic-density-height} and~\eqref{eqn:asymptotic-density-gradient}.
	
	\paragraph{Density of height.}
	The density of the height is
	\begin{equation}
		\label{eqn:density-height}
		f_{Y_t}(y) = \frac{1}{\sqrt{2\pi}} \cdot \exp\Big(-\frac{1}{2} (y - \mu_t)^2\Big).
	\end{equation}
	We substitute the second-order Taylor expansion of $\mu_t$ about $t = t^*$ into the quadratic, incurring error
	\begin{equation*}
		\Big|\big(y - \mu_t\big)^2 - (y - \mu_{t^*} + \frac{1}{2}h'\nabla^2\mu_{t^*} h)^2\Big| \leq C\Big(\lambda_n^2\|h\|^5 + |y - \mu_{t^*}| \lambda_n\|h\|^3\Big).
	\end{equation*}
	Expanding the square gives
	\begin{equation*}
		\Big|(y - \mu_{t^*}- \frac{1}{2}h'\nabla^2\mu_{t^*} h)^2  - (y - \mu_{t^*})^2 - (y - \mu_{t^*})h'\nabla^2\mu_{t^*} h\Big| \leq C\Big(\lambda_n^2\|h\|^4\Big),
	\end{equation*}
	so that by the triangle inequality
	\begin{equation*}
		\Big|\big(y - \mu_t\big)^2 - (y - \mu_{t^*})^2 - (y - \mu_{t^*})h'\nabla^2\mu_{t^*} h\Big| \leq C\Big(\lambda_n^2\|h\|^4 + |y - \mu_{t^*}| \lambda_n\|h\|^3\Big).
	\end{equation*}
	The right hand side of the above equation is at most $C\delta_n^2(\log \lambda_n)^2$, which is less than $1/2$ for all $n$ sufficiently large. In that case, exponentiating both sides of the above inequality and invoking the bound $|\exp(x) - 1| \leq \exp(1/2) |x|$ for all $|x| \leq 1/2$ yields
	\begin{equation*}
		\frac{f_{Y_t}(y) - \tilde{f}_{Y_t}(y)}{\tilde{f}_{Y_t}(y)} \leq C(\lambda_n^2\|h\|^4 + |y - \mu_{t^*}| \lambda_n\|h\|^3),
	\end{equation*}
	where 
	\begin{equation*}
		\begin{aligned}
		\tilde{f}_{Y_t}(y) 
		& := \frac{1}{\sqrt{2\pi}} \exp\Big(-\frac{1}{2}(y - \mu_{t^*})^2 + \frac{1}{2}(y - \mu_{t^*})h'\nabla^2\mu_{t^*}h\Big) \\
		& = \frac{1}{\sqrt{2\pi}} \exp\Big(-\frac{1}{2}(y - \mu_{t^*})^2 +  \frac{1}{2}(\bar{u}_{t^*} - \mu_{t^*})h'\nabla^2\mu_{t^*}h\Big) \exp\Big(\frac{1}{2}(y - \bar{u}_{t^*})h'\nabla^2\mu_{t^*}h\Big).
		\end{aligned}
	\end{equation*}
	Finally, invoking the bound 
	\begin{equation*}
	\Big|\exp\Big(\frac{1}{2}(y - \bar{u}_{t^*})h'\nabla^2\mu_{t^*}h\Big) - \big(1 + \frac{1}{2}(y  - \bar{u}_{t^*})h'\nabla^2\mu_{t^*}h\big) \Big| \leq C |y - \bar{u}_{t^*}|^2 \lambda_n^2 \|h\|^4,
	\end{equation*}
	again valid for all $n$ sufficiently large, gives the desired approximation $\bar{f}_{Y_t}(y)$.
	
	\paragraph{Density of gradient.}
	The density of the gradient $\nabla Y_t$ evaluated at $0$ is
	\begin{equation}
		\label{eqn:density-gradient}
		f_{\nabla Y_t}(0) = \frac{1}{\sqrt{(2\pi)^d\det( \Lambda_t)}} \exp\Big(-\frac{1}{2} \nabla \mu_t' \Lambda_t^{-1}\nabla\mu_t\Big).
	\end{equation}
	A first-order Taylor expansion of $\Lambda_t^{-1}$ about $t = t^*$ implies
	\begin{equation*}
		\|\Lambda_{t^* + h}^{-1} - \Lambda_{t^*}^{-1} +  \Lambda_{t^*}^{-1}\{\dot{\Lambda}_{t^*}(h)\}\Lambda_{t^*}^{-1}\| \leq C\|h\|^2,
	\end{equation*}
	where the bound on the remainder term holds as~\eqref{eqn:non-degenerate} implies that $\|\Lambda_{t}^{-1}\|$ is uniformly bounded away from $0$. Inserting this along with a second-order Taylor expansion of $\nabla \mu_{t}$ about $t = t^*$ into the quadratic, we obtain
	\begin{equation*}
		\Big|\nabla \mu_t' \Lambda_t^{-1}\nabla\mu_t - \big(\nabla^2\mu_{t^*} h + \frac{1}{2}\nabla^3\mu(h,h)\big)'\big(\Lambda_{t^*}^{-1} - \Lambda_{t^*}^{-1}\{\dot{\Lambda}_{t^*}(h)\}\Lambda_{t^*}^{-1}\big)\big(\nabla^2\mu_{t^*} h + \frac{1}{2}\nabla^3\mu(h,h)\big)\Big| \leq C\big(\lambda_n^2\|h\|^4\big).
	\end{equation*}
	Expanding the resulting quadratic gives
	\begin{equation*}
		\Big|(\nabla^2\mu_{t^*} h + \frac{1}{2}\nabla^3\mu(h,h))'(\Lambda_{t^*}^{-1} - \Lambda_{t^*}^{-1}\{\dot{\Lambda}_{t^*}(h)\}\Lambda_{t^*}^{-1})(\nabla^2\mu_{t^*} h + \frac{1}{2}\nabla^3\mu(h,h)) - \Big(h'\nabla^2\mu_{t^*} \Lambda_{t^*}^{-1}\nabla^2 \mu_{t^*}h + T_{30}^{\nabla Y}(h)\Big)\Big| \leq C\Big(\lambda_n^2\|h\|^4\Big),
	\end{equation*}
	and so by the triangle inequality
	\begin{equation*}
		\Big|\nabla \mu_t' \Lambda_t^{-1}\nabla\mu_t - \Big(h'\nabla^2\mu_{t^*} \Lambda_{t^*}^{-1}\nabla^2 \mu_{t^*}h + T_{30}^{\nabla Y}(h)\Big)\Big| \leq C\Big(\lambda_n^2\|h\|^4\Big).
	\end{equation*}
	On the other hand, a first-order Taylor expansion of $\{\det( \Lambda_t)\}^{-1/2}$ about $t = t^*$ implies
	\begin{equation*}
		\Big|\frac{1}{\sqrt{\det( \Lambda_t)}} - \frac{1}{\sqrt{\det(\Lambda_{t^*})}}(1 - \frac{1}{2} \tr(\Lambda_{t^*}^{-1}\dot{\Lambda}_{t^*}(h)))\Big| \leq C\|h\|^2,
	\end{equation*}
	where~\eqref{eqn:non-degenerate} implies that $\det( \Lambda_t)$ is uniformly bounded away from $0$. We have thus established an approximation to the log of $f_{\nabla Y_t}(0)$; from here, analysis similar to that used in tackling the density of the height yields the claimed result~\eqref{eqn:asymptotic-density-gradient}.
	
	\subsection{Finishing off the proof of Theorem~\ref{thm:approximate-joint-intensity}}
	\label{subsec:pf-approximate-joint-intensity-2}
	After applying Lemmas~\ref{lem:approximation-peak-intensity} and~\ref{lem:asymptotic-expansion-density}, some routine calculations yield the upper bound
	\begin{equation*}
		\frac{\rho(t,y) - \bar{\rho}(t,y)}{\bar{\rho}(t,y)} \leq C\Big(\{T_{01}^{\rho}(y)\}^2 + \{T_{10}^{\rho}(h)\}^2 + \{T_{30}^{\rho}(h)\}^2 + \{T_{21}^{\rho}(h,y)\}^2 + \Err_{\det}(h,y) + \Err_{Y}(h,y) + \Err_{{\nabla Y}}(h)\Big). 
	\end{equation*}
	As one would anticipate, each of the squared first-order terms above is second-order:
	\begin{align*}
		\{T_{01}^{\rho}(y)\}^2 & \leq C |y - \bar{u}_{t^*}|^2 \delta_n^2 \\
		\{T_{10}^{\rho}(h)\}^2 & \leq C\|h\|^2 \\
		\{T_{30}^{\rho}(h)\}^2 & \leq C\lambda_n^4\|h\|^6 \\
		\{T_{21}^{\rho}(h,y)\}^2 & \leq C|y - \bar{u}_{t^*}|^2 \lambda_n^2\|h\|^4,
	\end{align*}
	all of which are at most a constant times $\Err_{\det}(h,y) + \Err_{\nabla Y}(h)$. Thus we have shown that $\frac{\rho(t,y) - \bar{\rho}(t,y)}{\bar{\rho}(t,y)}$ is upper bounded by $\Err_{\rho}(h,y)$, and the corresponding lower bound follows similarly.

	\section{Proofs for Section~\ref{sec:peak-estimation}}
	\label{sec:pf-peak-estimation}

	\subsection{Proof of Proposition~\ref{prop:null-false-positive-rate}}
	\label{subsec:pf-null-false-positive-rate}
	We begin by deriving a second-order accurate Truncated Gaussian approximation to the Palm density defined in~\eqref{eqn:density-height-null}. The approximation is derived by local expansion, and will be valid at points $y \in v \pm \sqrt{3 \log v}$ for all $v \in \R$ sufficiently large. (Here and henceforth, we use the phrase ``for all $v$ sufficiently large'' to mean ``for all $v \geq C$'', where as always $C$ is a constant that depends only on the covariance kernel $K$, dimension $d$ and level $\alpha$.)
	\begin{lemma}
		\label{lem:density-height-null}
		For all $v$ sufficiently large, for any $t \in \mc{T}_0$ and $y \in v \pm \sqrt{3 \log v}$,
		\begin{equation}
			\frac{|\bar{p}_v(y|t) - p_v(y|t)|}{\bar{p}_v(y|t)} \leq C\Big(\frac{(y - v)^2}{v^2} + \frac{1}{v^2}\Big) := \Err_{p}(y|t)
		\end{equation}
		where
		\begin{equation}
			\bar{p}_v(y|t) = \frac{1}{\Psi(v - d/v)} \cdot \frac{1}{\sqrt{2\pi}}\exp\Big(-\frac{1}{2}(y - d/v)^2\Big).
		\end{equation}
	\end{lemma}
	The proof of Lemma~\ref{lem:density-height-null} is given in Section~\ref{subsec:pf-density-height-null}.
	The survival function corresponding to this density is $\P(Y_t > u|t \in \wh{T}_v) = \int_{u}^{\infty} p_v(y|t) \,dy$. Using Lemma~\ref{lem:density-height-null} and a truncation argument it can be shown that this survival function is close to a Truncated Gaussian survival function with mean parameter $d/v$.
	\begin{lemma}
		\label{lem:survival-height-null}
		For any constant $B > 0$, the following statement holds: for all $v$ sufficiently large, for any $v \leq u \leq v + B/v$,
		\begin{equation*}
			\sup_{t \in \mc{T}_0}\frac{\Big|\P(Y_t > u|t \in \wh{T}_v) - \S_{v}(u,d/v)\Big|}{\S_v(u,d/v)} \leq \frac{C}{v^2},
		\end{equation*}
		where $C$ depends on $B$. 
	\end{lemma}
	The proof of Lemma~\ref{lem:survival-height-null} is given in Section~\ref{subsec:pf-survival-height-null}. We assume an upper bound on $u$ in Lemma~\ref{lem:survival-height-null} in order to get an upper bound on the relative, rather than absolute, error. It is likely that the result can be strengthened to hold for a wider range of $u$ -- thus allowing $\alpha \to 0$ in Proposition~\ref{prop:null-false-positive-rate} --  but we will not pursue this further.  
	
	The final step in our analysis is to use Lemma~\ref{lem:survival-height-null} to show that the null false positive rate is close to $\alpha$. This is straightforward. The expected number of null discoveries made by the TG test is
	\begin{equation}
		\label{pf:null-false-positive-rate-1}
		\E[N_{u_{\TG}(\alpha,v)}(\mc{T}_0)] = \int_{\mc{T}_0} \int_{u_{\TG}(\alpha,v)}^{\infty} \rho_v(t,y) \,dy \,dt =  \int_{\mc{T}_0} \Big\{\int_{u_{\TG}(\alpha,v)}^{\infty} p_v(y|t) \,dy \Big\} \cdot \int_{v}^{\infty} \rho_v(t,y) \,dy \,dt.
	\end{equation}
	For any fixed $\alpha \in (0,1)$, let $B_{\alpha}$ be the twice the $(1 - \alpha)$th quantile of the $\Exp(1)$ distribution. For all $n \in \mathbb{N}$ sufficiently large $u_{\TG}(\alpha,v_n) \leq v_n + \frac{B_{\alpha}}{v_n}$. We can therefore apply Lemma~\ref{lem:survival-height-null} to conclude that 
	\begin{equation*}
		\Big|\int_{u_{\TG}(\alpha,v_n)}^{\infty} p_v(y|t) \,dy - \int_{u_{\TG}(\alpha,v_n)}^{\infty} \bar{p}_v(y|t) \,dy \Big| \leq C \frac{\S_{v_n}(u_{\TG}(\alpha,v_n),d/v_n)}{v_n^2} \leq \frac{C}{v_n^2}.
	\end{equation*}
	Replacing $p_v(y|t)$ by $\bar{p}_v(y|t)$ in~\eqref{pf:null-false-positive-rate-1}, and recognizing that $\int_{u_{\TG}(\alpha,v)}^{\infty} \bar{p}_v(y|t) \,dy = \S_{v}(u_{\TG}(\alpha,v),d/v) = \alpha$ and $\int_{\mc{T}_0}\int_{v}^{\infty} \rho_v(t,y) \,dy \,dt = \E[N_v(\mc{T}_0)]$ , we deduce that
	\begin{equation*}
		\Big|\E[N_{u_{\TG}(\alpha,v_n)}(\mc{T}_0)] -  \alpha \cdot \E[N_{v_n}(\mc{T}_0)] \Big| \leq \frac{C}{v_n^2} \E[N_{v_n}(\mc{T}_0)].
	\end{equation*}
	After dividing both sides of the previous inequality by $\E[N_{v_n}(\mc{T}_0)]$, and recognizing that $\PCER_0(u,v) = \E[N_u(\mc{T}_0)]/N_{v}(\mc{T})] \leq \E[N_u(\mc{T}_0)]/N_{v}(\mc{T}_0)] $, we arrive at the claim of Proposition~\ref{prop:null-false-positive-rate}.
	
	\subsection{Proof of Lemma~\ref{lem:density-height-null}}
	
	\label{subsec:pf-density-height-null}
	Recall the Palm density defined in~\eqref{eqn:density-height-null},
	\begin{equation*}
		p_v(y|t) = \frac{\E[\det(-\nabla^2 Y_t) \cdot \1(\nabla^2 Y_t \prec 0)|\nabla Y_t = 0, Y_t = y]}{\E[\det(-\nabla^2 Y_t) \cdot \1(Y_t > v, \nabla^2 Y_t \prec 0)| \nabla Y_t = 0]} \cdot f_{Y_t}(y) \cdot \1(y > v) = \frac{\rho_v(t,y)}{\int_{v}^{\infty} \rho_v(t,y) \,dy}.
	\end{equation*}
	To prove Lemma~\ref{lem:density-height-null} we give a local expansion of $\rho_v(t,y)$ about $y = v$ that is accurate at points $t \in \mc{T}_0$.
	\begin{lemma}
		\label{lem:approximate-intensity-null}
		For all $v$ sufficiently large, for any $t \in \mc{T}_0$ and $y \in v \pm \sqrt{3 \log v}$,
		\begin{equation}
			\frac{\Big|\rho_v(t,y) - \bar{\rho}(t,y)\Big|}{\bar{\rho}(t,y)} \leq C\Big(\frac{(y - v)^2}{v^2} + \frac{1}{v^2}\Big).
		\end{equation}
		where
		\begin{equation}
			\label{eqn:approximate-intensity-null}
			\begin{aligned}
			\bar{\rho}(t,y) 
			& := \frac{v^d \sqrt{\det(\Lambda_t)} \exp(-d)}{(2\pi)^{(d + 1)/2}} \exp\Big(-\frac{1}{2}(y - d/v)^2\Big) .
			\end{aligned}
		\end{equation}
	\end{lemma}
	
	The proof of Lemma~\ref{lem:approximate-intensity-null} is given in Section~\ref{subsec:pf-approximate-intensity-null}. It follows similar steps to the proof of Theorem~\ref{thm:approximate-joint-intensity}, but is simpler because we only compute a local expansion in the second argument of $\rho_v(t,y)$.
	
	To compute the normalizing constant $\rho_v(t) = \int_{v}^{\infty} \rho_v(t,y) \,dy$ -- which can be thought of as the marginal intensity function of the counting process $N_v$ -- we combine Lemma~\ref{lem:approximate-intensity-null} with a truncation argument.
	\begin{lemma}
		\label{lem:marginal-intensity-null}
		For all $v$ sufficiently large, for any $t \in \mc{T}_0$,
		\begin{equation*}
			\frac{|\rho_v(t) - \bar{\rho}_v(t)|}{\bar{\rho}_v(t)} \leq \frac{C}{v^2}.
		\end{equation*}
		where
		\begin{equation*}
			\bar{\rho}_v(t) = \frac{v^d \sqrt{\det(\Lambda_t)} \exp(-d)}{(2\pi)^{d/2}} \cdot  \Psi(v - d/v).
		\end{equation*}
	\end{lemma}
	The proof of Lemma~\ref{lem:marginal-intensity-null} is given in Section~\ref{subsec:pf-marginal-intensity-null}. Lemma~\ref{lem:density-height-null} follows immediately from Lemmas~\ref{lem:approximate-intensity-null} and~\ref{lem:marginal-intensity-null}. 
	
	\subsection{Proof of Lemma~\ref{lem:survival-height-null}}
	\label{subsec:pf-survival-height-null}
	The survival function of $Y_t|t \in \wh{T}_v$ can be written in terms of the marginal intensity $\rho_v(t)$:
	\begin{equation*}
		\P(Y_t > u|t \in \wh{T}_v) = \frac{\rho_u(t)}{\rho_v(t)}.
	\end{equation*}
	Lemma~\ref{lem:marginal-intensity-null} upper bounds the relative error incurred by substituting $\bar{\rho}_v(t)$ for $\rho_v(t)$. Exactly the same analysis used to prove Lemma~\ref{lem:marginal-intensity-null}, but substituting $u$ for $v$ everywhere in the proof of that lemma, also implies that
	\begin{equation}
		\label{pf:survival-height-null-1}
		\frac{|\bar{\rho}_u(t) - \rho_u(t)|}{\bar{\rho}_{u}(t)} \leq \frac{C}{u^2},
	\end{equation}
	where to be perfectly explicit 
	$$\
	\bar{\rho}_u(t) = \frac{u^d \sqrt{\det(\Lambda_t)} \exp(-d)}{(2\pi)^{d/2}} \cdot \Psi(u - d/u).
	$$ It remains to upper bound the relative error between $\bar{\rho}_u(t)$ and 
	$$
	\tilde{\rho}_u(t) = \frac{v^d \sqrt{\det(\Lambda_t)} \exp(-d)}{(2\pi)^{d/2}}\cdot \Psi(u - d/v),
	$$and it is here we use the condition that $v \leq u \leq v + \frac{B}{v}$. Specifically, this condition along with Taylor expansion implies that for all $v \geq C \sqrt{\log v}$ sufficiently large:
	\begin{equation*}
		\big|u^d- v^d\big| \leq Cv^d\frac{(u - v)}{v} \leq  \frac{Cv^d}{v^2},
	\end{equation*}
	and
	\begin{equation*}
		\Big|\Psi\Big(u - \frac{d}{u}\Big) - \Psi\Big(u - \frac{d}{v}\Big)\Big| \leq \Psi\Big(u - \frac{d}{v}\Big) \cdot \frac{(u - v)}{v} \leq C\Psi\Big(u - \frac{d}{v}\Big) \frac{1}{v^2},
	\end{equation*}
	from which we deduce that
	\begin{equation}
		\label{pf:survival-height-null-2}
		\frac{|\tilde{\rho}_u(t)- \bar{\rho}_u(t)|}{\tilde{\rho}_u(t)} \leq \frac{C}{v^2}.
	\end{equation}
	Together~\eqref{pf:survival-height-null-1},~\eqref{pf:survival-height-null-2} and Lemma~\ref{lem:marginal-intensity-null} imply the result of Lemma~\ref{lem:survival-height-null}. 
	
	\subsection{Proof of Lemma~\ref{lem:approximate-intensity-null}}
	\label{subsec:pf-approximate-intensity-null}
	At null points $t \in \mc{T}_0$, the Kac-Rice formula of the intensity $\rho(t,y)$ simplifies:
	\begin{equation}
		\begin{aligned}
			\label{eqn:kac-rice-intensity-null}
			\rho(t,y) 
			& = \E\Big[\det(-\nabla^2 Y_t) \cdot \1(\nabla^2 Y_t \prec 0)|\nabla Y_t = 0, Y_t = y\Big] \cdot f_{Y_t}(y) \cdot f_{\nabla Y_t}(0)  \\
			& = \E\Big[\det(-\nabla^2 Y_t) \cdot \1(\nabla^2 Y_t \prec 0)|\nabla Y_t = 0, Y_t = y\Big] \cdot \frac{1}{\sqrt{2\pi}} \exp\Big(-\frac{1}{2}y^2\Big) \cdot \frac{1}{(2\pi\det(\Lambda_t))^{d/2}} .
		\end{aligned}
	\end{equation}
	To prove Lemma~\ref{lem:approximate-intensity-null}, we compute a local expansion of the determinant term in the above expression.
	\begin{lemma}
		\label{lem:jacobian-null}
		For all $v$ sufficiently large, for any $t \in \mc{T}_0$ and $y \in v \pm \sqrt{3 \log v}$,
		\begin{equation}
			\label{eqn:jacobian-null}
			\frac{\Bigg|\E\Big[\det(-\nabla^2 Y_t) \cdot \1(\nabla^2 Y_t \prec 0)|\nabla Y_t = 0, Y_t = y\Big] - \Big(1 + (y - v)\frac{d}{v}\Big)\det(v \cdot \Lambda_t)\Bigg|}{\det(v \cdot \Lambda_t)} \leq C \Big(\frac{(y - v)^2}{v^2} + \frac{1}{v^2}\Big).
		\end{equation}
	\end{lemma}
	The proof of Lemma~\ref{lem:jacobian-null} is given in Section~\ref{subsec:pf-jacobian-null}. Now we show how the first-order expansion of the determinant term translates into a mean shift in the Gaussian density in~\eqref{eqn:kac-rice-intensity-null}. For all $v$ large enough so that $d\sqrt{3 \log v}/v \leq 1$, we have that for all $y \in v \pm \sqrt{3 \log v}$,	
	\begin{equation}
		\Big|1 + (y - v)\frac{d}{v} - \exp\Big((y - v)\frac{d}{v}\Big)\Big| \leq C \frac{(y - v)^2}{v^2}.
	\end{equation}
	On the other hand,
	\begin{equation}
		\label{pf:approximate-intensity-null-1}
		\Big|\frac{\exp(-\frac{1}{2}y^2)}{\sqrt{2 \pi}} \cdot \exp\Big((y - v)\frac{d}{v}\Big) -  \frac{\exp(-d)}{\sqrt{2\pi}}\exp\Big(-\frac{1}{2}(y - d/v)^2\Big)\Big| \leq \frac{C}{v^2}\exp\Big(-\frac{1}{2}(y - d/v)^2\Big).
	\end{equation}
	Combining~\eqref{eqn:jacobian-null}-\eqref{pf:approximate-intensity-null-1} gives the claim of Lemma~\ref{lem:approximate-intensity-null}.
	
	\subsection{Proof of Lemma~\ref{lem:marginal-intensity-null}}
	\label{subsec:pf-marginal-intensity-null}
	We expand the marginal intensity $\rho_v(t)$ of the counting process $N_v$ as follows:
	\begin{equation*}
		\rho_v(t) = \int_{v}^{\infty} \rho_v(t,y) \,dy = \int_{v}^{\infty} \bar{\rho}(t,y) \,dy + \int_{v + \sqrt{3 \log v}}^{\infty} \bar{\rho}(t,y) \,dy + \int_{v}^{v + \sqrt{3 \log v}} \rho_v(t,y) - \bar{\rho}(t,y) \,dy + \int_{v + \sqrt{3 \log v}}^{\infty} \rho_v(t,y) \,dy.
	\end{equation*}
	Above $\bar{\rho}(t,y)$ is the approximate intensity defined in~\eqref{eqn:approximate-intensity-null}. The first term on the right hand side of the previous equation is the main term:
	\begin{equation}
		\int_{v}^{\infty} \bar{\rho}(t,y) \,dy = \det(H_{t|v}) \exp(-d) \Psi(v - d/v) f_{\nabla Y_t}(0) = \bar{\rho}_v(t).
	\end{equation}
	The remaining terms represent error due to truncating the integral to $[v,v + \sqrt{3 \log v}]$, replacing $\rho_v(t,y)$ by the local expansion $\bar{\rho}(t,y)$, and undoing the truncation. The error due to truncation is at most
	\begin{equation*}
		\begin{aligned}
			\int_{v + \sqrt{3 \log v}}^{\infty} \rho_v(t,y) \,dy 
			& = \int_{v + \sqrt{3 \log v}}^{\infty} \E[\det(-\nabla^2 Y_t) \cdot \1(\nabla^2 Y_t \prec 0)|Y_t = y,\nabla Y_t = 0] \cdot \frac{1}{\sqrt{2\pi}}\exp\Big(-\frac{1}{2}y^2\Big) f_{\nabla Y_t}(0) \,dy \\
			& \overset{(i)}{\leq} C \int_{v + \sqrt{3 \log v}}^{\infty} y^d \cdot \exp\big(-\frac{1}{2}y^2\big) \,dy \\ 
			& \overset{(ii)}{\leq} C v^d \Psi(v + \sqrt{3 \log v}) \\
			& \overset{(iii)}{\leq} C v^d \Psi(v) \cdot \exp(-\frac{1}{2}v) \\
			& \leq C \bar{\rho}_v(t) \cdot \exp(-\frac{1}{2}v).
		\end{aligned}
	\end{equation*}
	Above, $(i)$ uses the fact that $f_{\nabla Y_t}(0) \leq C/\sqrt{\det(\Lambda_t)} \leq C$ is upper bounded by a constant, while the determinant term grows polynomially in $y$, i.e. for all $v \geq 1$ and $y > v + \sqrt{3 \log v}$,
	\begin{equation*}
		\E[\det(-\nabla^2 Y_t)|\nabla Y_t = 0,Y_t = y] \leq C(y^d + 1) \leq Cy^d;
	\end{equation*}
	$(ii)$ follows from the bound on the moments of a Truncated Gaussian as given in~\eqref{eqn:normalized-overshoot-moments}; and $(iii)$ applies the upper bound on Mills' ratio as stated in~\eqref{eqn:mills-ratio}. Very similar arguments imply that 
	\begin{equation*}
		\int_{v + \sqrt{3 \log v}}^{\infty} \bar{\rho}(t,y) \,dy = C \det(v \cdot \Lambda_t) \cdot \Psi(v + \sqrt{3 \log v} - d/v) \leq C \bar{\rho}_v(t) \cdot \exp(-\frac{1}{2}v).
	\end{equation*}
	Finally, Lemma~\ref{lem:approximate-intensity-null} upper bounds the relative error between $\bar{\rho}(t,y)$ and $\rho_v(t,y)$. Integrating this yields the following upper bound:
	\begin{align*}
		\int_{v}^{v + \sqrt{3 \log v}} |\rho_v(t,y) - \bar{\rho}(t,y)| \,dy & \leq 
		C \int_{v}^{\infty} \bar{\rho}(t,y) \Big(\frac{(y - v)^2}{v^2} + \frac{1}{v^2}\Big) \,dy \\
		& \leq C \det(v \cdot \Lambda_t) \cdot \frac{\Psi(v - d/v)}{v^2} \\
		& \leq 
		C \bar{\rho}_v(t) \cdot \frac{1}{v^2},
	\end{align*}
	with the second and third inequalities following by~\eqref{eqn:normalized-overshoot-moments} and~\eqref{eqn:mills-ratio} respectively. Combining these bounds yields the claim of Lemma~\ref{lem:marginal-intensity-null}. 
	
	\subsection{Proof of Lemma~\ref{lem:jacobian-null}}
	\label{subsec:pf-jacobian-null}
	The proof of Lemma~\ref{lem:jacobian-null} is similar to the proof of Lemma~\ref{lem:approximation-peak-intensity}. The determinant term is 
	\begin{equation*}
		\int_{\R^{d \times d}} \det(H_{t|y} + R) \cdot \1(H_{t|y} + R \succ 0) \cdot f_{R_t}(R) \,dR,
	\end{equation*}
	where we recall that the density of the pinned residual $f_{R_t}(R)$ is defined in~\eqref{eqn:hessian-covariance}. Let $\mc{R}_0 := B_{d \times d}(0, d \cdot \sqrt{3 \log v})$. We expand the determinant term as follows:
	\begin{equation*}
		\begin{aligned}
			& \int \det(H_{t|y} + R) \cdot \1(H_{t|y} + R \succ 0) \cdot f_{R_t}(R) \,dR \\
			& \quad = \det({H}_{t|v})\Big(1 + \frac{d(y - v)}{v}\Big) \\
			& \quad - \int_{\mc{R}_{0}^c}  \det({H}_{t|v})\Big(1 + \frac{d(y - v)}{v}\Big) \cdot f_{R_t}(R) \,dR \\
			& \quad + \int_{\mc{R}_{0}}  \det({H}_{t|v})\Big(\tr(\{H_{t|v}\}^{-1}R)\Big) \cdot f_{R_t}(R) \,dR \\
			& \quad + \int_{\mc{R}_{0}} \Big(\det(H_{t|y} + R)  \1(H_{t|y} + R \succ 0) -\det(H_{t|v})\Big(1 + \frac{d(y - v)}{v} + \tr(\{H_{t|v}\}^{-1}R)\Big) \cdot f_{R_t}(R) \,dR  \\
			& \quad + \int_{\mc{R}_{0}^c} \det(H_{t|v} + R) \cdot \1(H_{t|v} + R \succ 0) \cdot f_{R_t}(R) \,dR.
		\end{aligned}
	\end{equation*}
	The first term on the right hand side of the equality above is the claimed local expansion of the determinant term. The second term on the right hand side is the error due to ``undoing'' truncation to $\mc{R}_0$. The third term integrates an odd function over a set that is symmetric about the origin, and is thus zero. The fourth term is error due to local expansion, and the last term is error due to truncating to $\mc{R}_0$. Thus it  remains only to bound the various sources of error, due to local expansion, truncation, and undoing the truncation. 
	
	In bounding the error due to local expansion, we apply a first-order Taylor expansion of $\det(H_{t|y} + R)$ about $y = v$, $R = 0$ to conclude that for all $v > C\sqrt{\log v}$ the following is true: for all $t \in \mc{T}_0, y \in v \pm \sqrt{3 \log v}$ and $R \in B_{d \times d}(0,\sqrt{3 \log v})$
	\begin{equation*}
		\Big|\det(H_{t|y} + R) - \det(H_{t|v}) \cdot \Big(1 + \frac{d(y - v)}{v} + \tr(\{H_{t|v}\}^{-1}R)\Big)\Big| \leq C\Big(\frac{(y - v)^2}{v^2} + \frac{\|R\|^2}{v^2}\Big),
	\end{equation*}
	and moreover $H_{t|y} + R \succ 0$. Thus error due to local expansion is at most
	\begin{equation*}
		\begin{aligned}
			& \int_{\mc{R}_{0}} \Big(\det(H_{t|y} + R)  \1(H_{t|y} + R \succ 0) -\det(H_{t|v})\Big(1 + \frac{d(y - v)}{v} + \tr(\{H_{t|v}\}^{-1}R)\Big) \cdot f_{R_t}(R) \,dR  \\
			& \leq C \int_{\mc{R}_0} \Big(\frac{(y - v)^2}{v^2} + \frac{\|R\|^2}{v^2}\Big)  \cdot f_{R_t}(R) \,dR \\
			& \leq C \Big(\frac{(y - v)^2}{v^2} + \frac{1}{v^2}\Big).
		\end{aligned}
	\end{equation*}
	Arguments similar to those used in the proof of Lemma~\ref{lem:approximation-peak-intensity} show that the error due to truncation and undoing the truncation are both at most $C \det(H_{t|v})/v^{2}$. This completes the proof of Lemma~\ref{lem:jacobian-null}. 
	
	
	\subsection{Proof of Proposition~\ref{prop:expectation-counting-process}}
	\label{subsec:pf-expectation-counting-process}
	We begin by deriving the upper bound on the relative error between $\bar{\E}[N(\mc{S}_{t^*})]$ and $\E[N(\mc{S}_{t^*})]$, and then proceed to upper bound the error between $\E[N_u(\mc{B}_{t^*})]$ and $\E[N(\mc{S}_{t^*})]$. 
	
	\subsubsection{Error between $\bar{\E}[N(\mc{S}_{t^*})]$ and $\E[N(\mc{S}_{t^*})]$.}
	Define $\rho(y) := \int_{\mc{B}_{t^*}} \rho(t,y) \,dt$. One can think of this as the intensity function of peaks of height $y$ across $\mc{B}_{t^*}$.  The following Lemma computes a local expansion $\bar{\rho}(y)$ to $\rho(y)$ by integrating $\bar{\rho}(t,y)$ over $\Rd$. There are two types of error incurred: error due to replacing $\rho(t,y)$ by $\bar{\rho}(t,y)$, and error due to integrating over $\Rd$ rather than $\mc{B}_{t^*}$. These are bounded by the following: 
	\begin{align*}
		\Err_{\rho}(\mc{B}_{t^*},y) & := C\Big(|y - \bar{u}_{t^*}|^2\delta_n^2 + |\bar{u}_{t^*} - \mu_{t^*}|\delta_n^2 +  \delta_n^2 \Big) \\
		\Trunc_{\rho}(\mc{B}_{t^*},y) & := C\delta_n^2
	\end{align*}
	\begin{lemma}
		\label{lem:height-intensity}
		Under the assumptions of Theorem~\ref{thm:approximate-joint-intensity}, for all $n \in \mathbb{N}$ sufficiently large, for any $y \in \mc{I}_{t^*}$,
		\begin{equation}
			\label{eqn:height-intensity}
			\begin{aligned}
				& \frac{\big|\rho(y) - \bar{\rho}(y)\big|}{\bar{\rho}(y)}
				\leq \Err_{\rho}(\mc{B}_{t^*},y) + \Trunc_{\rho}(\mc{B}_{t^*},y)  := \Err_{\rho}(y),
			\end{aligned}
		\end{equation}
		where
		\begin{equation}
			\label{eqn:approximate-height-intensity}
			\begin{aligned}
				\bar{\rho}(y) & := \bar{\E}[N(\mc{S}_{t^*})] \cdot \frac{1}{\sqrt{2\pi}}\frac{\exp\Big(-\frac{1}{2}\big(y - \mu_{t^*} - \frac{1}{2}\tr(\bar{H}_{t^*}^{-1}\Lambda_{t^*})\big)^2\Big)}{\Psi(u - \mu_{t^*} - \frac{1}{2}\tr(\bar{H}_{t^*}^{-1}\Lambda_{t^*}))} \cdot \1(y > u).
			\end{aligned}
		\end{equation}
	\end{lemma}
	The proof of Lemma~\ref{lem:height-intensity} is given in Section~\ref{subsec:pf-height-intensity}.  Now we want to use Lemma~\ref{lem:height-intensity} to bound the difference between $\E[N(\mc{S}_{t^*})] = \int_{\mc{I}_{t^*}} \rho(y) \,dy$ and $\bar{\E}[N(\mc{S}_{t^*})]$.  To do so we substitute $\bar{\rho}(y)$ for $\rho(y)$ and undo the truncation in the integral defining $\E[N(\mc{S}_{t^*})]$, yielding the decomposition:
	\begin{equation}
		\begin{aligned}
			\label{pf:theorem-joint-density-1}
			\int_{\mc{I}_{t^*}} {\rho}(y) \,dy 
			& = \int_{u}^{\infty} \bar{\rho}(y) \,dy - \int_{(u,\infty) \setminus \mc{I}_{t^*}} \bar{\rho}(y) \,dy + \int_{\mc{I}_{t^*}} \rho(y) - \bar{\rho}(y) \,dy.
		\end{aligned}
	\end{equation}
	Direct calculation shows that the first term on the right hand side above is exactly $\int_{u}^{\infty} \bar{\rho}(y) \,dy = \bar{\E}[N(\mc{S}_{t^*})]$. The remaining terms represent error due to local expansion and undoing the truncation, and we now bound each. 
	
	\paragraph{Error due to local expansion.}
	The error due to approximation of $\rho(y)$ by $\bar{\rho}(y)$ is at most 
	\begin{equation*}
		\begin{aligned}
			\int_{\mc{I}_{t^*}} |\rho(y) - \bar{\rho}(y)| \,dy \leq \int_{\mc{I}_{t^*}} \Err_{\rho}(y) \cdot \bar{\rho}(y) \,dy \leq
			C\delta_n^2\Big((1 + |\bar{u}_{t^*} - \mu_{t^*}| ) \cdot \bar{\E}[N(\mc{S}_{t^*})] + \int_{\mc{I}_{t^*}} |y - \bar{u}_{t^*}|^2 \bar{\rho}(y) \,dy\Big).
		\end{aligned}
	\end{equation*}
	The remaining integral can be upper bounded using the Truncated Gaussian moment bounds in Section~\ref{subsec:gaussian-tail-univariate}: letting $m = \mu_{t^*} + \frac{1}{2}\tr(\bar{H}_{t^*}^{-1}\Lambda_{t^*})$,
	\begin{equation*}
		\begin{aligned}
			\int_{\mc{I}_{t^*}} |y - \bar{u}_{t^*}|^2 \bar{\rho}(y) \,dy 
			& \leq \frac{\bar{\E}[N(\mc{S}_{t^*})]}{\Psi(\bar{u}_{t^*} - m)} \int_{u}^{\infty} \frac{|y - \bar{u}_{t^*}|^2}{\sqrt{2\pi}} \cdot \exp\Big(-\frac{1}{2}\big(y - m\big)^2\Big) \,dy \\
			& = \frac{\bar{\E}[N(\mc{S}_{t^*})]}{\Psi(\bar{u}_{t^*} - m)} \int_{u - m}^{\infty} |z - (\bar{u}_{t^*} - m)|^2 \phi(z) \,dz \\
			& \leq C \cdot \bar{\E}[N(\mc{S}_{t^*})].
		\end{aligned}
	\end{equation*}
	Thus, the relative error due to local expansion is at most
	\begin{equation*}
		\int_{\mc{I}_{t^*}} |\rho(y) - \bar{\rho}(y)| \,dy \leq C \cdot \bar{\E}[N(\mc{S}_{t^*})] \cdot \Big( |\bar{u}_{t^*} - \mu_{t^*}| \delta_n^2 + \delta_n^2\Big).
	\end{equation*}
	\paragraph{Error due to undoing the truncation.}
	Now we bound the error due to integrating over $(u,\infty)$ rather than over $\mc{I}_{t^*}$. In bounding this error we will use the fact that $\tr(\bar{H}_{t^*}^{-1}\Lambda_{t^*}) \leq C \delta_n$, and we will assume that $n \in \mathbb{N}$ is sufficiently large so that $C\delta_n < 1$. We partition the complement of this interval, $\mc{I}_{t^*}^c = (u,\infty) \setminus \mc{I}_{t^*}$, into $\mc{I}_{-}^{c} := [u, \max(\bar{u}_{t^*} - \Delta_n,u))$ and $\mc{I}_{+}^{c} := (\bar{u}_{t^*} + \Delta_n,\infty)$, so that the truncation error is the sum of lower and upper truncation error terms,
	\begin{equation*}
		\int_{(u,\infty) \setminus \mc{I}_{t^*}} \bar{\rho}(y) \,dy  = \int_{\mc{I}_{-}^c} \bar{\rho}(y) \,dy + \int_{\mc{I}_{+}^c} \bar{\rho}(y) \,dy.
	\end{equation*}
	The lower truncation error is zero if $u \geq \mu_{t^*}$, since in that case $u > \bar{u}_{t^*} - \Delta_n$. Otherwise $\mu_{t^*} > u$ and $\bar{u}_{t^*} = \mu_{t^*}$, and so $\Psi(u - \mu_{t^*} - \frac{1}{2}\tr(\bar{H}_{t^*}^{-1}\Lambda_n)) \geq \Psi(-1/2) > 1/4$. In this case the lower truncation error is at most
	\begin{equation*}
		\begin{aligned}
			\int_{\mc{I}_{-}^c} \bar{\rho}(y) \,dy  
			& =  \frac{\bar{\E}[N(\mc{S}_{t^*})]}{\sqrt{2\pi}}  \int_{\mc{I}_{-}^c}  \frac{\exp\Big(-\frac{1}{2}\big(y - \mu_{t^*} - \frac{1}{2}\tr(\bar{H}_{t^*}^{-1}\Lambda_{t^*})\big)^2\Big)}{\Psi(\bar{u}_{t^*} - \mu_{t^*} - \frac{1}{2}\tr(\bar{H}_{t^*}^{-1}\Lambda_{t^*}))} \,dy \\
			& \leq  \frac{\bar{\E}[N(\mc{S}_{t^*})]}{\sqrt{2\pi}}  \int_{\mc{I}_{-}^c}  \frac{\exp\Big(-\frac{1}{2}(y - \mu_{t^*})^2\Big)}{\Psi(\bar{u}_{t^*} - \mu_{t^*} - \frac{1}{2}\tr(\bar{H}_{t^*}^{-1}\Lambda_{t^*}))} \,dy \\
			& \leq \bar{\E}[N(\mc{S}_{t^*})] \cdot \frac{\Phi(-\Delta_n)}{\Psi(u - \mu_{t^*} - \frac{1}{2}\tr(\bar{H}_{t^*}^{-1}\Lambda_{t^*}))} \\ 
			& \leq 4 \cdot \bar{\E}[N(\mc{S}_{t^*})] \Phi(-\Delta_n) \\
			& \overset{(i)}{\leq} 4\cdot \bar{\E}[N(\mc{S}_{t^*})] \exp\Big(-\frac{1}{2}\Delta_n^2\Big) \\
			& \leq 4\cdot \bar{\E}[N(\mc{S}_{t^*})] \cdot \delta_n^3,
		\end{aligned}
	\end{equation*}
	with $(i)$ following from a standard bound on Mills' ratio -- recorded in~\eqref{eqn:mills-ratio} for convenience. 
	
	A corresponding bound on upper truncation error can again be deduced using standard bounds on Mills' ratio:
	\begin{equation*}
		\int_{\mc{I}_{+}^c} \bar{\rho}(y) \,dy = \frac{\bar{\E}[N(\mc{S}_{t^*})]}{\sqrt{2\pi}} \cdot \frac{\Psi(\bar{u}_{t^*} + \Delta_n - \mu_{t^*} - \frac{1}{2}\tr(\bar{H}_{t^*}^{-1}\Lambda_{t^*}))}{\Psi(\bar{u}_{t^*} - \mu_{t^*} - \frac{1}{2}\tr(\bar{H}_{t^*}^{-1}\Lambda_{t^*}))} \leq \bar{\E}[N(\mc{S}_{t^*})] \exp\Big(-\frac{1}{2}\Delta_n^2\Big) \leq \bar{\E}[N(\mc{S}_{t^*})]  \delta_n^3.
	\end{equation*}
	We conclude that the error due to undoing truncation to $\mc{I}_{t^*}$ is at most $C \bar{\E}[N(\mc{S}_{t^*})] \delta_n^3$, of a lower order than the error due to local expansion. 
	
	\subsubsection{Upper bound on $\E[N_u(\mc{B}_{t^*})] - \E[N(\mc{S}_{t^*})]$.}
	It remains to upper bound the difference between the expected number of $\varepsilon_n$-consistent discoveries, and the expected number of $\varepsilon_n$-consistent discoveries that also have height $Y_t \in \bar{u}_{t^*} \pm \Delta_n$:
	\begin{equation*}
		\E[N_u(\mc{B}_{t^*})] - \E[N(\mc{S}_{t^*})] = \int_{\mc{B}_{t^*}} \int_{\mc{I}_{t^*}^c} \rho(t,y)  \,dy \,dt = \int_{\mc{B}_{t^*}} \int_{\mc{I}_{-}^c} \rho(t,y) \,dy \,dt + \int_{\mc{B}_{t^*}} \int_{\mc{I}_{+}^c} \rho(t,y) \,dy \,dt.
	\end{equation*}
	We will upper bound each of the two error terms above, the lower truncation error and the upper truncation error, separately. In obtaining these bounds, we cannot rely on Theorem~\ref{thm:approximate-joint-intensity} because the range of integration is over $y \not\in \bar{u}_{t^*} \pm \Delta_n$. Instead, we will bound each term in the Kac-Rice formula for $\rho(t,y)$, and integrate these bounds. For both lower and upper truncation error, the following bound on the determinant term in the Kac-Rice formula will be useful:
	\begin{align}
		\E[\det(-\nabla^2 Y_t) \cdot \1(\nabla^2 Y_t \prec 0)|Y_t = y, \nabla Y_t = 0] 
		& = \int \det(H_{t|y} + R) \cdot \1(H_{t|y} + R \succ 0) f_{R_t}(R) \,dR \nonumber \\
		& \leq C\Big(\det(H_{t|y}) + 1\Big) \nonumber \\
		& \leq C\Big(\det(H_{t^*|y}) + 1\Big) \nonumber \\
		& \leq C\Big(\det(\bar{H}_{t^*}) + (y - \bar{u}_{t^*})_{+}^d  + 1\Big) \nonumber \\
		& \leq C\Big(\det(\bar{H}_{t^*}) + (y - \bar{u}_{t^*})_{+}^d\Big) \label{pf:expectation-counting-process-2-1}
	\end{align}
	where in the last inequality we have assumed that $n \in \mathbb{N}$ is large enough so that $\lambda_n \geq 1$. We will also assume that $n \in \mathbb{N}$ is large enough that 
	\begin{equation}
		\label{pf:expectation-counting-process-2-2}
		\begin{aligned}
			\nabla \mu_{t^* + h}'\Lambda_{t^*}^{-1}\nabla \mu_{t^* + h} 
			& \geq \frac{1}{2} h'\nabla^2\mu_{t^*}\Lambda^{-1}\nabla^2\mu_{t^*}h, \\
			\mu_{t^*} - \mu_{t^* + h} 
			& \leq 1/2 \leq \sqrt{1/3} \cdot \Delta_n,
		\end{aligned}
	\end{equation}
	for all $h \in B_d(0,\varepsilon_n)$.
	
	\paragraph{Lower truncation term.}
	If $u \geq \mu_{t^*}$ then $\bar{u}_{t^*} = u, \bar{u}_{t^*} - \Delta_n < u$ and so the lower truncation term is $0$. Otherwise $u < \mu_{t^*}, \bar{u}_{t^*} = \mu_{t^*}, \bar{H}_{t^*} = -\nabla^2\mu_{t^*}$.  We apply the upper bound~\eqref{pf:expectation-counting-process-2-1} on the determinant term -- noting that $y < \bar{u}_{t^*}$ for all $y \in \mc{I}_{-}^c$ -- and use the bounds from~\eqref{pf:expectation-counting-process-2-2} to upper bound the density of the height and gradient, giving
	\begin{equation*}
		\begin{aligned}
			& \int_{\mc{B}_{t^*}} \int_{\mc{I}_{-}^c} \rho(t,y) \,dt \,dy \\
			& \quad \leq C \det(\bar{H}_{t^*})  \int_{B_d(0,\varepsilon_n)} \int_{-\infty}^{\mu_{t^*} - \Delta_n}  \exp\Big(-\frac{1}{2}(y - \mu_{t^*} + \sqrt{1/3} \cdot \Delta_n)^2\Big) \cdot \exp(-\frac{1}{4}h'\nabla^2\mu_{t^*}\Lambda_{t^*}^{-1}\nabla^2\mu_{t^*}h) \,dy \,dh \\
			& \quad \leq C \det(\bar{H}_{t^*})  \int_{\mc{B}_{t^*}} \int_{-\infty}^{\mu_{t^*} - \Delta_n}  f_{Y_t}(y) \cdot f_{\nabla Y_t}(0) \,dy \,dt \\
			& \quad \leq C \det(\bar{H}_{t^*}) \Phi\Big(-\sqrt{\frac{2}{3}} \Delta_n \Big) \int_{B(0,\varepsilon_n)} \exp(-\frac{1}{4}h'\nabla^2\mu_{t^*}\Lambda_{t^*}^{-1}\nabla^2\mu_{t^*}h) \,dh \\
			& \quad \leq C \Phi\Big(-\sqrt{\frac{2}{3}} \Delta_n \Big)  \\
			& \quad \leq C \delta_n^2.
		\end{aligned}
	\end{equation*}
	with the final inequality following by Mills' inequality~\eqref{eqn:mills-ratio} and the definition of $\Delta_n$. Finally, since $\mu_{t^*} > u$ it can be verified that $\bar{\E}[N(\mc{S}_{t^*})] \geq c > 0$, so that the lower truncation term is at most $C \cdot \bar{\E}[N(\mc{S}_{t^*})] \cdot \delta_n^2$.
	
	\paragraph{Upper truncation term.}
	For the upper truncation term, applying~\eqref{pf:expectation-counting-process-2-1} and integrating over $y \in [\bar{u}_{t^*} + \Delta_n,\infty)$:
	\begin{equation*}
		\begin{aligned}
			\int_{\mc{B}_{t^*}} \int_{\mc{I}_{+}^c} \rho(t,y) \,dt \,dy 
			& \leq C \int_{\mc{B}_{t^*}}\int_{\bar{u}_{t^*} + \Delta_n}^{\infty} (\det(\bar{H}_{t^*}) + (y - \bar{u}_{t^*})^{d}) \cdot f_{Y_{t}}(y) \cdot f_{\nabla Y_t}(0) \,dy \,dt \\
			& \leq C \int_{\mc{B}_{t^*}}\int_{\bar{u}_{t^*} + \Delta_n}^{\infty} (\det(\bar{H}_{t^*}) + \Delta_n^d + (y - \bar{u}_{t^*} - \Delta_n)^{d}) \cdot f_{Y_{t}}(y) \cdot f_{\nabla Y_t}(0) \,dy \,dt \\
			& \overset{(i)}{\leq} C (\det(\bar{H}_{t^*}) + \Delta_n^d) \int_{\mc{B}_{t^*}} \Psi(\bar{u}_{t^*} - \mu_t + \Delta_n) \cdot f_{\nabla Y_t}(0) \,dt \\
			& \leq C \det(\bar{H}_{t^*})\int_{\mc{B}_{t^*}} \Psi(\bar{u}_{t^*} - \mu_t + \Delta_n) \cdot f_{\nabla Y_t}(0) \,dt.
		\end{aligned}
	\end{equation*}
	where $(i)$ relies on the bound on Truncated Gaussian moments in~\eqref{eqn:normalized-overshoot-moments}. If $\bar{u}_{t^*} - \mu_{t^*} \not\to \infty$, then the remaining argument is very similar to the analysis of the lower truncation term, and so we focus on the case where $\bar{u}_{t^*} - \mu_{t^*} \to \infty$. Applying the upper bound~\eqref{eqn:mills-ratio} on Mills' ratio again:
	\begin{equation*}
		\begin{aligned}
			\Psi(\bar{u}_{t^*} - \mu_t + \Delta_n) 
			& \leq C \frac{\exp(-\frac{1}{2}\Delta_n^2 - \frac{1}{2}(\bar{u}_{t^*} - \mu_t)^2)}{\bar{u}_{t^*} - \mu_t + \Delta_n} \\
			& \leq C \exp\Big(-\frac{1}{2}\Delta_n^2\Big) \frac{\exp\big(-\frac{1}{2}(\bar{u}_{t^*} - \mu_{t^*})^2\big)}{\bar{u}_{t^*} - \mu_{t^*}}  \\
			& \leq C \exp\Big(-\frac{1}{2}\Delta_n^2\Big) \Psi(\bar{u}_{t^*} - \mu_{t^*}).
		\end{aligned}
	\end{equation*}
	Plugging this into our previous bound on the upper truncation term, using the bound from~\eqref{pf:expectation-counting-process-2-2} to upper bound $f_{\nabla Y_t}(0)$, and integrating over $t \in \mc{B}_{t^*}$, we conclude that the upper truncation error is at most
	\begin{equation*}
		C \exp\Big(-\frac{1}{2}\Delta_n^2\Big) \Psi(\bar{u}_{t^*} - \mu_{t^*}) \leq C \cdot \bar{\E}[N(\mc{S}_{t^*})] \cdot \exp\Big(-\frac{1}{2}\Delta_n^2\Big)  \leq C \cdot \bar{\E}[N(\mc{S}_{t^*})] \cdot \delta_n^3.
	\end{equation*}
	
	\subsection{Proof of Lemma~\ref{lem:height-intensity}}
	\label{subsec:pf-height-intensity}
	Recall that $\bar{\rho}(t^* + h,y)$ is defined in~\eqref{eqn:approximate-joint-intensity}. Let
	\begin{equation*}
		\tilde{\rho}(t^* + h,y) := \frac{\det(\bar{H}_{t^*})\Big(1 + T_{01}^{\rho}(y) + T_{21}^{\rho}(h,y)\Big)}{\sqrt{(2\pi)^{d + 1}\det(\Lambda_{t^*})}} \cdot \exp\Big(-\frac{(y - \mu_{t^*})^2}{2}\Big) \cdot \exp\bigg(-\frac{h' \bar{G}_{t^*}h}{2}\bigg) \cdot \1(y > u).
	\end{equation*}
	Note that the only difference between $\tilde{\rho}(t^* + h,y)$ and $\bar{\rho}(t^* + h,y)$ is that the former omits the first-order terms $T_{10}^{\rho}(h)$ and $T_{30}^{\rho}(h)$.	
	We write $\rho(y)$ in terms of a main term and three error terms:
	\begin{equation*}
		\rho(y) = \int_{\Rd} \tilde{\rho}(t^* + h,y) \,dh - \int_{\mc{B}_{t^*}^c} \tilde{\rho}(t^* + h,y) \,dh + \int_{\mc{B}_{t^*}} \tilde{\rho}(t^* + h,y) - \bar{\rho}(t^* + h,y) \,dh + \int_{\mc{B}_{t^*}} \rho(t^* + h,y) - \bar{\rho}(t^* + h,y) \,dh.
	\end{equation*}
	The second of the three error terms is zero, 
	\begin{equation*}
		\int_{\mc{B}_{t^*}} \tilde{\rho}(t^* + h,y)  = \int_{\mc{B}_{t^*}}\bar{\rho}(t^* + h,y) \,dh,
	\end{equation*}
	because $T_{10}^{\rho}(h), T_{30}^{\rho}(h)$ are odd functions of $h$ and $\mc{B}_{t^*}$ is symmetric about the origin. We proceed to bound the relative error between the main term and $\bar{\rho}(y)$, and then bound the magnitude of the remaining two error terms. In doing so, we will at times use the fact that the first-order terms in the definition of $\bar{\rho}(t^* + h,y)$ are indeed asymptotically vanishing, as
	\begin{equation*}
		\begin{aligned}
			T_{01}^{\rho}(y) & \leq C \Delta_n \delta_n \leq  C\varepsilon_n \\
			T_{10}^{\rho}(h) & \leq C \varepsilon_n \\
			T_{30}^{\rho}(h) & \leq C \delta_n \{\log(\lambda_n)\}^{3/2} \leq  C \varepsilon_n  \log(\lambda_n) \\
			T_{21}^{\rho}(h) & \leq C \delta_n \{\log(\lambda_n)\}^{3/2} \leq  C \varepsilon_n  \log(\lambda_n).
		\end{aligned}
	\end{equation*}
	We assume $n \in \mathbb{N}$ is large enough that $C \varepsilon_n \log(\lambda_n) \leq 1/2$.
	
	\paragraph{Main term.}
	Direct computation shows that the main term is
	\begin{equation}
		\label{pf:height-intensity-1}
		\begin{aligned}
			\int_{\Rd} \tilde{\rho}(t^* + h,y) \,dh 
			& = \Big(1 + T_{01}^{\rho}(y) + \frac{(y - \bar{u}_{t^*})}{2}\tr(\bar{G}_{t^*}^{-1}\nabla^2\mu_{t^*})\Big) \cdot \frac{\det(\bar{H}_{t^*})}{\sqrt{\det(\bar{G}_{t^*}\Lambda_{t^*})}} \cdot \frac{1}{\sqrt{2\pi}}\exp\Big(-\frac{(y - \mu_{t^*})^2}{2}\Big) \\
			& = \Big(1 + \frac{1}{2}T_{01}^{\rho}(y)\Big) \cdot \sqrt{\frac{\det(\bar{H}_{t^*})}{\det(-\nabla^2\mu_{t^*})}} \cdot \frac{1}{\sqrt{2\pi}}\exp\Big(-\frac{(y - \mu_{t^*})^2}{2}\Big) =: \tilde{\rho}(y).
		\end{aligned} 
	\end{equation}
	with the second equality following upon recalling that $\bar{G}_{t^*} = -\nabla^2\mu_{t^*} \Lambda_{t^*}^{-1}\bar{H}_{t^*}$.
	
	We have shown that the main term $\tilde{\rho}(y) = \int_{\Rd} \tilde{\rho}(t^* + h,y) \,dh$ is a valid first-order expansion of $\rho(y)$ about $y = \bar{u}_{t^*}$. We now show that this expansion can be rewritten as a Gaussian density with a first-order shift in the mean, at the cost of another second-order error term. 
	Assuming $n \in \mathbb{N}$ is large enough that $C \varepsilon_n < 1/2$, it follows from first-order Taylor expansion that $\sup_{x \leq C \varepsilon_n}|1 + x - \exp(x)| \leq 2 x^2$. Plugging in $x = \frac{1}{2}T_{01}^{\rho}(y)$, we conclude that
	\begin{equation*}
		\bigg|\Big(1 + \frac{1}{2}T_{01}^{\rho}(y)\Big) - \exp\Big(\frac{1}{2}(y - \mu_{t^*})\cdot \tr(\bar{H}_{t^*}^{-1}\Lambda_{t^*}) - \frac{1}{2}(\bar{u}_{t^*} - \mu_{t^*})\tr(\bar{H}_{t^*}^{-1}\Lambda_{t^*})\Big)\bigg| \leq C\Big((y - \bar{u}_{t^*})^2 \delta_n^2\Big).
	\end{equation*}
	Multiplying the second term inside the absolute value by $\exp(-\frac{1}{4}\tr(\bar{H}_{t^*}^{-1}\Lambda_{t^*})^2)$ costs only an additional factor of $C\delta_n^{2}$. Plugging the result back into~\eqref{pf:height-intensity-1}, we conclude that
	\begin{equation}
		\label{pf:approximate-joint-distribution-4}
		\begin{aligned}
			\big|\tilde{\rho}(y) - \bar{\rho}(y)\big| & \leq C\frac{\sqrt{\det(\bar{H}_{t^*})}}{\sqrt{(2\pi)\det(-\nabla^2\mu_{t^*})}} \cdot \exp\Big(-\frac{1}{2}(y - \mu_{t^*})^2\Big)\Big((y - \bar{u}_{t^*})^2 \delta_n^2 + \delta_n^2\Big) \\
			& \leq C \bar{\rho}(y) \cdot \Big((y - \bar{u}_{t^*})^2 \delta_n^2 + \delta_n^2\Big).
		\end{aligned}
	\end{equation}
	\paragraph{Error due to local expansion.}
	The error due to integrating the local expansion $\bar{\rho}(t^* + h,y)$ rather than $\rho(t^* + h,y)$ is at most
	\begin{equation}
		\label{pf:approximate-joint-distribution-1}
		\begin{aligned}
			& \int_{\mc{B}_{t^*}} \bar{\rho}(t^* + h,y) - \rho(t^* + h,y) \,dh \\
			& \leq  \int_{\mc{B}_{t^*}} \Err_{\rho}(h,y) \cdot \bar{\rho}(h,y) \,dh \\
			& \overset{(i)}{\leq} C \det(\bar{H}_{t^*}) \exp\Big(-\frac{1}{2}(y - \mu_{t^*})^2\Big)\int_{\mc{B}_{t^*}} \Err_{\rho}(h,y) \exp\Big(-\frac{1}{2}h'\bar{G}_{t^*}h\Big) \cdot  \,dh  \\
			& \leq C \det(\bar{H}_{t^*}) \exp\Big(-\frac{1}{2}(y - \mu_{t^*})^2\Big)\int_{\mc{B}_{t^*}} \Big(\Err_{\det}(h,y) + \Err_{Y}(h,y) + \Err_{\nabla Y}(h,y)\Big)  \exp\Big(-\frac{1}{2}h'\bar{G}_{t^*}h\Big) \,dh \\
			& \overset{(ii)}{\leq} C\frac{\det(\bar{H}_{t^*})}{\sqrt{\det(\bar{G}_{t^*})}} \exp\Big(-\frac{1}{2}(y - \mu_{t^*})^2\Big) \Big(|y - \bar{u}_{t^*}|^2\delta_n^2 + |\bar{u}_{t^*} - \mu_{t^*}|\delta_n^2 + \delta_n^2\Big) \\
			& \overset{(iii)}{\leq} C\bar{\rho}(y)\Big(|y - \bar{u}_{t^*}|^2\delta_n^2 + |\bar{u}_{t^*} - \mu_{t^*}|\delta_n^2 + \delta_n^2\Big).
		\end{aligned}
	\end{equation}
	In obtaining $(i)$ we have applied the upper bound $\bar{\rho}(t,y) \leq C \det(\bar{H}_{t^*})\exp(-\frac{1}{2}h'G(y)h) \cdot \exp(-\frac{1}{2}(y - \mu_{t^*})^2)$; in $(ii)$ we have integrated each component of $\Err_{\rho}(t,y)$, and in $(iii)$ we have applied the upper bound 
	\begin{equation*}
		\label{pf:approximate-joint-distribution-5}
		\frac{\det(\bar{H}_{t^*})}{\sqrt{\det(\bar{G}_{t^*})}} \exp\Big(-\frac{1}{2}(y - \mu_{t^*})^2\Big) \leq C\bar{\rho}(y).
	\end{equation*}
	\paragraph{Error due to truncation.}
	In bounding the truncation error, observe that $\varepsilon_n \geq \{\lambda_{\min}(\bar{G}_{t^*})\}^{-1} \sqrt{6 \log \lambda_n}$. Therefore,
	\begin{equation}
		\label{pf:approximate-joint-distribution-7}
		\sqrt{\det(\bar{G}_{t^*})} \int_{\mc{B}_{t^*}^c} \exp\Big(-\frac{h'\bar{G}_{t^*}h}{2}\Big) \,dh \leq \int_{\Rd \setminus B_d(0,\sqrt{6 \log \lambda_n})} \exp\Big(-\frac{\|z\|^2}{2}\Big) \,dz \leq C \lambda_n^{-2} = C \delta_n^2,
	\end{equation}
	the latter inequality following by~\eqref{eqn:gaussian-integrals-concentration}. Similarly, using the Cauchy-Schwarz inequality,
	\begin{align*}
		\sqrt{\det(\bar{G}_{t^*})} \int_{\mc{B}_{t^*}^c} \|h\|^2 \exp\Big(-\frac{h'\bar{G}_{t^*}h}{2}\Big) \,dh
		& \leq \Big(\int_{\Rd} \|h\|^4 \phi_{0,\bar{G}_{t^*}^{-1}}(h) \,dt\Big)^{1/2} \Big(\int_{\mc{B}_{t^*}^c} \phi_{0,\bar{G}_{t^*}^{-1}}(h) \,dt\Big)^{1/2} \leq C \delta_n^{4},
	\end{align*}
	where we have written $\phi_{0,\bar{G}_{t^*}^{-1}}(h)$ for the density of a $N_d(0,\bar{G}_{t^*}^{-1})$ distribution. 
	The overall error due to truncation is thus at most
	\begin{equation*}
		\begin{aligned}
			& \frac{\det(\bar{H}_{t^*})}{\sqrt{(2\pi)\det(\Lambda_{t^*} \bar{G}_{t^*})}} \cdot \exp\Big(-\frac{1}{2}(y - \mu_{t^*})^2\Big) \int_{\mc{B}_{t^*}^c} \Big(1 + T_{01}^{\rho}(y) + T_{21}^{\rho}(h,y)\Big) \phi_{0,\bar{G}_{t^*}^{-1}}(h) \,dh \\
			& \leq C\bar{\rho}(y) \int_{\mc{B}_{t^*}^c} \Big(1 + T_{01}^{\rho}(y) + T_{21}^{\rho}(h,y)\Big) \phi_{0,\bar{G}_{t^*}^{-1}}(h) \,dh \\
			& \leq C\bar{\rho}(y) \int_{\mc{B}_{t^*}^c} \Big(1 + \|h\|^2 \lambda_{\max}(-\nabla^2\mu_{t^*})\Big) \phi_{0,\bar{G}_{t^*}^{-1}}(h) \,dh \\
			& \leq C\bar{\rho}(y) \delta_n^2.
		\end{aligned}
	\end{equation*}
	This completes the proof.
	
	\subsection{Proof of Proposition~\ref{prop:high-gradient-peaks}}
	Throughout this proof we will assume that $n \in \mathbb{N}$ is large enough that $\lambda_n \geq 1$. Our goal is to upper bound the expected number of peaks in the high-gradient region. An exact identity for this expectation is given by the Kac-Rice formula:
	\begin{equation*}
		\begin{aligned}
			\E[N_u(\mc{G}_{\varepsilon_n})]
			& = \int_{\mc{G}_{\varepsilon_n}} \int_{u}^{\infty} \rho(t,y) \,dy \,dt \\
			& = \int_{\mc{G}_{\varepsilon_n}} \int_{u}^{\infty} \E[\det(-\nabla^2 Y_t) \cdot \1(\nabla^2 Y_t \prec 0)|\nabla Y_t = 0,Y_t = y] \cdot f_{\nabla Y_t}(0) \cdot f_{Y_t}(y) \,dy \,dt. 
		\end{aligned}
	\end{equation*}
	Intuitively, this should be small relative to $\E[N_u(\mc{T}_{\varepsilon_n})]$ because $\|\E[\nabla Y_t]\| = \|\nabla \mu_t\|$ is growing at all $t \in \mc{G}_{\varepsilon_n}$. However, obtaining the correct bound requires careful handling of the determinant and height density terms above. The determinant term can be upper bounded using Assumption~\ref{asmp:bounded-hessian}, which upper bounds the curvature of the signal:
	\begin{equation*}
		\begin{aligned}
			\E[\det(-\nabla^2 Y_t) \cdot \1(\nabla^2 Y_t \prec 0)|\nabla Y_t = 0,Y_t = y] 
			& = \int \det(H_{t|y} + R) \cdot \1(H_{t|y} + R \succ 0) f_{R_t}(R) \,dR \\
			& \leq \int (\|H_{t|y}\|^d + \|R\|^d) f_{R_t}(R) \,dR \\
			& \leq C(\det(H_{t|y}) + 1) \\
			& \leq C(\|\nabla^2 \mu_t\|^d + |y - \mu_t|^d + \|\nabla \mu_t\|^d + 1) \\
			& \leq C(\lambda_n^d + |y - \mu_t|^d + \|\nabla \mu_t\|^d).
		\end{aligned}
	\end{equation*}
	Plugging this in to our exact expression for $\E[N_u(\mc{G}_{\varepsilon_n})]$ and marginalizing over $y \in (u,\infty)$ yields the upper bound
	\begin{equation}
		\label{pf:high-gradient-peaks-2}
		\E[N_u(\mc{G}_{\varepsilon_n})] \leq C \int_{\mc{G}_{\varepsilon_n}} (\lambda_n^d + (u - \mu_t)_{+}^d + \|\nabla \mu_t\|^d) \Psi(u - \mu_t) f_{\nabla Y_t}(0) \,dt,
	\end{equation}
	where we have used~\eqref{eqn:normalized-overshoot-moments} to upper bound $\int_{u}^{\infty} |y - \mu_t|^d f_{Y_t}(y) \,dy \leq C(1 + (u - \mu_t)_{+}^d)\Psi(u - \mu_t)$. 	
	
	From here, we will apply Assumption~\ref{asmp:well-separated} to upper bound the density of the gradient $f_{\nabla Y_t}(0)$. The precise nature of the analysis, and how it contributes to the ultimate upper bound, depends on how far $t$ is from $T^*$. So we further partition $\mc{G}_{\varepsilon_n}$ into two subregions, and bound the expectation over each. By~\eqref{eqn:signal-taylor-expansion}, there exist constants $\varepsilon_0 > 0, n_0$ such that for all $n \geq n_0$,
	\begin{equation*}
		\nabla \mu_{t^* + h}' \Lambda_{t^*}^{-1} \nabla \mu_{t^* + h} \geq \frac{3}{4}h'\nabla^2\mu_{t^*}\Lambda_{t^*}^{-1}\nabla^2\mu_{t^*}h, \quad \|\nabla \mu_{t^* + h}\| \leq 2 \|\nabla^2\mu_{t^*}\| \|h\|, \quad 0 \leq \mu_{t^*} - \mu_{t^* + h} \leq 1/2,
	\end{equation*}
	for all $t^* \in T^*, h \in B_d(0,\varepsilon_0)$. We partition $\mc{G}_{\varepsilon_n}$ into $\mc{G}_{\varepsilon_n}^{0}$ and $\mc{G}_{\varepsilon_n}^{1}$, where $\mc{G}_{\varepsilon_n}^0$ consists of points $t \in \mc{G}_{\varepsilon_n}$ that are within distance $\varepsilon_0$ of $T^*$, while $\mc{G}_{\varepsilon_n}^1$ consists of all other points $t \in \mc{G}_{\varepsilon_n}$:  
	\begin{equation}
		\mc{G}_{\varepsilon_n}^{0} = \mc{G}_{\varepsilon_n} \setminus \mc{G}_{\varepsilon_0}, \quad \textrm{and} \quad \mc{G}_{\varepsilon_n}^{1} = \mc{G}_{\varepsilon_n} \cap \mc{G}_{\varepsilon_0}.
	\end{equation}
	By Assumption~\ref{asmp:well-separated}, there exists a constant $n_1$ such that $\{g_{n}(\varepsilon_0)\}^2 \geq 6\sigma_1^2(\log |\mc{T}| + (d + 2) \log \lambda_n)$ for all $n \geq n_1$. Moreover, there exists a constant $n_2$ such that $3\{g_n(\varepsilon_0)\}^2 - 3d \sigma_1^2 \log g_n(\varepsilon_0) \geq \{g_n(\varepsilon_0)\}^2$ for all $n \geq n_2$. Hereafter we will assume $n \geq \max(n_0,n_1,n_2)$.
	
	\paragraph{Upper bound on $\E[N(\mc{G}_{\varepsilon_n}^0)]$.} Let $\mc{G}_{t^*} := B(t^*,\varepsilon_0) \setminus B(t^*,\varepsilon_n)$.  For all $t \in \mc{G}_{t^*}$, $(\lambda_n^d + (u - \mu_t)_{+}^d + \|\nabla \mu_t\|^d) \leq C\det(\bar{H}_{t^*})$. Therefore, the expected number of points in $\mc{G}_{t^*}$ is at most
	\begin{equation*}
		\begin{aligned}
			C \det(\bar{H}_{t^*}) \int_{\mc{G}_{t^*}}  \Psi(u - \mu_t) f_{\nabla Y_t}(0) \,dt 
			& \leq C \det(\bar{H}_{t^*}) \cdot \Psi(u - \mu_{t^*})  \cdot \int_{\mc{G}_{t^*}} f_{\nabla Y_t}(0) \,dt \\
			& \leq C \det(\bar{H}_{t^*}) \cdot \Psi(u - \mu_{t^*})  \cdot \int_{B(0,\varepsilon_0) \setminus B(0,\varepsilon_n) }\exp(-\frac{3}{8}h'\nabla^2\mu_{t^*}\Lambda_{t^*}^{-1}\nabla^2\mu_{t^*}h) \,dh \\
			&\leq C \det(\bar{H}_{t^*}) \cdot \Psi(u - \mu_{t^*}) \cdot \int_{B(0,\varepsilon_0) \setminus B(0,\varepsilon_n)} \exp(-\frac{3\lambda_n^2}{8\sigma_1^2}\|h\|^2) \,dh \\
			&\overset{(i)}{\leq} C \cdot \Psi(u - \mu_{t^*}) \cdot \int_{\Rd \setminus B(0,\sqrt{5 \log \lambda_n})} \exp(-\frac{1}{2}\|z\|^2) \,dz \\
			&\overset{(ii)}{\leq} C \cdot \Psi(u - \mu_{t^*}) \cdot \{\log(\lambda_n)\}^{d/2} \delta_n^{5/2} \\
			&\leq C \cdot \Psi(u - \mu_{t^*}) \cdot \delta_n^2,
		\end{aligned}
	\end{equation*}
	where $(i)$ follows by changing variables from $h$ to $z = \sqrt{3\lambda_n/4\sigma_1^2}h$, and $(ii)$ follows by a Chernoff bound as stated in~\eqref{eqn:gaussian-integrals-concentration}. 
	Summing over $t^* \in T^*$, we conclude that $\E[N(\mc{G}_{\varepsilon_n}^0)] \leq C \sum_{t^* \in T^*} \bar{\E}[N(\mc{S}_{t^*})] \delta_n^2$.
	
	\paragraph{Upper bound on $\E[N(\mc{G}_{\varepsilon_n}^1)]$.}
	To upper bound  $\E[N(\mc{G}_{\varepsilon_n}^1)]$, we will want to invoke Assumption~\ref{asmp:interior-global-max} to upper bound the product of the determinant term and the Gaussian survival function in~\eqref{pf:high-gradient-peaks-2}. Let $t^{**} = \argmax_{t^* \in T^*} \mu_{t}$. At all points $t \in \mc{T}$ such that $u - \mu_t < \sqrt{d}$, 
	\begin{equation}
		(\lambda_n^d + (u - \mu_{t})_{+}^d + \|\nabla \mu_t\|^d) \Psi(u - \mu_{t}) \leq C (\lambda_n^d + \|\nabla \mu_t\|^d) \Psi(u - \mu_{t}) \leq C (\lambda_n^d + \|\nabla \mu_t\|^d) \cdot \Psi(u - \mu_{t^{**}}),
	\end{equation}
	with the last inequality following by Assumption~\ref{asmp:interior-global-max}. 
	Otherwise if $u - \mu_t > \sqrt{d}$ then $\mu_{t} \mapsto ((u - \mu_t)_{+}^d + \lambda_n^d)\Psi(u - \mu_t)$ is monotonically increasing in $\mu_t$, and so for all such points $t \in \mc{T}$
	\begin{equation}
		(\lambda_n^d + (u - \mu_{t})_{+}^d + \|\nabla \mu_t\|^d) \Psi(u - \mu_{t}) \leq C (\lambda_n^d + (u - \mu_{t^{**}})_{+}^d + \|\nabla \mu_t\|^d) \cdot \Psi(u - \mu_{t^{**}}).
	\end{equation}
	Therefore the upper bound in~\eqref{pf:high-gradient-peaks-2} is at most
	\begin{equation}
		\label{pf:high-gradient-peaks-1}
		C \int_{\mc{G}_{\varepsilon_n}^{1}} (\lambda_n^d + (u - \mu_{t^{**}})_{+}^d + \|\nabla \mu_t\|^d) \Psi(u - \mu_{t^{**}}) f_{\nabla Y_t}(0) \,dt \leq C \int_{\mc{G}_{\varepsilon_n}^{1}} (\det(\bar{H}_{t^{**}})  + \|\nabla \mu_t\|^d) \Psi(u - \mu_{t^{**}}) f_{\nabla Y_t}(0) \,dt.
	\end{equation}	
	We now apply Assumption~\ref{asmp:well-separated}, which lower bounds $\inf_{t \in \mc{G}_{\varepsilon_n}^1} \|\nabla \mu_t\| \geq g_n(\varepsilon_0)$. As a result,
	\begin{equation*}
		\inf_{t \in \mc{G}_{\varepsilon_n}^{1}} \|\nabla\mu_t\|^d \cdot f_{\nabla Y_t}(0) \leq C \exp\Big(-\frac{\|\nabla\mu_t\|^2}{2\sigma_{1}^2} + d\log(\|\nabla\mu_t\|)\Big) \leq C \exp\Big(-\frac{1}{3\sigma_{1}^2} \{g_n(\varepsilon_0)\}^2\Big),
	\end{equation*}
	with the latter inequality following because $3\{g_n(\varepsilon_0)\}^2 - 3d \sigma_1^2 \log g_n(\varepsilon_0) \geq \{g_n(\varepsilon_0)\}^2$. As a result~\eqref{pf:high-gradient-peaks-1} is further upper bounded by
	\begin{equation*}
		\begin{aligned}
			C \cdot |\mc{T}| \Psi\big(u - \mu_{t^{**}}\big) \det(\bar{H}_{t^{**}}) \cdot \exp\Big(-\frac{1}{3\sigma_{1}^2} \{g_n(\varepsilon_0)\}^2\Big)
			& \leq C \cdot |\mc{T}| \cdot \lambda_n^d \cdot \Psi\big(u - \mu_{t^{**}}\big)  \cdot \exp\Big(-\frac{1}{3\sigma_{1}^2} \{g_n(\varepsilon_0)\}^2\Big) \\
			& \leq C \cdot \Psi\big(u - \mu_{t^{**}}\big)   \exp\Big(- \frac{1}{3\sigma_1^2} \{g_n(\varepsilon_0)\}^2 + \log |\mc{T}| + d\log \lambda_n \Big) \\
			& \leq C \cdot \Psi\big(u - \mu_{t^{**}}\big)  \delta_n^2, 
		\end{aligned}
	\end{equation*}
	with the final inequality following because $\{g_{n}(\varepsilon_0)\}^2 \geq 6\sigma_1^2(\log |\mc{T}| + (d + 2) \log \lambda_n)$. Thus $\E[N(\mc{G}_{\varepsilon_n}^1)] \leq C \bar{\E}[N(\mc{S}_{t^{**}})] \delta_n^2 \leq C \delta_n^2 \cdot \sum_{t^* \in T^*} \bar{E}[N(\mc{S}_{t^{*}})]$. Combining the upper bounds on $\E[N(\mc{G}_{\varepsilon_n}^0)]$ and $\E[N(\mc{G}_{\varepsilon_n}^1)]$ yields the claim of Proposition~\ref{prop:high-gradient-peaks}.
	
	\section{Proofs for Section~\ref{sec:peak-distribution}}
	Section~\ref{sec:peak-distribution} contains four main results: Theorem~\ref{thm:approximate-joint-distribution} on the asymptotic conditional distribution of peaks; Theorem~\ref{thm:pivot} on the asymptotic distribution of pivots; Corollary~\ref{cor:conditional-coverage} on conditional coverage; and Theorem~\ref{thm:miscoverage} which bounds the asymptotic PCMR. The proof of Theorem~\ref{thm:approximate-joint-distribution} is given in Sections~\ref{subsec:pf-approximate-joint-distribution}-\ref{subsec:pf-no-more-than-one-process-peak-per-signal-peak}. The proof of Theorem~\ref{thm:pivot} is given in Sections~\ref{subsec:pf-pivot-hessian}-\ref{subsec:pf-approximate-density-studentized-peak}. The proof of Corollary~\ref{cor:conditional-coverage} is given in~\ref{subsec:pf-conditional-coverage}. The proof of Theorem~\ref{thm:miscoverage} is given in Section~\ref{subsec:pf-miscoverage}.

	\subsection{Proof of Theorem~\ref{thm:approximate-joint-distribution}}
	\label{subsec:pf-approximate-joint-distribution}
	As mentioned in the main text, under the assumptions of Section~\ref{subsec:assumptions}, as $\delta_n \to 0$ the $\P$-probability of observing a unique peak $\hat{t} \in \wh{T}$ that consistently estimates $t^*$ and has height $\wh{Y} \in \mu_{t^*} \pm \Delta_n$ tends to one. This is a consequence of the following stronger result.
	\begin{proposition}
		\label{prop:no-more-than-one-process-peak-per-signal-peak}
		Under Assumptions~\ref{asmp:signal-holder}-\ref{asmp:curvature-asymptotics}, for all $n \in \mathbb{N}$ sufficiently large:
		\begin{equation}
			\label{eqn:no-more-than-one-peak-1}
			\begin{aligned}
				\sup_{\mc{A}}\frac{|\P\big(N(\mc{A}) = 1, N(\mc{S}_{t^*}) = 1\big) - \E[N(\mc{A})]|}{\E[N(\mc{A})]} \leq C\exp\big(-c \lambda_n^2\big),
			\end{aligned}
		\end{equation}
		where the supremum is over all measurable subsets $\mc{A}$ of $\mc{S}_{t^*}$.
	\end{proposition}
	The proof of Proposition~\ref{prop:no-more-than-one-process-peak-per-signal-peak} is given in Section~\ref{subsec:pf-no-more-than-one-process-peak-per-signal-peak}. Note that taking $u = -\infty, \mc{A} = \mc{S}_{t^*}$ in Proposition~\ref{prop:no-more-than-one-process-peak-per-signal-peak} shows that  
		\begin{equation*}
		\P(N_{\mu_{t^*} \pm \Delta_n}(\mc{B}_{t^*}) = 1) \geq \E[N_{\mu_{t^*} \pm \Delta_n}(\mc{B}_{t^*}))]\big(1 - O(\exp(-\lambda_n^2))\big) \geq 1 - O(\delta_n^2),
	\end{equation*}	
	with the latter inequality following from Proposition~\ref{prop:expectation-counting-process}. In other words, each true peak $t^*$ will be $\varepsilon_n$-consistently estimated by a unique observed peak $\hat{t} \in \wh{T}$, with asymptotic probability one. 
	
	We now use Proposition~\ref{prop:no-more-than-one-process-peak-per-signal-peak} to  derive~\eqref{eqn:joint-density-error}, which upper bounds the relative error between $p(t,y)$ and $\bar{p}(t,y)$. First, application of Proposition~\ref{prop:no-more-than-one-process-peak-per-signal-peak} with $\mc{A} = \mc{S}_{t^*}$ upper bounds $|\P(\mc{S}_{t^*} = 1)/\E[N(\mc{S}_{t^*})] - 1|$. A second application of Proposition~\ref{prop:no-more-than-one-process-peak-per-signal-peak} with $\mc{A} = B_d(t,r) \times [y - r, y + r]$ then implies
	\begin{equation}
		\begin{aligned}
			\frac{\Big|\P(N(\mc{A}_r) = 1, N(\mc{S}_{t^*}) = 1)/\E[N(\mc{S}_{t^*})] - \E[N(\mc{A}_r)]/\E[N(\mc{S}_{t^*})]\Big|}{E[N(\mc{A}_r)]/\E[N(\mc{S}_{t^*})]} \leq C\exp\big(-c \lambda_n^2\big).
		\end{aligned}
	\end{equation}
	Combining these two results and taking the limit as $r \to 0$ gives
	\begin{equation}
		\label{eqn:approximate-distribution-after-selection-1}
		\begin{aligned}
			\frac{\Big|p(t,y) - \rho(t,y)/\E[N(\mc{S}_{t^*})]\Big|}{\rho(t,y)/\E[N(\mc{S}_{t^*})]} \leq C\exp\big(-c \lambda_n^2\big).
		\end{aligned}
	\end{equation}
	The local expansion $\bar{p}(t,y) = \bar{\rho}(t,y)/\bar{\E}[N(\mc{S}_{t^*})]$, with the bound on relative error~\eqref{eqn:joint-density-error} following from~\eqref{eqn:approximate-distribution-after-selection-1}, Theorem~\ref{thm:approximate-joint-intensity} and Proposition~\ref{prop:expectation-counting-process}. We do not give the specific calculations involved explicitly, because they are similar to the bookkeeping carried out explicitly in the proof of Theorem~\ref{thm:approximate-joint-intensity}. Similarly, the local expansion $\bar{p}(y) = \bar{\rho}(y)/\bar{\E}[N(\mc{S}_{t^*})]$, with the bound on relative error~\eqref{eqn:density-height-error} following from~\eqref{eqn:approximate-distribution-after-selection-1}, Lemma~\ref{lem:height-intensity} and Proposition~\ref{prop:expectation-counting-process}. 
	
	The same is essentially also true for $\bar{p}(t|y)$ and $\Err_{p}(t|y)$, but the details are slightly more complicated. Recall that in the proof of Lemma~\ref{lem:height-intensity}, we introduced
	\begin{equation*}
		\tilde{\rho}(y) = \frac{\sqrt{\det(\bar{H}_{t^*})}}{\sqrt{(2\pi)\det(-\nabla^2\mu_{t^*})}} \cdot \exp\Big(-\frac{1}{2}(y - \mu_{t^*})^2\Big)\Big(1 + \frac{1}{2}T_{01}^{\rho}(y)\Big).
	\end{equation*}
	Now additionally define
	\begin{equation}
		\label{pf:approximate-joint-distribution-6}
		\begin{aligned}
			\tilde{p}(t^* + h|y) := \frac{\bar{\rho}(t^* + h,y)}{\tilde{\rho}(y)} = \sqrt{\frac{\det(G_{t^*|y})}{(2\pi)^d}} \exp\Big(-\frac{1}{2}h'G_{t^*|y}h\Big) \bigg(\frac{1 + T_{10}^{\rho}(h) - \frac{1}{2} R_{30}^{\rho}(h) + T_{01}^{\rho}(y) + T_{21}^{\rho}(h,y)}{1 + \frac{1}{2}T_{01}^{\rho}(y)}\bigg).
		\end{aligned}
	\end{equation}
	Inspecting the proof of Lemma~\ref{lem:height-intensity} shows that $|\tilde{\rho}(y) - \bar{\rho}(y)|/\tilde{\rho}(y) \leq \Err_{\rho}(y)$. Combining this observation with Lemma~\ref{lem:height-intensity},~\eqref{eqn:approximate-distribution-after-selection-1} and Theorem~\ref{thm:approximate-joint-intensity} implies that for all $n \in \mathbb{N}$ sufficiently large:
	\begin{equation}
		\label{pf:approximate-joint-distribution-2}
		\frac{|\tilde{p}(t^* + h|y) - p(t^* + h|y)|}{\tilde{p}(t^* + h|y)} \leq C\big(\delta_n^2 + \Err_{\rho}(h,y) + \Err_{\rho}(y)\big).
	\end{equation}
	Thus $\tilde{p}(t|y)$ is a second-order accurate local expansion of $p(t|y)$. An upper bound on the relative error between $\tilde{p}(t|y)$ and $\bar{p}(t|y)$ follows from Taylor expansion of several simple functions involving first-order terms in $\tilde{p}(t|y)$: for all $n \in \mathbb{N}$ sufficiently large,
	\begin{align*}
		\Big|\frac{1 + T_{01}^{\rho}(y)}{1 + \frac{1}{2}T_{01}^{\rho}(y)} - \Big(1 + \frac{1}{2}T_{01}^{\rho}(y)\Big)\Big| 
		& \leq  C\{T_{01}^{\rho}(y)\}^2 \leq C\cdot\Err_{\rho}(h,y), \\
		\frac{\Big|\sqrt{\det(G_{t^*|y})} - \sqrt{\det(\bar{G}_{t^*})}(1 + \frac{1}{2}T_{01}^{\rho}(y))\Big|}{\sqrt{\det(\bar{G}_{t^*})}} 
		& \leq C\{T_{01}^{\rho}(y)\}^2 \leq C\cdot\Err_{\rho}(h,y), \\
		\Big|\exp(T_{21}^{\rho}(h,y)) - (1 + T_{21}^{\rho}(h,y))\Big| 
		& \leq C\{T_{21}^{\rho}(h,y)\}^2 \leq C\cdot\Err_{\rho}(h,y),
	\end{align*}
	and therefore
	\begin{equation}
		\label{pf:approximate-joint-distribution-3}
		\frac{|\bar{p}(t^* + h|y) - \tilde{p}(t^* + h|y)|}{\tilde{p}(t^* + h|y)} \leq C \cdot \Err_{\rho}(h,y).
	\end{equation}
	Combining~\eqref{pf:approximate-joint-distribution-2} and~\eqref{pf:approximate-joint-distribution-3} yields~\eqref{eqn:density-location-error}.
		
	\subsection{Proof of Proposition~\ref{prop:no-more-than-one-process-peak-per-signal-peak}}
	\label{subsec:pf-no-more-than-one-process-peak-per-signal-peak}
	Our goal is to upper bound
	\begin{equation*}
		\Big|\E[N(\mc{A})] - \P(N(\mc{A}) = 1,N(\mc{S}_{t^*}) = 1)\Big|.
	\end{equation*}
	We begin with a deterministic upper bound: $N(\mc{A}) - \1(N(\mc{A}) = 1,N(\mc{S}_{t^*}) = 1)$ -- which is always non-negative -- is at most
	\begin{equation*}
		N(\mc{A}) - \1(N(\mc{A}) = 1,N(\mc{S}_{t^*}) = 1) \leq N(\mc{A}) \cdot \1(N(\mc{S}_{t^*}) > 1) \leq N(\mc{A}) \cdot \1\bigg(\sup_{s \in \mc{B}_{t^*}} \lambda_{\max}(\nabla^2 Y_s) \geq 0\bigg).
	\end{equation*}
	Here we are using the fact that if the maximum eigenvalue of $\nabla^2 Y_{s}$ over $s \in \mc{B}_{t^*}$ is negative, then $Y$ is strongly concave in $\mc{B}_{t^*}$ and hence has at most one local maximum. Decomposing the Hessian $-\nabla^2 Y_s = H_s^{t}(y) + R_s^{t}$ -- where we recall the projection/residual notation from~\eqref{eqn:hessian-decomposition} -- allows us to further upper bound the final quantity on the right hand side by the counting process
	\begin{equation*}
		N\Big((t,y) \in \mc{A}: \nabla Y_t = 0, Y_t = y, \sup_{s \in \mc{B}_{t^*}} \lambda_{\max}(R_s^t) \geq \inf_{s \in \mc{B}_{t^*}} \lambda_{\min}(H_s^t(y)) \Big),
	\end{equation*} 
	The results of Section~\ref{subsec:hessian-asymptotics}, specifically~\eqref{eqn:deterministic-hessian-relative-error}, further imply that  for all $n \in \mathbb{N}$ large enough such that $\lambda_n \geq C \Delta_n$, for all $t \in \mc{B}_{t^*},y \in \mc{I}_{t^*}$:
	$$
	\inf_{s \in \mc{B}_{t^*}} \lambda_{\min}(H_s^t(y)) \geq \lambda_{\min}(\bar{H}_{t^*}) - C \Delta_n \geq \lambda_n - C \Delta_n \geq \frac{1}{2}\lambda_n,
	$$ 
	Taking expectations, we conclude that
	\begin{equation*}
		\Big|\E[N(\mc{A})] - \P(N(\mc{A}) = 1,N(\mc{S}_{t^*}) = 1)\Big| \leq \E\Big[N\Big((t,y) \in \mc{A}: \nabla Y_t = 0, Y_t = y, \sup_{s \in \mc{B}_{t^*}} \lambda_{\max}(R_s^t) \geq \frac{1}{2}\lambda_n \Big)\Big]
	\end{equation*}
	By the Kac-Rice formula, the expectation on the right hand side of the previous display is
	\begin{equation*}
		\int_{\mc{A}} \E\Big[\det(-\nabla^2 Y_{t}) \cdot \1\big(\nabla^2 Y_t \prec 0, \sup_{s \in \mc{B}_{t^*}} \lambda_{\max}(R_s^t) \geq \frac{1}{2} \lambda_{n}\big)|Y_t = y,\nabla Y_t = 0\Big] \cdot f_{Y_t}(y) \cdot f_{\nabla Y_t}(0) \,dy \,dt.
	\end{equation*}
	We upper bound this via H\"{o}lder's inequality:
	\begin{equation*}
		\begin{aligned}
			& \int_{\mc{A}} \E\Big[\det(-\nabla^2 Y_{t}) \cdot \1\big(\nabla^2 Y_t \prec 0, \sup_{s \in \mc{B}_{t^*}} \lambda_{\max}(R_s^t) \geq \frac{1}{2} \lambda_{n}\big)|Y_t = y,\nabla Y_t = 0\Big] \cdot f_{Y_t}(y) \cdot f_{\nabla Y_t}(0) \,dy \,dt \\
			& \leq \int_{\mc{A}} \Big\{\E\Big[|\det(-\nabla^2 Y_{t})|^{q_1}|Y_t = y,\nabla Y_t = 0\Big]\Big\}^{1/q_1} \cdot \Big\{\P\Big(\sup_{s \in \mc{B}_{t^*}} \lambda_{\max}(R_s^t) \geq \frac{1}{2}\lambda_n|Y_t = y,\nabla Y_t = 0\Big)\Big\}^{1/q_2} \cdot f_{Y_t}(y) \cdot f_{\nabla Y_t}(0) \,dy \,dt, \\
			& = \int_{\mc{A}} \Big\{\E\Big[|\det(-\nabla^2 Y_{t})|^{q_1}|Y_t = y,\nabla Y_t = 0\Big]\Big\}^{1/q_1} \cdot \Big\{\P\Big(\sup_{s \in \mc{B}_{t^*}} \lambda_{\max}(R_s^t) \geq \frac{1}{2}\lambda_n\Big)\Big\}^{1/q_2} \cdot f_{Y_t}(y) \cdot f_{\nabla Y_t}(0) \,dy \,dt,
		\end{aligned}
	\end{equation*}               
	where $q_1 > 1, q_2 > 1, 1/q_1 + 1/q_2 = 1$ are H\"{o}lder conjugates, and in the equality we have used the fact that $R_s^t$ is independent of $(Y_t,\nabla Y_t)$. We proceed to upper bound each term in the previous display. In the proof of Lemma~\ref{lem:approximation-peak-intensity}, it was shown that for all $(t,y) \in \mc{S}_{t^*}$, for all $n \in \mathbb{N}$ sufficiently large,
	\begin{equation*}
		\E\Big[|\det(-\nabla^2 Y_{t})|^{q_1}|Y_t = y,\nabla Y_t = 0\Big]\Big\}^{1/q_1} \leq C(\lambda_n^d + 1) \leq C |\det (\bar{H}_{t^*})|.
	\end{equation*} 
	Lemma~\ref{lem:asymptotic-expansion-density} implies that for all $n \in \mathbb{N}$ sufficiently large,
	\begin{equation*}
		f_{Y_t}(y) \cdot f_{\nabla Y_{t}}(0) \leq C\exp\Big(-\frac{(y - \mu_{t^*})^2}{2} - \frac{(t - t^*)'G_{t^*|y}(t - t^*)}{2}\Big).
	\end{equation*}              
	Finally, it is easy to see that~\eqref{eqn:hessian-residual-tail-behavior} holds with $\xi_n$ replaced by $\frac{1}{2}\lambda_n^2$. Taking $\eta = \sigma_{22}^2/8$ in that equation yields
	\begin{equation*}
		\P\Big(\sup_{s \in \mc{B}_{t^*}} \lambda_{\max}(R_s^t) \geq  \frac{1}{2}\lambda_n^2\Big) \leq C\exp\Big(-\frac{\lambda_n^2}{8\sigma_{22}^2}\Big).
	\end{equation*}
	Thus, for all $n \in \mathbb{N}$ sufficiently large,
	\begin{equation*}
		\begin{aligned}
			& \Big|\E\big[N(\mc{A})\big] - \P\big(N(\mc{A}) = 1,N(\mc{S}_{t^*}) = 1\big)\Big| \\
			& \leq C\exp\Big(-\frac{\lambda_n^2}{8q_2\sigma_{22}^2}\Big)\int_{\mc{A}} |\det(\bar{H}_{t^*})| \cdot \exp\Big(-\frac{(y - \mu_{t^*})^2}{2} - \frac{(t - t^*)'G_{t^*|y}(t - t^*)}{2}\Big) \,dy \,dt \\
			& \leq C\exp\Big(-\frac{\lambda_n^2}{8q_2\sigma_{22}^2}\Big) \int_{\mc{A}} \bar{\rho}(t,y) \,dy \,dt \\
			& \leq C\exp\Big(-\frac{\lambda_n^2}{8q_2\sigma_{22}^2}\Big) \int_{\mc{A}} \rho(t,y) \,dy \,dt \\
			& = C\exp(-c\lambda_n^2) \cdot \E[N(\mc{A})],
		\end{aligned}
	\end{equation*}
	which is exactly the claim of Proposition~\ref{prop:no-more-than-one-process-peak-per-signal-peak}.
	
	\subsection{Proof of Theorem~\ref{thm:pivot}: deterministic limit of observed Hessian}
	\label{subsec:pf-pivot-hessian}
	 In this section, we build on the results of Section~\ref{subsec:hessian-asymptotics} to show that under $\Q^{\mc{S}_{t^*}}$ the relative error between $\wh{H}$ and $\bar{H}_{t^*}$ converges to $0$ in conditional probability, given $\{\hat{t} = t, \wh{Y} = y\}$.

	We begin by writing $\wh{H}$ in terms of the projection/residual decomposition defined in Section~\ref{subsec:hessian-asymptotics}: 
	$$\wh{H} = H_{\hat{t}|\wh{Y}} + \wh{R},$$
	where $\wh{R} = R_{\hat{t}}$. The joint distribution of the triple $(\wh{R},\hat{t},\wh{Y})$ under the law $\Q^{\mc{S}_{t^*}}$ is continuous, and we denote its density by $p(R,t,y)$, and define the conditional density $p(R|t,y) := p(R,t,y)/p(t,y)$, where $p(t,y)$ is the joint density of location and height defined in~\eqref{eqn:density-after-selection}.

	Now we recall the result of~\eqref{eqn:deterministic-hessian-relative-error}:
	\begin{equation*}
		\sup_{(t,y) \in \mc{S}_{t^*}} \|H_{t|y} - \bar{H}_{t^*}\| \leq C \Delta_n.
	\end{equation*} 
	If additionally $R \in \mc{R}_{t^*} = B_{d \times d}(0,\xi_n) \cap \R_{sym}^{d \times d}$ then 
	\begin{equation}
		\label{eqn:hessian-deterministic-equivalent}
		\sup_{(t,y) \in \mc{S}_{t^*}, R \in \mc{R}_{t^*}} \frac{\|H_{t|y} + R - \bar{H}_{t^*}\|}{\|\bar{H}_{t^*}\|} \leq C\frac{(\Delta_n + \xi_n)}{\lambda_n} \leq C\varepsilon_n.
	\end{equation}
	It remains to upper bound the probability that $\wh{R} \in \mc{R}_{t^*}^c$ under $\Q^{\mc{S}_{t^*}}$ and given $\{\hat{t} = t,\wh{Y} = y\}$. This probability can be written as 
	\begin{equation*}
		\Q^{\mc{S}_{t^*}}(\wh{R} \in \mc{R}_{t^*}^c|\hat{t} = t,\wh{Y} = y) = \lim_{r \to 0} \frac{\P(\wh{R} \in \mc{R}_{t^*}^c, N(\mc{A}_r) = 1)}{\P(N(\mc{S}_{t^*}) = 1, N(\mc{A}_r) = 1)}.
	\end{equation*}
	Here $\mc{A}_r = B_d(t,r) \times [y - r, y + r]$. By Proposition~\ref{prop:no-more-than-one-process-peak-per-signal-peak}, for all $n \in \mathbb{N}$ sufficiently large the denominator is at least $\frac{1}{2} \cdot \E[N(\mc{A}_r)]$ for all $r > 0$. On the other hand, the numerator is upper bounded by
	$$
	\E\bigg[N\Big((t,y) \in \mc{A}_r: \nabla Y_t = 0,Y_t = y, \nabla^2 Y_t \prec 0, \|R_t\| > \xi_n\Big)\bigg] \leq C \delta_n^2 \cdot \E[N(\mc{A}_r)],
	$$
	with the inequality following by arguments similar to those used to prove Proposition~\ref{prop:no-more-than-one-process-peak-per-signal-peak}. As a result, for all $n \in \mathbb{N}$ sufficiently large
	\begin{equation}
		\label{eqn:hessian-residual-truncation-error}
		\Q^{\mc{S}_{t^*}}(\wh{R} \in \mc{R}_{t^*}^c|\hat{t} = t,\wh{Y} = y) \leq C \delta_n^2,
	\end{equation}
	and in combination with~\eqref{eqn:hessian-deterministic-equivalent} this implies
	\begin{equation*}
		\sup_{(t,y) \in \mc{S}_{t^*}} \Q^{\mc{S}_{t^*}}\Big(\frac{\|\wh{H} - \bar{H}_{t^*}\|}{\|\bar{H}_{t^*}\|} \geq C\varepsilon_n|\hat{t} = t,\wh{Y} = y\Big) \leq C \delta_n^2.
	\end{equation*}
	

	\subsection{Proof of Theorem~\ref{thm:pivot}: pivot for height}
	We will work under the assumptions of Theorem~\ref{thm:pivot} and establish~\eqref{eqn:height-pivot}. Lemma~\ref{lem:height-pivot} shows that $\bar{\S}_{\mu}(\wh{Y})$ is approximately pivotal.
	\begin{lemma}
		\label{lem:height-pivot}
		Under the assumptions of Theorem~\ref{thm:pivot},
		\begin{equation}
			\label{eqn:height-pivot-3}
			\Big|\Q^{\mc{S}_{t^*}}\big(\bar{\S}_{\mu}(\wh{Y}) \leq \alpha\big) - \alpha\Big| \leq C \big( |\bar{u}_{t^*} - \mu_{t^*}| \delta_n^2 + \delta_n^2\big) := \Err_{\bar{\S}}.
		\end{equation}
	\end{lemma}
	The proof of Lemma~\ref{lem:height-pivot} is given in Section~\ref{subsec:pf-height-pivot}. Lemma~\ref{lem:height-pivot-2} bounds the difference between $\wh{\S}_{\mu}$ and $\bar{\S}_{\mu}$ on a high probability set.
	\begin{lemma}
		\label{lem:height-pivot-2}
		For all $(t,y) \in \mc{S}_{t^*}$ and $R \in \mc{R}_{t^*}$, 
		\begin{equation}
			\label{eqn:height-pivot-2}
			\Big|\wh{\S}_{\mu}(y,t,H_{t|y} + R) - \bar{\S}_{\mu}(y)\Big| \leq C(|\bar{u} - \mu_{t^*}| + \Delta_n + 1)\varepsilon_n \delta_n := {\rm PlugIn}_{\S}.
		\end{equation}
	\end{lemma}
	The proof of Lemma~\ref{lem:height-pivot-2} is given in Section~\ref{subsec:pf-height-pivot-2}. Now we use Lemmas~\ref{lem:height-pivot} and~\ref{lem:height-pivot-2} to establish~\eqref{eqn:height-pivot}. Recall that the quantity we are interested in is
	\begin{equation*}
		\Q^{\mc{S}_{t^*}}\Big(\wh{\S}_{\mu}(\wh{Y}, \hat{t}, \wh{H}) \leq \alpha\Big) = \int_{\mc{S}_{t^*}} \int_{\R_{sym}^{d \times d}} \1\{\wh{\S}_{\mu}(y, t, H_{t|y} + R) \leq \alpha\} \cdot p(R|t,y) \cdot p(t,y) \,dR \,dt \,dy.
	\end{equation*}
	We truncate this integral to the range at which the approximation of Lemma~\ref{lem:height-pivot-2} applies; applying this lemma and then Lemma~\ref{lem:height-pivot} gives:
	\begin{align*}
		\Q^{\mc{S}_{t^*}}\Big(\wh{\S}_{\mu}(\wh{Y}, \hat{t}, \wh{H}) \leq \alpha \Big)  
		& \geq \int_{\mc{S}_{t^*}} \int_{\mc{R}_{t^*}} \1\{\wh{\S}_{\mu}(y, t, H_{t|y} + R) \leq \alpha\} \cdot p(R|t,y) \cdot p(t,y) \,dR \,dt \,dy \\
		& \overset{(i)}{\geq} \int_{\mc{S}_{t^*}} \int_{\mc{R}_{t^*}} \1(\bar{\S}_{\mu}(y) \leq \alpha - \PlugIn_{\S}) \cdot p(R|t,y) \cdot p(t,y) \,dR \,dt \,dy \\
		& \geq \int_{\mc{S}_{t^*}} \1(\bar{\S}_{\mu}(y) \leq \alpha - \PlugIn_{\S})  p(t,y) \,dt \,dy - \sup_{(t,y) \in \mc{S}_{t^*}} \int_{\mc{R}_{t^*}^c} p(R|t,y) \,dR \\
		& \overset{(ii)}{\geq} \alpha - \PlugIn_{\S} - \Err_{\bar{\S}} - \sup_{(t,y) \in \mc{S}_{t^*}} \int_{\mc{R}_{t^*}^c} p(R|t,y) \,dR \\
		& \overset{(iii)}{\geq} \alpha - \PlugIn_{\S} - \Err_{\bar{\S}} - C\delta_n^2
	\end{align*}
	Above, (i) follows from Lemma~\ref{lem:height-pivot-2}, (ii) follows from Lemma~\ref{lem:height-pivot},  and (iii) follows from~\eqref{eqn:hessian-residual-truncation-error}. A symmetric upper bound follows from essentially equivalent reasoning. Notice that $\PlugIn_{\S} + \Err_{\bar{\S}} + C\delta_n^2 \leq C \cdot \Err_{\S}$. Thus, to establish~\eqref{eqn:height-pivot} it remains to prove Lemma~\ref{lem:height-pivot} and Lemma~\ref{lem:height-pivot-2}.
	
	\subsection{Proof of Lemma~\ref{lem:height-pivot}}
	\label{subsec:pf-height-pivot}
	The quantity we are after is
	\begin{equation*}
		\Q^{\mc{S}_{t^*}}(\bar{\S}_{\mu}(\wh{Y}) \leq \alpha) = \int_{\mc{S}_{t^*}} \1(\bar{\S}_{\mu}(y) \leq \alpha) \cdot p(t,y) \,dt \,dy = \int_{\mc{I}_{t^*}} \1(\bar{\S}_{\mu}(y) \leq \alpha) \cdot p(y) \,dy. 
	\end{equation*}
	Replacing $p(y)$ by $\bar{p}(y)$, and adjusting the limits of integration, we have
	\begin{equation*}
		\Big|\Q^{\mc{S}_{t^*}}(\bar{\S}_{\mu}(\wh{Y}) \leq \alpha) - \int_{u}^{\infty} \1(\bar{\S}_{\mu}(y) \leq \alpha) \cdot \bar{p}(y) \,dy\Big| \leq \int_{\mc{I}_{t^*}} |p(y) - \bar{p}(y)| \,dy +\int_{(u,\infty) \setminus \mc{I}_{t^*}} \bar{p}(y) \,dy.
	\end{equation*}
	By definition
	$
	\int_{u}^{\infty} \1(\bar{\S}_{\mu}(y) \leq \alpha) \cdot \bar{p}(y) \,dy = \bar{\S}_{\mu}(\bar{\S}_{\mu}^{-1}(\alpha)) = \alpha,
	$
	and so it remains only to bound the error terms in the previous display. In the proof of Proposition~\ref{prop:expectation-counting-process}, it is shown that 
	\begin{equation*}
		\int_{\mc{I}_{t^*}} |y - \bar{u}_{t^*}|^2 \bar{\rho}(y) \,dy \leq C \cdot \bar{\E}[N(\mc{S}_{t^*})] \quad \textrm{and} \quad \int_{(u,\infty) \setminus \mc{I}_{t^*}} \bar{\rho}(y) \,dy \leq C \cdot \bar{\E}[N(\mc{S}_{t^*})] \delta_n^3
	\end{equation*}
	Consequently, the error due to local expansion is at most
	\begin{equation*}
		\int_{\mc{I}_{t^*}} |p(y) - \bar{p}(y)| \,dy \leq \int_{\mc{I}_{t^*}} \Err_{p}(y) \bar{p}(y) \,dy \leq C \delta_n^2 \int_{\mc{I}_{t^*}} \big(|y - \bar{u}_{t^*}|^2 + |\bar{u}_{t^*} - \mu_{t^*}| + 1\big) \bar{p}(y) \,dy \leq C \delta_n^2( |\bar{u}_{t^*} - \mu_{t^*}|^2 + 1),
	\end{equation*}
	while  the error due to truncation is at most
	$$
	\int_{(u,\infty) \setminus \mc{I}_{t^*}} \bar{p}(y) \,dy = \int_{(u,\infty) \setminus \mc{I}_{t^*}} \frac{\bar{\rho}(y)}{\bar{\E}[N(\mc{S}_{t^*})]} \,dy \leq C \delta_n^3,
	$$
	which is negligible by comparison.
	
	\subsection{Proof of Lemma~\ref{lem:height-pivot-2}}
	\label{subsec:pf-height-pivot-2}
	
	We find it convenient to introduce the following notation:
	\begin{align*}
		m & = \mu_{t^*} + \frac{1}{2}\tr(\bar{H}_{t^*}^{-1}\Lambda) \\
		E_1(t,y,R) & = \frac{1}{2}\tr((\{H_{t|y} + R\}^{-1} - \bar{H}_{t^*}^{-1})\Lambda_{t}) \\
		E_2(t) & = \frac{1}{2}\tr(\bar{H}_{t^*}^{-1}\{\Lambda_t - \Lambda\}),
	\end{align*}
	so that the estimated survival function is
	\begin{equation*}
		\wh{\S}_{\mu}(y,t,H_{t|y} + R) = \frac{\Psi(y - m - E_1(t,y,R) - E_2(t))}{\Psi(u - m - E_1(t,y,R) - E_2(t))}.
	\end{equation*}
	For all $n \in \mathbb{N}$ sufficiently large, the following bounds on the magnitude of $|E_1(t,y,R)|$ and $|E_2(t)|$ hold uniformly over $(t,y) \in \mc{S}_{t^*}, R \in \mc{R}_{t^*}$:
	\begin{align*}
		|E_1(t,y,R)| \leq C \|\{H_{t|y} + R\}^{-1} - \bar{H}_{t^*}^{-1}\| \cdot \|\Lambda_t\| \overset{(i)}{\leq} C \frac{\|H_{t|y} + R - \bar{H}_{t^*}\|}{\|\bar{H}_{t^*}\|^2} \|\Lambda_t\| \overset{(ii)}{\leq} C \varepsilon_n \delta_n.
	\end{align*}
	where $(i)$ follows from a first-order Taylor expansion of $A \mapsto A^{-1}$ (stated in~\eqref{eqn:matrix-inverse-taylor-expansion} for completeness), while (ii) follows from~\eqref{eqn:hessian-deterministic-equivalent}; and more simply
	\begin{equation*}
		|E_2(t)| \leq \|\bar{H}_{t^*}^{-1}\| \cdot\|\Lambda_{t} - \Lambda\| \leq C \delta_n \varepsilon_n.
	\end{equation*}
	An implication of the upper bounds on $|E_1(t,y,R)|$ and $|E_2(t)|$ is that uniformly over $(t,y) \in \mc{S}_{t^*}$ and $R \in \mc{R}_{t^*}$, $|y - m| \cdot (|E_1(t,y,R)| + |E_2(t)|) \to 0$ and $|u - m| \cdot (|E_1(t,y,R)| + |E_2(t)|) \to 0$. We can therefore apply Lemma~\ref{lem:gaussian-survival-function-perturbation} to both numerator and denominator of $\wh{\S}_{\mu}(y,t,H_{t|y} + R)$, and conclude that for all $n \in \mathbb{N}$ sufficiently large
	\begin{equation}
		\begin{aligned}
			\frac{\Psi(y - m - E_1(t,y,R) - E_2(t)) - \Psi(y - m)}{\Psi(y - m)} & \leq C((y - m)_{+} + 1) \varepsilon_n \delta_n, \\
			\frac{\Psi(u - m - E_1(t,y,R) - E_2(t)) - \Psi(u - m)}{\Psi(u - m)} & \leq C((u - m)_{+} + 1) \varepsilon_n \delta_n,
		\end{aligned}
	\end{equation}
	and therefore
	\begin{equation}
		\label{pf:height-pivot-3}
		\Big|\wh{\S}_{\mu}(y,t,H_t(y) + R) - \bar{\S}_{\mu}(y)\Big| \leq C ((y - m)_{+} + (u - m)_{+} + 1) \varepsilon_n \delta_n \leq C(\Delta_n  + |\bar{u} - \mu_{t^*}| + 1)\varepsilon_n \delta_n.
	\end{equation}
	
	\subsection{Proof of Theorem~\ref{thm:pivot}: pivot for location}
	We work under the assumptions of Theorem~\ref{thm:pivot} and establish~\eqref{eqn:pivot-location}. Consider the \emph{studentized peak}
	\begin{equation*}
		\wh{Z} = Z_{\hat{t}}(\wh{Y},\wh{R}) = \wh{F}(\hat{t} - t^*), \quad \textrm{where} \quad\wh{F} = F_{\hat{t}}(\wh{Y},\wh{R}) =  \Lambda_{\hat{t}}^{-1/2} \wh{H} = \Lambda_{\hat{t}}^{-1/2} (H_{\hat{t}|\wh{Y}} + \wh{R})
	\end{equation*}
	The Wald pivot is simply $\wh{W} := \wh{W}_{t^*}(\hat{t}) = \wh{Z}'\wh{Z}$. We begin with a high-level outline of the derivation of~\eqref{eqn:pivot-location}. First we give a second-order accurate local expansion $\bar{p}(R|t,y)$ of the conditional density of  $\wh{R}$ given $\hat{t} = t,\wh{Y} = y$. Then, we use $\bar{p}(R|t,y)$ and Theorem~\ref{thm:approximate-joint-distribution} to derive an approximation to $p(t,R|y)$, the joint density of $(\wh{t},\wh{R})$ given $\wh{Y} = y$.  We use this to deduce first a local expansion of the density of $\wh{Z}$, and then a local expansion of the distribution of $\wh{W}$. Finally, this limiting distribution of $\wh{W}$ is a generalized chi-squared distribution, that can be approximated by a $\chi_d^2$ distribution -- with loss of second-order accuracy -- under the conditions of Theorem~\ref{thm:pivot}.
	
	\paragraph{Local expansion of density of Hessian.}
	As in the proof of Theorem~\ref{thm:approximate-joint-distribution}, we derive a local expansion of $p(R|t,y)$ by first relating this conditional density to the intensity function of a relevant counting process. Concretely, arguments similar to the proof of Proposition~\ref{prop:no-more-than-one-process-peak-per-signal-peak} can be used to show that at any $R \in \mc{R}_{t^*}, (t,y) \in \mc{S}_{t^*}$,
	\begin{equation*}
	\bigg|\frac{p(R,t,y) - \rho(R,t,y)/\E[N(\mc{S}_{t^*})]}{\rho(R,t,y)/\E[N(\mc{S}_{t^*})]}\bigg| \leq C\exp(-c\lambda_n^2).
	\end{equation*}
	where 
	\begin{equation}
		\rho(R,t,y) = \det(H_{t|y} + R) \cdot \1(H_{t|y} + R \succ 0) \cdot f_{R_t}(R) \cdot f_{Y_t}(y) \cdot f_{\nabla Y_t}(0),
	\end{equation}
	is the intensity function of a process that counts peaks of a given height, location, and residual Hessian; recall that we use $f_{R_t}(R)$ to denote the density of $R_t$, which does not depend on $(Y_t,\nabla Y_t)$ by construction.
	In combination with~\eqref{eqn:approximate-distribution-after-selection-1} this implies
	\begin{equation}
		\label{pf:approximate-palm-density-hessian-1}
		\bigg|\frac{p(R|t,y) - \rho(R,t,y)/\rho(t,y)}{\rho(R,t,y)/\rho(t,y)}\bigg| \leq C\exp(-c\lambda_n^2).
	\end{equation}
	Thus we have related the density of $\wh{R}|\hat{t},\wh{Y}$ under $\Q^{\mc{S}_{t^*}}$ to the Palm distribution of $\wh{R}$, which has density
	\begin{equation}
		\label{eqn:palm-density-hessian}
		\frac{\rho(R,t,y)}{\rho(t,y)} = \frac{\det(R + H_{t|y}) \cdot \1(R + H_{t|y} \succ 0)}{\E[\det(-\nabla^2 Y_t) \cdot \1(-\nabla^2 Y_t \prec 0)|\nabla Y_t = 0,Y_t = y]} \cdot f_{R_t}(R),
	\end{equation}
	where $f_{R_t}(R)$ is the Gaussian density
	\begin{equation}
		\label{eqn:gaussian-density-hessian}
		f_{R_t}(R) = \frac{1}{\sqrt{(2\pi)^{d(d + 1)/2}\det(\Theta_t)}} \exp\Big(-\frac{1}{2}\vech(R)'\Theta_t^{-1}\vech(R)\Big),
	\end{equation}
	and $\Theta_t = \Var[\vech(R_t)]$ is invertible as a result of the condition $\tau_{2}^2 > 0$. We now derive an approximation to $p(R|t,y)$ by local expansion of~\eqref{eqn:palm-density-hessian} about $t = t^*, R = 0$. We introduce the following notation for the first-order terms in this expansion: 
	\begin{align*}
		T_{10}^{R}(h) & := -\frac{1}{2}\tr(\Theta_{t^*}^{-1}\dot{\Theta}_{t^*}(h)) \\
		T_{12}^{R}(h,R) & := \frac{1}{2}\vech(R)'\Theta_{t^*}^{-1}\{\dot{\Theta}_{t^*}(h)\}\Theta_{t^*}^{-1}\vech(R)\\
		T_{01}^{R}(R) & := \tr(\bar{H}_{t^*}^{-1}R).
	\end{align*}
	\begin{lemma}
		\label{lem:approximate-palm-density-hessian}
		Under the assumptions of Theorem~\ref{thm:pivot}, at any $t^* \in T^*, h \in B_d(0,\varepsilon_n), y \in \mc{I}_{t^*}$ and $R \in \mc{R}_{t^*}$,
		\begin{equation*}
			\bigg|\frac{p(R|t^* + h,y) - \bar{p}(R|t^* + h,y)}{\bar{p}(R|t^* + h,y)}\bigg| \leq C\Big(\|h\|^2 + |y - \bar{u}_{t^*}|^2 \delta_n^2 + \|R\|^2\delta_n^2 + \|R\|^2 \|h\|^2 + \delta_n^2\Big) := \Err_{p}(R|h,y),
		\end{equation*}
		where
		\begin{equation}
			\label{eqn:approximate-palm-density-hessian}
			\bar{p}(R|t^* + h,y) := \Big(1 + T_{10}^{R}(h) + T_{12}^{R}(h,R) + T_{01}^{R}(R)\Big) \cdot \frac{1}{\sqrt{(2\pi)^{d(d + 1)/2} \cdot \det(\Theta_{t^*})}} \exp\Big(-\frac{1}{2}\vech(R)'\Theta_{t^*}^{-1}\vech(R)\Big).
		\end{equation}
	\end{lemma}
	Lemma~\ref{lem:approximate-palm-density-hessian} is proved in Section~\ref{subsec:pf-approximate-palm-density-hessian}. 
	
	\paragraph{Joint density of location and residual Hessian.}
	Lemma~\ref{lem:approximate-palm-density-hessian} and Theorem~\ref{thm:approximate-joint-distribution} can be used to compute an approximation to the joint density of $(\wh{R},\hat{t})$ given $\wh{Y}$, the density of $\wh{R}$ given $\wh{Y}$, and the density of $\hat{t}$ given $(\wh{R},\wh{Y})$, which are respectively
	\begin{equation*}
		p(R,t|y) = p(R|t,y) \cdot p(t|y), \quad p(R|y) = \int_{\mc{B}_{t^*}} p(R,t|y) \,dt, \quad p(t|R,y) = \frac{p(R,t|y)}{p(R|y)}.
	\end{equation*}
	To more concisely denote the error involved in these expansions, let
	\begin{align*}
		\Err_{p}(R,y) & := C \big(|y - \bar{u}_{t^*}|^2 + |\bar{u}_{t^*} - \mu_{t^*}| + \|R\|^2\big)\delta_n^2\\
		\wb{\Trunc}_{\rho}(\mc{B}_{t^*}) & := C \delta_n^2
	\end{align*}
	\begin{lemma}
		\label{lem:approximate-joint-distribution-2}
		Under the assumptions of Theorem~\ref{thm:pivot}, the following statements hold at any $t^* \in T^*, h \in B_d(0,\varepsilon_n), y \in \mc{I}_{t^*}$ and $R \in \mc{R}_{t^*}$:
		\begin{itemize}
			\item The joint density of $(\wh{R},\wh{t})$ given $\wh{Y}$ satisfies
			\begin{equation*}
				\bigg|\frac{p(R,t^* + h|y) - \bar{p}(R|t^* + h,y)}{\bar{p}(R|t^* + h,y)}\bigg| \leq C\Big(\Err_{p}(R|h,y) + \Err_{p}(h|y)\Big) := \Err_{p}(R,h|y),
			\end{equation*}
			where
			\begin{equation}
				\label{eqn:approximate-joint-density-hessian-location}
				\begin{aligned}
					\bar{p}(R,t^* + h|y) := & \Big(1 + T_{10}^{R}(h) + T_{12}^{R}(h,R) + T_{01}^{R}(R) + T_{10}^{\rho}(h) - \frac{1}{2}T_{30}^{\rho}(h)\Big) \\ 
					& \times \frac{1}{\sqrt{(2\pi)^{d(d + 1)/2}\det(\Theta_{t^*})}} \exp\Big(-\frac{1}{2}\vech(R)'\Theta_{t^*}^{-1}\vech(R)\Big) \cdot \frac{\sqrt{\det(G_{t^*|y})}}{\sqrt{(2\pi)^d}} \cdot \exp\Big(-\frac{1}{2}h'G_{t^*|y} h\Big).
				\end{aligned}
			\end{equation}
			\item The conditional density of $\wh{R}$ given $\wh{Y}$ satisfies
			\begin{equation*}
				\frac{|\bar{p}(R|y) - p(R|y)|}{\bar{p}(R|y)} \leq C\Big(\Err_{p}(R,y) + \wb{\Trunc}_{\rho}(\mc{B}_{t^*})\Big) := \Err_{p}(R|y),
			\end{equation*}
			where 
			\begin{equation*}
				\bar{p}(R|y) = \Big(1 + T_{01}^{R}(R)\Big) \cdot \frac{1}{\sqrt{(2\pi)^{d(d + 1)/2} \cdot \det(\Theta_{t^*})}} \exp\Big(-\frac{1}{2}\vech(R)'\Theta_{t^*}^{-1}\vech(R)\Big).
			\end{equation*}
			\item The conditional density of $\hat{t}$ given $\wh{R},\wh{Y}$ satisfies
			\begin{equation*}
				\frac{|\bar{p}(t^* + h|R,y) - p(t^* + h|R,y)|}{\bar{p}(t^* + h|R,y)} \leq C\Big(\Err_{p}(R,h|y) + \Err_{p}(h|y)\Big) := \Err_{p}(h|R,y),
			\end{equation*}
			where
			\begin{equation*}
				\bar{p}(t^* + h|R,y) := \Big(1 + T_{10}^{R}(h) + T_{12}^{R}(h,R) + T_{10}^{\rho}(h) - \frac{1}{2}T_{30}^{\rho}(h)\Big) \cdot  \frac{\sqrt{\det(G_{t^*|y})}}{\sqrt{(2\pi)^d}} \cdot \exp\Big(-\frac{1}{2}h'G_{t^*|y} h\Big).
			\end{equation*}
		\end{itemize}
	\end{lemma}
	Lemma~\ref{lem:approximate-joint-distribution-2} is proved in Section~\ref{subsec:pf-approximate-joint-distribution-2}. 
	
	\paragraph{Density of studentized peak.}
	Recall the studentized peak $\wh{Z} = Z_{\hat{t}}(\wh{R},\wh{Y})$. Lemma~\ref{lem:approximate-joint-distribution-2}, combined with a change of variables, yields an approximation to the joint density of $(\wh{Z},\wh{R})$ given $\wh{Y}$, the conditional density of $\wh{Z}$ given $(\wh{R},\wh{Y})$, and the conditional density of $\wh{Z}$ given $\wh{Y}$ and $\wh{R} \in \mc{R}_{t^*}$. We introduce some notation to compactly denote these approximations. We write $F_{t^*} = F_{t^*}(y,0)$ and $\dot{F}_{t^*}$ for $\nabla_t F_t(y,0)|_{t = t^*}$. We let $\tilde{L}_t$ be the differential of $\Lambda_t^{-1/2}$, meaning any array in $\R^{d \times d \times d}$ that satisfies $\{L_t(h)\}'\Lambda_{t}^{-1/2} + \Lambda_t^{-1/2}\{L_t(h)\} = -\Lambda_t^{-1}\{\dot{\Lambda}_t(h)\}\Lambda_t^{-1}$ for all $h \in \Rd$. We write $L_t \in \R^{d \times d \times d}$ for the array that swaps the second and third indices of $\tilde{L}_t$, so that $[L_{t}]_{ijk} = [\tilde{L}_{t}]_{ikj}$, and write $\bar{L}_t = L_t + \tilde{L}_t$. Finally, we denote the first-order terms in these expansions by
	\begin{equation*}
		\begin{aligned}
			T_{01}^{p_{Z|R}}(z)
			& = T_{10}^R(F_{t^*}^{-1}z) + T_{10}^{\rho}(F_{t^*}^{-1}z) - 2 \cdot \tr(F_{t^*}^{-1}\Lambda_{t^*}^{-1/2}\{\dot{H}_{t^*}(F_{t^*}^{-1}z)\}) - \tr(F_{t^*}^{-1}\{\bar{L}_{t^*}(H_{t^*}F_{t^*}^{-1}z)\})  \\
			& = -\frac{1}{2}\tr\big(\Theta_{t^*}^{-1}\dot{\Theta}_{t^*}(F_{t^*}^{-1}z)\big) + \tr\big(\bar{H}_{t^*}^{-1} \dot{H}_{t^*|\bar{u}_{t^*}}(F_{t^*}^{-1}z)\big) + \frac{1}{2}\tr\big(\Lambda_{t^*}^{-1}\dot{\Lambda}_{t^*}(F_{t^*}^{-1}z)\big) \\
			& \quad \; - 2 \cdot \tr\big(F_{t^*}^{-1}\dot{F}_{t^*}(F_{t^*}^{-1}z)\big) - \tr(F_{t^*}^{-1}\{\bar{L}_{t^*}(H_{t^*}F_{t^*}^{-1}z)\}) \\
			T_{30}^{p_{Z|R}}(z) & = -\frac{1}{2}T_{30}^{\rho}(F_{t^*}^{-1}z) - z' F_{t^*}^{-1}G_{t^*|y}F_{t^*}^{-1}\dot{F}_{t^*}(F_{t^*}^{-1}z, F_{t^*}^{-1}z) \\
			T_{21}^{p_{Z|R}}(z,R) & = z'(F_{t^*}^{-1})'G_{t^*|y}F_{t^*}^{-1}\Lambda_{t^*}^{-1/2}RF_{t^*}^{-1}z \\
			T_{12}^{p_{Z|R}}(z,R) & = T_{12}^R(F_{t^*}^{-1}z,R) =  \frac{1}{2}\vech(R)'\Theta_{t^*}^{-1}\{\dot{\Theta}_{t^*}(F_{t^*}^{-1}z)\}\Theta_{t^*}^{-1}\vech(R), \\
			T_{01}^{p_{Z|R}}(R) & = -\tr(F_{t^*}^{-1}\Lambda_{t^*}^{-1/2}R),
		\end{aligned}
	\end{equation*}
	and for the relative error incurred by the expansion
	\begin{equation*}
		\begin{aligned}
			\Err_{p}(z, \mc{R}_{t^*}|y) & := C\delta_n^2\big(1 + \|z\|^6\big)\big(1 + |y - \bar{u}_{t^*}|^2 + |\bar{u}_{t^*} - \mu_{t^*}|\big) \\
			\Trunc_{p}(\mc{R}_{t^*}) & :=C \delta_n^2 \\
			\wb{\Trunc}_{p}(\mc{R}_{t^*}) & := C \delta_n^{2.5}.
		\end{aligned}
	\end{equation*}
	Finally, let $\mc{Z}_{t^*}(y) = \{Z_{t^* + h}(0,y): 2h \in B_d(0,\varepsilon_n)\}$.
	\begin{lemma}
		\label{lem:approximate-density-studentized-peak}
		Under the assumptions of Theorem~\ref{thm:pivot}, the following statements hold at any $y \in \mc{I}_{t^*}, z \in \mc{Z}_{t^*}(y)$ and $R \in \mc{R}_{t^*}$: 
		\begin{itemize}
			\item The conditional density of $\wh{Z}$ given $\wh{R} = R,\wh{Y} = y$ satisfies
			\begin{equation}
				\label{eqn:density-studentized-peak-given-hessian}
				\Big|\frac{p(z|R,y) - \bar{p}(z|R,y)}{\bar{p}(z|R,y)}\Big| \leq C\Big(\Err_{p}(F_{t^*}^{-1}z|R,y) + \delta_n^2(1 + \|z\|^6)(1 + \|R\|^2)\Big) := \Err_{p}(z|R,y),
			\end{equation}
			where 
			\begin{equation*}
				\begin{aligned}
					\bar{p}(z|R,y) & := \big(1 + T_{01}^{p_{Z|R}}(z) + T_{30}^{p_{Z|R}}(z)+ T_{21}^{p_{Z|R}}(z,R) + T_{12}^{p_{Z|R}}(z,R) + T_{01}^{p_{Z|R}}(R)\big) \\
					& \times \frac{\sqrt{\det(G_{t^*|y})}}{\sqrt{(2\pi)^d} \det(F)} \exp\Big(-\frac{1}{2}(F_{t^*}^{-1}z)'G_{t^*|y}F_{t^*}^{-1}z\Big).
				\end{aligned}
			\end{equation*}
			\item The joint density of $\wh{Z}$ and $\wh{R}$ given $\wh{Y} = y$ satisfies
			\begin{equation*}
				\Big|\frac{p(z,R|y) - \bar{p}(z,R|y)}{\bar{p}(z,R|y)}\Big| \leq C\Big(\Err_{p}(z|R,y) + \Err_{p}(R|y)\Big) := \Err_{p}(z,R|y)
			\end{equation*}
			where
			\begin{equation}
				\label{eqn:density-studentized-peak-and-hessian}
				\begin{aligned}
					\bar{p}(z,R|y) 
					& : = \big(1 + T_{01}^{p_{Z|R}}(z) + T_{30}^{p_{Z|R}}(z)+ T_{21}^{p_{Z|R}}(z,R) + T_{12}^{p_{Z|R}}(z,R) + T_{01}^{p_{Z|R}}(R) + T_{01}^{R}(R)\big) \\
					& \times \frac{\sqrt{\det(G_{t^*|y})}}{\sqrt{(2\pi)^d} \det(F_{t^*})} \exp\Big(-\frac{1}{2}(F_{t^*}^{-1}z)'G_{t^*|y}F_{t^*}^{-1}z\Big) \cdot \frac{1}{\sqrt{(2\pi)^{d(d + 1)/2} \cdot \det(\Theta_{t^*})}} \exp\Big(-\frac{1}{2}\vech(R)'\Theta_{t^*}^{-1}\vech(R)\Big).
				\end{aligned}
			\end{equation}
			\item The conditional density of $\wh{Z}$ given $\wh{Y} = y$ and $\wh{R} \in \mc{R}_{t^*}$ satisfies
			\begin{equation}
				\label{eqn:density-studentized-peak-2}
				\Big|\frac{p(z|y) - \bar{p}(z|y)}{\bar{p}(z|y)}\Big| \leq C\Big(\Err_{p}(z,\mc{R}_{t^*}|y) + \Trunc_{p}(\mc{R}_{t^*}) + \wb{\Trunc}_{p}(\mc{R}_{t^*})\Big) := \Err_{p}(z|y),
			\end{equation}
			where
			\begin{equation*}
				\bar{p}(z|y) = \Big(1 + T_{01}^{p_{Z|R}}(z) + T_{30}^{p_{Z|R}}(z) + \frac{1}{2}\tr\big(\Theta_{t^*}^{-1}\{\dot{\Theta}_{t^*}(F_{t^*}^{-1}z)\}\big)\Big) \cdot \frac{\sqrt{\det(G_{t^*|y})}}{\sqrt{(2\pi)^d} \det(F_{t^*})} \exp\Big(-\frac{1}{2}(F_{t^*}^{-1}z)'G_{t^*|y}F_{t^*}^{-1}z\Big).
			\end{equation*}
		\end{itemize}
	\end{lemma}
	Lemma~\ref{lem:approximate-density-studentized-peak} is proved in Section~\ref{subsec:pf-approximate-density-studentized-peak}. To be perfectly clear, the density $p(z|y)$ conditions both on $N(\mc{S}_{t^*}) = 1$ and on $\wh{R} \in \mc{R}_{t^*}$, i.e.:
	\begin{equation*}
		p(z|y) := \frac{1}{\Q^{\mc{S}_{t^*}}(\wh{R} \in \mc{R}_{t^*})}\int_{\mc{R}_{t^*}} p(z,R|y) \,dR.
	\end{equation*}
	
	\paragraph{Distribution of studentized Wald pivot.} We now turn to the conditional distribution of the studentized Wald pivot $\wh{W}$, given  the selection event $\mc{S}_{t^*}$ and conditional on the height $\wh{Y}$. To additionally condition on the high-probability event $\wh{R} \in \mc{R}_{t^*}$, we sandwich 
	\begin{equation*}
		\Q^{\mc{S}_{t^*}}(\wh{W} \leq q, \wh{R} \in \mc{R}_{t^*}|\wh{Y} =  y) \leq \Q^{\mc{S}_{t^*}}(\wh{W} \leq q|\wh{Y} = y) \leq \Q^{\mc{S}_{t^*}}(\wh{W} \leq q, \wh{R} \in \mc{R}_{t^*}|\wh{Y} =  y) + \Q^{\mc{S}_{t^*}}(\wh{R} \notin \mc{R}_{t^*}|\wh{Y} = y).
	\end{equation*}
	We note that $\Q^{\mc{S}_{t^*}}(\wh{R} \notin \mc{R}_{t^*}|\wh{Y} = y) \leq C \delta_n^2$ and 
	$$
	\Big|\Q^{\mc{S}_{t^*}}(\wh{W} \leq q, \wh{R} \in \mc{R}_{t^*}|\wh{Y} =  y) - \Q^{\mc{S}_{t^*}}(\wh{W} \leq q|\wh{R} \in \mc{R}_{t^*},\wh{Y} = y)\Big| \leq C \delta_n^2,
	$$
	both by~\eqref{eqn:hessian-residual-truncation-error}. Since $\Q^{\mc{S}_{t^*}}(\wh{W} \leq q|\wh{R} \in \mc{R}_{t^*},\wh{Y} = y) = \int_{B_d(0,q)} p(z|y) \,dz$, we have succeeded in showing that 
	\begin{equation*}
		\bigg|\Q^{\mc{S}_{t^*}}(\wh{W} \leq q|\wh{Y} = y) - \int_{B_d(0,q)}  p(z|y) \,dz\bigg| \leq C \delta_n^2.
	\end{equation*}
	Notice that $B_d(0,q) \subseteq \mc{Z}_{t^*}$ for all $n \in \mathbb{N}$ sufficiently large. We can therefore apply Lemma~\ref{lem:approximate-density-studentized-peak} to conclude that
	\begin{equation}
		\label{pf:studentized-wald-1}
		\Big|\int_{B_d(0,q)} p(z|y) \,dz - \int_{B_d(0,q)} \bar{p}(z|y) \,dz\Big| \leq \int_{B_d(0,q)} \Err_{p}(z|y) \bar{p}(z|y) \,dz.
	\end{equation}
	Since each of the first-order terms in the local expansion $\bar{p}(z|y)$ is an odd function of $z$, they vanish when integrated over $B_d(0,q)$, leaving us with
	\begin{equation}
		\label{pf:studentized-wald-2}
		\int_{B_d(0,q)} \bar{p}(z|y) \,dz = \int_{B_d(0,q)} \frac{\sqrt{\det(G_{t^*|y})}}{\sqrt{(2\pi)^d} \det(F_{t^*})} \exp\Big(-\frac{1}{2}(F_{t^*}^{-1}z)'G_{t^*|y}F_{t^*}^{-1}z\Big) \,dz.
	\end{equation}
	In other words $\wh{W}$ has a limiting generalized chi-squared distribution, up to second-order relative error
	\begin{equation}
		\label{eqn:pivot-location-error}
		\begin{aligned}
			C\delta_n^2 + \int_{B_d(0,q)} \Err_{p}(z|y) \bar{p}(z|y) \,dz 
			& \leq C \Big(\delta_n^2(1 + |y - \bar{u}_{t^*}|^2 + |\bar{u}_{t^*} - \mu_{t^*}|)\Big).
		\end{aligned}
	\end{equation}
	
	\paragraph{Derivation of~\eqref{eqn:pivot-location}.}
	Equation~\eqref{pf:studentized-wald-2} gives a second-order accurate approximation to $\Q^{\mc{S}_{t^*}}(\wh{W} \leq q|\wh{Y} = y)$. We now expand this approximation in $y$, around $y = \mu_{t^*}$, to obtain the ultimate claim~\eqref{eqn:pivot-location}, valid when $(\bar{u}_{t^*} - \mu_{t^*})\delta_n \to 0, y \in \bar{u}_{t^*} \pm \Delta_n$ and therefore $(y - \mu_{t^*})\delta_n \to 0$. We begin by simplifying the integrand of~\eqref{pf:studentized-wald-2}. We have that the precision matrix 
	\begin{equation*}
		F_{t^*}^{-1}G_{t^*|y}F_{t^*}^{-1} = \Lambda_{t^*}^{-1/2}(-\nabla^2 \mu_{t^*})H_{t^*|y}^{-1}\Lambda_{t^*}^{1/2} = I - (y - \mu_{t^*})\Lambda_{t^*}^{1/2}H_{t^*|y}^{-1}\Lambda_{t^*}^{1/2}.
	\end{equation*}
	Applying the matrix perturbation bound~\eqref{eqn:matrix-inverse-taylor-expansion} allows us to bound $\|H_{t^*|y}^{-1} - \bar{H}_{t^*}\| \leq C|y - \bar{u}_{t^*}|\delta_n^2$, and in turn implies
	\begin{equation*}
		\|F_{t^*}^{-1}G_{t^*|y}F_{t^*}^{-1} - \Big(I - (y - \mu_{t^*})\Lambda_{t^*}^{1/2}\bar{H}_{t^*}^{-1}\Lambda_{t^*}^{1/2}\Big)\| \leq C\Big( |y - \bar{u}_{t^*}| |y - \mu_{t^*}| \delta_n^2\Big).
	\end{equation*}
	Then  a subsequent Taylor expansion of $\exp(x)$ around $x = 0$ -- noting that the magnitude of the first-order term in the previous expression is $O((y - \mu_{t^*})\delta_n )$ -- implies
	\begin{equation*}
		\Big\|\exp\Big(-\frac{1}{2}(F_{t^*}^{-1}z)'G_{t^*|y}F_{t^*}^{-1}z\Big)  - \Big(1 + \frac{(y - \mu_{t^*})}{2}z'\Lambda_{t^*}^{1/2}\bar{H}_{t^*}^{-1}\Lambda_{t^*}^{1/2}z\Big)\exp\Big(-\frac{1}{2}z'z\Big) \Big\| \leq C\Big( (y - \mu_{t^*})^2\delta_n^2 + |y - \bar{u}_{t^*}| |y - \mu_{t^*}|\delta_n^2\Big).
	\end{equation*}
	Similarly,
	\begin{equation}
		\frac{\sqrt{\det(G_{t^*|y})}}{\det(F_{t^*})} = \sqrt{\frac{\det(-\nabla^2\mu_{t^*})}{\det(H_{t^*|y})}} = \sqrt{\det(I - (y  - \mu_{t^*})H_{t^*|y}^{-1}\Lambda_{t^*})}.
	\end{equation}
	Using a first-order Taylor expansion of $\det(I + E)$ about $E = 0$, as recorded in~\eqref{eqn:matrix-det-taylor-expansion}, and again applying the matrix perturbation bound~\eqref{eqn:matrix-inverse-taylor-expansion} to upper bound $\|H_{t^*|y}^{-1} - \bar{H}_{t^*}\|$, we have
	\begin{equation*}
		\Big\|\sqrt{\det(I - (y  - \mu_{t^*})H_{t^*|y}^{-1}\Lambda_{t^*})} - \Big(1 - \frac{(y - \mu_{t^*})}{2}\tr(\bar{H}_{t^*}^{-1}\Lambda_{t^*})\Big)\Big\| \leq C\Big( (y - \mu_{t^*})^2\delta_n^2 +  |y - \mu_{t^*}| |y - \bar{u}_{t^*}| \delta_n^2\Big).
	\end{equation*}
	Putting these two approximations together yields the following approximation to the integrand of~\eqref{pf:studentized-wald-2}:
	\begin{equation}
		\label{pf:studentized-wald-3}
		\bigg|\phi_{0,F_{t^*}^{-1}G_{t^*|y}F_{t^*}^{-1}}(z)- \Big(1 -\frac{(y - \mu_{t^*})}{2}\tr(\bar{H}_{t^*}^{-1}\Lambda_{t^*}) + \frac{(y - \mu_{t^*})}{2}z'\Lambda_{t^*}^{1/2}\bar{H}_{t^*}^{-1}\Lambda_{t^*}^{1/2}z \Big) \phi_{0,I}(z)\bigg| \leq C\Big( (y - \mu_{t^*})^2\delta_n^2 +  |y - \mu_{t^*}| |y - \bar{u}_{t^*}| \delta_n^2\Big).
	\end{equation}
	Integrating this approximation over $B_d(0,q)$ gives
	\begin{equation*}
		\int_{B_d(0,q)} \Big(1 -\frac{(y - \mu_{t^*})}{2}\tr(\bar{H}_{t^*}^{-1}\Lambda_{t^*}) + \frac{(y - \mu_{t^*})}{2}z'\Lambda_{t^*}^{1/2}\bar{H}_{t^*}^{-1}\Lambda_{t^*}^{1/2}z \Big) \phi_{0,I}(z) \,dz = \chi_d^2(q) - A \cdot \frac{(y - \mu_{t^*})}{2}\tr(\bar{H}_{t^*}^{-1}\Lambda_{t^*}),
	\end{equation*}
	which equals $\P(\wh{W} \leq q|\wh{Y} = y)$ up to the second-order error terms in~\eqref{eqn:pivot-location-error} and~\eqref{pf:studentized-wald-3}, which are at most $C\delta_n^2(1 + |y - \bar{u}_{t^*}|^2 + |\bar{u}_{t^*} - \mu_{t^*}|^2)$. 
	
	\subsection{Proof of Lemma~\ref{lem:approximate-palm-density-hessian}}
	\label{subsec:pf-approximate-palm-density-hessian}
	The density of the Palm distribution of $R_t|\{\hat{t} = t,\wh{Y} = y\}$ is given in~\eqref{eqn:palm-density-hessian}; up to an exponential error term, this agrees with the density of $R_t||\{\hat{t} = t,\wh{Y} = y\}$ under $\Q^{\mc{S}_{t^*}}$ as shown in~\eqref{pf:approximate-palm-density-hessian-1}. For all $n \in \mathbb{N}$ large enough and $(t,y) \in \mc{S}_{t^*}, R \in \mc{R}_{t^*}$, the matrix $H_{t^*|y} + R$ will be positive definite, and then the Palm density is simply the product of a Gaussian density and a determinant term, divided by a normalizing constant:
	\begin{equation}
		\label{eqn:palm-density-hessian-no-psd}
		\begin{aligned}
			p(R|t,y) = \frac{\det(R + H_{t|y})}{\E[\det(-\nabla^2 Y_t) \cdot \1(-\nabla^2 Y_t \succ 0)|\nabla Y_t = 0,Y_t = y]} f_{R_t}(R).
		\end{aligned}
	\end{equation}
	We separately approximate the Gaussian density, and the determinant term and normalizing constant. Throughout we write $r = \vech(R)$ and $h = t - t^*$ for notational convenience. 
	
	\paragraph{Gaussian density of residual.}
	Recall that the Gaussian density of the residual Hessian $R_t$ is 
	\begin{equation*}
		f_{R_t}(R) = \frac{1}{\sqrt{(2\pi)^{d(d + 1)/2}\det(\Theta_t)}} \exp\Big(-\frac{1}{2}r'\Theta_t^{-1}r\Big),
	\end{equation*}
	A first-order Taylor expansion of $\Theta_t^{-1}$ about $t = t^*$ implies
	\begin{equation*}
		\Big|r'\Theta_{t}^{-1}r - r'(\Theta_{t^*}^{-1} - \Theta_{t^*}^{-1}\{\dot{\Theta}_{t^*}(h)\}\Theta_{t^*}^{-1})r\Big| \leq C\|h\|^2\|R\|^2,
	\end{equation*}
	while first-order Taylor expansion of $\{\det(\Theta_t)\}^{-1/2}$ about $t = t^*$ implies
	\begin{equation*}
		\Big|\frac{1}{\sqrt{\det(\Theta_t)}} - \frac{1}{\sqrt{\det(\Theta_{t^*})}}(1 - \frac{1}{2}\tr(\Theta_{t^*}^{-1}\dot{\Theta}_{t^*}(h))\Big| \leq C\|h\|^2;
	\end{equation*}
	in both cases we are using \eqref{eqn:non-degenerate} which implies that $\lambda_{\min}(\Theta_t) > \tau_2^2 > 0$. These combine to yield an approximation to the log of $f_{R_t}(R)$. From here, arguments similar to those of Section~\ref{subsec:asymptotic-analysis-density} yield
	\begin{equation*}
		\frac{|f_{R_t}(R) - \bar{f}_{R_t}(R)|}{\bar{f}_{R_t}(R)} \leq C\Big(\|R\|^2 \|h\|^2 + \|h\|^2\Big)
	\end{equation*}
	where 
	\begin{equation}
		\label{pf:palm-distribution-hessian-1}
		\bar{f}_{R_t}(R) := \frac{\Big(1 - \frac{1}{2}\tr\big(\Theta_{t^*}^{-1}\dot{\Theta}_{t^*}(h)\big) + \frac{1}{2}r'\Theta_{t^*}^{-1}\{\dot{\Theta}_{t^*}(h)\}\Theta_{t^*}^{-1}r\Big)}{\sqrt{(2\pi)^{d(d + 1)/2}\det(\Theta_{t^*})}} \exp\Big(-\frac{1}{2}r'\Theta_{t^*}^{-1}r\Big).
	\end{equation}
	
	\paragraph{Determinant term and normalizing constant.} Recall the result of \eqref{eqn:hessian-determinant-local-expansion} in Section~\ref{subsec:hessian-asymptotics}:
	\begin{equation*}
		\begin{aligned}
			& \frac{\Big|\det(H_{t|y} + R) - \det(\bar{H}_{t^*})\big(1 + T_{10}^{\det}(h) + T_{01}^{\det}(y) + \tr(\bar{H}_{t^*}^{-1}R)\big)\Big|}{\det(\bar{H}_{t^*})} \leq C\Big(\|h\|^2 + |y - \bar{u}_{t^*}|^2\delta_n^2+ \|R\|^2\delta_n^2\Big).
		\end{aligned}
	\end{equation*}
	Combined with Lemma~\ref{lem:approximation-peak-intensity}, this implies
	\begin{equation}
		\label{pf:palm-distribution-hessian-2}
		\Big|\frac{\det(H_{t|y} + R)}{\E[\det(-\nabla^2 Y_t) \cdot \1(-\nabla^2 Y_t \succ 0)|\nabla Y_t = 0,Y_t = y]} - \big(1 + \tr(\bar{H}_{t^*}^{-1}R)\big)\Big| \leq C\Big(\|h\|^2 + |y - \bar{u}_{t^*}|^2\delta_n^2 + \|R\|^2\delta_n^2 + \delta_n^2\Big).
	\end{equation}
	Combining~\eqref{pf:palm-distribution-hessian-1} and~\eqref{pf:palm-distribution-hessian-2} give the approximation $\bar{p}(R|t,y)$ to the claimed order of accuracy.
	
	\subsection{Proof of Lemma~\ref{lem:approximate-joint-distribution-2}}
	\label{subsec:pf-approximate-joint-distribution-2}
	
	Throughout this proof we denote $r = \vech(R)$ and $t = t^* + h$.
	
	\subsubsection{Joint density of residual Hessian and location}
	Lemma~\ref{lem:approximate-palm-density-hessian} and Theorem~\ref{thm:approximate-joint-distribution} imply that for all $n$ sufficiently large:
	\begin{equation*}
		\frac{\Big|p(R,t|y) - \bar{p}(R|t,y) \bar{p}(t|y)\Big|}{ \bar{p}(R|t,y) \bar{p}(t|y)} \leq C\Big(\Err_p(R|t,y) + \Err_{p}(t|y)\Big).
	\end{equation*}
	The product $\bar{p}(R|t,y) \bar{p}(t|y)$ is 
	\begin{align*}
		\bar{p}(R|t,y) \bar{p}(t|y) 
		& = \Big(1 + T_{10}^{\rho}(h) - \frac{1}{2}T_{30}^{\rho}(h)\Big) \cdot \frac{\sqrt{\det(G_{t^*|y})}}{\sqrt{(2\pi)^d}} \cdot \exp\Big(-\frac{1}{2}h'G_{t^*|y} h\Big) \\
		& \times \Big(1 + T_{10}^{R}(h) + T_{12}^{R}(h,R) + T_{01}^{R}(R)\Big) \cdot \frac{1}{\sqrt{(2\pi)^{d(d + 1)/2} \cdot \det(\Theta_{t^*})}} \exp\Big(-\frac{1}{2}r'\Theta_{t^*}^{-1}r\Big).
	\end{align*}
	Ignoring the product of all first-order terms in the above equation incurs at most a second-order error, meaning precisely:
	\begin{equation*}
		\frac{|\bar{p}(R|t,y)\bar{p}(t|y) - \bar{p}(R,t|y)|}{\bar{p}(R,t|y)} \leq C\Big(\Err_p(R|t,y) + \Err_{p}(h|y)\Big),
	\end{equation*}
	and establishing~\eqref{eqn:approximate-joint-density-hessian-location}. We give an explicit bound on the magnitude of the error term $\Err_{p}(R,h|y)$ in terms of $(\delta_n,h,y - \bar{u}_{t^*},R)$, for future reference:
	\begin{equation}
		\begin{aligned}
			\label{eqn:density-hessian-location-error}
			& \Err_p(R|h,y) + \Err_{p}(h|y) \\
			& \leq C(\Err_{p}(R|h,y) + \Err_{\rho}(h,y) + \Err_{\E[N(\mc{S}_{t^*})]}) \\ 
			& \leq C(\Err_{p}(R|h,y) + \Err_{\det}(h,y) + \Err_{Y}(h,y) + \Err_{\nabla Y}(h) + |y - \bar{u}_{t^*}|^2 \lambda_n^2 \|h\|^4 + \Err_{\E[N(\mc{S}_{t^*})]}) \\
			& \leq C\Big(\|h\|^2 + |y - \bar{u}_{t^*}|^2 \delta_n^2 + \|R\|^2\delta_n^2 + \|R\|^2 \|h\|^2 + \delta_n^2  + \big(|y - \bar{u}_{t^*}| + |\bar{u}_{t^*} - \mu_{t^*}|\big) \lambda_n \|h\|^3 + \lambda_n^2\|h\|^4 +
			|y - \bar{u}_{t^*}|^2 \lambda_n^2 \|h\|^4\Big).
		\end{aligned}
	\end{equation}
	
	\subsubsection{Density of residual Hessian.}
	Let 
	\begin{equation*}
		\tilde{p}(R,t^* + h|y) = \Big(1 + T_{01}^{R}(R)\Big) \cdot \frac{1}{\sqrt{(2\pi)^{d(d + 1)/2}\det(\Theta_{t^*})}} \exp\Big(-\frac{1}{2}\vech(R)'\Theta_{t^*}^{-1}\vech(R)\Big) \cdot \frac{\sqrt{\det(G_{t^*|y})}}{\sqrt{(2\pi)^d}} \cdot \exp\Big(-\frac{1}{2}h'G_{t^*|y} h\Big).
	\end{equation*}
	The difference between $\tilde{p}(R,t^* + h|y)$ and $\bar{p}(R,t^* + h|y)$ is that the former omits first-order terms that are linear in $h$. The conditional density of the residual Hessian is $p(R|y) = \int_{\mc{B}_{t^*}} p(R,t|y) \,dt$. We expand this as the sum of a main term and three error terms:
	\begin{equation*}
		p(R|y) = \int_{\Rd} \tilde{p}(R,t|y) \,dh - \int_{\mc{B}_{t^*}^c} \tilde{p}(R,t|y) \,dt + \int_{\mc{B}_{t^*}} \bar{p}(R,t|y) - \tilde{p}(R,t|y) \,dt + \int_{\mc{B}_{t^*}} {p}(R,t|y) - \bar{p}(R,t|y) \,dt.
	\end{equation*}
	The main term above $\int_{\Rd} \tilde{p}(R,t|y) \,dt = \bar{p}(R|y)$. The second error term $\int_{\mc{B}_{t^*}} \bar{p}(R,t|y) - \tilde{p}(R,t|y) \,dt = 0$, since the integral of any first-order term that is linear in $h$ over $\mc{B}_{t^*}$ vanishes. It remains to bound the two sources of error, due to local expansion and truncation.
	
	\paragraph{Error due to local expansion.} 
	For all $n \in \mathbb{N}$ sufficiently large,
	\begin{equation*}
		\bar{p}(R,t|y) \leq C \cdot \bar{p}(R|y) \cdot \sqrt{\det(G_{t^*|y})} \cdot \exp\Big(-\frac{1}{2}h'G_{t^*|y}h\Big).
	\end{equation*}
	Combined with \eqref{eqn:approximate-joint-density-hessian-location}, this implies that the error due to local expansion is at most
	\begin{equation*}
		\begin{aligned}
			\int_{\mc{B}_{t^*}} \big|{p}(R,t|y) - \bar{p}(R,t|y)\big| \,dt & \leq \int_{\mc{B}_{t^*}} \Err_{p}(R,h|y) \cdot \bar{p}(R,t|y) \,dt \\
			& \leq C \cdot \bar{p}(R|y)  \sqrt{\det(G_{t^*|y})} \int_{B(0,\varepsilon_n)} \Err_{p}(R,h|y) \cdot \exp\Big(-\frac{1}{2}h'G_{t^*|y}h\Big) \,dh.
		\end{aligned}
	\end{equation*}
	Inserting the upper bound on $\Err_{p}(R,h|y)$ from~\eqref{eqn:density-hessian-location-error} into the previous display, and integrating over $h \in B_d(0,\varepsilon_n)$, shows that the error due to local expansion is at most
	\begin{equation*}
		C \cdot \bar{p}(R|y) \Big(|y - \bar{u}_{t^*}|^2 + |\bar{u}_{t^*} - \mu_{t^*}| + \|R\|^2 + 1)\delta_n^2.
	\end{equation*}
	
	\paragraph{Truncation error.} 
	For all $n \in \mathbb{N}$ sufficiently large,
	\begin{equation*}
		\tilde{p}(R,t|y) \leq C \cdot \bar{p}(R|y) \cdot \sqrt{\det(G_{t^*|y})} \cdot \exp\Big(-\frac{1}{2}h'G_{t^*|y}h\Big).
	\end{equation*}
	Moreover, for all $n \in \mathbb{N}$ sufficiently large, $\varepsilon_n \geq \{\lambda_{\min}(G_{t^*|y})\}^{-1} \sqrt{5 \log \lambda_n}$. Thus $\int_{\mc{B}_{t^*}^c} \tilde{p}(R,t|y) \,dt$ is at most
	\begin{equation*}
		C \cdot \bar{p}(R|y) \sqrt{\det(G_{t^*|y})} \int_{\mc{B}_{t^*}^c} \exp\Big(-\frac{1}{2}(t - t^*)'G_{t^*|y}(t - t^*)\Big) \,dt \leq C \cdot \bar{p}(R|y) \int_{\R^d \setminus B_d(0,\sqrt{5 \log \lambda_n})} \exp\Big(-\frac{\|z\|^2}{2}\Big) \,dz \leq C\delta_n^2,
	\end{equation*}
	with the final inequality following by~\eqref{eqn:gaussian-integrals-concentration}.
	
	\subsubsection{Conditional density of location}
	Finally, the local expansion $\bar{p}(t|R,y)$ to the conditional density of $\hat{t}|\{\wh{R} = R,\wh{Y} = y\}$ is obtained by dividing $\bar{p}(t,R|y)$ by $\bar{p}(R|y)$. Proper bookkeeping shows that the ratio of first-order terms can be ignored, at the cost of a second-order error term that is absorbed into $\Err_{p}(t|R,y)$.
	
	\subsection{Proof of Lemma~\ref{lem:approximate-density-studentized-peak}}
	\label{subsec:pf-approximate-density-studentized-peak}
	
	We begin with some calculus preliminaries before deriving the local expansions of Lemma~\ref{lem:approximate-density-studentized-peak}. For ease of reading, we suppress all notational dependence on $y$, writing $F_t(R) = F_t(y,R)$, $Z_t(R) = Z_{t}(y,R)$ and $H_t = H_{t|y}$.
	
	\subsubsection{Calculus preliminaries.}
	
	\paragraph{Change of variables.}
	The Jacobian of the mapping $t \mapsto Z_{t}(R)$ is given by
	\begin{equation*}
			\dot{Z}_t(R) = F_{t}(R) +L_t\big((H_t + R)h\big)+ \Lambda_{t}^{-1/2} \{\dot{H}_t(h)\}
	\end{equation*}
	and we recall that $[L_{t}]_{ijk} = [\tilde{L}_{t}]_{ikj}$, and $\tilde{L}_t$ is the differential of $\Lambda_t^{-1/2}$, meaning it is any array in $\R^{d \times d \times d}$ that satisfies $\{L_t(h)\}'\Lambda_{t}^{-1/2} + \Lambda_t^{-1/2}\{L_t(h)\} = -\Lambda_t^{-1}\{\dot{\Lambda}_t(h)\}\Lambda_t^{-1}$ for all $h \in \Rd$. Note that $D_t$ is not the same as $\dot{F}_t$. We will use the fact that $\sup_{t \in \Rd} \|L_t(v)\| \leq C\|v\|$ for any $v \in \Rd$, as a result of~\eqref{eqn:non-degenerate}.  
	
	For all $n \in \mathbb{N}$ sufficiently large and any $y \in \mc{I}_{t^*}, R \in \mc{R}_{t^*}$, the minimum eigenvalue of $\dot{Z}_t(R)$ is uniformly bounded away from $0$ over all $t \in \mc{B}_{t^*}$: the minimum eigenvalue of $F_t(R)$ is at least
	\begin{equation*}
		\inf_{t \in \mc{B}_{t^*}} \lambda_{\min}(F_t(R)) \geq \inf_{t \in \mc{B}_{t^*}} \frac{1}{\sqrt{\lambda_{\max}(\Lambda_t)}} \lambda_{\min}(H_t + R) \geq c (\lambda_n - \xi_n) \geq c \lambda_n,
	\end{equation*}
	while the norm of $\dot{F}_t(R)(t - t^*)$ is at most
	\begin{equation*}
		\sup_{t \in \mc{B}_{t^*}}\|\dot{F}_t(R)(t - t^*)\| \leq \sup_{t \in \mc{B}_{t^*}} \|A_t(t - t^*)\| + C \|\dot{H}_t(t - t^*)\| \leq C(\lambda_n + \xi_n) \varepsilon_n \leq C \lambda_n \varepsilon_n,
	\end{equation*}
	and therefore for all $n \in \mathbb{N}$ sufficiently large:
	\begin{align*}
		\inf_{t \in \mc{B}_{t^*}}\|\dot{Z}_t(R)\| 
		\geq \inf_{t \in \mc{B}_{t^*}}\|F_t(R)\| - \sup_{t \in \mc{B}_{t^*}}\|\dot{F}_t(R)(t - t^*)\| \geq c \lambda_n(1 - C\varepsilon_n) \geq c\lambda_n.
	\end{align*}
	We conclude that the vector field $t \mapsto Z_t(R)$ is strongly monotonic: for all $s,t \in \mc{B}_{t^*}$,
	$$
	(Z_{s}(R) - Z_{t}(R))'(s - t) \geq c\lambda_n\|s - t\|^2,
	$$
	and hence $t \mapsto Z_t(R)$ is injective over $\mc{B}_{t^*}$. The change of variables formula for one-to-one mappings  therefore gives the density of $\wh{Z}|\wh{R},\wh{Y}$: letting $\mc{Z}_{t^*}(R) := \{Z_t(R):t \in \mc{B}_{t^*}\}$, we have that
	\begin{equation}
		\label{eqn:change-of-variables}
		p(z|R,y) = \frac{p(T_z(R)|R,y)}{|\det(\dot{Z}_{T_z(R)}(R))|},  \quad \textrm{for $z \in \mc{Z}_{t^*}(R)$,}
	\end{equation}
	where $T_z(R)$ is the unique point $t \in \mc{B}_{t^*}$ satisfying $Z_{T_z(R)}(R) = z$.
	
	\paragraph{Taylor approximation of $T_z(R)$.}
	We now provide a local expansion of the inverse $T_z(R)$ around $z = 0, R = 0$ by inverting a local expansion of $Z_t(R)$ around $t = t^*, R = 0$. The first-order Taylor expansion of $F_t(R)$ around $t = t^*, R = 0$ is 
	\begin{equation*}
		F_{t^*} + \dot{F}_{t^*}(t - t^*) + \Lambda_{t^*}^{-1/2} R,  
	\end{equation*}
	where we recall that $F_{t^*} = F_{t^*}(0) = \Lambda_{t^*}^{-1/2}H_{t^*}$ and $\dot{F}_{t^*} = \dot{F}_{t^*}(0)$. Applying~\eqref{eqn:deterministic-hessian-taylor-expansion-first-order} and $\|L_t(v)\| \leq C\|v\|$, we conclude that
	\begin{equation*}
		\Big\|F_t(R) - F_{t^*} - \dot{F}_{t^*}(t - t^*) - \Lambda_{t^*}^{-1/2} R\Big\| \leq C\lambda_n\|t - t^*\|^2.
	\end{equation*}
	Consequently, the first-order local expansion of $Z_t(R)$ around $t = t^*, R = 0$ is 
	\begin{equation*}
		F_{t^*}(t - t^*) + \dot{F}_{t^*}(t - t^*,t - t^*) + \Lambda_{t^*}^{-1/2}R(t - t^*),
	\end{equation*}
	with error
	\begin{equation*}
		\|Z_t(R) - F_{t^*}(t - t^*) - \dot{F}_{t^*}(t - t^*,t - t^*) - \Lambda_{t^*}^{-1/2}R(t - t^*)\| \leq C(\lambda_n \|t - t^*\|^3 + \|R\| \|t - t^*\|).
	\end{equation*}
	Now we ``invert'' this expansion. To begin, evaluating $Z_t(R)$ at $t = T_z(R)$ and applying the above bound:
	\begin{equation}
		\label{pf:studentized-location-2}
		\begin{aligned}
			& \Big\|z - F_{t^*}h_z(R) - \dot{F}_{t^*}(h_z(R),h_z(R)) - \Lambda_{t^*}^{-1/2}R h_z(R) \Big\| \leq C(\lambda_n \|h_z(R)\|^3 + \|R\| \cdot \|h_z(R)\|),
		\end{aligned}
	\end{equation}
	where $h_z(R) = T_z(R) - t^*$. Substituting $F_{t^*}^{-1}z$ for $h_z(R)$ in the first-order terms in~\eqref{pf:studentized-location-2}, and observing that $\|h_z(R)\| \leq C \delta_n \|z\|$, we obtain
	\begin{equation*}
		\begin{aligned}
			\Big\|z - F_{t^*}h_z(R) - \dot{F}_{t^*}(F_{t^*}^{-1}z,F_{t^*}^{-1}z) - \Lambda_{t^*}^{-1/2}R(F_{t^*}^{-1}z)\Big\| 
			& \leq C\Big(\lambda_n \|h_z(R)\|^3 + \|R\| \cdot \|h_z(R)\|^2 + \|R\|^2 \cdot \|h_z(R)\|\Big) \\
			& \leq C\delta_n^2\Big(\|z\|^3 + \|R\| \cdot \|z\|^2 + \|R\|^2 \|z\|\Big).
		\end{aligned}
	\end{equation*}
	The first inequality above follows since $\|F_{t^*}^{-1}z - h_z(R)\| \leq C\|h_z(R)\|^2$ by~\eqref{pf:studentized-location-2}, and the second inequality follows since $\|h_z(R)\| \leq C \delta_n \|z\|$. 
	Finally, application of $F_{t^*}^{-1}$ gives the first-order expansion to $h_z(R)$:
	\begin{equation}
		\label{pf:approximate-density-studentized-peak-1}
		\begin{aligned}
			\Big\|h_z(R) - \bar{h}_z(R)\Big\| & \leq C\delta_n^3\Big(\|z\|^3 + \|R\| \|z\|^2 + \|R\|^2 \|z\|\Big) := \Err_{h}(z,R), \\
			\bar{h}_z(R) & := F_{t^*}^{-1}z - F_{t^*}^{-1}\dot{F}_{t^*}(F_{t^*}^{-1}z,F_{t^*}^{-1}z) - F_{t^*}^{-1}\Lambda_{t^*}^{-1/2}RF_{t^*}^{-1}z
		\end{aligned}
	\end{equation}
	The magnitude of the first-order terms in this expansion are at most
	\begin{equation*}
		\|F_{t^*}^{-1}\dot{F}_{t^*}(F_{t^*}^{-1}z,F_{t^*}^{-1}z)\| \leq C\delta_n^2\|z\|^3, \quad \|F_{t^*}^{-1}\Lambda_{t^*}^{-1/2}RF_{t^*}^{-1}z\| \leq C\delta_n^2\|R\|\|z\|.
	\end{equation*}
	
	\subsubsection{Conditional distribution of studentized peak given residual Hessian.}
	We now derive~\eqref{eqn:density-studentized-peak-given-hessian}, which gives a local expansion of the conditional density $p(z|R,y)$, by separately providing approximations to the density term and determinant term in~\eqref{eqn:change-of-variables}. 
	
	\paragraph{Density term.}
	Lemma~\ref{lem:approximate-joint-distribution-2} implies that the relative error between $p(T_z(R)|R,y)$ and 
	\begin{equation*}
		\bar{p}(T_z(R)|R,y) = \Big(1 + T_{10}^R(h_z(R)) + T_{12}^R(h_z(R),R) + T_{10}^{\rho}(h_z(R)) - \frac{1}{2}T_{30}^{\rho}(h_z(R))\Big) \cdot  \frac{\sqrt{\det(G_{t^*|y})}}{\sqrt{(2\pi)^d}} \cdot \exp\Big(-\frac{1}{2}(h_z(R))'G_{t^*|y} h_z(R)\Big),
	\end{equation*}
	is at most
	\begin{equation*}
		\begin{aligned}
			\frac{|p(T_z(R)|R,y) - \bar{p}(T_z(R)|R,y)|}{\bar{p}(T_z(R)|R,y)} \leq \Err_{p}(h_z(R)|R,y) \leq C \cdot \Err_{p}(F_{t^*}^{-1}z|R,y),
		\end{aligned}
	\end{equation*}
	with the last inequality following because $\|h_z(R)\| \leq C \|F_{t^*}^{-1}z\|$, while $\Err_{p}(h|R,y)$ is increasing in $\|h\|$. Substituting $\bar{h}_z(R)$ for $h_z(R)$ in the above expression incurs error that can be bounded above using~\eqref{pf:approximate-density-studentized-peak-1}; doing this, Taylor expanding $\exp(1 + x)$ around $x = 0$, and disregarding second-order terms, we conclude that the relative error between $\bar{p}(T_z(R)|R,y)$ and 
	\begin{equation*}
		\begin{aligned}
			\tilde{p}(z|R,y) & = \Big(1 + T_{10}^R(F_{t^*}^{-1}z) + T_{12}^R(F_{t^*}^{-1}z,R) + T_{10}^{\rho}(F_{t^*}^{-1}z) - \frac{1}{2}T_{30}^{\rho}(F_{t^*}^{-1}z) + \\
			& (F_{t^*}^{-1}z)'G_{t^*|y} \big(F_{t^*}^{-1}\dot{F}_{t^*}(F_{t^*}^{-1}z,F_{t^*}^{-1}z) + F_{t^*}^{-1}\Lambda_{t^*}^{-1/2}RF_{t^*}^{-1}z\big)\Big) \cdot \frac{\sqrt{\det(G_{t^*|y} )}}{\sqrt{(2\pi)^d}} \cdot \exp\Big(-\frac{1}{2}(F_{t^*}^{-1}z)'G_{t^*|y} F_{t^*}^{-1}z\Big)
		\end{aligned}
	\end{equation*}
	is at most 
	\begin{equation*}
		\frac{|\bar{p}(T_z(R)|R,y) - \tilde{p}(z|R,y)|}{\tilde{p}(z|R,y)} \leq C\Big((1 + \|R\|^2) \cdot \Err_{h}(z,R) + \delta_n^2\big(1 + \|z\|^6\big)\big(1 + \|R\|^2\big)\Big) \leq C\delta_n^2\big(1 + \|z\|^6\big)\big(1 + \|R\|^2\big).
	\end{equation*}
	For future reference, we note that 
	$$
	\Err_{p}(F_{t^*}^{-1}z|R,y) \leq C\delta_n^2\big(1 + \|z\|^4\big)\big(1 + \|R\|^2\big)\big(1 + |y - \bar{u}_{t^*}| + |\bar{u}_{t^*} - \mu_{t^*}|\big),
	$$
	
	\paragraph{Determinant term.}
	A first order Taylor expansion of $\dot{Z}_t(R)$ about $t = t^*, R = 0$ gives
	\begin{equation}
		\Big\|\dot{Z}_t(R) - F_{t^*}(h) - \dot{F}_{t^*}(h) - \Lambda_{t^*}^{-1/2}R - L_{t^*} H_{t^*}(h) - \Lambda_{t^*}^{-1/2}\{\dot{H}_{t^*}(h)\}\big\| \leq C\Big(\lambda_n \|h\|^2 + \|R\| \|h\| \Big) .
	\end{equation}
	As $\dot{F}_{t^*}(h) = \tilde{L}_{t^*}(H_{t^*}h) + \Lambda_{t^*}^{-1/2}\{\dot{H}_{t^*}(h)$, we can just as easily write the first-order Taylor expansion as 
	$$
	F_{t^*} + 2 \Lambda_{t^*}^{-1/2}\{\dot{H}_{t^*}(h)\} + \bar{L}_{t^*}(H_{t^*}h) + \Lambda_{t^*}^{-1/2}R.
	$$
	From this, we see that the determinant term can be locally expanded as
	\begin{equation*}
		\begin{aligned}
		& \Big|\det\big(\dot{Z}_{t}(R)\big) - \det(F_{t^*})\big(1 +2 \cdot \tr(F_{t^*}^{-1}\Lambda_{t^*}^{-1/2}\{\dot{H}_{t^*}(h)\}) + \tr(F_{t^*}^{-1}\{\bar{L}_{t^*}(H_{t^*}h)\}) + \tr(F_{t^*}^{-1}\Lambda_{t^*}^{-1/2}R)\big)\Big| \\
		& \leq C\det(F_{t^*}) (\lambda_n\|h\|^3 + \delta_n\|R\|\cdot\|h\|).
		\end{aligned}
	\end{equation*}
	Evaluating this at $t = T_z(R)$, substituting $\bar{h}_z(R)$ for $h_z(R)$, and disregarding second-order terms:
	\begin{equation*}
		\begin{aligned}
			& \Big|\det(\dot{Z}_{T_z(R)}(R)) - \det(F_{t^*})\big(1 +2 \cdot \tr(F_{t^*}^{-1}\Lambda_{t^*}^{-1/2}\{\dot{H}_{t^*}(F_{t^*}^{-1}z)\}) + \tr(F_{t^*}^{-1}\{\bar{L}_{t^*}(H_{t^*}F_{t^*}^{-1}z)\}) + \tr(F_{t^*}^{-1}\Lambda_{t^*}^{-1/2}R)\big)\Big| \\
			& \leq C \det(F_{t^*})\delta_n^2\Big(\|z\|^3 + \|R\|\cdot\|z\|\Big).
		\end{aligned}
	\end{equation*}
	Notice that this can further be upper bounded by $C\det(F_{t^*})\delta_n^2(1 + \|z\|^6)(1 + \|R\|^2)$. Combined our analyses of the density term and determinant term yield the claimed result~\eqref{eqn:density-studentized-peak-given-hessian}.
	
	\subsubsection{Joint distribution of studentized peak, residual Hessian.}
	The joint density of $\wh{Z},\wh{R}$ given $\wh{Y} = y$ is 
	\begin{equation*}
		p(z,R|y) = p(z|R,y) \cdot p(R|y).
	\end{equation*}
	Equation~\eqref{eqn:density-studentized-peak-given-hessian} gives an approximation to $p(z|R,y)$, while Lemma~\ref{lem:approximate-joint-distribution-2} gives an approximation to $p(R|y)$. Combined, these results imply~\eqref{eqn:density-studentized-peak-and-hessian}. The relative error incurred by these approximations is at most
	\begin{equation*}
		\begin{aligned}
		& \Err_{p}(z,R|y) \\
		& \quad = C\big(\Err_{p}(z|R,y) + \Err_{p}(R|y)\big) \\
		& \quad \leq C\big(\Err_{p}(F_{t^*}^{-1}z|R,y) + \delta_n^2(1 + \|z\|^6)(1 + \|R\|^2) + \delta_n^2(|y - \bar{u}_{t^*}|^2 + |\bar{u}_{t^*} - \mu_{t^*}| + \|R\|^2)\big) \\
		& \quad \leq C\delta_n^2\big((1 + \|z\|^4)(1 + \|R\|^2)(1 + |y - \bar{u}_{t^*}| + |\bar{u}_{t^*} - \mu_{t^*}|) + (1 + \|z\|^6)(1 + \|R\|^2) + |y - \bar{u}_{t^*}|^2 + |\bar{u}_{t^*} - \mu_{t^*}| + \|R\|^2 \big) \\
		& \quad \leq C \delta_n^2\big((1 + \|z\|^6)(1 + \|R\|^2)(1 + |y - \bar{u}_{t^*}|^2 + |\bar{u}_{t^*} - \mu_{t^*}|)\big).
		\end{aligned}
	\end{equation*}
	
	\subsubsection{Validity in $\mc{Z}_{t^*}$.}
	So far we have proven the validity of the local expansions $\bar{p}(z|R,y)$ and $\bar{p}(z|y)$ at points $z \in \mc{Z}_{t^*}(R)$. We now show that for $n$ sufficiently large $\mc{Z}_{t^*} \subseteq \mc{Z}_{t^*}(R)$ for all $R \in \mc{R}_{t^*}$. By definition $z \in \mc{Z}_{t^*}$ means that $h_{z}(0) = F_{t^*}^{-1}z \leq \varepsilon_n/2$. Application of~\eqref{pf:studentized-location-2} then implies that for all $n$ sufficiently large 
	\begin{equation*}
		\|h_z(R)\| \leq \|F_{t^*}^{-1}z\| + \|F_{t^*}^{-1}z - h_z(R)\| \leq \frac{\varepsilon_n}{2} + C\|h_z(R)\|^2 \leq   \frac{\varepsilon_n}{2}+ C\delta_n^2\|z\|^2 \leq \frac{\varepsilon_n}{2}+ C\varepsilon_n^2 \leq \varepsilon_n;
	\end{equation*}
	in the second-to-last inequality we have used the fact that $\mc{Z}_{t^*} \subseteq B(0,\lambda_n \varepsilon_n)$. Thus $T_z(R) \in \mc{B}_{t^*}$, implying that $z \in \mc{Z}_{t^*}(R)$.
	
	\subsubsection{Distribution of studentized peak.}
	Let 
	\begin{equation*}
		\begin{aligned}
			\tilde{p}(z,R|y) 
			& : = \big(1 + T_{01}^{p_{Z|R}}(z) + T_{30}^{p_{Z|R}}(z)+ T_{12}^{p_{Z|R}}(z,R)\big)) \\
			& \times \frac{\sqrt{\det(G_{t^*|y})}}{\sqrt{(2\pi)^d} \det(F_{t^*})} \exp\Big(-\frac{1}{2}(F_{t^*}^{-1}z)'G_{t^*|y}F_{t^*}^{-1}z\Big) \cdot \frac{1}{\sqrt{(2\pi)^{d(d + 1)/2} \cdot \det(\Theta_{t^*})}} \exp\Big(-\frac{1}{2}\vech(R)'\Theta_{t^*}^{-1}\vech(R)\Big).
		\end{aligned}
	\end{equation*}
	The difference between $\tilde{p}(z,R|y)$ and $\bar{p}(z,R|y)$ is that the former omits all first-order terms that are linear in $R$. We decompose the density of $\wh{Z}$ given $\wh{Y} = y$ into the sum of a main term and four error terms:
	\begin{equation*}
		\begin{aligned}
			p(z|y) 
			& = \frac{1}{\Q^{\mc{S}_{t^*}}(\wh{R} \in \mc{R}_{t^*})}\int_{\mc{R}_{t^*}} p(z,R|y) \,dR \\
			& = \int_{\R_{sym}^{d \times d}} \tilde{p}(z,R|y) \,dR - \int_{\mc{R}_{t^*}^c} \tilde{p}(z,R|y) \,dR + \int_{\mc{R}_{t^*}} \bar{p}(z,R|y) - \tilde{p}(z,R|y) \,dR + \int_{\mc{R}_{t^*}} \bar{p}(z,R|y) - p(z,R|y)  \,dR + \\
			& \quad \quad \Big(1 - \frac{1}{\Q^{\mc{S}_{t^*}}(\wh{R} \in \mc{R}_{t^*})}\Big) \int_{\mc{R}_{t^*}} p(z,R|y) \,dR.
		\end{aligned}
	\end{equation*}
	The main term above is $\int_{\R_{sym}^{d \times d}} \tilde{p}(z,R|y) \,dR = \bar{p}(z|y)$. The second error term above is $\int_{\mc{R}_{t^*}} \bar{p}(z,R|y) - \tilde{p}(z,R|y) \,dR = 0$, since $\mc{R}_{t^*}$ is symmetric about the origin and so the integral of any odd function in $R$ vanishes over $\mc{R}_{t^*}$. It remains to bound the three sources of error, due to local expansion, undoing the effects of truncation, and using the wrong normalizing constant. 
	
	\paragraph{Error due to local expansion.}
	
	Applying~\eqref{eqn:density-studentized-peak-and-hessian}, we can upper bound the error due to local expansion:
	\begin{equation*}
		\begin{aligned}
			\int_{\mc{R}_{t^*}} \bar{p}(z,R|y) - p(z,R|y) \,dR 
			& \leq \int_{\mc{R}_{t^*}} \Err_{p}(z,R|y) \cdot \bar{p}(z,R|y) \,dR \\
			& = \int_{\mc{R}_{t^*}} \Err_{p}(z,R|y) \cdot \bar{p}(R|z,y) \cdot \bar{p}(z|y) \,dR \\
			& \leq C\delta_n^2\int_{\mc{R}_{t^*}} (1 + \|z\|^6)(1 + \|R\|^2)(1 + |y - \bar{u}_{t^*}|^2 + |\bar{u}_{t^*} - \mu_{t^*}|) \cdot \bar{p}(R|z,y) \cdot \bar{p}(z|y) \,dR \\ 
			& \leq C\delta_n^2(1 + \|z\|^6)(1 + |y - \bar{u}_{t^*}|^2 + |\bar{u}_{t^*} - \mu_{t^*}|) \bar{p}(z|y) := \Err_{p}(z,\mc{R}_{t^*}|y) \cdot \bar{p}(z|y).
		\end{aligned}
	\end{equation*} 
	
	\paragraph{Error due to truncation.}
	In bounding the truncation error, observe that $\sigma_{22}^2 \geq \lambda_{\min}(\Theta_{t^*})$ and therefore $\xi_n \geq \{\lambda_{\min}(\Theta_{t^*})\}^{1/2}\sqrt{6 \log \lambda_n}$. Therefore,
	\begin{equation}
		\int_{\mc{R}_{t^*}^c} \exp(-\frac{1}{2}\vech(R)'\Theta_{t^*}^{-1}\vech(R)) \,dR \leq \int_{B_{d \times d}(0,\sqrt{6 \log \lambda_n})^c} \exp(-\frac{1}{2}\|\vech(R)\|^2) \,dR \leq C \delta_n^{2.5},
	\end{equation}
	the latter inequality following by~\eqref{eqn:gaussian-integrals-concentration} for all $n \in \mathbb{N}$ sufficiently large. This upper bound, along with an application of H\"{o}lder's inequality with conjugate exponents $q_1 = 5, q_2 = 1.25$, implies
	\begin{equation}
		\begin{aligned}
			& \int_{\mc{R}_{t^*}^c} \|R\|^2 \exp(-\frac{1}{2}\vech(R)'\Theta_{t^*}^{-1}\vech(R)) \,dR \\
			& \leq \Big(\int \|R\|^{10} \exp(-\frac{1}{2}\vech(R)'\Theta_{t^*}^{-1}\vech(R)) \,dR\Big)^{1/5} \cdot \Big(\int_{\mc{R}_{t^*}^c} \exp(-\frac{1}{2}\vech(R)'\Theta_{t^*}^{-1}\vech(R)) \,dR\Big)^{1/1.25} \\
			& \leq C\delta_n^2.
		\end{aligned}
	\end{equation}
	Finally the approximate density 
	$$
	\tilde{p}(z,R|y) \leq C(1 + \delta_n\|R\|^2) \exp(-\frac{1}{2} \vech(R)'\Theta_{t^*}^{-1}\vech(R)) \cdot \bar{p}(z|y).
	$$
	Thus the error due to truncation is at most
	\begin{equation*}
		\begin{aligned}
			\int_{\mc{R}_{t^*}^c} \tilde{p}(z,R|y) \,dR \leq \int_{\mc{R}_{t^*}^c} C (1 + \delta_n\|R\|^2) \exp(-\frac{1}{2} \vech(R)'\Theta_{t^*}^{-1}\vech(R)) \bar{p}(z|y)\,dR  \leq C\delta_n^{2.5} \bar{p}(z|y) := \wb{\Trunc}_p(\mc{R}_{t^*}) \cdot \bar{p}(z|y).
		\end{aligned}
	\end{equation*}
	
	\paragraph{Error due to normalizing constant.}
	Equation~\eqref{eqn:hessian-residual-truncation-error} shows that $\Q^{\mc{S}_{t^*}}(\wh{R} \in \mc{R}_{t^*}) \geq 1 - C\delta_n^2$. On the other hand, the arguments of the previous paragraph imply that $\int_{\mc{R}_{t^*}} p(z,R|y) \,dR \leq C \bar{p}(z|y)$. Therefore, the error due to using the wrong normalizing constant is at most
	\begin{equation*}
		\Big(\frac{1}{\Q^{\mc{S}_{t^*}}(\wh{R} \in \mc{R}_{t^*})} - 1\Big) \int_{\mc{R}_{t^*}} p(z,R|y) \,dR \leq C \delta_n^2 \cdot \bar{p}(z|y) := \Trunc_{p}(\mc{R}_{t^*}) \cdot \bar{p}(z|y),
	\end{equation*}
	for all $n \in \mathbb{N}$ sufficiently large.
	
	\subsection{Proof of Corollary~\ref{cor:conditional-coverage}}
	\label{subsec:pf-conditional-coverage}
	
	Conditional coverage of the height $\mu_{t^*}$ by $\mc{I}_{\hat{t}}$ follows immediately from~\eqref{eqn:height-pivot}. The probability that $C_{\hat{t}}$ covers the true location $t^*$ is 
	\begin{equation*}
		\Q^{\mc{S}_{t^*}}(t^* \in C_{\hat{t}}) = \int_{\mc{I}_{t^*}} \Q^{\mc{S}_{t^*}}(t^* \in C_{\hat{t}}|\wh{Y} = y) p(y) \,dy.
	\end{equation*}
	Theorem~\ref{thm:pivot} gives the approximation
	\begin{equation*}
		\begin{aligned}
		\wb{\Q}^{\mc{S}_{t^*}}(t^* \in C_{\hat{t}}) 
		& :=  \int_{\mc{I}_{t^*}} \big(1 - \alpha - A\cdot\frac{(y - \bar{u}_{t^*}) + (\bar{u}_{t^*} - \mu_{t^*})}{2}\tr(\bar{H}_{t^*}^{-1}\Lambda_{t^*})\big) p(y) \,dy \\
		& = 1 - \alpha - A\cdot \frac{(\E^{\mc{S}_{t^*}}[\wh{Y}]  - \mu_{t^*})}{2}\tr(\bar{H}_{t^*}^{-1}\Lambda_{t^*}).  
		\end{aligned}
	\end{equation*}
	The difference between this approximation and the true coverage probability is at most
	\begin{equation*}
		\begin{aligned}
			\Big|\Q^{\mc{S}_{t^*}}(t^* \in C_{\hat{t}})  - \wb{\Q}^{\mc{S}_{t^*}}(t^* \in C_{\hat{t}}) \Big| 
			& \leq \int_{\mc{I}_{t^*}} \Err_{W}(y) \cdot p(y) \,dy \\
			& \leq C\delta_n^2\int_{\mc{I}_{t^*}} \Big(1 + |y - \bar{u}_{t^*}| ^2+ |\bar{u}_{t^*} - \mu_{t^*}|^2)\Big) p(y) \,dy.
		\end{aligned}
	\end{equation*}
	To bound the integral in the last line above, we first apply Theorem~\ref{thm:approximate-joint-distribution} which implies that 
	\begin{equation*}
		\Big|\int_{\mc{I}_{t^*}} |y - \bar{u}_{t^*}|^2 p(y) \,dy  - \int_{\mc{I}_{t^*}} |y - \bar{u}_{t^*}|^2 \bar{p}(y) \,dy \Big| \leq \int_{\mc{I}_{t^*}} |y - \bar{u}_{t^*}|^2 \Err_{p}(y) \bar{p}(y) \,dy \leq C(|\bar{u}_{t^*} - \mu_{t^*}| + 1)\delta_n^2,
	\end{equation*}
	with the latter inequality following from~\eqref{eqn:normalized-overshoot-moments}. On the other hand, 
	\begin{equation*}
		\int_{\mc{I}_{t^*}} |y - \bar{u}_{t^*}|^2 \bar{p}(y) \,dy \leq C,
	\end{equation*}
	again by~\eqref{eqn:normalized-overshoot-moments}. Thus we conclude that 
	\begin{equation*}
		\Big|\Q^{\mc{S}_{t^*}}(t^* \in C_{\hat{t}}) - \big(1 - \alpha - A\cdot \frac{(\E^{\mc{S}_{t^*}}[\wh{Y}]  - \mu_{t^*})}{2}\tr(\bar{H}_{t^*}^{-1}\Lambda_{t^*})\big)\Big| \leq C\delta_n^2(1+ |\bar{u}_{t^*} - \mu_{t^*}|^2),
	\end{equation*}
	which is the desired claim.
	
	\subsection{Proof of Theorem~\ref{thm:miscoverage}}
	\label{subsec:pf-miscoverage}
	We will explicitly derive the upper bound on $\Miscov_{T^*}$ in terms of $\wb{\Miscov}_n$ as stated in~\eqref{eqn:miscoverage-location}. The upper bound on $\Miscov_{\mu_{T^*}}$ in terms of $\wb{\Miscov}_n$, given in~\eqref{eqn:miscoverage-height}, is derived in a very similar way and we omit the proof. Throughout this proof, we write $v \equiv v_n$ and $u_{\TG} \equiv u_{\TG}(\alpha,v_n)$.
	
	The proof of~\eqref{eqn:miscoverage-location} will proceed similarly to the proof of Theorem~\ref{thm:epsilon-fpr}: we partition the parameter space $\mc{T}$ into regions and consider the expected number of regions produced that fail to cover in each region. We introduce notation to count the number of points at which a region is produced that fails to cover: for a set $\mc{A} = (\mc{B} \times \mc{I}) \subset \mc{T} \times \R$,
	\begin{equation*}
		V(\mc{A}) := N(t \in \wh{T} \cap \mc{B}: Y_t \in \mc{I}, t^*(t) \not\in C_t),
	\end{equation*}
	and abbreviate $V_u(\mc{B}) = V(\mc{B} \times (u,\infty))$. The overall miscoverage rate of our method is 
	\begin{equation*}
		\frac{\E[V_{u_{\TG}}(\mc{T})]}{\E[N_v(\mc{T})]} = \frac{\E[V_{u_{\TG}}(\mc{T}_{\varepsilon_n}^{*}) + V_{u_{\TG}}(\mc{G}_{\varepsilon_n}) + V_{u_{\TG}}(\mc{T}_0)]}{\E[N_v(\mc{T}_{\varepsilon_n}^{*}) + N_v(\mc{G}_{\varepsilon_n}) + N_v(\mc{T}_0)]}.
	\end{equation*}
	As $V_u(\mc{S}) \leq N_u(\mc{S})$, the miscoverage rate is at most
	\begin{equation}
		\label{pf:miscoverage-1}
		\frac{\E[V_{u_{\TG}}(\mc{T}_{\varepsilon_n}^{*}) + N_{u_{\TG}}(\mc{T}_{\varepsilon_n}) + N_{u_{\TG}}(\mc{T}_0)]}{\E[N_v(\mc{T}_{\varepsilon_n}^{*}) + N_v(\mc{T}_0)]}.
	\end{equation}
	Proposition~\ref{prop:null-false-positive-rate} upper bounds the expected number of null false positives $\E[N_{u_{\TG}(\alpha,v)}(\mc{T}_0)]$ in terms of $\E[N_v(\mc{T}_0)]$:
	\begin{equation*}
		\E[N_{u_{\TG}}(\mc{T}_0)] \leq (1 + C/v^2) \cdot \alpha \cdot \E[N_v(\mc{T}_0)].
	\end{equation*} 
	Propositions~\ref{prop:expectation-counting-process} and~\ref{prop:high-gradient-peaks} upper bound the expected number of high-gradient peaks $\E[N_{u_{\TG}}(\mc{G}_{\varepsilon_n})]$ in terms of the expected number of $\varepsilon_n$-consistent estimates $\E[N_v(\mc{T}_{\varepsilon_n}^{*})]$: 
	\begin{equation*}
		\E[N_{u_{\TG}}(\mc{G}_{\varepsilon_n})] \leq C \delta_n^2 \E[N_{u_{\TG}}(\mc{T}_{\varepsilon_n}^{*})] \leq C \delta_n^2 \E[N_{v}(\mc{T}_{\varepsilon_n}^{*})].
	\end{equation*}
	For all $n \in \mathbb{N}$ sufficiently large, $\mc{T}_{\varepsilon_n}^{*} = \cup_{t^* \in T^*} B(t^*,\varepsilon_n)$ is a union of disjoint balls of radius $\varepsilon_n$. The expected number of $\varepsilon_n$-consistent estimates is therefore bounded below by
	\begin{align*}
		\E[N_v(\mc{T}_{\varepsilon_n}^{*})] = \sum_{t^* \in T^*} \E[N_v(\mc{B}_{t^*})] \geq \sum_{t^* \in T^*} \P(N_v(\mc{B}_{t^*}) = 1).
	\end{align*}
	What is left is to upper bound the expected number of $\varepsilon_n$-consistent estimates with corresponding confidence regions that fail to cover the truth. Applying Proposition~\ref{prop:expectation-counting-process} and Proposition~\ref{prop:no-more-than-one-process-peak-per-signal-peak}:  
	\begin{align*}
		\E[V_{u_{\TG}}(\mc{T}_{\varepsilon_n}^{*})] 
		& = \sum_{t^* \in T^*} \E[V_{u_{\TG}}(\mc{B}_{t^*})] \\
		& \leq \sum_{t^* \in T^*} \E[V(\mc{S}_{t^*})] + \E[N_{\mc{I}_{t^*}^{c}}(\mc{B}_{t^*})] \\
		& \leq \sum_{t^* \in T^*} \E[V(\mc{S}_{t^*})] + C\Big((\bar{u}_{t^*} - \mu_{t^*})\delta_n^2 + \delta_n^2\Big) \E[N(\mc{S}_{t^*})] \\
		& \leq \sum_{t^* \in T^*} \E[V(\mc{S}_{t^*})] + C\Big((\bar{u}_{t^*} - \mu_{t^*})\delta_n^2 + \delta_n^2\Big) \P(N(\mc{S}_{t^*}) = 1).
	\end{align*} 
	Finally, we apply Corollary~\ref{cor:conditional-coverage} to upper bound the expected number of $\varepsilon_n$-consistent estimates of a particular $t^* \in T^*$ with corresponding confidence regions that fail to cover $t^*$:
	\begin{equation*}
		\begin{aligned}
			\E[V(\mc{S}_{t^*})] 
			& = \E[V(\mc{S}_{t^*}) \cdot \1(N(\mc{S}_{t^*}) = 1)] +  \E[V(\mc{S}_{t^*}) \cdot \1(N(\mc{S}_{t^*}) \geq 1)] \\
			& = \Q^{\mc{S}_{t^*}}(t^* \not\in C_{\hat{t}}) +  \E[V(\mc{S}_{t^*}) \cdot \1(N(\mc{S}_{t^*}) \geq 1)] \\
			& \leq C\Big(\alpha + A\cdot \frac{(\E^{\mc{S}_{t^*}}[\wh{Y}]  - \mu_{t^*})}{2}\tr(\bar{H}_{t^*}^{-1}\Lambda_{t^*}) + C\delta_n^2\big(1 + (\bar{u}_{t^*} - \mu_{t^*})^2\big) +  \E[N(\mc{S}_{t^*}) \cdot \1(N(\mc{S}_{t^*}) \geq 1)]\Big) \\
			& \leq C\Big(\alpha + A\cdot \frac{(\E^{\mc{S}_{t^*}}[\wh{Y}]  - \mu_{t^*})}{2}\tr(\bar{H}_{t^*}^{-1}\Lambda_{t^*}) + C\delta_n^2\big(1 + (\bar{u}_{t^*} - \mu_{t^*})^2\big)+  C\exp(-c\lambda_n^2)\Big),
		\end{aligned}
	\end{equation*}
	where the last inequality is established in the proof of Proposition~\ref{prop:no-more-than-one-process-peak-per-signal-peak}. Combining all these bounds, basic algebra implies that 
	\begin{equation*}
		\begin{aligned}
			\frac{\E[V_{u_{\TG}}(\mc{T})]}{\E[N_v(\mc{T})]} & \leq \wb{\Miscov}_n + C\Big(\frac{1}{v^2} + \delta_n^2 + A \cdot \sup_{t^* \in T^*} (\E^{\mc{S}_{t^*}}[\wh{Y}]  - \mu_{t^*}) \delta_n + \sup_{t^* \in T^*} (\bar{u}_{t^*} - \mu_{t^*})^2\delta_n^2\Big),
		\end{aligned}
	\end{equation*}
	establishing~\eqref{eqn:miscoverage-location}.
	
	\section{Peak inference with randomized selection: details and derivations}
	\label{sec:randomized-peak-inference-analysis}

	In this section we explain in more detail the data splitting and data carving methods for peak inference after randomized selection discussed in Section~\ref{sec:randomized-peak-inference}, and provide heuristic derivations of the local expansions~\eqref{eqn:density-height-carve} and~\eqref{eqn:density-location-carve}. Recall that the peak detection portion of both methods is identical, and both methods will select the same randomized peaks $\wc{T}_u$. Thus, for both methods we will be interested in the distribution of peaks nearby $t^* \in T^*$, given that $t^*$ is uniquely $\varepsilon_n^{\gamma}$-consistently estimated by some randomized peak $\check{t} \in \wc{T}_u$. (Recall that $\varepsilon_n^{\gamma} = \sigma_{\gamma} \cdot \varepsilon_n$.) Formally speaking, the event we will condition on is $N^{\sel}(\mc{S}_{t^*}^{\sel}) = 1$, where
		\begin{equation}
			N^{\sel}(\mc{A}) = \Big\{(t',y') \in \mc{A}: Y_{t'}^{\sel} = y', Y_{t'}^{\sel} > u, \nabla Y_{t'}^{\sel} = 0, \nabla^2 Y_{t'}^{\sel} \prec 0\Big\},
		\end{equation}
		and 
		\begin{equation*}
			\mc{S}_{t^*}^{\sel} := \Big\{(t',y') \in \R^{d + 1}: t' \in B_d(t^*,\varepsilon_n^{\gamma}), y' \in \mc{I}_{t^*}^{\sel}\}, \quad  \mc{I}_{t^*}^{\sel} := (\bar{u}_{t^*} \pm \Delta_n) \cap (u,\infty).
		\end{equation*}Section~\ref{subsec:split} discusses data splitting, and Sections~\ref{subsec:carve-densities}-\ref{subsec:density-location-carve} deal with data carving.

	\subsection{Peak inference by data splitting}
	\label{subsec:split}
	As discussed in the main text, the general principle of data splitting is that the data used for inference be independent of the data used for selection. In our peak inference context, the data used for selection is $Y^{\sel}$. Thus, our data splitting method conducts inference using the location and height of peaks of the independent field $Y^{\inf}$,
	\begin{equation*}
		\wt{T} = \{t \in \mc{T}: \nabla Y_t^{\inf} = 0,\nabla^2 Y_t^{\inf} \prec 0\}.
	\end{equation*}
	As always, we would like our confidence regions for the location and height of a true peak $t^* \in T^*$ to be valid conditional on $t^*$ being discovered; in the context of randomized selection, this means valid conditional on $N^{\sel}(\mc{S}_{t^*}) = 1$. To construct these confidence regions, we will be interested in the asymptotic distribution of the height and location of a peak $\tilde{t} \in \wt{T}$ nearby $t^*$ conditional on $N^{\sel}(\mc{S}_{t^*}) = 1$; but since $Y^{\sel}$ and $Y^{\inf}$ are independent, this is the same as the unconditional asymptotic distribution of $\tilde{t}$.
	
	The asymptotic distribution of $\tilde{t}$ and $\wt{Y} := Y_{\tilde{t}}$ can be obtained by applying the results of Sections~\ref{sec:peak-estimation} and~\ref{sec:peak-distribution}, with threshold $u = -\infty$, to the field $Y_{t}^{\inf}/\tau_{\gamma}$ where $\tau_{\gamma} = \sqrt{1 + 1/\gamma}$. From this we deduce the following: (1) for each $t^* \in T^*$, there will be a unique $\tilde{t} \in \wt{T}$ that $\tau_{\gamma} \cdot \varepsilon_n$-consistently estimates $t^*$, (2) letting $\wt{H} = -\nabla^2 Y_{\tilde{t}}^{\inf}$ and $\wt{G} = \wt{H}\Lambda_{\tilde{t}}^{-1}\wt{H}$,
	\begin{equation*}
		\Psi\Big(\frac{\wt{Y} - \mu_{t^*} - \frac{\tau_{\gamma}^2}{2}\tr(\wt{H}^{-1}\Lambda_{\tilde{t}})}{\tau_{\gamma}}\Big) \convweak {\rm Unif}(0,1), \quad (\tilde{t} - t^*)'\wt{G}(\tilde{t} - t^*) \convweak \chi_d^2.
	\end{equation*}
	In each case, the rate of convergence is $\wt{O}(\delta_n^2)$, meaning these are asymptotic second-order pivotal.
	Inverting tests based on these asymptotic pivots thus lead to confidence regions that have second-order accurate asymptotic coverage, both marginally and conditionally on $N^{\sel}(\mc{S}_{t^*}^{\sel}) = 1$. Our overall method for post-selection peak inference via data splitting is summarized in Algorithm~\ref{alg:split}. 
	
	\begin{algorithm}[t]
		\caption{Selective peak inference via data splitting}
		\label{alg:split}
		\begin{algorithmic}[1]
			\REQUIRE Field $Y$, gradient covariance $\Cov[\nabla \epsilon_t]$, pre-threshold $v \in \mathbb{R}$, significance level $\alpha \in (0,1)$, randomization level $\gamma > 0$.
			\STATE Select peaks $\wc{T}_v$ by pre-thresholding $Y^{\sel}$.
			\STATE At each peak $t \in \wc{T}_v$, test the null hypothesis $H_{0,t}$ by applying the TG test to $Y^{\sel}$ with significance threshold $\sigma_{\gamma} \cdot u_{\TG}(\alpha,v)$, making discoveries $\wc{T}_{u_{\TG}(\alpha,v)}$.
			\STATE For each discovery $\check{t} \in \wc{T}_{u_{\TG}(\alpha,v)}$, let $\tilde{t} := \argmin_{t \in \wt{T}} \|\check{t} - \tilde{t}\|$ be the nearest peak in $Y^{\inf}$. Return confidence regions
			\begin{equation}
				\label{eqn:confidence-regions-split}
				\begin{aligned}
					I_{\hat{t}}^{{\rm split}} & := \Big\{\mu: \frac{\alpha}{2} \leq \Psi\Big(\frac{\wt{Y} - \mu_{t^*} - \frac{\tau_{\gamma}^2}{2}\tr(\wt{H}^{-1}\Lambda_{\tilde{t}})}{\tau_{\gamma}}\Big) \leq 1 - \frac{\alpha}{2}\Big\} \\
					C_{\hat{t}}^{{\rm split}} & := \Big\{t: (\tilde{t} - t)'\wt{G}(\tilde{t} - t) \leq q_{\chi_d^2}(1 - \alpha) \Big\}. \\
				\end{aligned}
			\end{equation}
		\end{algorithmic}
	\end{algorithm}

	\subsection{Peak inference by data carving}
	\label{subsec:carve-densities}
	We calibrate post-selection peak inference via data carving using the conditional distribution
	\begin{equation}
		\label{eqn:distribution-after-randomized-selection}
		\Q^{\carve}(\mc{A}) := \P\Big((\hat{t},\wh{Y}) \in \mc{A}|N^{\sel}(\mc{S}_{t^*}^{\sel}) = 1, N(\mc{S}_{t^*}^{\gamma}) = 1\Big),
	\end{equation}
	where
	\begin{equation*}
		\mc{S}_{t^*}^{\gamma} := \Big\{(t,y) \in \R^{d + 1}: t \in B_d(t^*,\varepsilon_n), y \in \mc{I}_{t^*}^{\gamma}\Big\}, \quad \mc{I}_{t^*}^{\gamma} := \bar{u}_{t^*}^{\gamma} \pm \Delta_n, \bar{u}_{t^*}^{\gamma} := (1 - \pi) \bar{u}_{t^*} + \pi \mu_{t^*},
	\end{equation*}
	and we recall that $\pi = \gamma /(1 + \gamma)$ corresponds to the fraction of information ``left behind'' for inference. The event $N^{\sel}(\mc{S}_{t^*}^{\sel}) = 1$ corresponds to the event that $t^*$ is $\varepsilon_n^{\gamma}$-consistently discovered by a randomized peak $\check{t}$. The event $N(\mc{S}_{t^*}^{\gamma}) = 1$ occurs with asymptotic probability tending to one, and ensures there is a unique full-data peak nearby $t^*$.

	A formula for the $\Q^{\carve}$-distribution of $(\hat{t},\wh{Y})$  can be derived from the joint distribution of peaks $\wh{T},\wc{T}$. Let $p^{\pair}(t,y,t',y')$ be the probability density of $(\hat{t}, \wh{Y}, \check{t},\wc{Y})$ under the event $N^{\sel}(\mc{S}_{t^*}^{\sel}) = 1, N(\mc{S}_{t^*}^{\gamma}) = 1$: then 
	\begin{equation*}
		\Q^{\carve}(\mc{A}) = \int_{\mc{A}} \int_{\mc{S}_{t^*}^{\sel}} p^{\pair}(t,y,t',y') \,dt' \,dy', \quad \mc{A} \subseteq \mc{S}_{t^*}^{\gamma}
	\end{equation*}
	and thus the probability density of $(\hat{t},\wh{Y})$ under the same event is
	\begin{equation*}
		p^{\carve}(t,y) = \int_{\mc{S}_{t^*}^{\sel}} p^{\pair}(t,y,t',y') \,dt' \,dy'.
	\end{equation*}
	The density of the height, and the conditional density of the location given the height, are correspondingly given by
	\begin{equation*}
		p^{\carve}(y) := \int_{\mc{B}_{t^*}} p^{\carve}(t, y) \,dt, \quad
		p^{\carve}(t|y) := \frac{p^{\carve}(t,y)}{p^{\carve}(y)}.
	\end{equation*}
	
	\paragraph{Roadmap.}
	We now provide a high-level roadmap of the heuristic derivations of our approximations to $p^{\carve}(t,y), p^{\carve}(y)$ and $p^{\carve}(t|y)$. As in our analysis of the non-randomized case, a central role is played by the intensity function of a point process that counts peaks. In this case the relevant point process counts \emph{pairs} of peaks of $Y$ and $Y^{\sel}$. We introduce the vector-valued random field $\Phi_{t,t',y,y'}= (Y_{t'}^{\sel} - y',Y_t - y, \nabla Y_{t'}^{\sel},\nabla Y_t)$, and notice that pairs of peaks correspond to zeros of $\Phi$:
	\begin{equation}
		N^{\pair}(\mc{A},\mc{A}') := N\Big\{(t,y) \in \mc{A}, (t',y') \in \mc{A}':\Phi_{t,t',y,y'} = 0, \nabla^2 Y_t \prec 0, \nabla^2 Y_{t'}^{\sel} \prec 0\Big\}.
	\end{equation}
	The Kac-Rice theorem yields a formula for the intensity $\rho^{\pair}(t,y,t',y')$ of $N^{\pair}(\mc{A},\mc{A'})$. 
	In Section~\ref{subsec:pairs-intensity}, we provide a local expansion of $\rho^{\pair}(t,y,t',y')$ around $t' = t, y' = \bar{u}_{t^*}$. Then in Sections~\ref{subsec:joint-intensity-carve} and \ref{subsec:density-height-carve}  we marginalize this local expansion over $(t',y')$ to obtain approximations $\bar{\rho}^{\carve}(t,y)$ and $\bar{\rho}^{\carve}(y)$ to the ``marginal'' intensity functions 
	\begin{equation*}
		\begin{aligned}
			\rho^{\carve}(t,y) := \int_{\mc{I}_{t^*}^{\gamma}} \int_{\mc{B}_{t^*}^{\gamma}} \rho^{\pair}(t,y,t',y') \,dt' \,dy', \quad \rho^{\carve}(y) := \int_{\mc{B}_{t^*}} \rho^{\carve}(t,y) \,dt.
		\end{aligned}
	\end{equation*}
	Finally, our approximations to the conditional densities of interest will be 
	\begin{equation*}
		\bar{p}^{\carve}(t|y) = \frac{\bar{\rho}^{\carve}(t,y)}{\bar{\rho}^{\carve}(y)}, \quad \bar{p}^{\carve}(y) \propto \bar{\rho}^{\carve}(y).
	\end{equation*}
	yielding the explicit local expansions stated in~\eqref{eqn:density-height-carve} and~\eqref{eqn:density-location-carve}.
	
	Our calculations will be heuristic in two senses. First, we will not provide explicit upper bounds on the relative error due to local expansion and truncation, though we expect that the error terms are ultimately second-order. Second, we are (implicitly) approximating the density $p^{\carve}$ (up to a constant of
	proportionality) by the intensity $\rho^{\pair}$, through
	\begin{equation*}
		p^{\carve}(t,y,t',y') = \lim_{r \to 0} \frac{1}{\nu_r} \frac{\P(N^{\pair}(\mc{A}_r,\mc{A}_{r}') = 1)}{\P(N^{\pair}(\mc{S}_{t^*}^{\gamma},\mc{S}_{t^*}^{\sel}) = 1)}\approx \lim_{r \to 0}\frac{\E[N^{\pair}(\mc{A}_r)]}{\E[N^{\pair}(\mc{S}_{t^*}^{\gamma},\mc{S}_{t^*}^{\sel})]} = \frac{\rho^{\pair}(t,y,t',y')}{\E[N^{\pair}(\mc{S}_{t^*}^{\gamma},\mc{S}_{t^*}^{\sel})]},
	\end{equation*}
	where $\mc{A}_r = B(t,r) \times [y - r,y + r]$ and $\mc{A}_r' = B(t',r) \times [y' - r,y' + r]$ and $\nu_r$ is the (Lebesgue) measure of $\mc{A}_r \times \mc{A}_{r}'$. This approximation ignores the possibility that $N^{\pair}(\cdot,\cdot) > 1$.
	We expect that analysis similar to that used in the proof of Proposition~\ref{prop:no-more-than-one-process-peak-per-signal-peak} should confirm that the error thus incurred is at most $C\exp(-c\lambda_n^2)$ but we do not give an explicit bound to that effect.
	
	\subsection{Local expansion of $\rho^{\pair}$}
	\label{subsec:pairs-intensity}
	Under the smoothness assumptions of Section~\ref{subsec:model}, the counting process $N^{\pair}(\mc{A},\mc{A}')$ satisfies all of the conditions of the Kac-Rice theorem, so long as $\gamma > 0$ to rule out degenerate cases. Thus we have a Kac-Rice formula for its intensity function,	
	\begin{equation}
		\label{eqn:kac-rice-pairs}
		\rho^{\pair}(t,y,t',y') = \E\Big[|\det(\nabla \Phi_{t,t',y,y'})| \cdot \1(\nabla^2 Y_t \prec 0, \nabla^2 Y_{t'}^{\sel} \prec 0)| \Phi_{t,t',y,y'} = 0\Big] \cdot f_{\Phi_{t,t',y,y'}}(0).
	\end{equation}
	Above $\nabla \Phi_{t,t',y,y'} \in \R^{2(d + 1) \times 2(d + 1)} $ is the Jacobian of $\Phi$, while $f_{\Phi_{t,t',y,y'}}(0)$ is the (Gaussian) density of $\Phi_{t,t',y,y'}$ evaluated at $0$.  From this formula we derive a local expansion $\bar{\rho}^{\pair}(t,y,t + h,y')$ around $h = 0,y' = \bar{u}_{t^*}$, as usual by separately computing local expansions of the density and determinant terms.
	
	\subsubsection{Determinant term}
	The determinant term is
	\begin{equation*}
		\E\Big[|\det(\nabla \Phi_{t,t',y,y'})| \cdot \1(\nabla^2 Y_t \prec 0, \nabla^2 Y_{t'}^{\sel} \prec 0)|\Phi_{t,t',y,y'} = 0\Big].
	\end{equation*}
	The Jacobian $\nabla \Phi_{t,t',y,y'}$ is block diagonal and direct calculation shows that
	$\det(\nabla \Phi_{t,t',y,y'}) = \det(\nabla^2 Y_t) \cdot \det(\nabla^2 Y_{t'}^{\sel})$,
	and therefore the determinant term is
	\begin{equation*}
		\E\Big[\det(\nabla^2 Y_t) \cdot \det(\nabla^2 Y_{t'}^{\sel})  \cdot \1(\nabla^2 Y_t \prec 0, \nabla^2 Y_{t'}^{\sel} \prec 0)|\Phi_{t,t',y,y'} = 0\Big].
	\end{equation*}
	Ignoring the negative-definite indicators and swapping determinant and expectation gives
	\begin{equation}
		\label{eqn:pairs-intensity-jacobian-term-1}
		\begin{aligned}
			& \E\Big[\det(\nabla^2 Y_t) \cdot \det(\nabla^2 Y_{t'}^{\sel})  \cdot \1(\nabla^2 Y_t \prec 0, \nabla^2 Y_{t'}^{\sel} \prec 0)|\Phi_{t,t',y,y'} = 0\Big] \\
			& \approx \det\Big(\E[\nabla^2 Y_t|\Phi_{t,t',y,y'} = 0]\Big) \cdot \det\Big(\E[\nabla^2 Y_{t'}^{\sel}|\Phi_{t,t',y,y'} = 0]\Big).
		\end{aligned}
	\end{equation}
	Reasoning similar to that used in the proof of Lemma~\ref{lem:approximation-peak-intensity} should justify that the error incurred in~\eqref{eqn:pairs-intensity-jacobian-term-1} is second-order, but we do not prove this. Instead we move to computing a local expansion of~\eqref{eqn:pairs-intensity-jacobian-term-1} about $t' = t, y' = \bar{u}_{t^*}$. 
	
	\paragraph{Conditional expectation of Hessian.}
	The conditional expectation of $-\nabla^2 Y^{\sel}_{t'}$ given $\Phi_{t,t',y,y'} = 0$ is
	\begin{equation*}
		J_{t'}^{t,t',y,y'} := \Big(-\nabla^2\mu_{t'} + \Cov[\nabla^2 Y_{t'}^{\sel},\Phi_{t,t',y,y'}]\Big(\big\{\Var[\Phi_{t,t',y,y'}]\big\}^{-1}\E[\Phi_{t,t',y,y'}]\Big),
	\end{equation*}
	where the array $\Cov[\nabla^2 Y_{t'}^{\sel},\Phi_{t,t',y,y'}]$ has entries $\Cov[\nabla^2 Y_{t'}^{\sel},\Phi_{t,t',y,y'}]_{ijk} = \Cov[\partial_{ij} Y_{t'}^{\sel}, (\Phi_{t,t',y,y'})_k]$. We will show that at points $t' = t + h, h = \wt{O}(\delta_n)$ and $y' = \bar{u}_{t^*} + \wt{O}(1)$, 
	\begin{equation*}
		J_{t + h}^{t, t + h,y,y'} \approx \bar{H}_t + T_{10}^{J}(h) + T_{01}^{J}(y'), \quad \textrm{where} \quad \bar{H}_{t} := H_{t|\bar{u}_{t^*}}.
	\end{equation*}
	(The first-order terms $T_{10}^{J}(h)$ and $T_{01}^{J}(y')$ are defined below.) To derive this approximation to $J_{t + h}^{t, t + h,y,y'}$, we start by taking a first-order Taylor expansion of each term in the definition of $J_{t'}^{t,t',y,y'}$, which gives the approximations:
	\begin{align*}
		\nabla^2\mu_{t + h} & \approx \nabla^2 \mu_t + \nabla^3 \mu_t(h) \\
		\Var[\Phi_{t,t + h}(y,y')] 
		& = 
		\begin{bmatrix}
			1 + \gamma & K(t + h,t) & 0 & K_{01}(t + h,t) \\
			K(t, t + h) & 1 & K_{01}(t,t + h) & 0 \\
			0 & K_{10}(t + h,t) & (1 + \gamma) \Lambda_{t + h} & K_{11}(t + h,t) \\
			K_{10}(t,t + h) & 0 & K_{11}(t,t + h) & \Lambda_t
		\end{bmatrix} \\
		& \approx
		\begin{bmatrix}
			1 + \gamma & 1 & 0 & 0 \\
			1 & 1 & 0 & 0 \\
			0 & 0 & (1 + \gamma) \Lambda_t  & \Lambda_t\\
			0 & 0 & \Lambda_t & \Lambda_t
		\end{bmatrix} +
		\begin{bmatrix}
			0 & 0 & 0 & \Lambda_t h \\
			0 & 0 & -\Lambda_t h & 0 \\
			0 & -\Lambda_t h & (1 + \gamma) \dot{\Lambda}_t(h) & K_{21}(t)(h) \\
			\Lambda_th & 0 & K_{21}(t)(h) & 0
		\end{bmatrix}
		\\
		& := V_t^{\Phi} + A_t(h)
		\\
		\E[\Phi_{t,t + h,y,y'}] 
		& =
		\begin{bmatrix}
			y' - \mu_{t + h} \\
			y - \mu_t \\
			-\nabla \mu_{t + h} \\
			-\nabla \mu_t
		\end{bmatrix} \\
		& \approx
		\begin{bmatrix}
			y' - \mu_{t}  \\
			y - \mu_t \\
			-\nabla \mu_{t} \\
			-\nabla \mu_t
		\end{bmatrix}
		+
		\begin{bmatrix}
			- \nabla\mu_t'h \\
			0 \\
			- \nabla^2 \mu_t h \\
			0
		\end{bmatrix} \\
		& := \E[\Phi_{t,t,y,y'}]  + B_t(h). 	
	\end{align*}
	Finally, we introduce some notation to compactly write the first-order Taylor expansion of $\Cov[\nabla^2 Y_{t + h}, \Phi_{t,t + h,y,y'}]$. For an array $A$, we write $A = \begin{pmatrix}D_1 \\ \vdots \\ D_m\end{pmatrix}$ to mean that $A_{ijk} = D_k$. In this notation, 
	\begin{align*}
		\Cov[\nabla^2 Y^{\sel}_{t + h},\Phi_{t,t + h,y,y'}] 
		& = 
		\begin{pmatrix}
			-(1 + \gamma) \cdot \Lambda_{t + h} \\
			K_{20}(t + h,t) \\
			(1 + \gamma) \cdot K_{21}(t + h) \\
			K_{21}(t + h,t)\\
		\end{pmatrix} \\
		& \approx
		\begin{pmatrix}
			-(1 + \gamma) \cdot \Lambda_t  \\
			-\Lambda_t \\
			(1 + \gamma) \cdot K_{21}(t) \\
			K_{21}(t) \\
		\end{pmatrix}
		+ 
		\begin{pmatrix}
			(1 + \gamma) \cdot \dot{\Lambda}_t(h) \\
			K_{30}(t)(h) \\
			(1 + \gamma) \cdot K_{31}(t)(h) + (1 + \gamma) \cdot K_{22}(t)(h) \\
			K_{31}(t)(h)
		\end{pmatrix} \\
		& := C_t^{\nabla^2 Y^{\sel}, \Phi} + C_t(h).
	\end{align*}
	Combining these approximations, we conclude that the following approximation of the conditional expectation of the Hessian is accurate up to (nearly) second-order relative error:
	\begin{equation*}
		\begin{aligned}
			J_{t}^{t,t + h,y,y'}
			& \approx -\nabla^2 \mu_t + C_t^{\nabla^2 Y^{\sel}, \Phi} \{V_t^{\Phi}\}^{-1}\E[\Phi_{t,t,y,y'}] \\
			& \;\; - \nabla^3 \mu_t(h) + C_t(h) \{V_t^{\Phi}\}^{-1}\E[\Phi_{t,t,y,y'}]  - C_t^{\nabla^2 Y^{\sel}, \Phi}\{V_{t}^{\Phi}\}^{-1}A_t(h)\{V_{t}^{\Phi}\}^{-1}\E[\Phi_{t,t,y,y'}] + C_t^{\nabla^2 Y^{\sel}, \Phi} \{V_t^{\Phi}\}^{-1}B_t(h).
		\end{aligned}
	\end{equation*}
	This can be further simplified. In the first-order terms above $\E[\Phi_{t,t,y,y'}]$ can be replaced by $\E[\Phi_{t,t,y,\bar{u}_{t^*}}]$, incurring second-order error $\wt{O}(\delta_n)$. Additionally, the leading-order term above simplifies to
	\begin{equation*}
		\begin{aligned}
			-\nabla^2 \mu_t + C_t^{\nabla^2 Y^{\sel}, \Phi} \{V_t^{\Phi}\}^{-1}\Phi_{t,t}(y,y')
			& =
			-\nabla^2 \mu_t +
			\begin{bmatrix}
				-(1 + \gamma) \cdot \Lambda_t  \\
				-\Lambda_t \\
				(1 + \gamma) \cdot K_{21}(t) \\
				K_{21}(t) \\
			\end{bmatrix}
			\Bigg(
			\begin{bmatrix}
				1 + \gamma & 1 & 0 & 0 \\
				1 & 1 & 0 & 0 \\
				0 & 0 & (1 + \gamma) \Lambda_t  & \Lambda_t\\
				0 & 0 & \Lambda_t & \Lambda_t
			\end{bmatrix}^{-1}
			\begin{bmatrix}
				y' - \mu_{t}  \\
				y - \mu_t \\
				-\nabla \mu_{t} \\
				-\nabla \mu_t
			\end{bmatrix}
			\Bigg)
			\\
			& = 
			-\nabla^2 \mu_t +
			\frac{1}{\gamma}
			\begin{bmatrix}
				-(1 + \gamma) \cdot \Lambda_t  \\
				-\Lambda_t \\
				(1 + \gamma) \cdot K_{21}(t) \\
				K_{21}(t) \\
			\end{bmatrix}
			\begin{bmatrix}
				(y' - \mu_t) - (y - \mu_t) \\
				-(y' - \mu_t) + (1 + \gamma)(y - \mu_t) \\
				0 \\
				\gamma \Lambda_{t}^{-1} \nabla \mu_t
			\end{bmatrix} \\
			& = -\nabla^2 \mu_t + (y' - \mu_t) \Lambda_t + \Gamma_t(\nabla \mu_t) \\
			& = \bar{H}_{t} + (y' - \bar{u}_{t^*}) \Lambda_t.
		\end{aligned}
	\end{equation*}
	Combining these yields our approximation to the Hessian of the selection field:
	\begin{equation*}
		\begin{aligned}
			J_t^{t,t+ h,y,y'} & \approx \bar{H}_t + T_{10}^{J}(h) + T_{01}^{J}(y'), \quad \textrm{where}\\
			T_{10}^{J}(h) & = -\nabla^3 \mu_t(h) + C_t(h) \{V_t^{\Phi}\}^{-1}\E[\Phi_{t,t,y,\bar{u}_{t^*}}] - C_t^{\nabla^2 Y^{\sel}, \Phi}\{V_{t}^{\Phi}\}^{-1}A_t(h)\{V_{t}^{\Phi}\}^{-1}\E[\Phi_{t,t,y,\bar{u}_{t^*}}] + C_t^{\nabla^2 Y^{\sel}, \Phi} \{V_t^{\Phi}\}^{-1}B_t(h) \\
			T_{01}^{J}(y') & = (y' - \bar{u}_{t^*}) \Lambda_t.
		\end{aligned}
	\end{equation*}
	Similar analysis implies that 
	$$
	\E[-\nabla^2 Y_t|\Phi_{t,t + h,y,y'} = 0] := H_{t}^{t,t + h}(y,y') \approx H_{t|y} + T_{10}^{H}(h).
	$$ 
	As $T_{10}^{H}(h) = \wt{O}(\delta_n), T_{10}^{J}(h) = \wt{O}(\delta_n), T_{01}^{J}(y') = \wt{O}(\delta_n)$, their product is second-order. Taylor expansion of the function $\det(A  + E)$ around $E = 0$ finally yields our first-order expansion of the overall determinant term:
	\begin{equation}
		\label{eqn:pairs-intensity-approximate-jacobian-term}
		\det(H_{t|y}) \cdot \det(\bar{H}_t) \cdot \Big(1 + \tr\big(\{H_{t|y}\}^{-1}T_{10}^{H}(h)\big) + \tr\big(\{\bar{H}_{t}\}^{-1}T_{10}^J(h)\big) + \tr\big(\{\bar{H}_t\}^{-1}T_{01}^J(y')\big)\Big).
	\end{equation} 
	
	\subsubsection{Density term}
	The density term can be factorized as 
	\begin{equation}
		\label{eqn:pairs-intensity-density-term}
		f_{\Phi_{t,t',y,y'}}(0) = f_{Y_t}(y) \cdot f_{\nabla Y_t}(0) \cdot f_{Y_{t'}^{\sel},\nabla Y_{t'}^{\sel}|Y_t,\nabla Y_t}(y',0|y,0),
	\end{equation}
	where $f_{Y_{t'}^{\sel},\nabla Y_{t'}^{\sel}|Y_t,\nabla Y_t}$ denotes the conditional density of $(Y_{t'}^{\sel},\nabla Y_{t'}^{\sel})$ given $(Y_t,\nabla Y_t)$, and is given by
	\begin{equation*}
		f_{Y_{t'}^{\sel},\nabla Y_{t'}^{\sel}|Y_t,\nabla Y_t}(y',0|y,0) = \frac{1}{\sqrt{(2\pi)^{d + 1} \det(V_{t'|t,y})}} \cdot \exp\Big(-\frac{1}{2}(\mu_{t'|t,y} ~ \nabla \mu_{t'|t,y})'V_{t'|t,y}^{-1}(\mu_{t'|t,y} ~ \nabla \mu_{t'|t,y})\Big).
	\end{equation*}
	Above 
	\begin{align*}
		V_{t'|t,y} := \Var[(Y_{t'}^{\sel}, \nabla Y_{t'}^{\sel})|Y_t,\nabla Y_t],  \; \mu_{t'|t,y} := \E[Y_{t'}^{\sel}|Y_t = y,\nabla Y_t = 0], \;
		\nabla \mu_{t'|t,y} := \E[\nabla Y_{t'}^{\sel}|Y_t = y,\nabla Y_t = 0].
	\end{align*}
	We now give expansions of each of these terms, and then plug back into~\eqref{eqn:pairs-intensity-density-term} to give an overall expansion of the density term.
	
	\paragraph{Conditional expectation.}
	The conditional expectation of $Y_{t + h}^{\sel}|Y_t = y,\nabla Y_t = 0$ is 
	\begin{equation*}
		\mu_{t + h|t,y} = \mu_{t + h} + K(t + h,t) (y - \mu_t) + K_{01}(t + h,t)\Lambda_{t}^{-1}(0 - \nabla\mu_t).
	\end{equation*}
	Taking a second-order Taylor expansion of $\mu_{t + h|t,y}$ around $h = 0$ and disregarding terms that are $\wt{O}(\delta_n^2)$ gives
	\begin{align*}
		\mu_{t + h|t,y} \approx \mu_{t} + \frac{1}{2}h'\nabla^2\mu_t h + (y - \mu_t) - \frac{1}{2}h'\Lambda_t h (y -\mu_t) - h'\nabla\mu_t - K_{21}(t)(h,h)'\Lambda_{t}^{-1}\nabla\mu_t =  y - \frac{1}{2}h'H_{t|y}h.
	\end{align*}
	Similarly, the conditional expectation of $\nabla {Y}_{t + h}^{\sel}|\{Y_t = y,\nabla Y_t = 0\}$ is 
	\begin{equation*}
		\nabla \mu_{t + h|t,y} = \nabla\mu_{t + h} + K_{10}(t + h,t) (y - \mu_t) + K_{11}(t + h,t)\Lambda_{t}^{-1}(0 - \nabla\mu_t).
	\end{equation*}
	Taking a first-order Taylor expansion of $\nabla \mu_{t + h|t,y}$ around $h = 0$ and disregarding terms that are $\wt{O}(\delta_n^2)$ gives
	\begin{align*}
		 \nabla \mu_{t + h|t,y} \approx -H_{t|y}h + \frac{1}{2}\nabla^3\mu_{t}(h,h) + \frac{1}{2}(y - \mu_t)K_{30}(t)(h,h).
	\end{align*}
	
	\paragraph{Conditional variance.}
	The conditional variance is
	\begin{equation*}
		\begin{aligned}
		& V_{t + h|t,y} \\
		& = 
		\begin{bmatrix}
			1 + \gamma - \{K(t + h,t)\}^2 - K_{01}(t + h,t)\Lambda_t^{-1}K_{10}(t,t + h) & - K(t + h,t)K_{01}(t,t + h)' - K_{01}(t + h,t)'\Lambda_{t}^{-1}K_{11}(t,t + h) \\
			- K_{11}(t + h,t)\Lambda_{t}^{-1}K_{01}(t,t + h) - K(t + h,t)K_{01}(t,t + h) & (1 + \gamma)\Lambda_{t} - K_{11}(t + h,t)\Lambda_{t}^{-1}K_{11}(t,t + h) 
		\end{bmatrix}.
		\end{aligned}
	\end{equation*}
	Taking a first-order Taylor expansion of $V_{t + h|t,y}$ around $h = 0$ and disregarding terms that are $\wt{O}(\delta_n^2)$ yields
	\begin{equation*}
		V_{t+h|t,y} \approx 
		\begin{bmatrix}
			\gamma & 0' \\
			0 & \gamma \Lambda_t - \dot{\Lambda}_t(h)
		\end{bmatrix}.
	\end{equation*}
	\paragraph{Local expansion of density term.}
	Plugging in our approximations to the conditional mean and variance, applying Taylor expansion of the matrix inverse and determinant functions, and further disregarding all terms that are $\wt{O}(\delta_n^2)$ yields our local expansion of the density term:
	\begin{equation}
		\label{eqn:pairs-intensity-approximate-density-term}
		\frac{1 + T_{10}^{p_{|t,y}}(h) + T_{30}^{p_{|t,y}}(h) + T_{21}^{p_{t|y}}(h,y') + T_{11}^{p_{t|y}}(h,y')}{\sqrt{(2\pi)^{d + 1}\gamma^{d + 1}\det(\Lambda_t)}} \exp\Big(-\frac{1}{2\gamma}(y' - y)^2 - \frac{1}{2\gamma}h'\bar{H}_t\Lambda_t^{-1}H_{t|y}h\Big),
	\end{equation}
	where 
	\begin{align*}
		T_{10}^{p_{|t,y}}(h) & = \frac{1}{2\gamma}\tr\big(\Lambda_t^{-1}\{\dot{\Lambda}_t(h)\}\big) - \frac{1}{2\gamma} \nabla \mu_t'h \\
		T_{30}^{p_{|t,y}}(h) & =  \frac{1}{\gamma}h'H_{t|y}\Lambda_{t}^{-1}\Big(\frac{1}{2}\nabla^3\mu_{t}(h,h) + \frac{1}{2}(y - \mu_t)K_{30}(t)(h,h)\Big) - \frac{1}{2\gamma}h'H_{t|y} \Lambda_t^{-1}\{\dot\Lambda_t(h)\}H_{t|y}h \\
		T_{21}^{p_{|t,y}}(h,y') & = \frac{-(y' - \bar{u}_{t^*})}{2\gamma}h'H_{t|y}h \\
		T_{11}^{p_{|t,y}}(h,y') & = -\frac{1}{2}(y' - \bar{u}_{t^*}) \nabla \mu_t'h.
	\end{align*} 
	
	\paragraph{Local expansion of pairs intensity.}
	Combining~\eqref{eqn:pairs-intensity-approximate-jacobian-term} and~\eqref{eqn:pairs-intensity-approximate-density-term} gives our local expansion of the pairs intensity around $h = 0,y' = \bar{u}_{t^*}$:
	\begin{equation*}
		\begin{aligned}
			& \bar{\rho}^{\pair}(t,y,t + h,y') \\
			& := \det(H_{t|y}) \cdot \det(\bar{H}_{t}) \cdot \frac{(1 + T_{10}^{\pair}(h) + T_{01}^{\pair}(y') + T_{30}^{\pair}(h) + T_{21}^{\pair}(h,y') + T_{11}^{\pair}(h,y')}{\sqrt{(2\pi)^{d + 1}\gamma^{d + 1}\det(\Lambda_t)}}  \\
			& \quad \times \exp\Big(-\frac{1}{2\gamma}(y' - y)^2 - \frac{1}{2\gamma}h'\bar{H}_t\Lambda_t^{-1}H_{t|y}h\Big) \cdot f_{Y_t,\nabla Y_t}(y,0),
		\end{aligned}
	\end{equation*}
	where
	\begin{equation*}
		\begin{aligned}
			T_{10}^{\pair}(h) 
			& = \tr\big(H_{t|y}^{-1}T_{1}^{H}(h)\big) + \tr\big(\bar{H}_{t}^{-1}T_{1}^{J}(h)\big) + T_{1}^{p_{|t,y}}(h) \\
			T_{30}^{\pair}(h) & = T_{3}^{p_{|t,y}}(h) \\
			T_{01}^{\pair}(y') & = (y' - \bar{u}_{t^*}) \cdot \tr({\bar{H}_{t}}^{-1}\Lambda_t) \\
			T_{21}^{\pair}(h,y') & = T_{21}^{p_{|t,y}}(h,y') = \frac{-(y' - \bar{u}_{t^*})}{2\gamma}h'H_{t|y}h \\
			T_{11}^{\pair}(h,y') & = T_{11}^{p_{|t,y}}(h,y').
		\end{aligned}
	\end{equation*}
	
	\subsection{Intensity of location and height}
	\label{subsec:joint-intensity-carve}
	We now marginalize over $(h,y')$ in $\bar{\rho}^{\pair}(t,y,t + h, y')$ to compute an approximation $\bar{\rho}^{\carve}(t,y)$ to the marginal intensity $\rho^{\carve}(t,y)$. Integrating over $h \in \Rd$ and $y' \in (u,\infty)$, the approximate marginal intensity we obtain is 
	\begin{equation*}
		\begin{aligned}
			& \int_{u}^{\infty} \int_{\Rd} \bar{\rho}^{\pair}(t,y,t + h,y') \,dh \,dy' \\
			& = \sqrt{\frac{\det(H_{t|y}) \cdot \det(\bar{H}_{t})}{2\pi \gamma}} \cdot f_{Y_t,\nabla Y_t}(y,0) \cdot \int_{u}^{\infty}  \Big(1 + \frac{1}{2}(y'  - \bar{u}_{t^*}) \cdot \tr(\bar{H}_{t}^{-1}\Lambda_t)\Big) \cdot \exp\Big(-\frac{1}{2\gamma}(y' - y)^2\Big) \,dy' \\
			& \approx \sqrt{\frac{\det(H_{t|y}) \cdot \det(\bar{H}_{t})}{2\pi \gamma}} \cdot f_{Y_t,\nabla Y_t}(y,0) \cdot \exp\Big(\frac{1}{2}(y - \bar{u}_{t^*})\tr(\bar{H}_t^{-1}\Lambda_t)\Big) \cdot \int_{u}^{\infty} \sqrt{\frac{1}{2\pi\gamma}} \exp\Big(-\frac{1}{2\gamma}(y' - y - \frac{\gamma}{2}\tr(\bar{H}_{t}^{-1}\Lambda_t))^2\Big) \,dy' \\
			& = \sqrt{\det(H_{t|y}) \cdot \det(\bar{H}_{t})} \cdot f_{Y_t,\nabla Y_t}(y,0) \cdot \Psi\Big(\frac{u - y - \frac{\gamma}{2}\tr(\bar{H}_{t}^{-1}\Lambda_t)}{\sqrt{\gamma}}\Big) \cdot \exp\Big(\frac{1}{2}(y - \bar{u}_{t^*})\tr(\bar{H}_t^{-1}\Lambda_t)\Big) \\
			& \approx \sqrt{\det(H_{t|y}) \cdot \det(\bar{H}_{t})} \cdot f_{Y_t,\nabla Y_t}(y,0) \cdot \Psi\Big(\frac{u - y - \frac{\gamma}{2}\tr(\bar{H}_{t^*}^{-1}\Lambda_{t^*})}{\sqrt{\gamma}}\Big) \cdot \exp\Big(\frac{1}{2}(y - \bar{u}_{t^*})\tr(\bar{H}_{t^*}^{-1}\Lambda_{t^*})\Big).
		\end{aligned}
	\end{equation*}
	Taylor expansion of $\det(\bar{H}_{t})$ about $t = t^*$ and $\det(H_{t|y})$ about $t = t^*, y = \bar{u}_{t^*}^{\gamma}$ gives
	\begin{equation*}
		\begin{aligned}
			\det(H_{t^* + h|y}) & \approx \det(\bar{H}_{t^*}^{\gamma})\Big(1 + \tr\big(\{\bar{H}_{t^*}^{\gamma}\}^{-1}\dot{H}_{t^*|\bar{u}_{t^*}^{\gamma}}(h)\big) +  (y - \bar{u}_{t^*}^{\gamma}) \tr\big(\{\bar{H}_{t^*}^{\gamma}\}^{-1}\Lambda_{t^*}\big)\Big) \\
			\det(\bar{H}_{t^* + h}) & \approx \det(\bar{H}_{t^*})\Big(1 + T_{10}^{\det}(h)\Big).
		\end{aligned}
	\end{equation*}
	where $\bar{H}_{t^*}^{\gamma} = H_{t^*|\bar{u}_{t^*}^{\gamma}}$. Calculations similar to those used in the proof of Lemma~\ref{lem:asymptotic-expansion-density} in  Section~\ref{subsec:asymptotic-analysis-density} imply
	\begin{equation*}
		\begin{aligned}
			f_{Y_t,\nabla Y_t}(0,y) 
			& \approx \Big(1 - \frac{T_{30}^{\nabla Y}(h)}{2} + T_{10}^{\nabla Y}(h) + \frac{1}{2}(y - \bar{u}_{t^*}^{\gamma})h'\nabla^2\mu_{t^*}h\Big) \exp\Big(-\frac{1}{2}(y - \mu_{t^*})^2\Big) \\
			& \times \frac{1}{\sqrt{(2\pi)^{d + 1} \det(\Lambda_{t^*})}}\exp\Big(-\frac{1}{2}h'\bar{H}_{t^*}^{\gamma}\Lambda_{t^*}^{-1}(-\nabla^2\mu_{t^*})h\Big).
		\end{aligned}
	\end{equation*}
	Finally, the exponential term in $y$ terms can be written as
	\begin{equation}
	\begin{aligned}
		\exp\Big(-\frac{1}{2}(y - \mu_{t^*})^2 + \frac{1}{2}(y - \bar{u}_{t^*})\tr(\bar{H}_{t^*}^{-1}\Lambda_{t^*})\Big) 
		& \approx \exp\Big(-\frac{1}{2}\big(y - {\mu}_{t^*} - \frac{1}{2}\tr(\bar{H}_{t^*}^{-1}\Lambda_{t^*})\big)^2 + \frac{1}{2}(\mu_{t^*} - \bar{u}_{t^*})\tr(\bar{H}_{t^*}^{-1}\Lambda_{t^*})\Big),
	\end{aligned}
	\end{equation}
	up to relative error ${O}(\delta_n^2)$. Altogether this yields the following first-order expansion of the joint intensity:
	\begin{equation}
		\label{eqn:approximate-intensity-carve}
		\begin{aligned}
			\bar{\rho}^{\carve}(t^* + h,y) & := \sqrt{\frac{\det(\bar{H}_{t^*}) \cdot \det(\bar{H}_{t^*}^{\gamma})}{(2\pi)^{d + 1} \det(\Lambda_{t^*})}}  \cdot \exp\Big(\frac{1}{2}(\mu_{t^*} - \bar{u}_{t^*})\tr(\bar{H}_{t^*}^{-1}\Lambda_{t^*})\Big) \\
			& \; \times  \Big(1 + T_{10}^{\pair}(h) + T_{01}^{\pair}(y) + T_{30}^{\pair}(h) + T_{21}^{\pair}(h,y) \Big) \cdot \Psi\Big(\frac{u - y - \frac{\gamma}{2}\tr(\bar{H}_{t^*}^{-1}\Lambda_{t^*})}{\sqrt{\gamma}}\Big) \\ 
			& \; \times  \exp\Big(-\frac{1}{2}\big(y - \mu_{t^*} - \frac{1}{2}\tr(\bar{H}_{t^*}^{-1}\Lambda_{t^*})\big)^2\Big) \cdot \exp\Big(-\frac{1}{2}h'\bar{H}_{t^*}^{\gamma}\Lambda_{t^*}^{-1}(-\nabla^2\mu_{t^*})h\Big),
		\end{aligned}
	\end{equation}
	where
	\begin{equation*}
		\begin{aligned}
			T_{10}^{\carve}(h) 
			& := \frac{T_{10}^{\det}(h)}{2} + \frac{1}{2}\tr\Big(\{\bar{H}_{t^*}^{\gamma}\}^{-1}\dot{H}_{t^*|\bar{u}_{t^*}^{\gamma}}(h)\Big) + T_{10}^{\nabla Y}(h)\\
			T_{01}^{\carve}(y) & := \frac{1}{2}(y - \bar{u}_{t^*}^{\gamma})\tr(\{\bar{H}_{t^*}^{\gamma}\}^{-1}\Lambda_{t^*})\\
			T_{30}^{\carve}(h) & := -\frac{T_{30}^{\nabla Y}(h)}{2} \\
			T_{21}^{\carve}(h,y) & := \frac{1}{2}(y - \bar{u}_{t^*}^{\gamma}) h'\nabla^2 \mu_{t^*}h.
		\end{aligned}
	\end{equation*}
	
	\subsection{Density of height}
	\label{subsec:density-height-carve}
	We obtain an approximation $\bar{\rho}^{\carve}(y)$ to $\rho^{\carve}(y)$ by marginalizing over $h$ in~\eqref{eqn:approximate-intensity-carve}. In marginalizing over $h$ all of the first order terms cancel: 
	$$\int_{\Rd} T_{10}^{\carve}(h) \cdot \phi_{0,\bar{H}_{t^*}^{\gamma}\Lambda_{t^*}^{-1}(-\nabla^2\mu_{t^*})}(h)\,dh = \int_{\Rd} T_{30}^{\carve}(h) \cdot \phi_{0,\bar{H}_{t^*}^{\gamma}\Lambda_{t^*}^{-1}(-\nabla^2\mu_{t^*})}(h) \,dh  = 0$$ 
	by symmetry, while 
	$$\int_{\Rd} T_{21}^{\carve}(h,y) \phi_{0,\bar{H}_{t^*}^{\gamma}\Lambda_{t^*}^{-1}(-\nabla^2\mu_{t^*})}(h) \,dh = - T_{01}^{\carve}(y).$$ 
	Thus marginalizing over $h$ in~\eqref{eqn:approximate-intensity-carve} yields
	\begin{equation*}
		\begin{aligned}
			 & \sqrt{\frac{\det(\bar{H}_{t^*})}{\det(-\nabla^2 \mu_{t^*})}} \cdot \frac{\exp\Big(\frac{1}{2}(\mu_{t^*} - \bar{u}_{t^*}) \tr(\bar{H}_{t^*}^{-1}\Lambda_{t^*})\Big)}{\sqrt{2\pi}} \cdot \Psi\Big(\frac{u - y - \frac{\gamma}{2}\tr(\bar{H}_{t^*}^{-1}\Lambda_{t^*})}{\gamma}\Big) \cdot \exp\Big(-\frac{1}{2}(y - \mu_{t^*} - \frac{1}{2}\tr(\bar{H}_{t^*}^{-1}\Lambda_{t^*}))^2\Big) \\
			& := \bar{\rho}^{\carve}(y).
		\end{aligned}
	\end{equation*}
	Our approximation to $p^{\carve}(y)$ is in turn
	\begin{equation*}
		\bar{p}^{\carve}(y) :\propto \bar{\rho}^{\carve}(y) \propto \Psi\Big(\frac{u - y - \frac{\gamma}{2}\tr(\bar{H}_{t^*}^{-1}\Lambda_{t^*})}{\gamma}\Big) \cdot \frac{1}{\sqrt{2\pi}} \cdot \exp\Big(-\frac{1}{2}(y - \mu_{t^*} - \frac{1}{2}\tr(\bar{H}_{t^*}^{-1}\Lambda_{t^*}))^2\Big),
	\end{equation*}
	as given in~\eqref{eqn:density-height-carve}. 
	
	\subsection{Density of location given height}
	\label{subsec:density-location-carve}
	We have given one local expansion of the joint intensity of height and location in~\eqref{eqn:approximate-intensity-carve}. Pushing the $T_{21}^{\carve}(h,y)$ term into the exponent incurs $\wt{O}(\delta_n^2)$ error, and yields the alternate local expansion:
	\begin{equation*}
		\begin{aligned}
			\tilde{\rho}^{\carve}(t^* + h,y) & =  \sqrt{\frac{\det(\bar{H}_{t^*}) \cdot \det(\bar{H}_{t^*}^{\gamma})}{(2\pi)^{d + 1} \det(\Lambda_{t^*})}}  \cdot \exp\Big(\frac{1}{2}(\mu_{t^*} - \bar{u}_{t^*})\tr(\bar{H}_{t^*}^{-1}\Lambda_{t^*})\Big) \\
			& \; \times  \Big(1 + T_{10}^{\carve}(h) + T_{01}^{\carve}(y) + T_{30}^{\carve}(h)\Big) \cdot \Psi\Big(\frac{u - y - \frac{\gamma}{2}\tr(\bar{H}_{t^*}^{-1}\Lambda_{t^*})}{\sqrt{\gamma}}\Big) \\ 
			& \; \times  \exp\Big(-\frac{1}{2}\big(y - \mu_{t^*} - \frac{1}{2}\tr(\bar{H}_{t^*}^{-1}\Lambda_{t^*})\big)^2\Big) \cdot \exp\Big(-\frac{1}{2}h'H_{t^*|y}\Lambda_{t^*}^{-1}(-\nabla^2\mu_{t^*})h\Big).
		\end{aligned}
	\end{equation*}
	Taking this and dividing by $\bar{\rho}^{\carve}(y)$, we are left with
	\begin{equation*}
		\begin{aligned}
			& \Big(1 + T_{10}^{\carve}(h) + T_{01}^{\carve}(y) + T_{30}^{\carve}(h)\Big)\cdot \sqrt{\frac{\det(\bar{H}_{t^*}^{\gamma}) \cdot \det(-\nabla^2 \mu_{t^*})}{\det(\Lambda_{t^*}) \cdot (2\pi)^{d}}} \cdot \exp\Big(-\frac{1}{2}h'H_{t^*|y}\Lambda_{t^*}^{-1}(-\nabla^2 \mu_{t^*})h\Big) \\
			& \approx \Big(1 + T_{10}^{\carve}(h) + T_{30}^{\carve}(h)\Big) \cdot \sqrt{\frac{\det(H_{t^*|y}) \cdot \det(-\nabla^2 \mu_{t^*})}{\det(\Lambda_{t^*}) \cdot (2\pi)^d}} \cdot \exp\Big(-\frac{1}{2}h'H_{t^*|y}\Lambda_{t^*}^{-1}(-\nabla^2 \mu_{t^*})h\Big) \\
			& = \Big(1 + T_{10}^{\carve}(h) + T_{30}^{\carve}(h)\Big) \cdot \sqrt{\frac{\det(G_{t^*|y})}{(2\pi)^d}} \cdot \exp\Big(-\frac{1}{2}h'G_{t^*|y}h\Big),
		\end{aligned}
	\end{equation*}
	which is the expression for $\bar{p}^{\carve}(t^* + h|y)$ given in~\eqref{eqn:density-location-carve}.

	\section{Technical results}
	
	\subsection{Gaussian tail behavior: univariate}
	\label{subsec:gaussian-tail-univariate}
	For any $a > 0$,
	\begin{equation}
		\label{eqn:mills-ratio}
		\frac{a}{a^2 + 1} \phi(a)\leq \Psi(a) \leq \frac{\phi(a)}{a}.
	\end{equation}
	Suppose $Z \sim N(0,1)$. Using~\eqref{eqn:mills-ratio} and integration by parts,
	\begin{equation}
		\label{eqn:normalized-overshoot-first-moment}
		\int_{a}^{\infty} (y - a) \phi(y) \,dy \leq \frac{\Psi(a)}{a}.
	\end{equation}
	Arguing recursively, it can be shown that for any $p \in \mathbb{N}, p \geq 1$, there exists a constant $C$ depending only on $p$ such that
	\begin{equation}
		\label{eqn:normalized-overshoot-moments}
		\int_{a}^{\infty} |y - a|^p \phi(y) \,dy \leq C \cdot \Psi(a). 
	\end{equation}
	Finally, the following lemma bounds the relative error in perturbations of the Gaussian survival function. 
	\begin{lemma}
		\label{lem:gaussian-survival-function-perturbation}
		For any sequence $u_n \in \R$ and $\varepsilon_n \to 0$,
		\begin{equation}
			\label{eqn:gaussian-survival-function-perturbation}
			\frac{|\Psi(u_n + \varepsilon_n) - \Psi(u_n) \exp(-u_n\varepsilon_n)|}{\Psi(u_n) \exp(-u_n\varepsilon_n)} = O(\varepsilon_n/u + \varepsilon_n^2).
		\end{equation}
		If additionally $u_n \varepsilon_n \to 0$, then
		\begin{equation}
			\label{eqn:gaussian-survival-function-perturbation-2}
			\frac{|\Psi(u_n + \varepsilon_n) - \Psi(u_n)|}{\Psi(u_n)} = O((u_n \vee 1)\varepsilon_n).
		\end{equation}
	\end{lemma}
	\begin{proof}
		The claims are obviously true if $u_n \not\to \infty$, and we thus focus on the case where $u_n \to \infty$. We write out 
		\begin{align*}
			\exp(\frac{1}{2}\varepsilon^2)\Psi(u + \varepsilon)\exp(u\varepsilon) 
			& = \exp(\frac{1}{2}\varepsilon^2)\int_{u}^{\infty} \phi(x + \varepsilon) \exp(u\varepsilon) \,dx \\
			& = \int_{u}^{\infty} \phi(x) \exp(-(x - u)\varepsilon) \,dx \\
			& = \Psi(u) + \int_{u}^{\infty} (\exp(-(x - u)\varepsilon) - 1) \phi(x) \,dx
		\end{align*}
		For all $x \in (u,\infty)$ we have that $0 \leq 1 - \exp(-(x - u)\varepsilon) \leq (x - u)\varepsilon$, and therefore
		\begin{equation*}
			\Big|\Psi(u) - \exp(\frac{1}{2}\varepsilon^2) \Psi(u + \varepsilon)\exp(u\varepsilon) \Big| \leq \varepsilon \int_{u}^{\infty} (x - u) \phi(x) \,dx \leq \frac{\varepsilon \Psi(u)}{u},
		\end{equation*}
		the latter inequality following from~\eqref{eqn:normalized-overshoot-first-moment}. Consequently, 
		\begin{equation*}
			\frac{\Big|\Psi(u) - \Psi(u + \varepsilon)\exp(u\varepsilon) \Big|}{\Psi(u)} \leq \frac{\varepsilon}{u} + \Big|1 - \exp(\frac{\varepsilon^2}{2})\Big| \frac{\Psi(u + \varepsilon) \exp(u \varepsilon)}{\Psi(u)} \leq \frac{\varepsilon}{u} + C \Big|1 - \exp(\frac{\varepsilon^2}{2})\Big|,
		\end{equation*}
		the latter inequality holding as $\frac{\Psi(u + \varepsilon) \exp(u \varepsilon)}{\Psi(u)} \leq C$ for all $\varepsilon$ sufficiently small. This proves~\eqref{eqn:gaussian-survival-function-perturbation}, with~\eqref{eqn:gaussian-survival-function-perturbation-2} following since $|1 - \exp(-u\varepsilon)| \leq u\varepsilon$. 
	\end{proof}
	
	\subsection{Gaussian tail behavior: multivariate}
	\label{subsec:gaussian-tail-multivariate}
	Consider a multivariate Gaussian $Z \sim N_d(0,G^{-1})$, and let $B(0,r)$ be a ball centered at $0$ of radius $r$. For any vector $x \in \Rd$ and array $A \in \R^{d \times d \times d}$,
	\begin{equation}
		\label{eqn:gaussian-integrals-truncated-moments}
		\E\Big[(\sum x_i Z_i) \cdot \1(Z \in B(0,r))\Big] = 0, \quad \E\Big[(\sum A_{ijk} Z_i Z_jZ_k) \cdot \1(Z \in B(0,r))\Big] = 0.
	\end{equation}
	Now consider $U \sim N_d(0,I)$. The following is a consequence of a standard Chernoff bound: for any $\eta > 0$, letting $r_n = \sqrt{(2 + \eta) \log n}$, for all $n$ large enough so that $r_n \geq d$:
	\begin{equation}
		\begin{aligned}
			\label{eqn:gaussian-integrals-concentration}
			\P\big(\|U\| \geq r_n\big) \leq C \{(2 + \eta) \log n\}^{d/2} n^{-(1 + \eta/2)}.
		\end{aligned}
	\end{equation}
	
	\subsection{Maximum of a Gaussian process}
	We record the \emph{Borell-TIS inequality}~\citep{adler2007random}, slightly restated for our purposes, to bound the asymptotic probability that the supremum of a Gaussian process exceeds some threshold $u_n$. Let $X \sim N(0,C)$ be a mean-zero Gaussian processes over $\mc{T}$, with covariance kernel $C$, and suppose $X$ is bounded a.s. over  $\mc{A} \subseteq \mc{T}$. Let $\sigma_{\mc{A}}^2 := \sup_{t \in \mc{A}}\Var[X_{t}]$ and $m_{\mc{A}} := \E[\sup_{t \in \mc{A}} X_t]$.  For any $\eta > 0$, 
	\begin{equation}
		\label{eqn:borell-tis}
		\P\Big(\sup_{t \in \mc{A}} X_{t} > u\Big) \leq \exp\Big\{\Big(\frac{1}{4\eta}-\frac{1}{2}\Big)\frac{m_{\mathcal{A}}^{2}}{\sigma_{\mathcal A}^{2}}\Big\} \cdot \exp\Big(\eta u^2 - \frac{u^2}{2\sigma_{\mc{A}}^2}\Big).
	\end{equation}
	The following consequence will be useful for our purposes: suppose $K(t,\cdot) \in C^1(\Rd)$ for all $t$, $\sup_{t \in \Rd} \|K(t,\cdot)\|_{C^1(\Rd)} \leq C_0 < \infty$ and $\sup_{t \in \Rd} \Var[X_t] \leq \sigma^2 < \infty$. Then there exists a constant $C$ depending only on $C_0$ such that $\sup_{t \in \Rd} \sup_{\mc{A} \subseteq B(t,1)} m_{\mc{A}} \leq C$. Consequently, there exists a constant $C(\eta)$ depending only on $\eta$ (and $C_0$) such that
	\begin{equation}
		\label{eqn:borell-tis-consequence}
		\sup_{t \in \Rd} \sup_{\mc{A} \subseteq B(t,1)} \P\Big(\sup_{t \in \mc{A}} X_{t} > u\Big) \leq  C(\eta) \exp\Big(\eta u^2 - \frac{u^2}{2\sigma^2}\Big).
	\end{equation}
	
	\subsection{Matrix perturbation}
	The following are standard results in matrix perturbation theory. Consider sequences of matrices $A_n, E_n$ such that $A_n$ is invertible for all $n$, and for which $\|E_n\|/\|A_n\| \to 0$. For all $n$ sufficiently large,
	\begin{equation}
		\label{eqn:matrix-inverse-taylor-expansion}
		\|(A_n + E_n)^{-1} - A_n^{-1}\| \leq 2 \|A_n^{-1}\|^2 \|E_n\|.
	\end{equation}
	A first-order Taylor expansion of $\det(A + E)$ about $E = 0$ yields the following: if $\|E_n\|/\|A_n\| \leq 1/2$, then \begin{equation}
		\label{eqn:matrix-det-taylor-expansion}
		|\det(A_n + E_n) - \det(A_n)| \leq 4 \det(A_n) \frac{\|E_n\|}{\lambda_{\min}(A_n)}, \quad |\det(A_n + E_n) - \det(A_n)(1 + \tr(A_n^{-1}E_n))| \leq 8 \det(A_n) \frac{\|E_n\|^2}{\lambda_{\min}(A_n)^2}.
	\end{equation}
	
	\section{Additional experiments}
	\label{sec:additional-experiments}
	
	\subsection{Validating asymptotic theory}
	
	\paragraph{Effect of $\alpha$.}
	To demonstrate that the conclusions of Section~\ref{subsec:experiment-1} are robust to the choice of $\alpha$, we plot the empirical distribution of all quantities, for $\mu_0 = 3$ (Figure~\ref{fig:experiment-1-ppplots-1}) and $\mu_0 = 11$ (Figure~\ref{fig:experiment-1-ppplots-2}).
	
	\begin{figure}[htbp]
		\centering
		\begin{subfigure}[t]{0.32\linewidth}
			\centering
			\includegraphics[width=\linewidth]{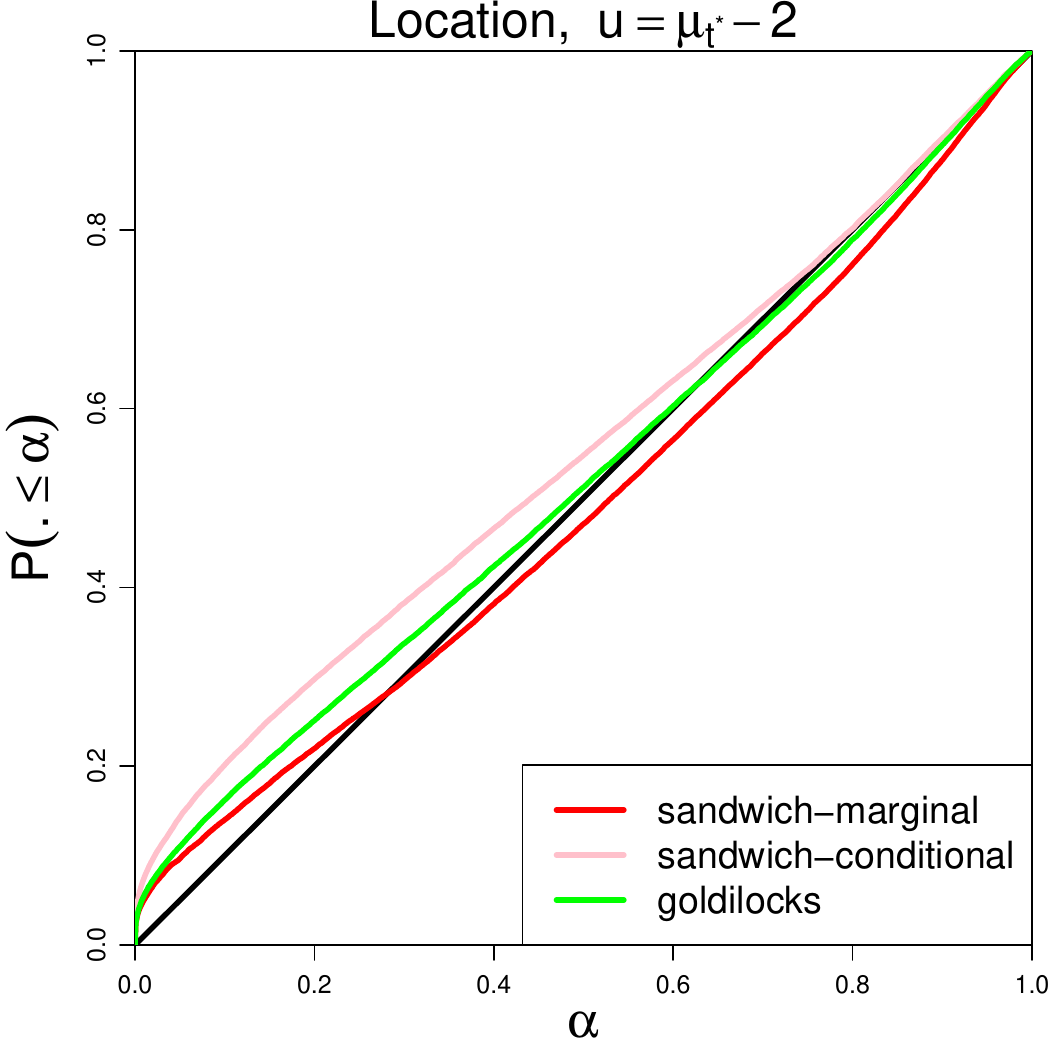}
			\label{fig:experiment-1-4}
		\end{subfigure}\hfill
		\begin{subfigure}[t]{0.32\linewidth}
			\centering
			\includegraphics[width=\linewidth]{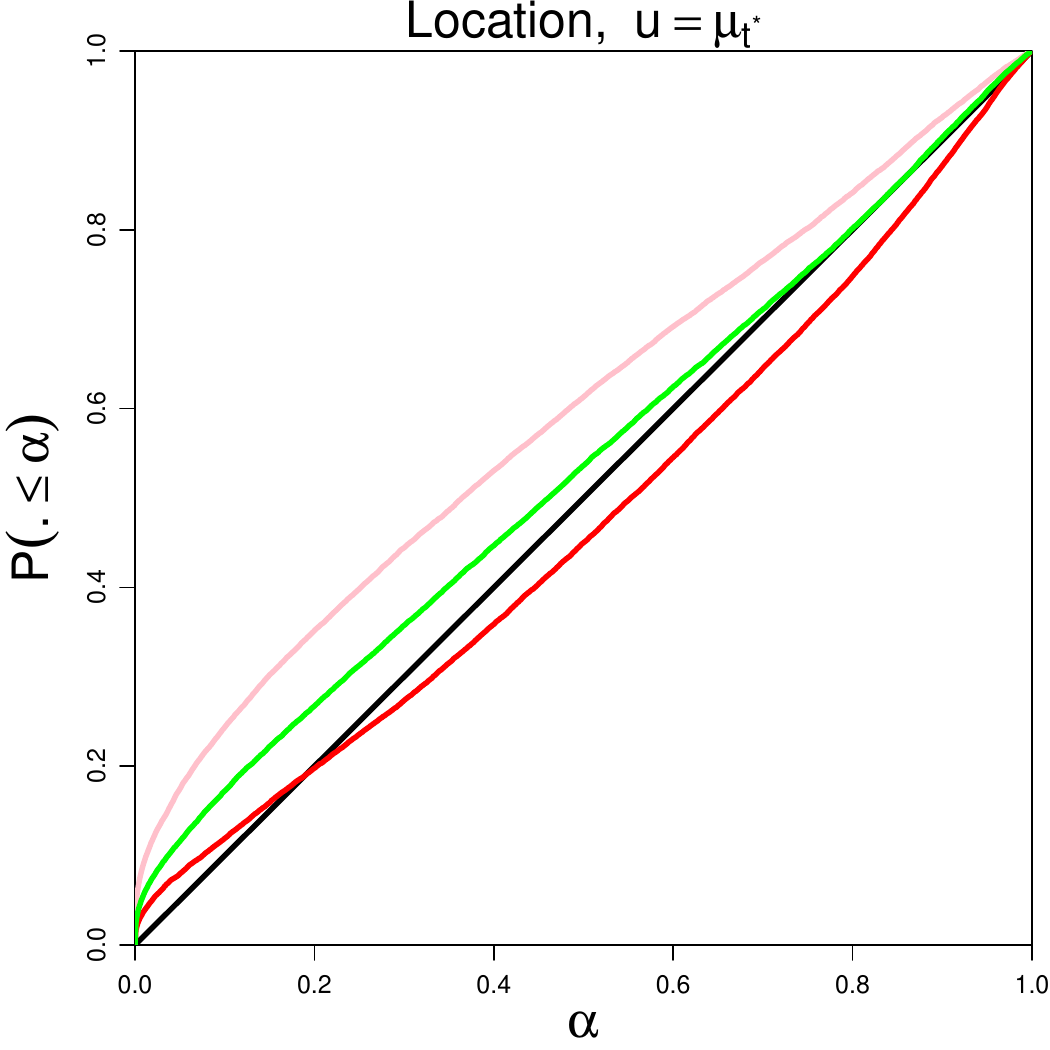}
			\label{fig:experiment-1-5}
		\end{subfigure}\hfill
		\begin{subfigure}[t]{0.32\linewidth}
			\centering
			\includegraphics[width=\linewidth]{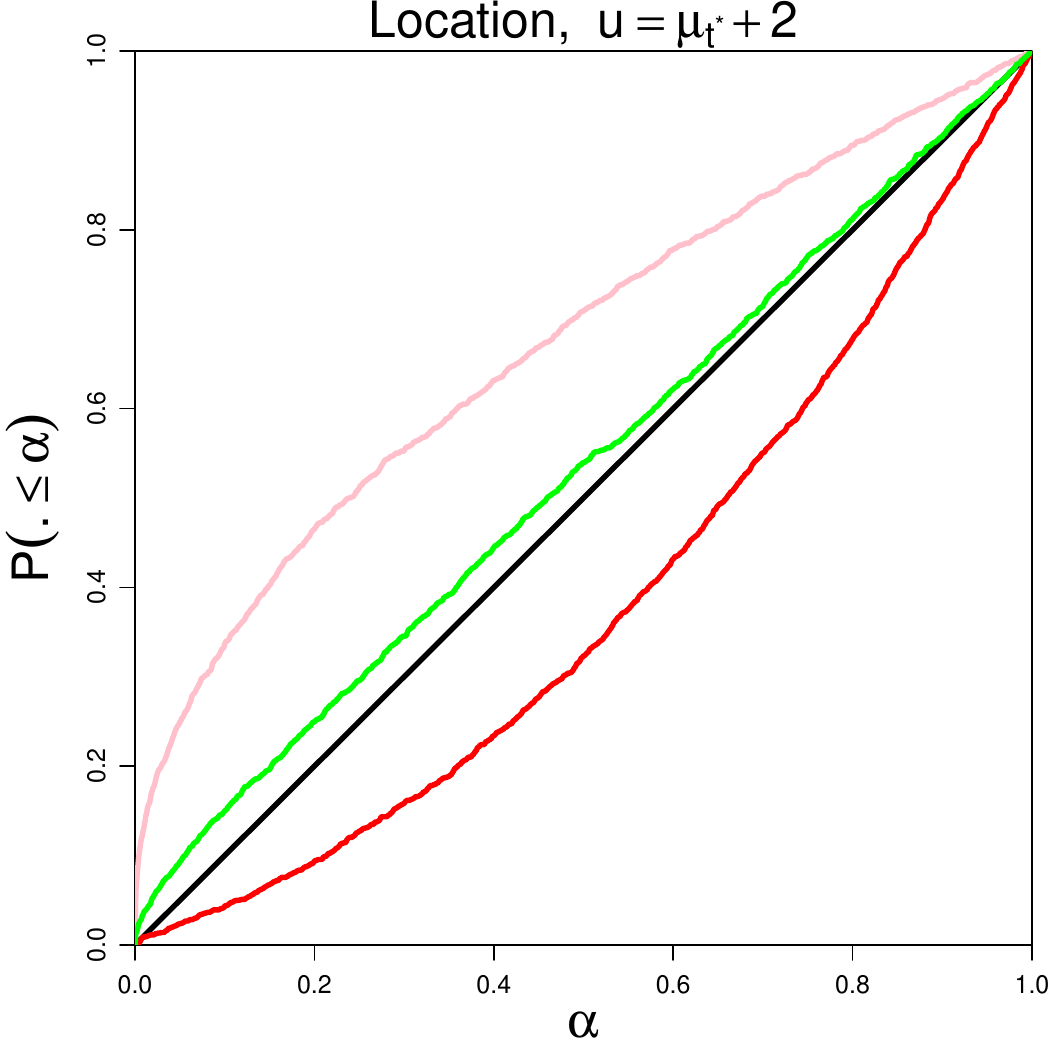}
			\label{fig:experiment-1-6}
		\end{subfigure}\hfill
		
		\begin{subfigure}[t]{0.32\linewidth}
			\centering
			\includegraphics[width=\linewidth]{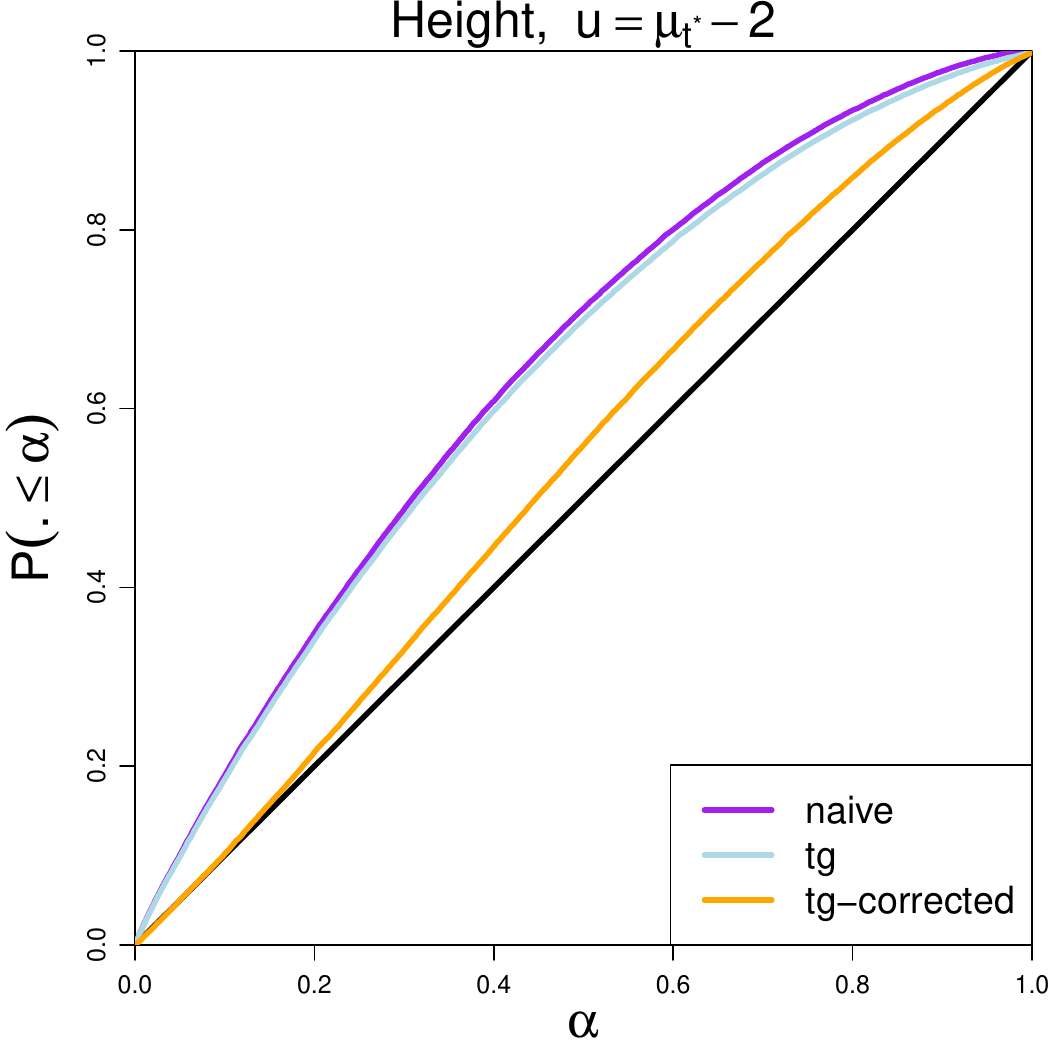}
			\label{fig:experiment-1-1}
		\end{subfigure}\hfill
		\begin{subfigure}[t]{0.32\linewidth}
			\centering
			\includegraphics[width=\linewidth]{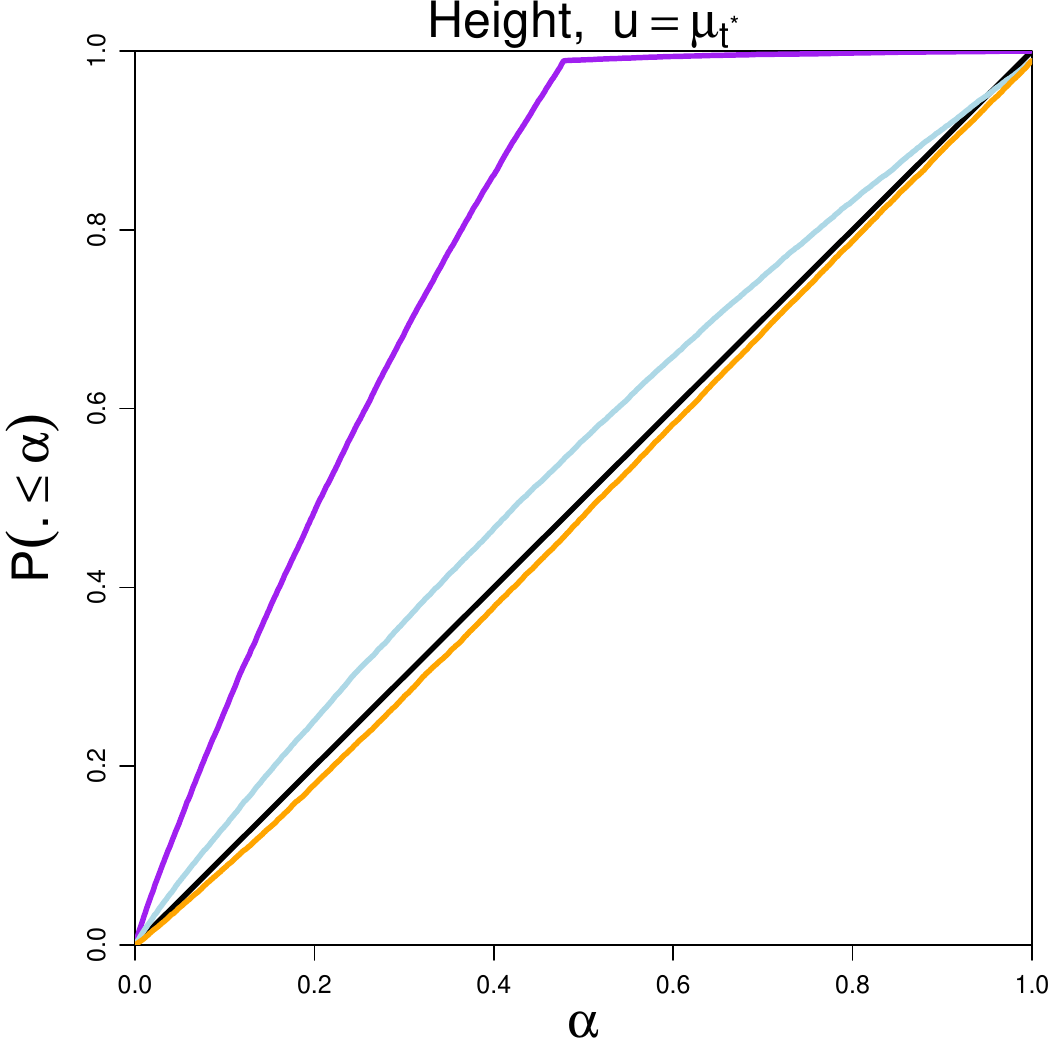}
			\label{fig:experiment-1-2}
		\end{subfigure}\hfill
		\begin{subfigure}[t]{0.32\linewidth}
			\centering
			\includegraphics[width=\linewidth]{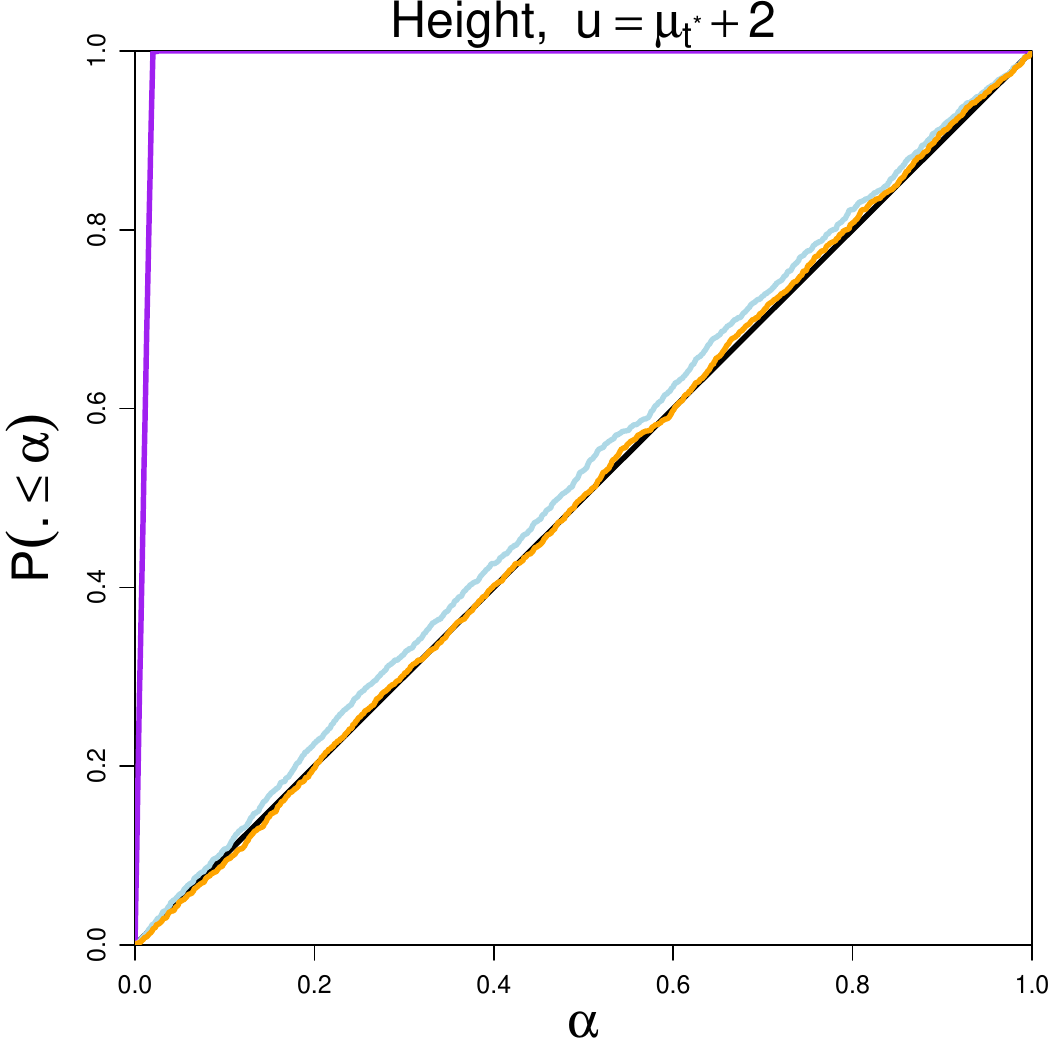}
			\label{fig:experiment-1-3}
		\end{subfigure}\hfill
		
		\caption{Distribution of candidate quantities for location (top row) and height (bottom row), when $\mu_0 = 3$. Different columns corresponds to different thresholds $u$. Details of experimental setup and takeaways are in the main text.}
		\label{fig:experiment-1-ppplots-1}
	\end{figure}
	
	\begin{figure}[htbp]
		\centering
		\begin{subfigure}[t]{0.32\linewidth}
			\centering
			\includegraphics[width=\linewidth]{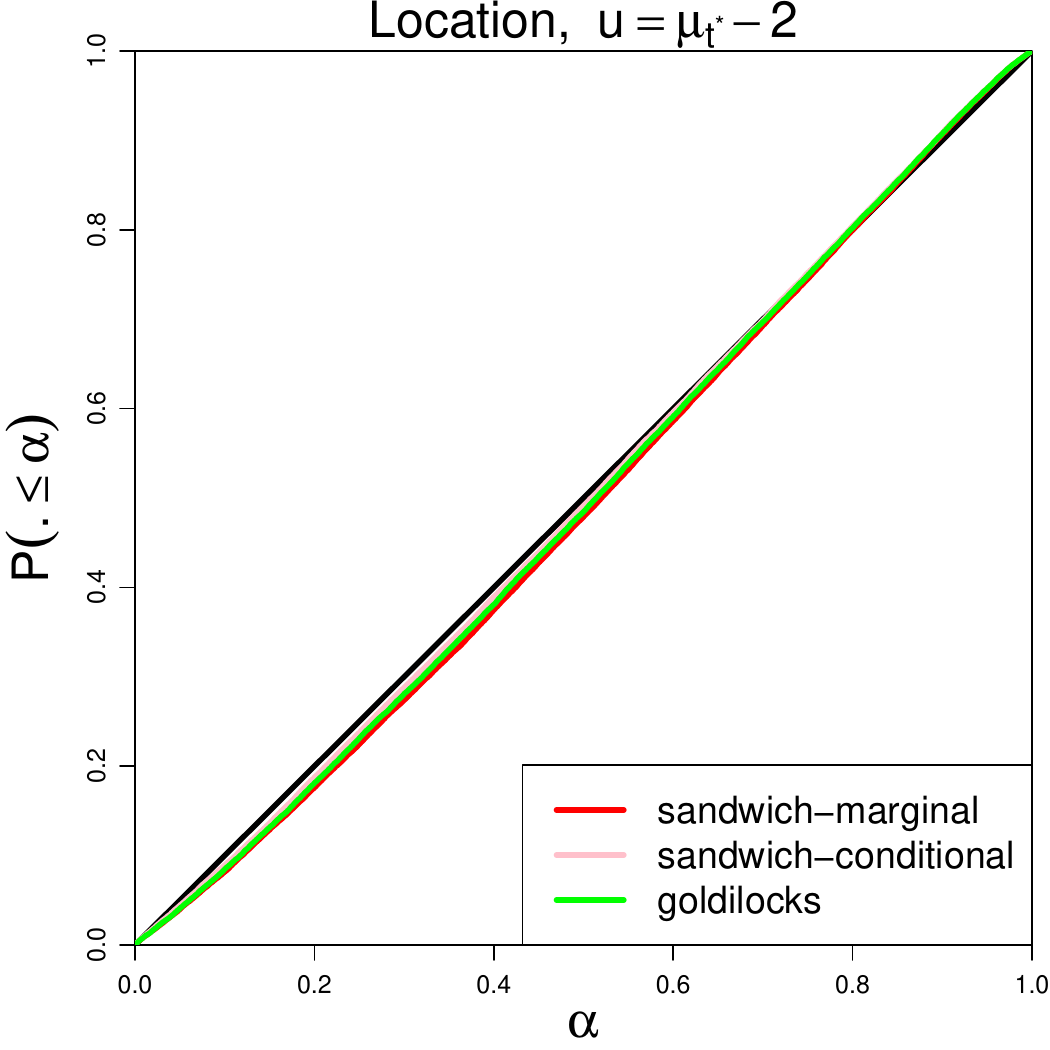}
			\label{fig:experiment-1-4}
		\end{subfigure}\hfill
		\begin{subfigure}[t]{0.32\linewidth}
			\centering
			\includegraphics[width=\linewidth]{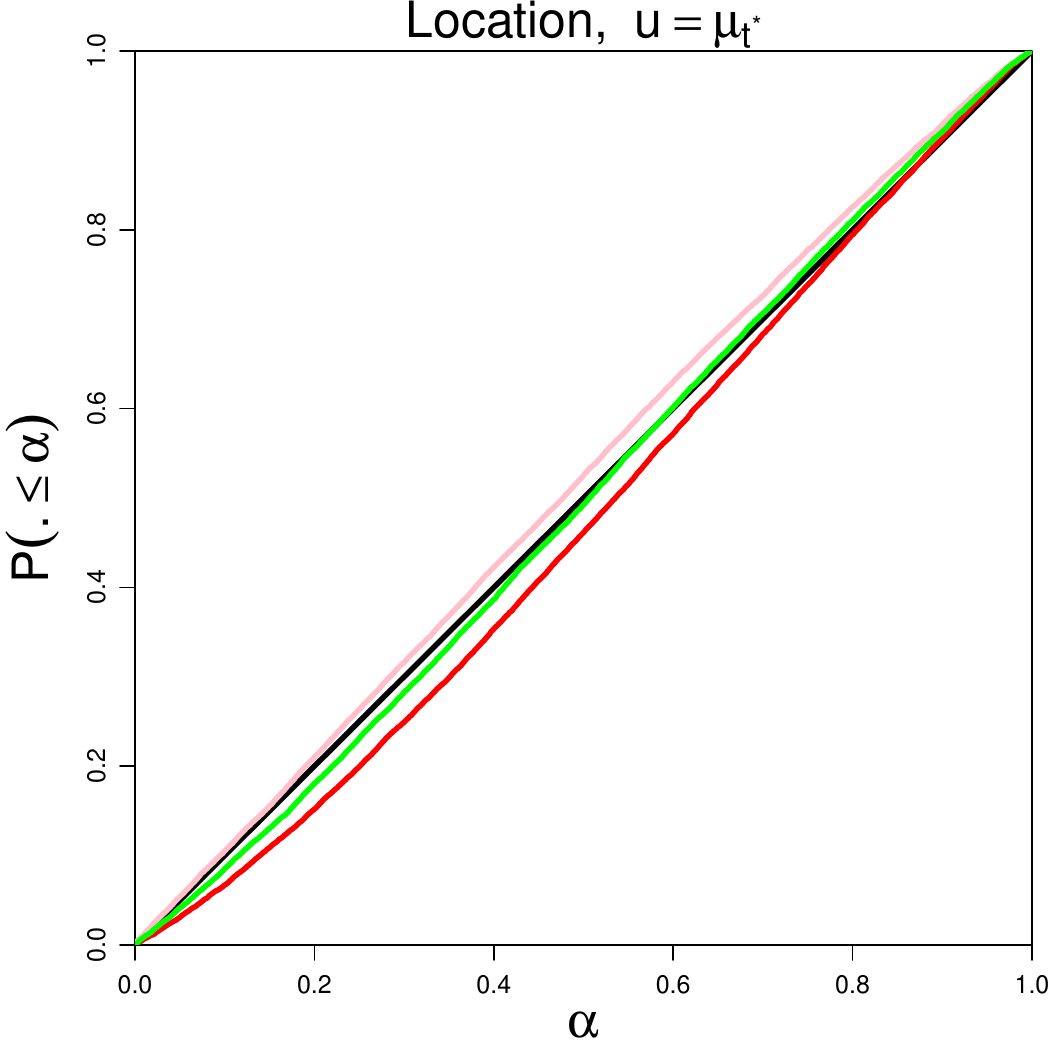}
			\label{fig:experiment-1-5}
		\end{subfigure}\hfill
		\begin{subfigure}[t]{0.32\linewidth}
			\centering
			\includegraphics[width=\linewidth]{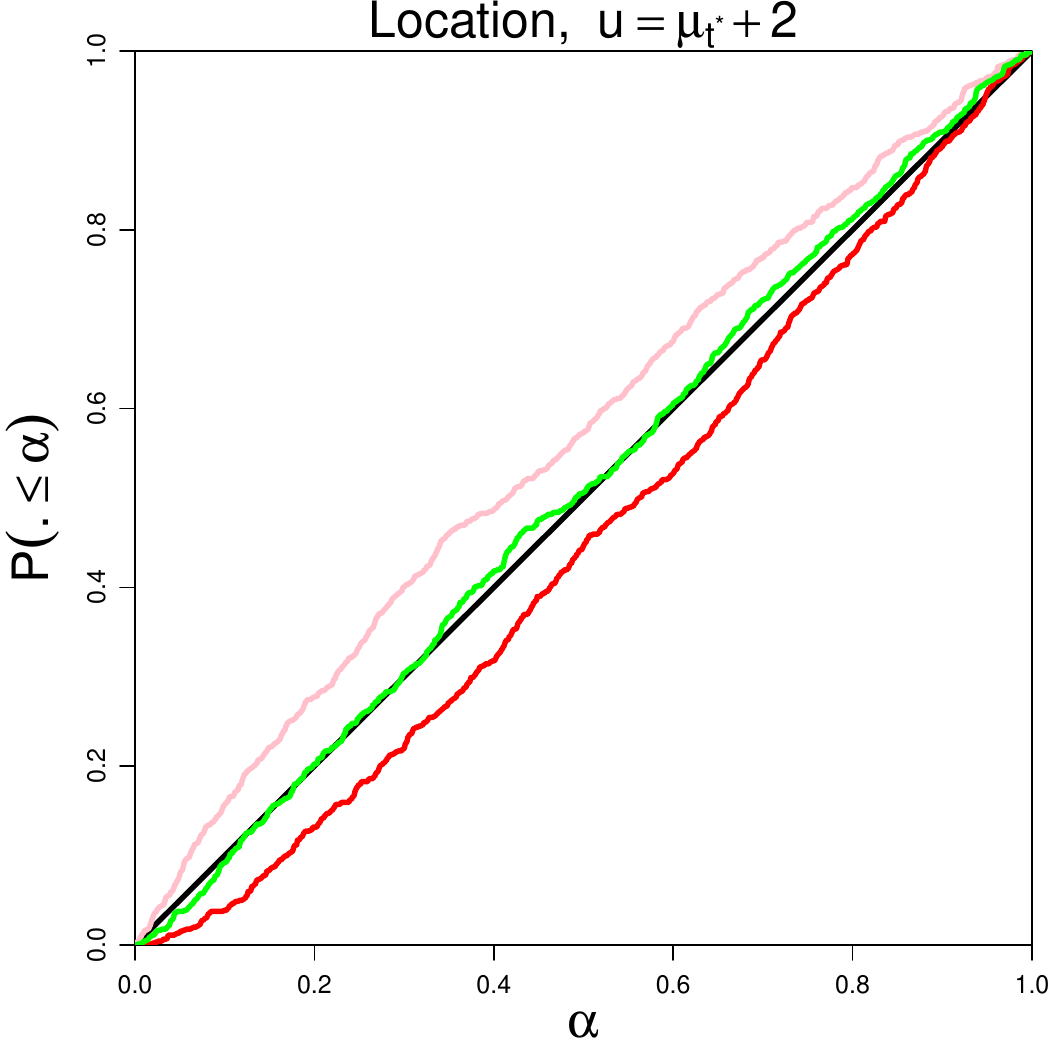}
			\label{fig:experiment-1-6}
		\end{subfigure}\hfill
		
		\begin{subfigure}[t]{0.32\linewidth}
			\centering
			\includegraphics[width=\linewidth]{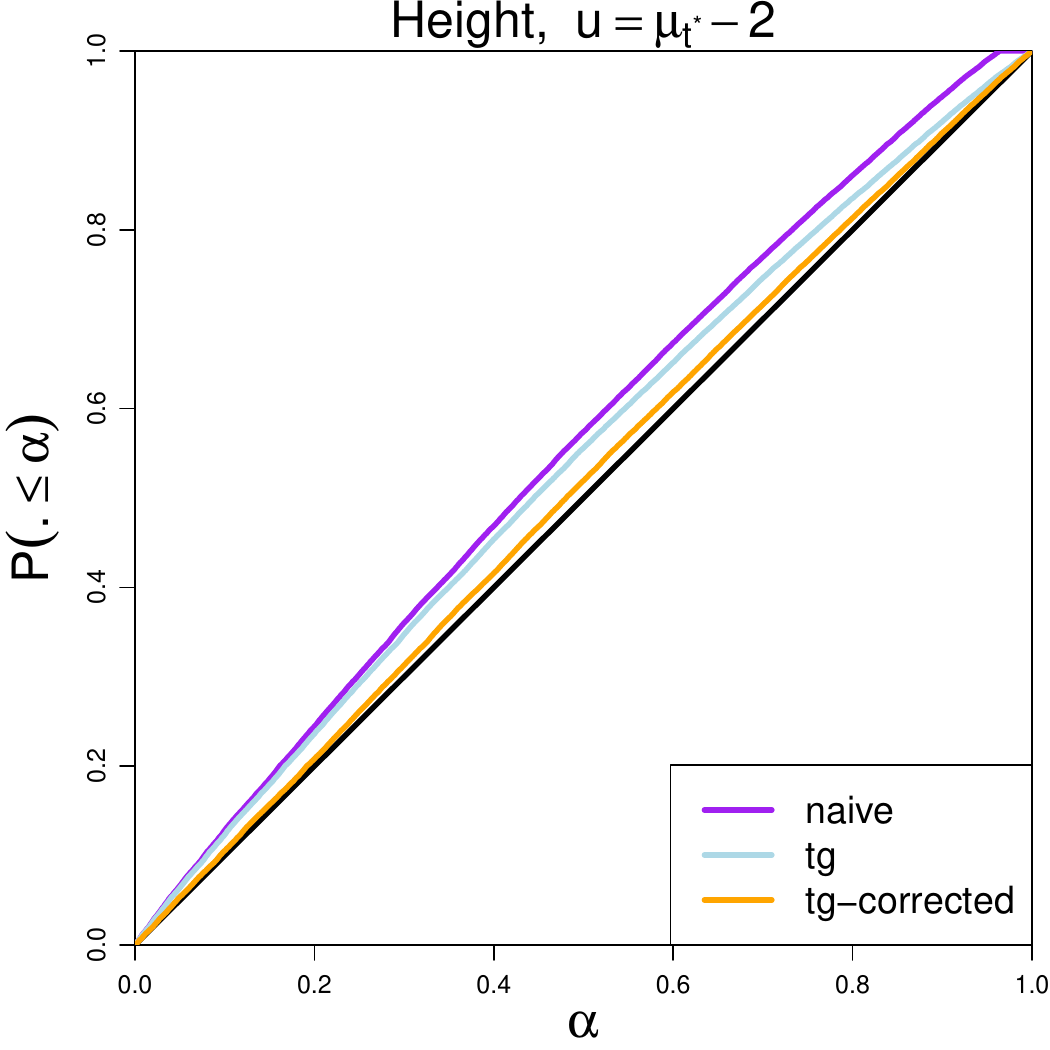}
			\label{fig:experiment-1-1}
		\end{subfigure}\hfill
		\begin{subfigure}[t]{0.32\linewidth}
			\centering
			\includegraphics[width=\linewidth]{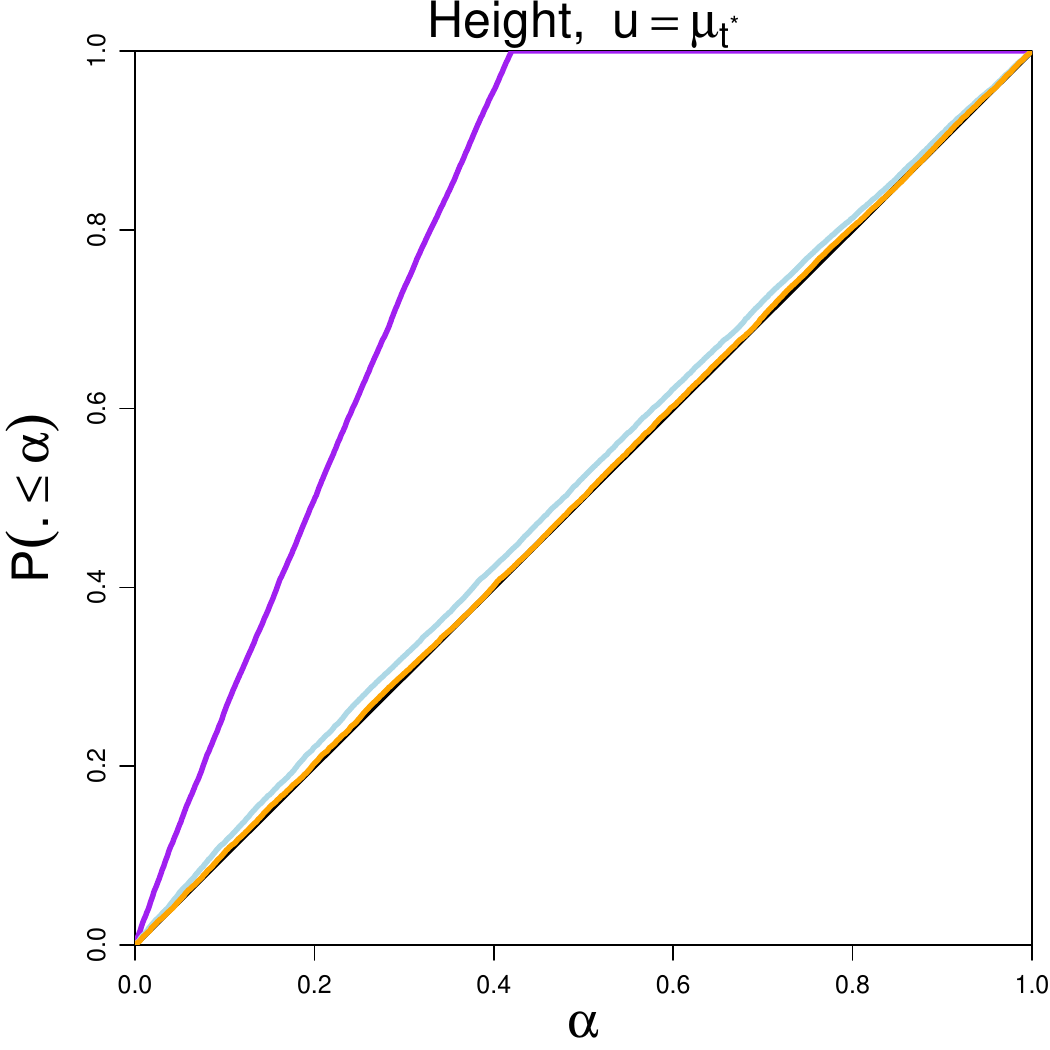}
			\label{fig:experiment-1-2}
		\end{subfigure}\hfill
		\begin{subfigure}[t]{0.32\linewidth}
			\centering
			\includegraphics[width=\linewidth]{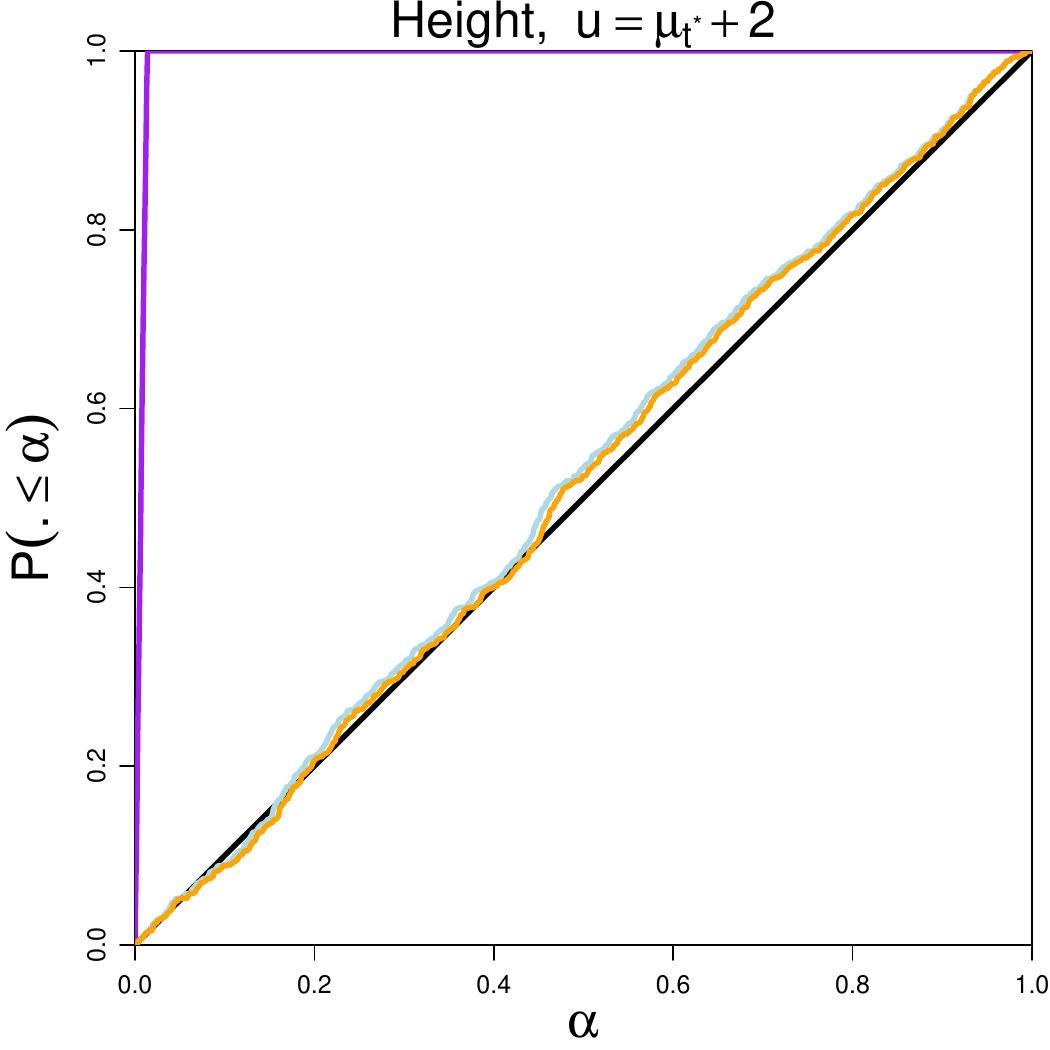}
			\label{fig:experiment-1-3}
		\end{subfigure}\hfill
		
		\caption{Same as Figure~\ref{fig:experiment-1-ppplots-1}, but $\mu_0 = 11$.}
		\label{fig:experiment-1-ppplots-2}
	\end{figure}
	
	\paragraph{One-dimensional experiment.}
	The results of an experiment where $d = 1$ are shown in Figure~\ref{fig:experiment-1-1d}. In this experiment, the signal and covariance kernel are
	\begin{equation}
		\label{eqn:experiment-1d}
	\mu_{t} = \mu_0 \cdot K(0,t), \quad K(s,t) = \exp\Big(-\frac{(s - t)^2}{2 \cdot 0.15^2}\Big).
	\end{equation}
	Again we vary $\mu_0 \in \{3,4,\cdots,11\}$ and consider thresholds $u \in \{\mu_0 - 2,\mu_0, \mu_0 + 2\}$. We draw very similar conclusions as in the two-dimensional experiment presented in Section~\ref{subsec:experiment-1}, except the agreement with asymptotic theory is even closer.
	
	\begin{figure}[htbp]
		\centering
		\begin{subfigure}[t]{0.32\linewidth}
			\centering
			\includegraphics[width=\linewidth]{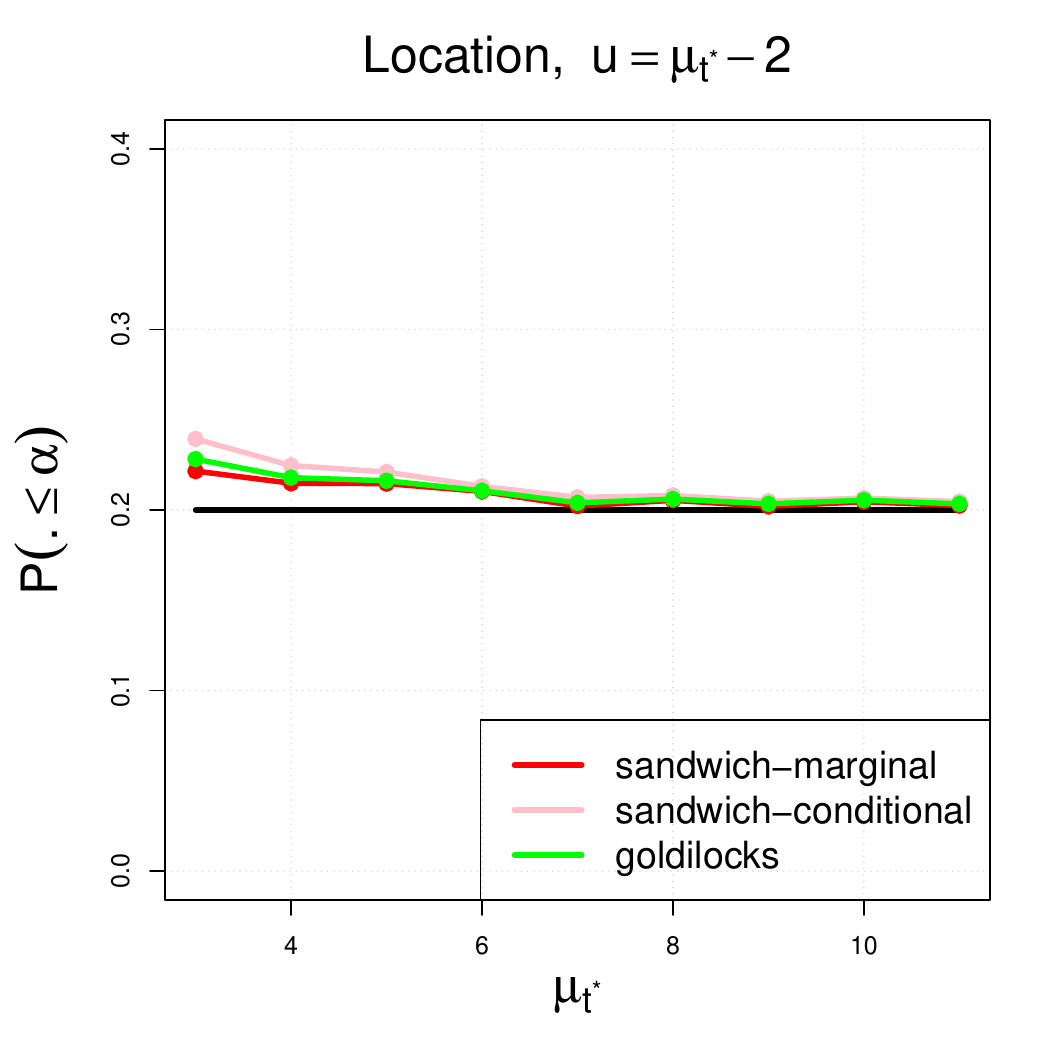}
			\label{fig:experiment-1-4}
		\end{subfigure}\hfill
		\begin{subfigure}[t]{0.32\linewidth}
			\centering
			\includegraphics[width=\linewidth]{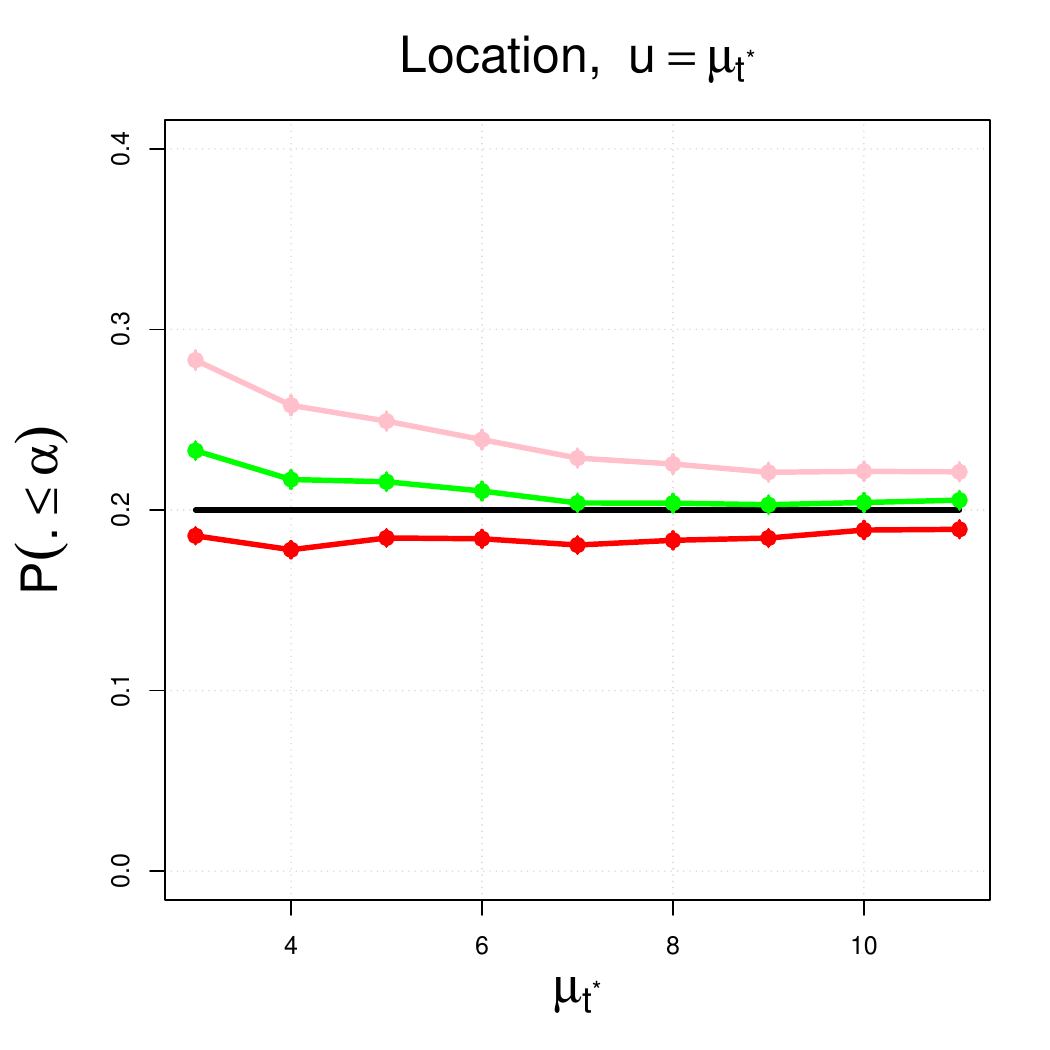}
			\label{fig:experiment-1-5}
		\end{subfigure}\hfill
		\begin{subfigure}[t]{0.32\linewidth}
			\centering
			\includegraphics[width=\linewidth]{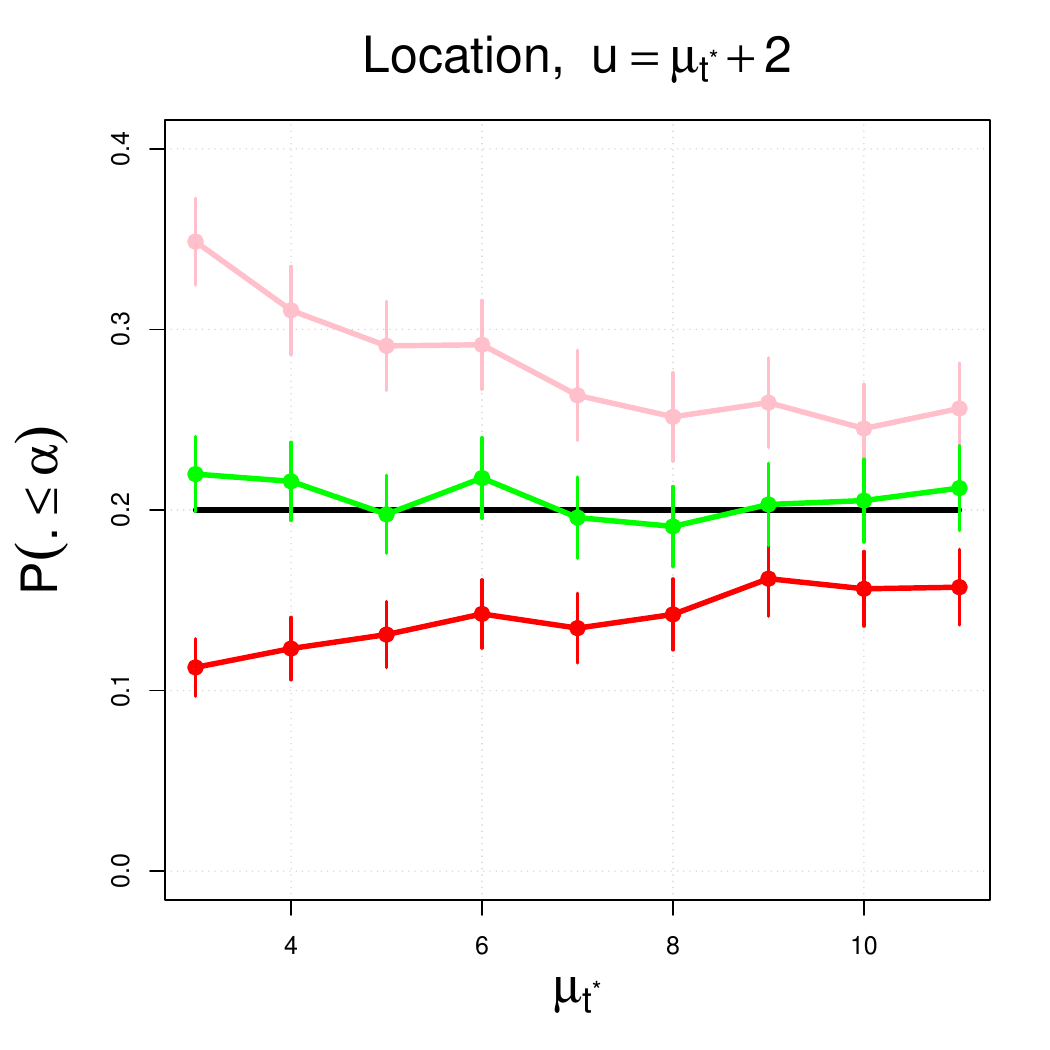}
			\label{fig:experiment-1-6}
		\end{subfigure}\hfill
		
		\begin{subfigure}[t]{0.32\linewidth}
			\centering
			\includegraphics[width=\linewidth]{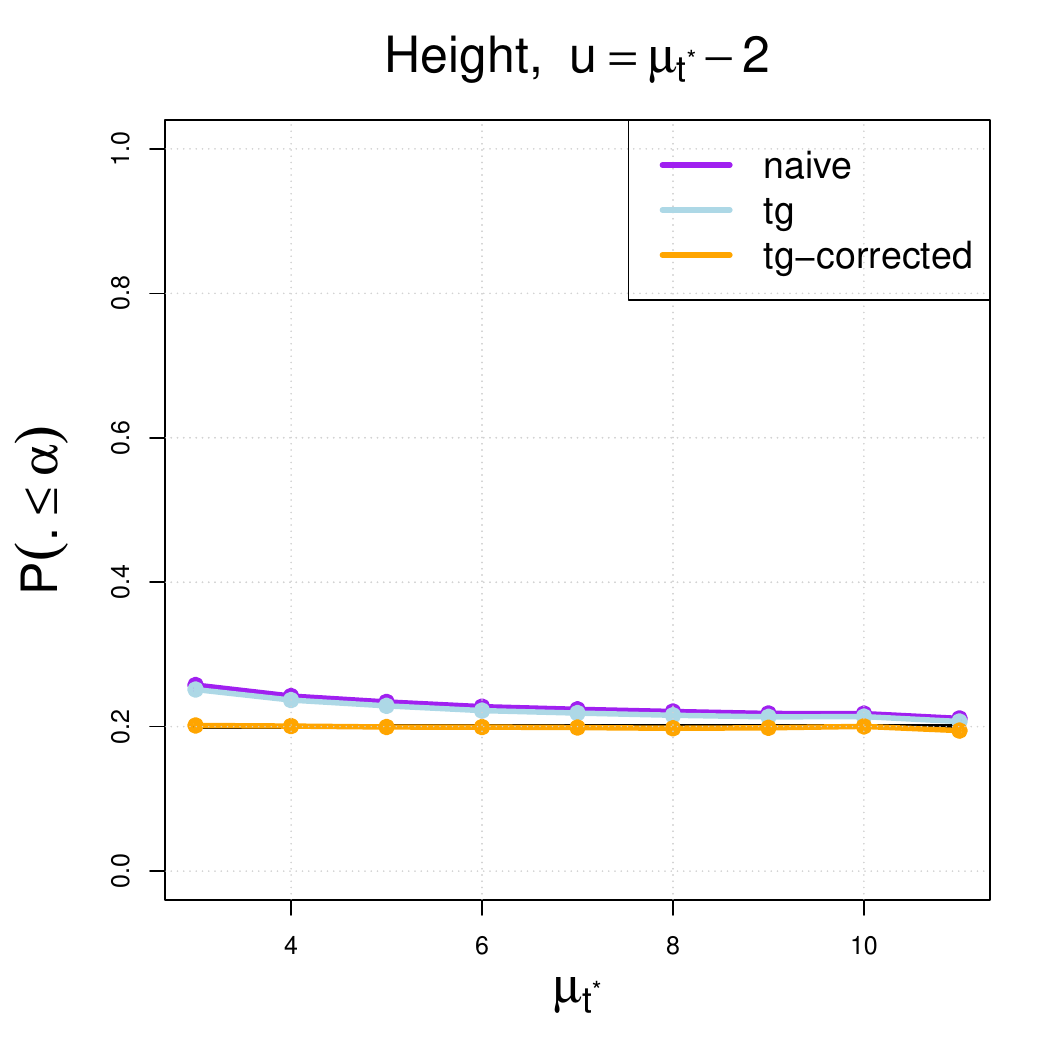}
			\label{fig:experiment-1-1}
		\end{subfigure}\hfill
		\begin{subfigure}[t]{0.32\linewidth}
			\centering
			\includegraphics[width=\linewidth]{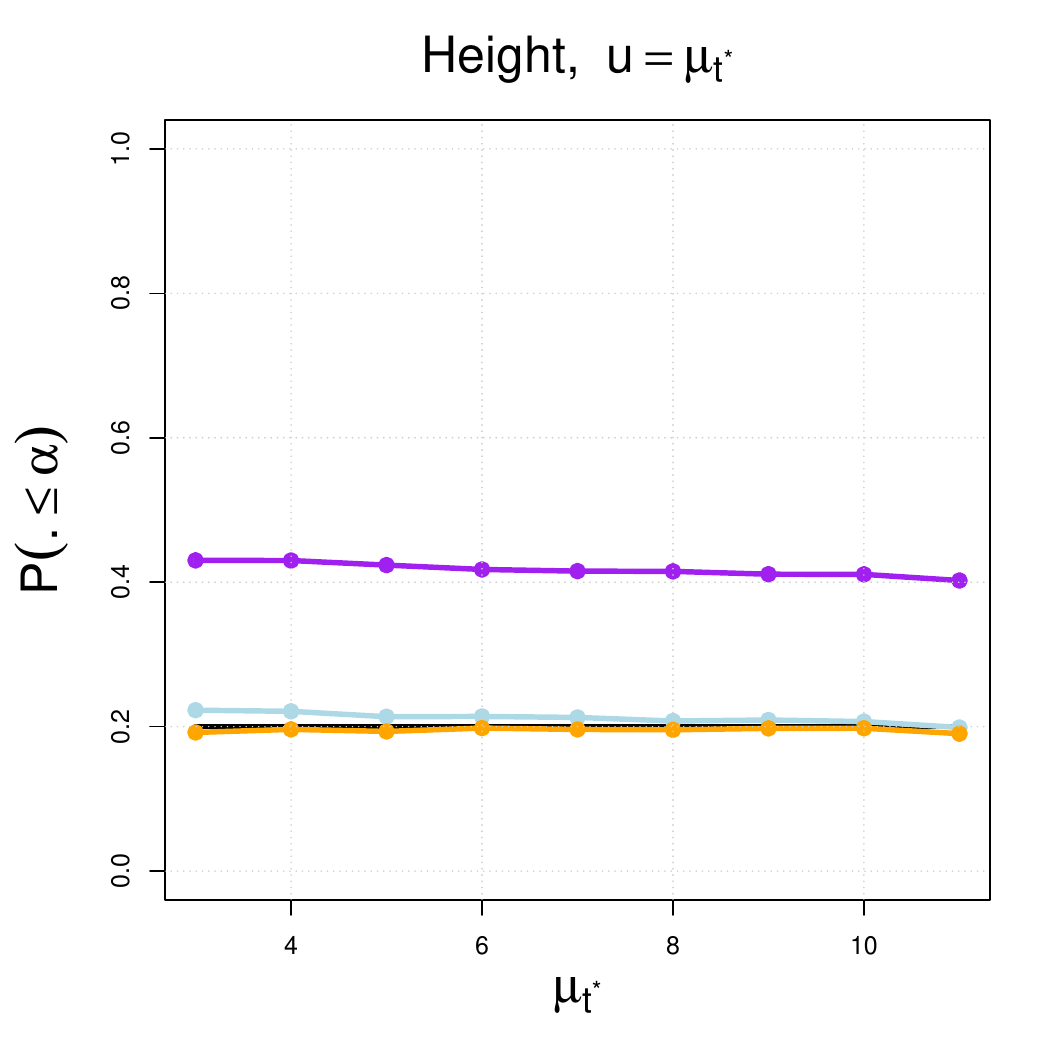}
			\label{fig:experiment-1-2}
		\end{subfigure}\hfill
		\begin{subfigure}[t]{0.32\linewidth}
			\centering
			\includegraphics[width=\linewidth]{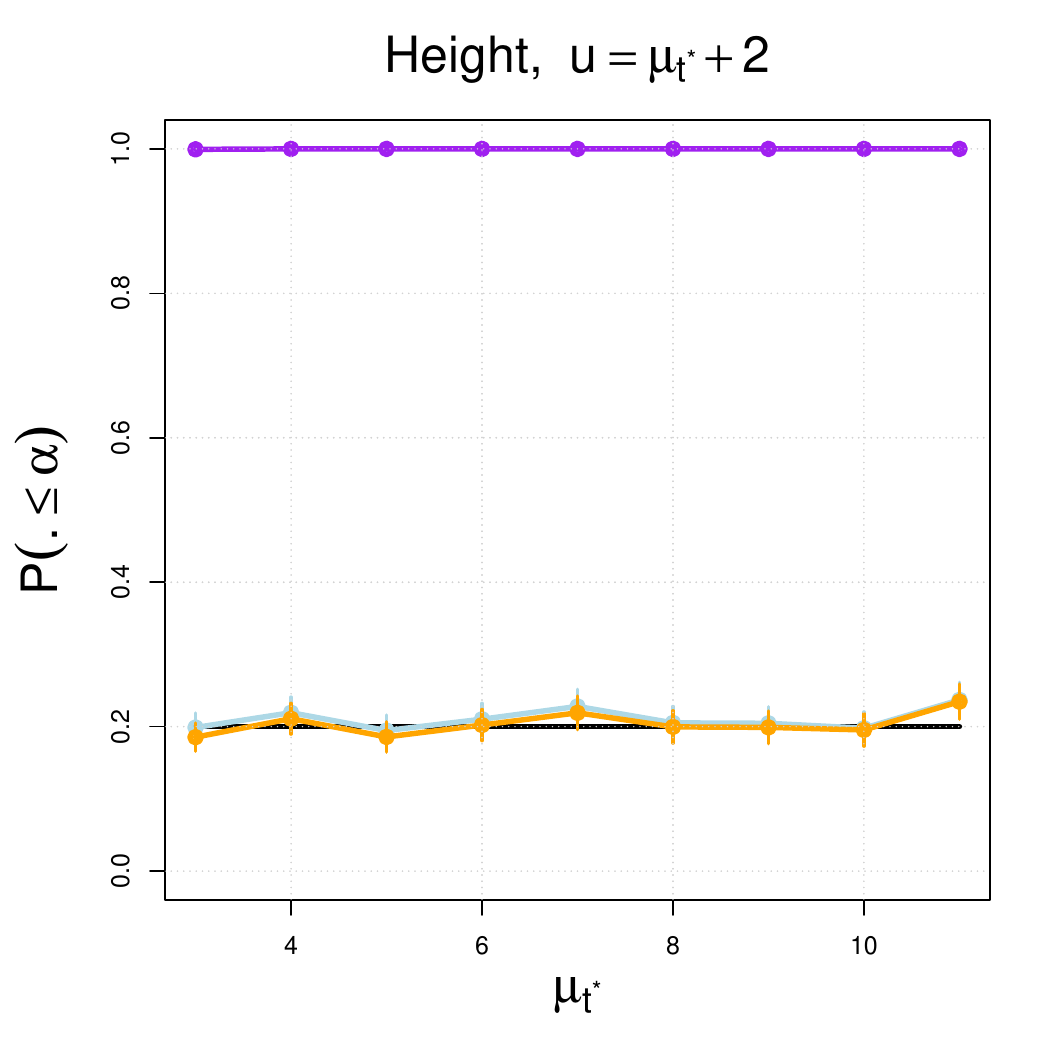}
			\label{fig:experiment-1-3}
		\end{subfigure}\hfill
		
		\caption{Distribution of candidate pivots for height (top row) and location (bottom). Different columns corresponds to different thresholds $u$. Details of experimental setup and takeaways are in the main text.}
		\label{fig:experiment-1-1d}
	\end{figure}
	
	\subsection{Conditional coverage}
	\label{subsec:experiment-2-additional}
	
	\paragraph{Wide signal.}
	Figure~\ref{fig:experiment-2-2dwide} presents the results of an experiment analogous to that of Section~\ref{subsec:experiment-2}, but where the signal is
	\begin{equation*}
	\mu_t = \mu_{0} \cdot \exp\Big(-\frac{\|t - 0\|^2}{2 \cdot 0.25^2}\Big),
	\end{equation*}
	i.e. it is wider than the covariance kernel. The relative performance of methods is qualitatively similar to those given in Section~\ref{subsec:experiment-2}, except the coverage of all methods is much poorer at the smaller signal strengths. This is because the curvature of the signal is much smaller. \citet{davenport2022confidence} discuss a similar phenomenon when conducting inference for location without selection, and offer a solution that calibrates inferences using Monte Carlo rather than asymptotic theory. It would be of interest to see whether this method could be incorporated into the post-selection inference framework considered in this paper, to improve performance in challenging low-curvature settings. 
	
	\paragraph{One-dimensional experiment.}
	Figure~\ref{fig:experiment-2-1d} presents the results of an experiment analogous to that of Section~\ref{subsec:experiment-2}, but with the one-dimensional data generating process defined in~\eqref{eqn:experiment-1d}.  The relative performance of methods is qualitatively similar to those given in Section~\ref{subsec:experiment-2}.
	
	\begin{figure}[htbp]
		\centering
		\begin{subfigure}[t]{0.32\linewidth}
			\centering
			\includegraphics[width=\linewidth]{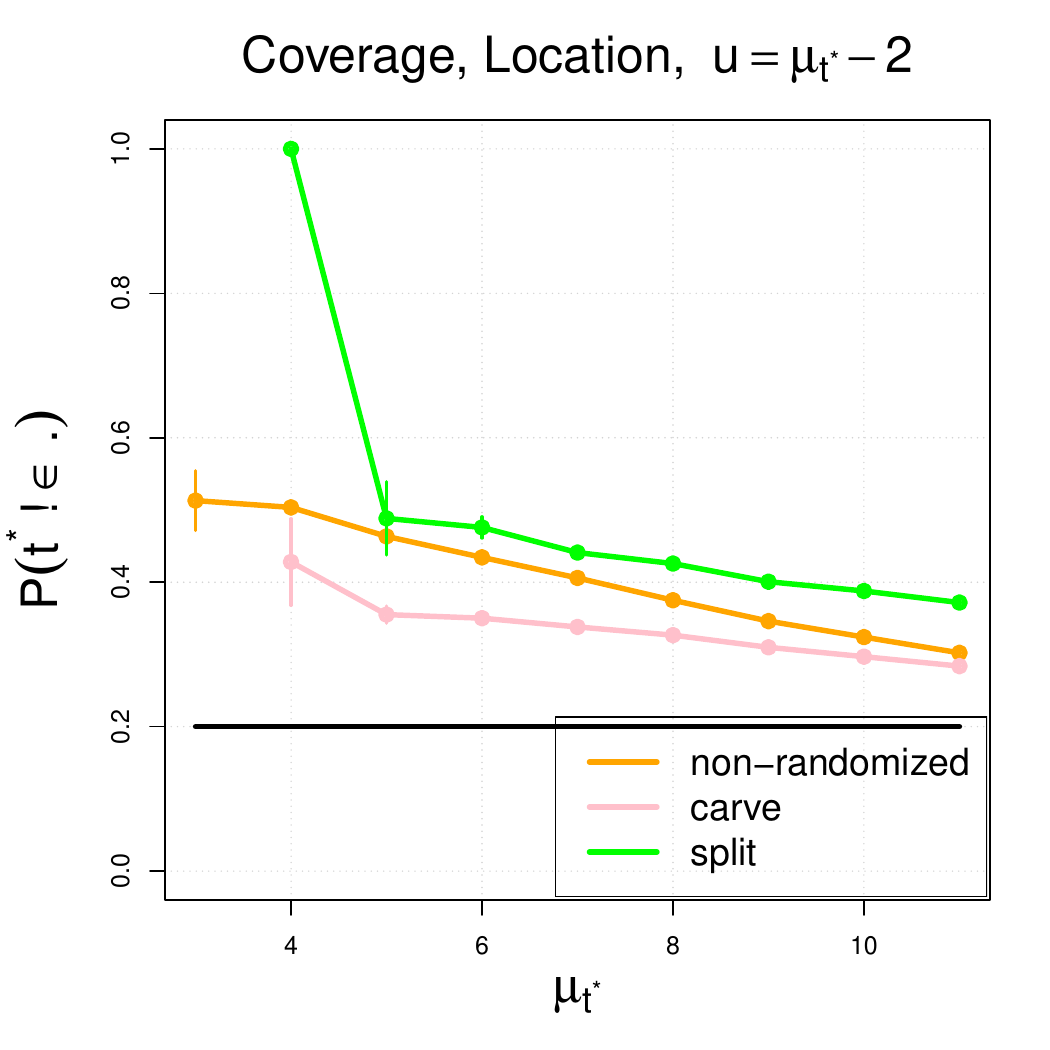}
		\end{subfigure}\hfill
		\begin{subfigure}[t]{0.32\linewidth}
			\centering
			\includegraphics[width=\linewidth]{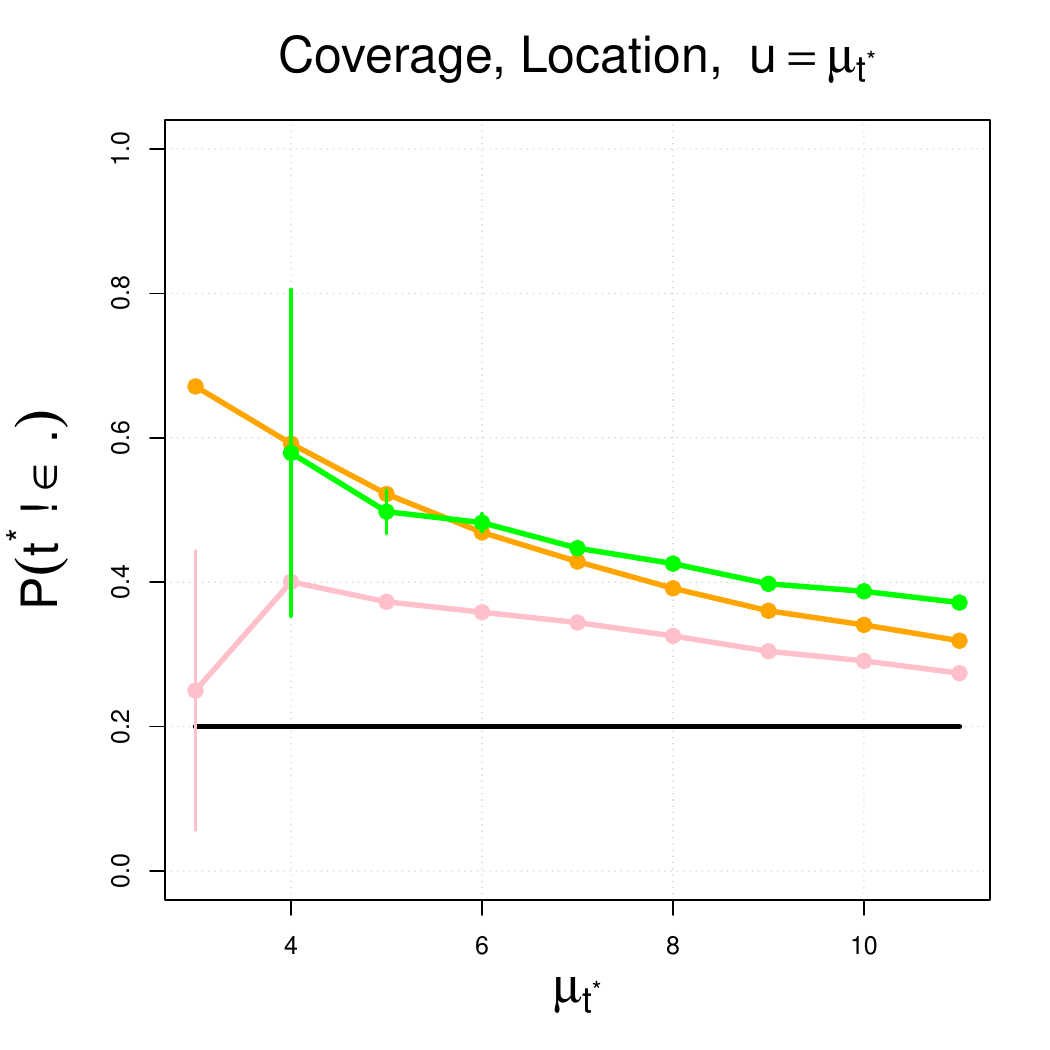}
		\end{subfigure}\hfill
		\begin{subfigure}[t]{0.32\linewidth}
			\centering
			\includegraphics[width=\linewidth]{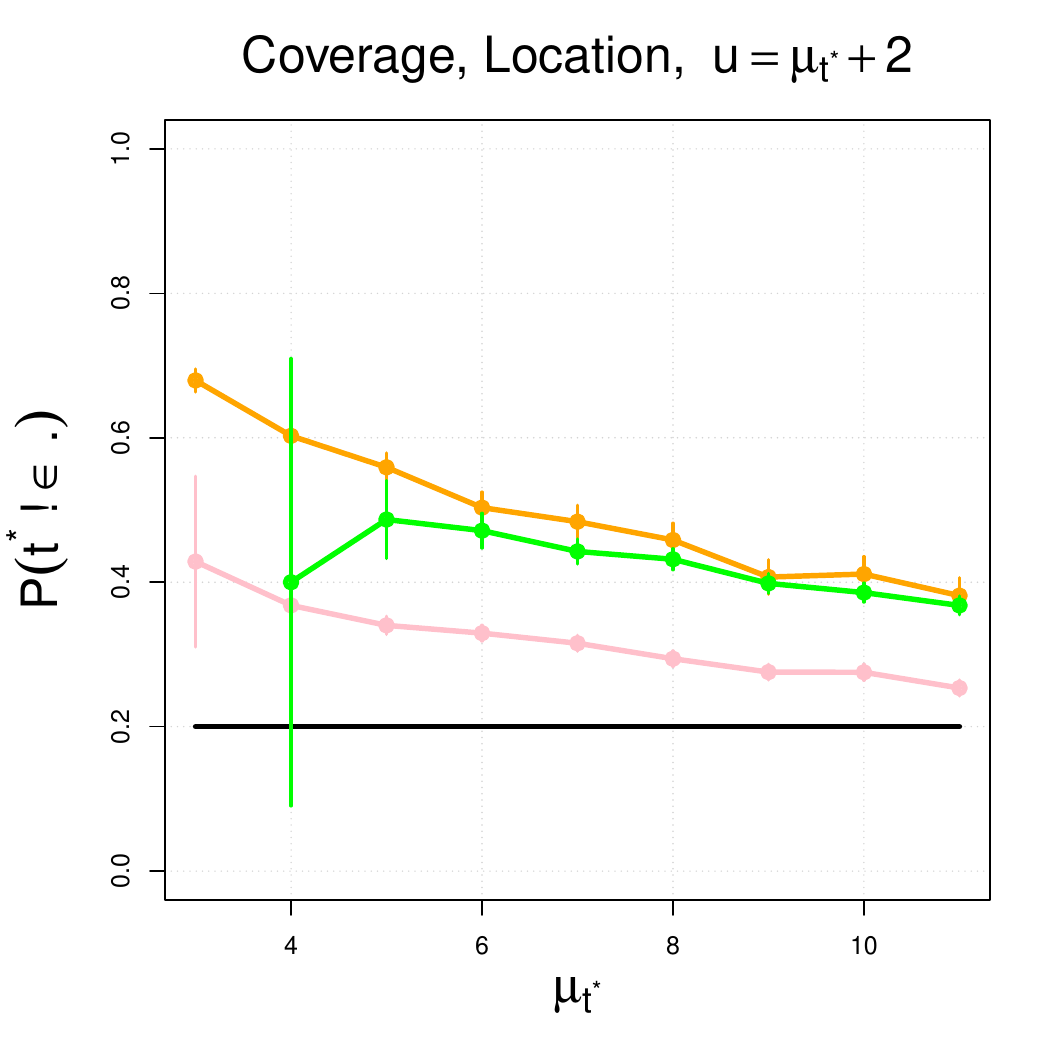}
		\end{subfigure}\hfill
		
		\begin{subfigure}[t]{0.32\linewidth}
			\centering
			\includegraphics[width=\linewidth]{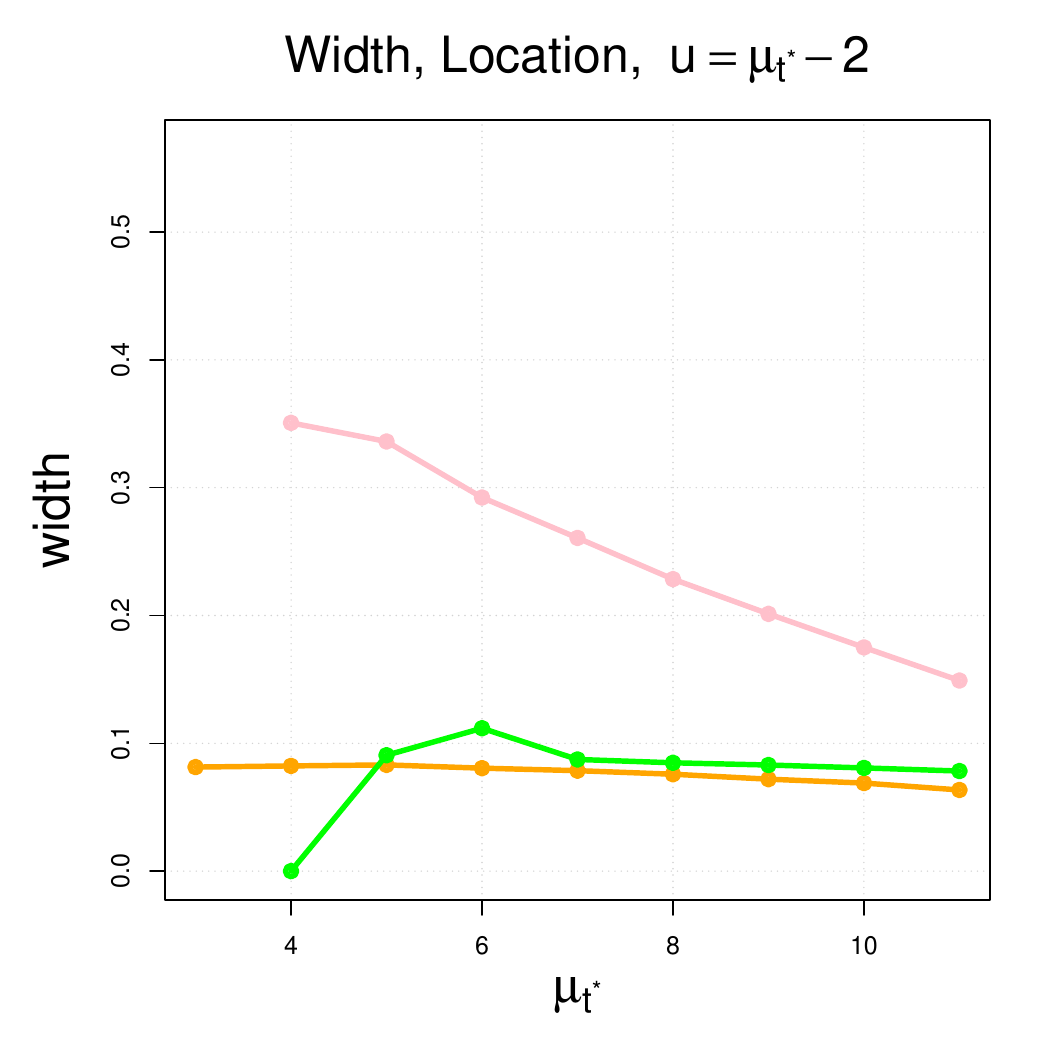}
		\end{subfigure}\hfill
		\begin{subfigure}[t]{0.32\linewidth}
			\centering
			\includegraphics[width=\linewidth]{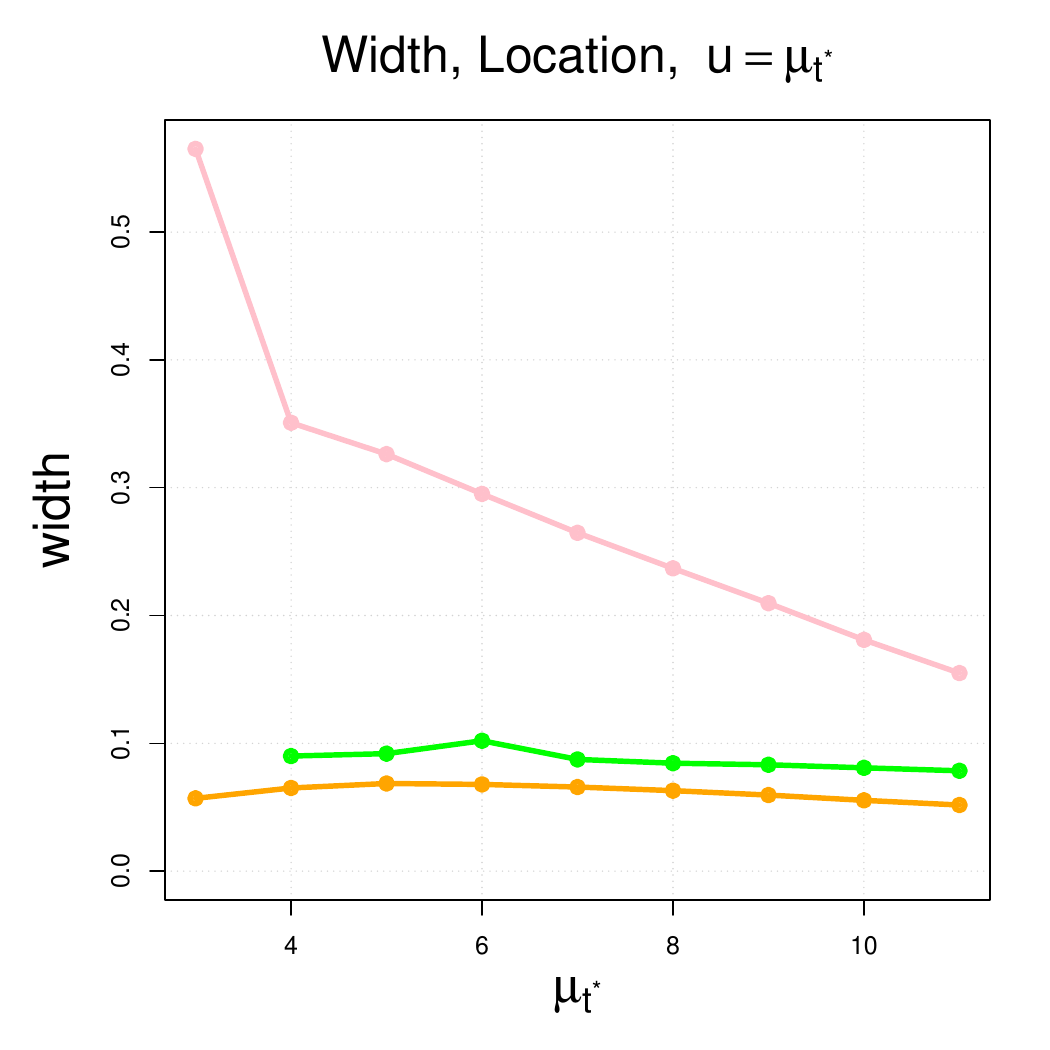}
		\end{subfigure}\hfill
		\begin{subfigure}[t]{0.32\linewidth}
			\centering
			\includegraphics[width=\linewidth]{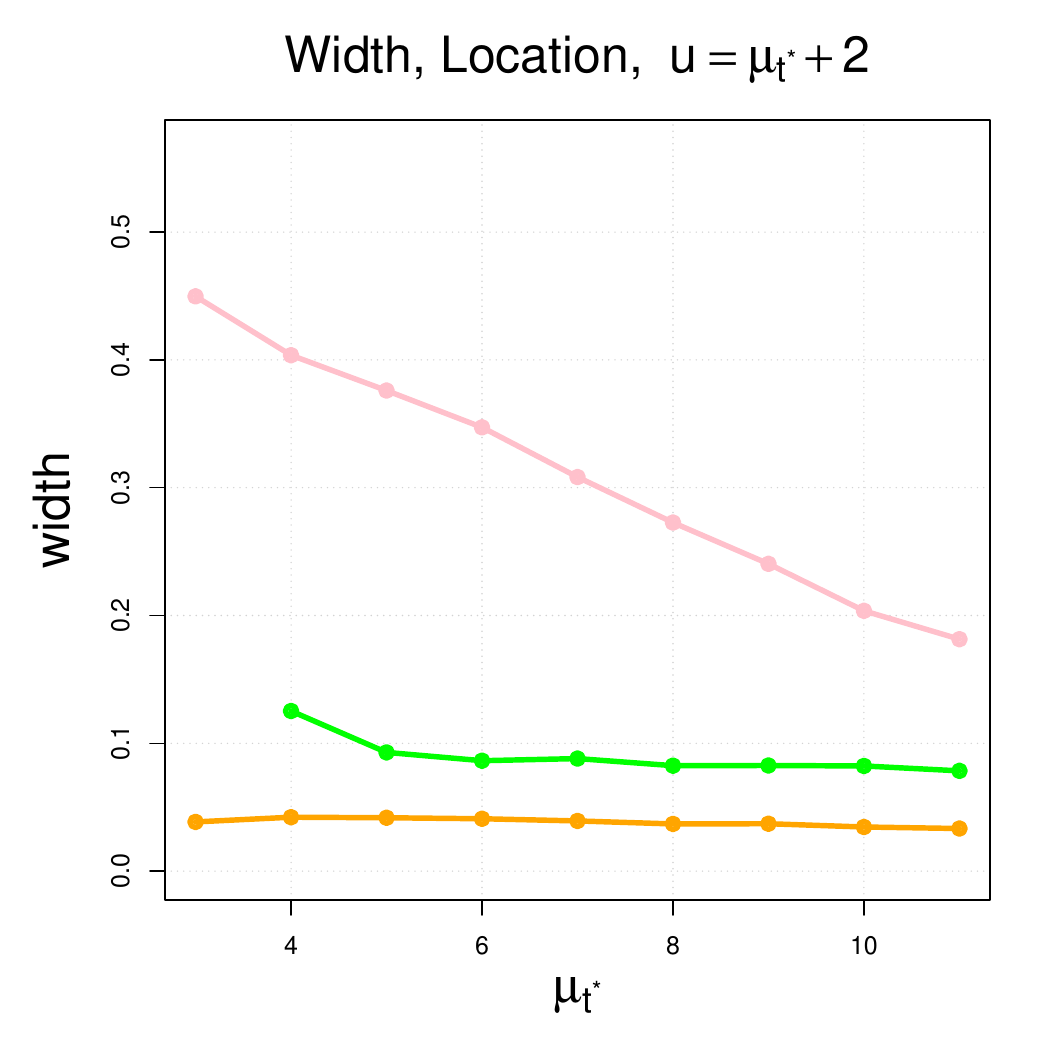}
		\end{subfigure}\hfill
		
		\begin{subfigure}[t]{0.32\linewidth}
			\centering
			\includegraphics[width=\linewidth]{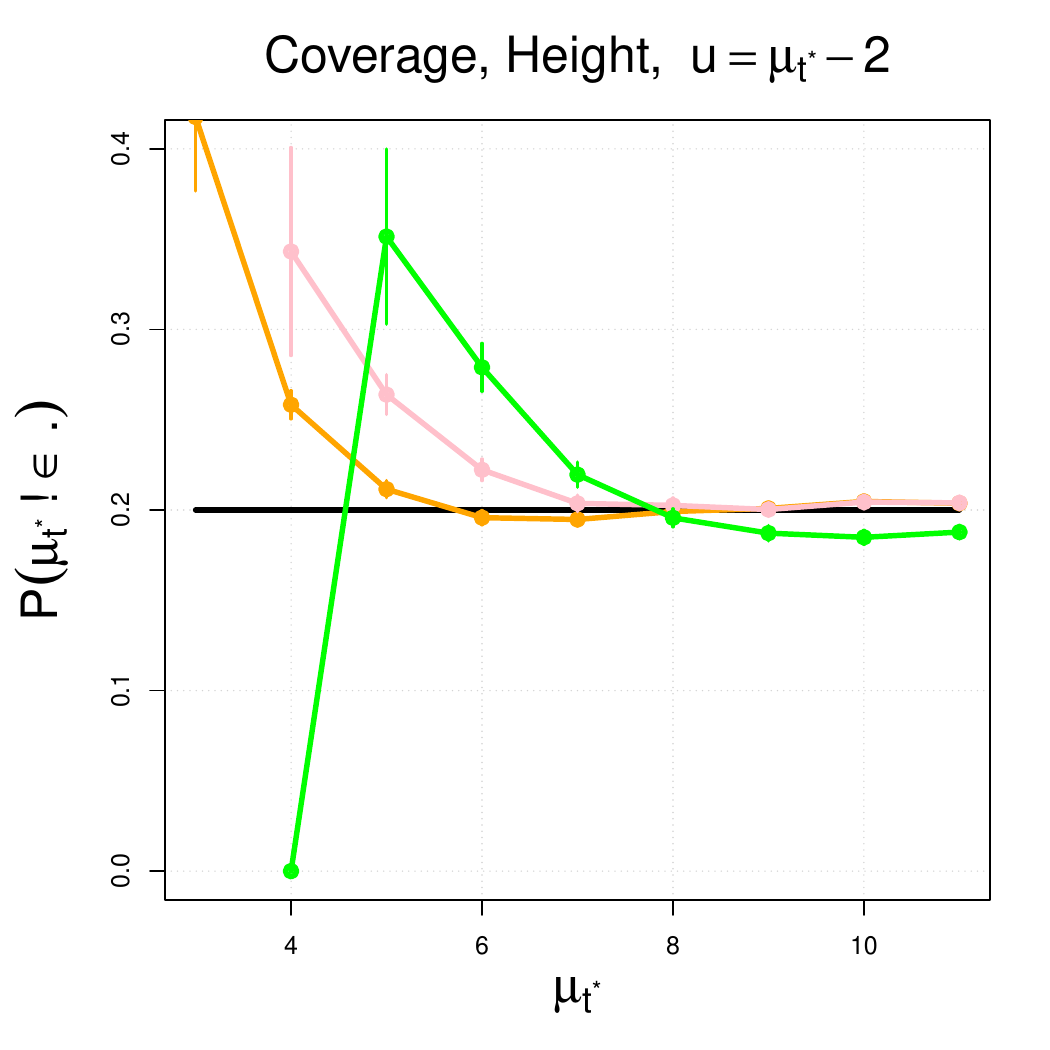}
		\end{subfigure}\hfill
		\begin{subfigure}[t]{0.32\linewidth}
			\centering
			\includegraphics[width=\linewidth]{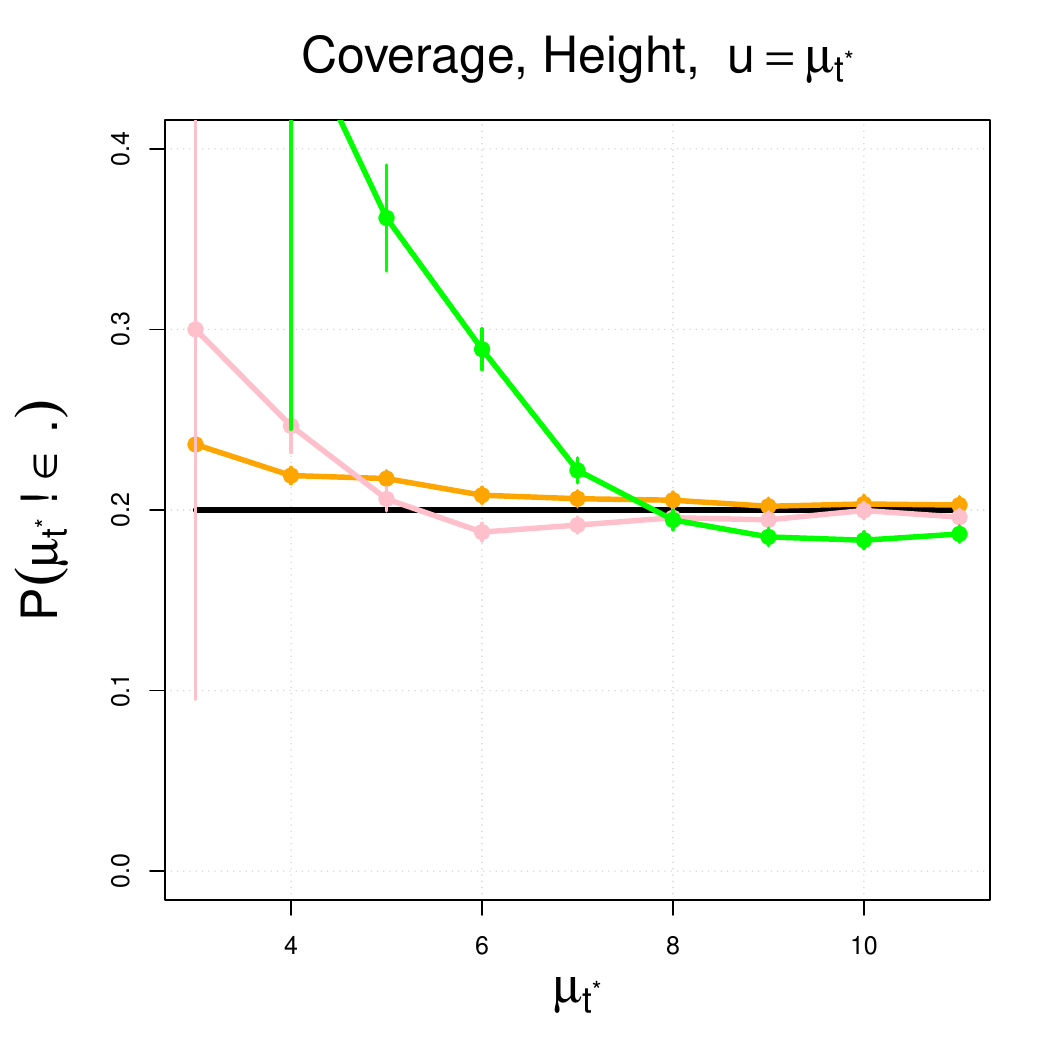}
		\end{subfigure}\hfill
		\begin{subfigure}[t]{0.32\linewidth}
			\centering
			\includegraphics[width=\linewidth]{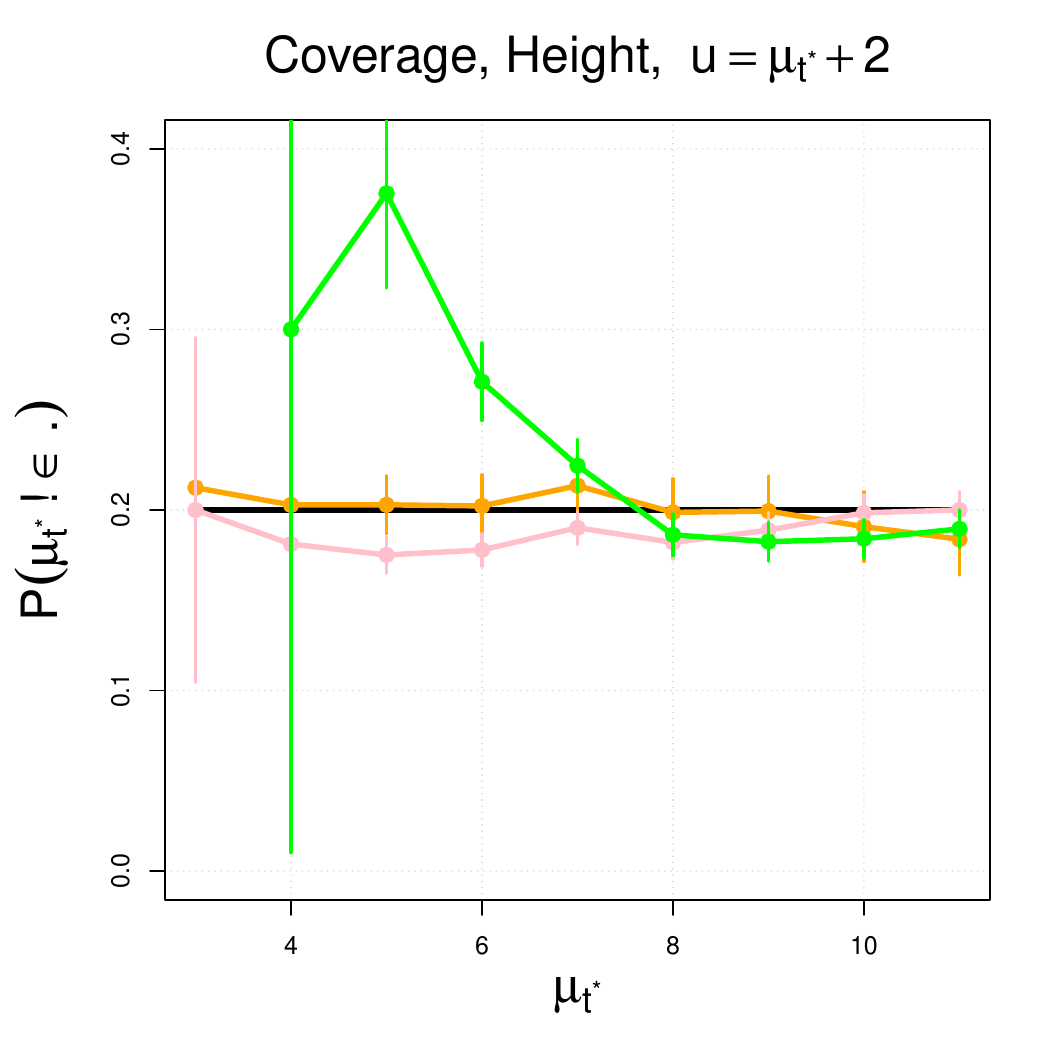}
		\end{subfigure}\hfill
		
		\begin{subfigure}[t]{0.32\linewidth}
			\centering
			\includegraphics[width=\linewidth]{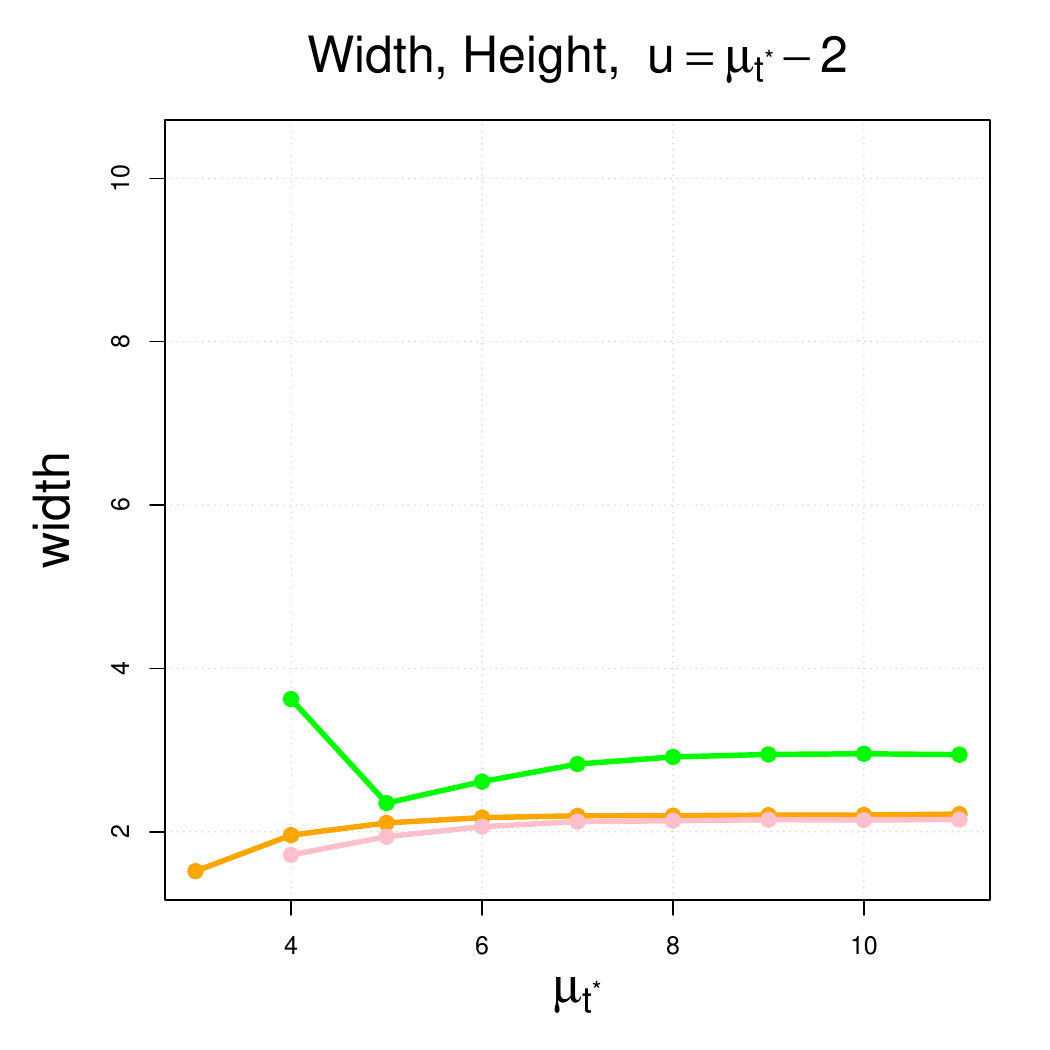}
		\end{subfigure}\hfill
		\begin{subfigure}[t]{0.32\linewidth}
			\centering
			\includegraphics[width=\linewidth]{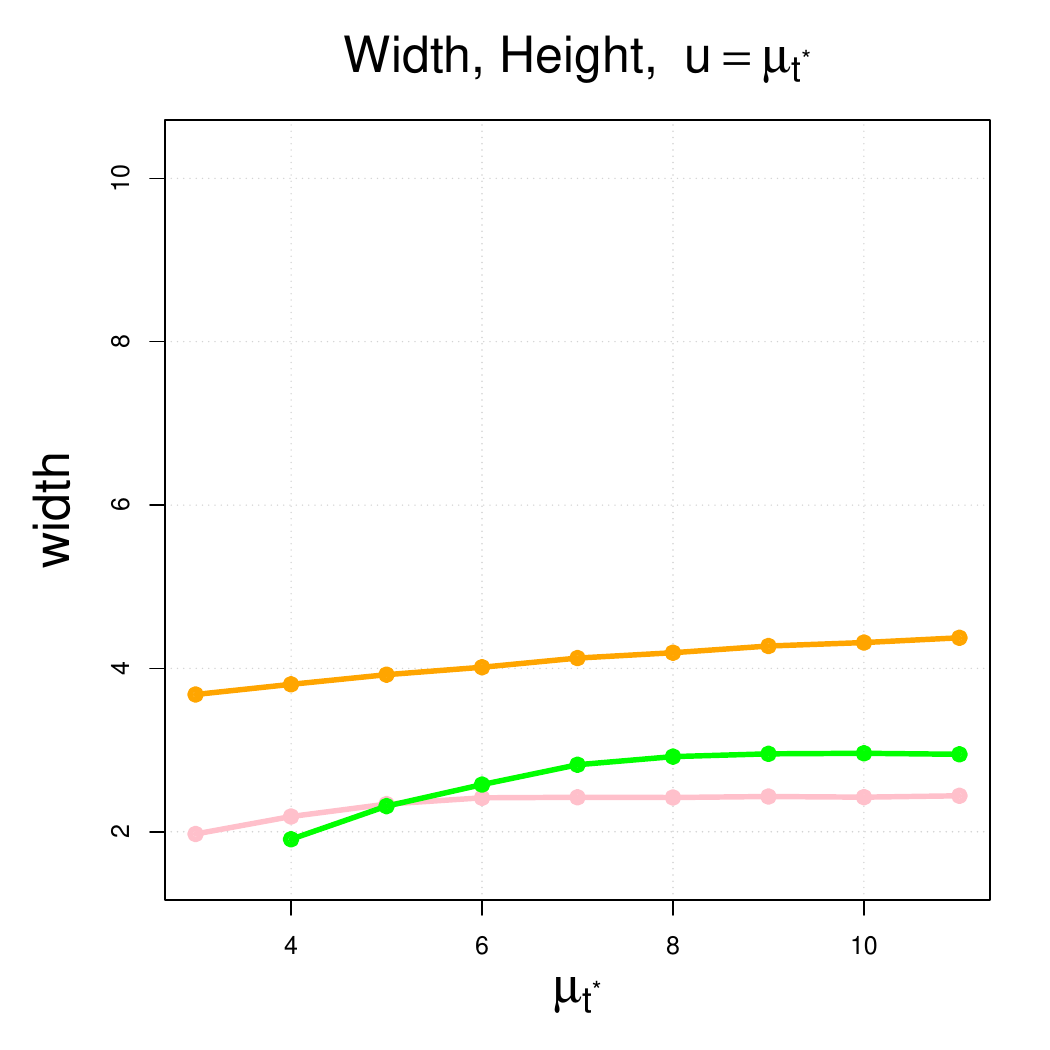}
		\end{subfigure}\hfill
		\begin{subfigure}[t]{0.32\linewidth}
			\centering
			\includegraphics[width=\linewidth]{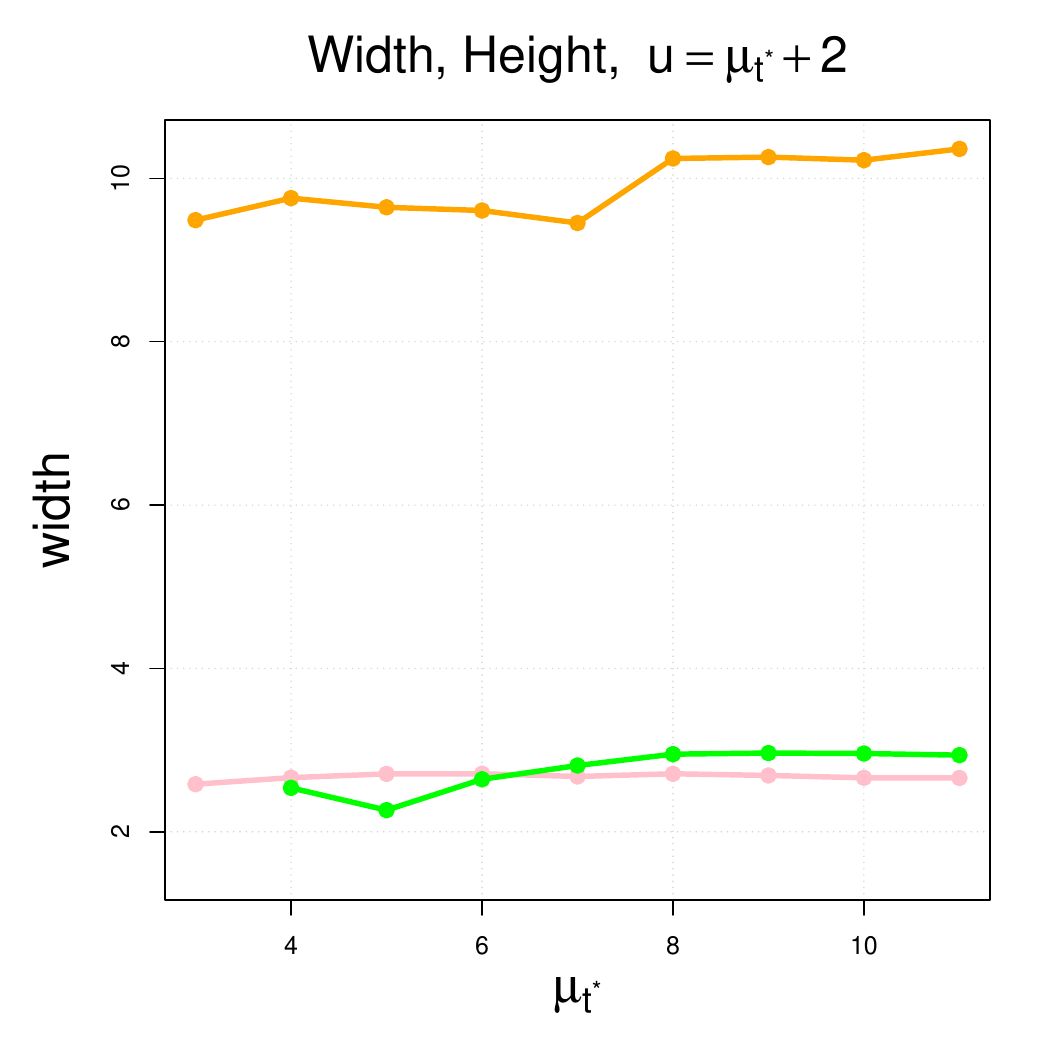}
		\end{subfigure}\hfill
		
		\caption{Comparing miscoverage and width of non-randomized, carve, and split methods for peak inference, in the wide signal experiment described in Section~\ref{subsec:experiment-2-additional}. Top two rows correspond to inference for location, bottom two rows to inference for height.}
		\label{fig:experiment-2-2dwide}
	\end{figure}
	
	\begin{figure}[htbp]
		\centering
		\begin{subfigure}[t]{0.32\linewidth}
			\centering
			\includegraphics[width=\linewidth]{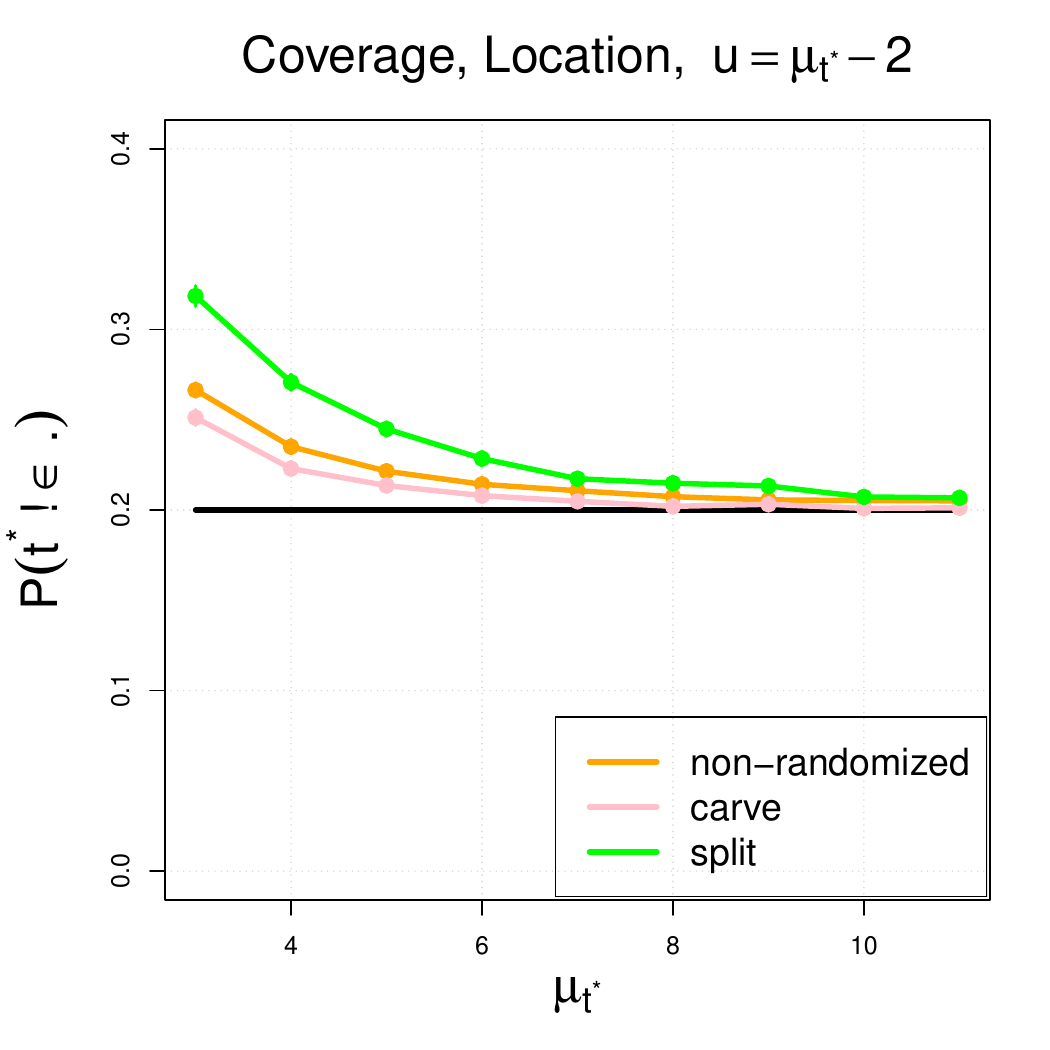}
		\end{subfigure}\hfill
		\begin{subfigure}[t]{0.32\linewidth}
			\centering
			\includegraphics[width=\linewidth]{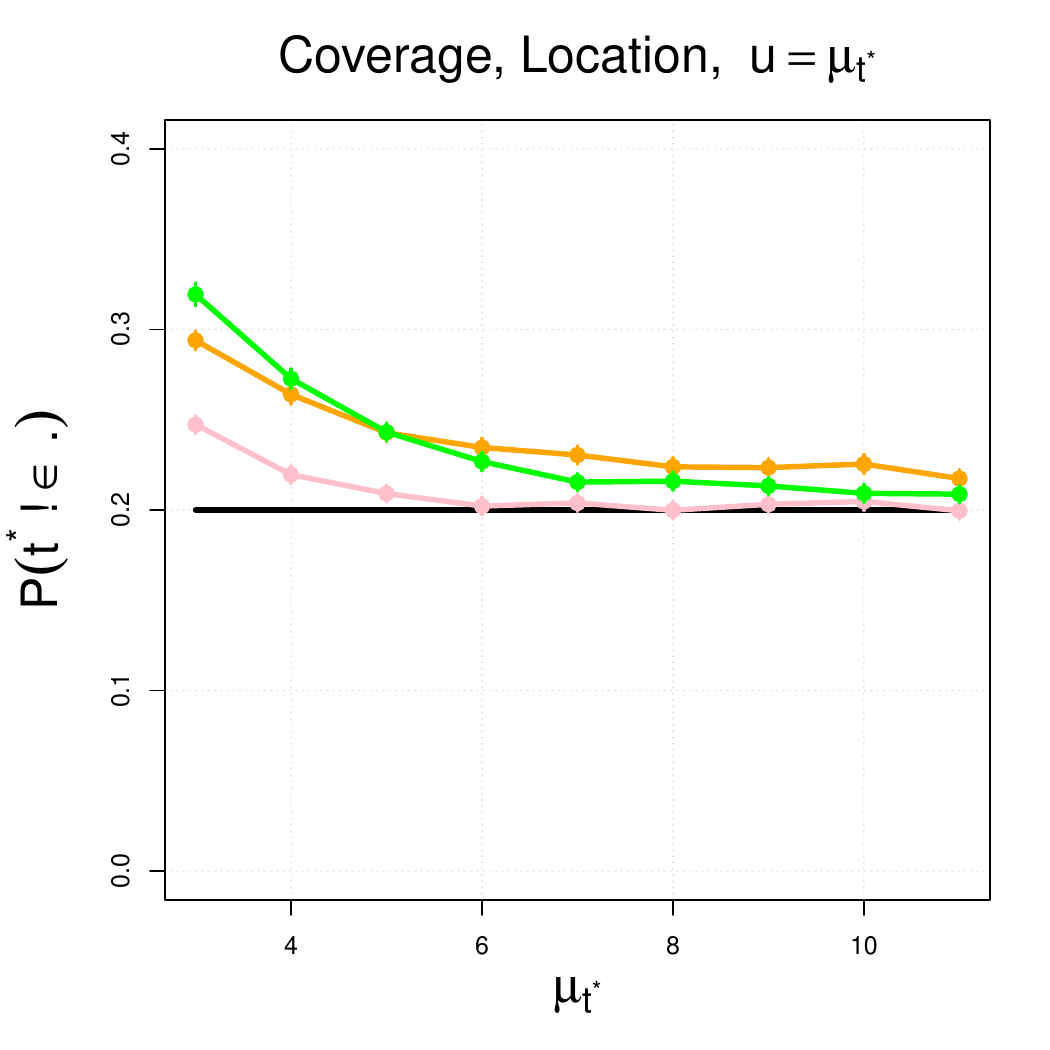}
		\end{subfigure}\hfill
		\begin{subfigure}[t]{0.32\linewidth}
			\centering
			\includegraphics[width=\linewidth]{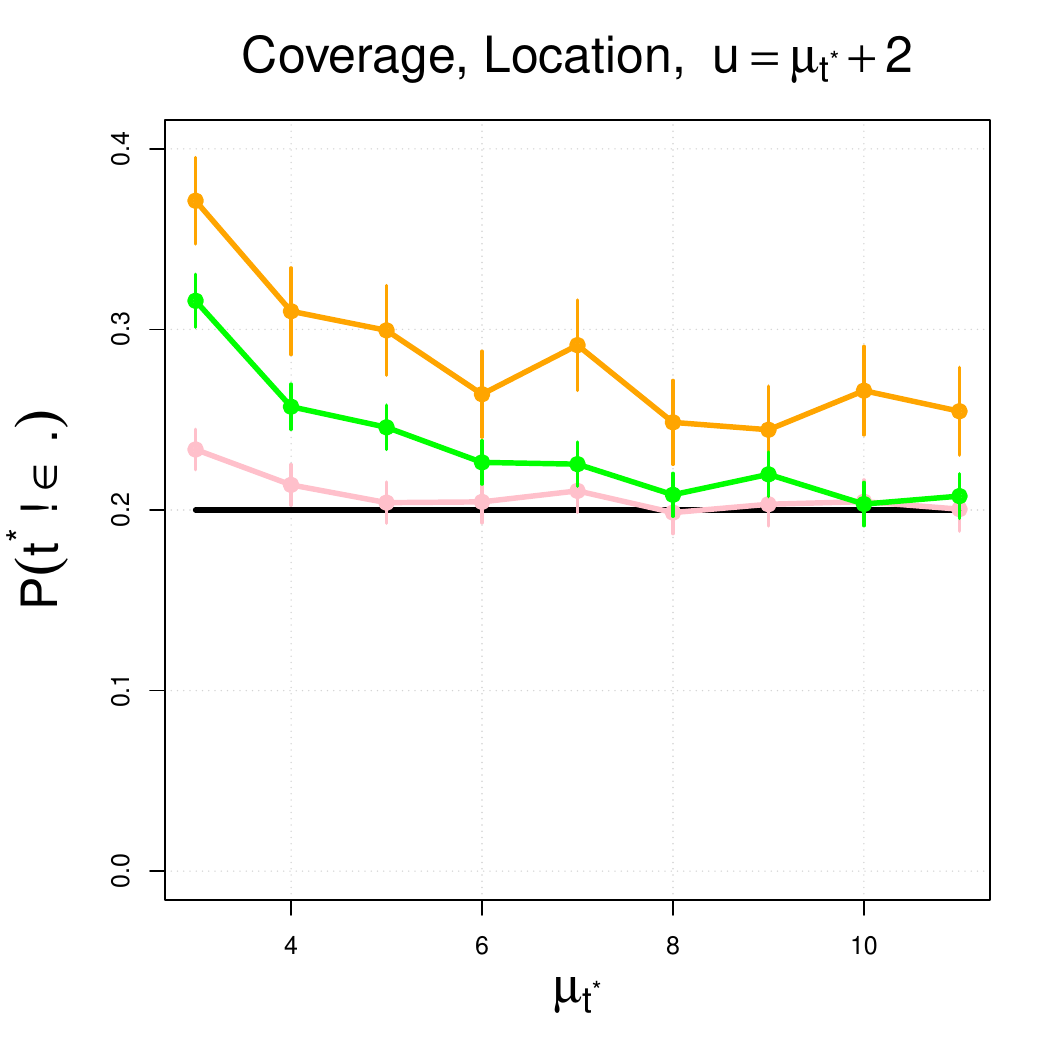}
		\end{subfigure}\hfill
		
		\begin{subfigure}[t]{0.32\linewidth}
			\centering
			\includegraphics[width=\linewidth]{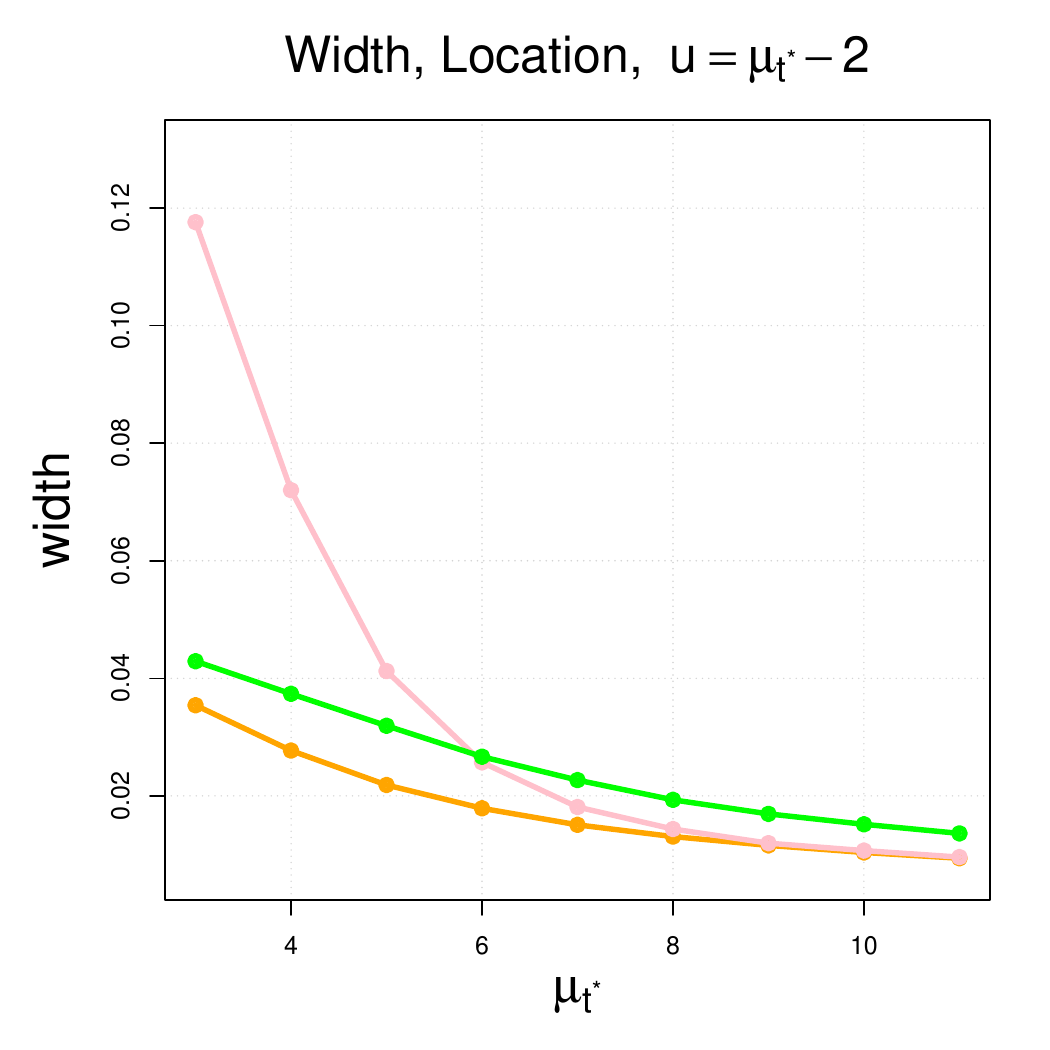}
		\end{subfigure}\hfill
		\begin{subfigure}[t]{0.32\linewidth}
			\centering
			\includegraphics[width=\linewidth]{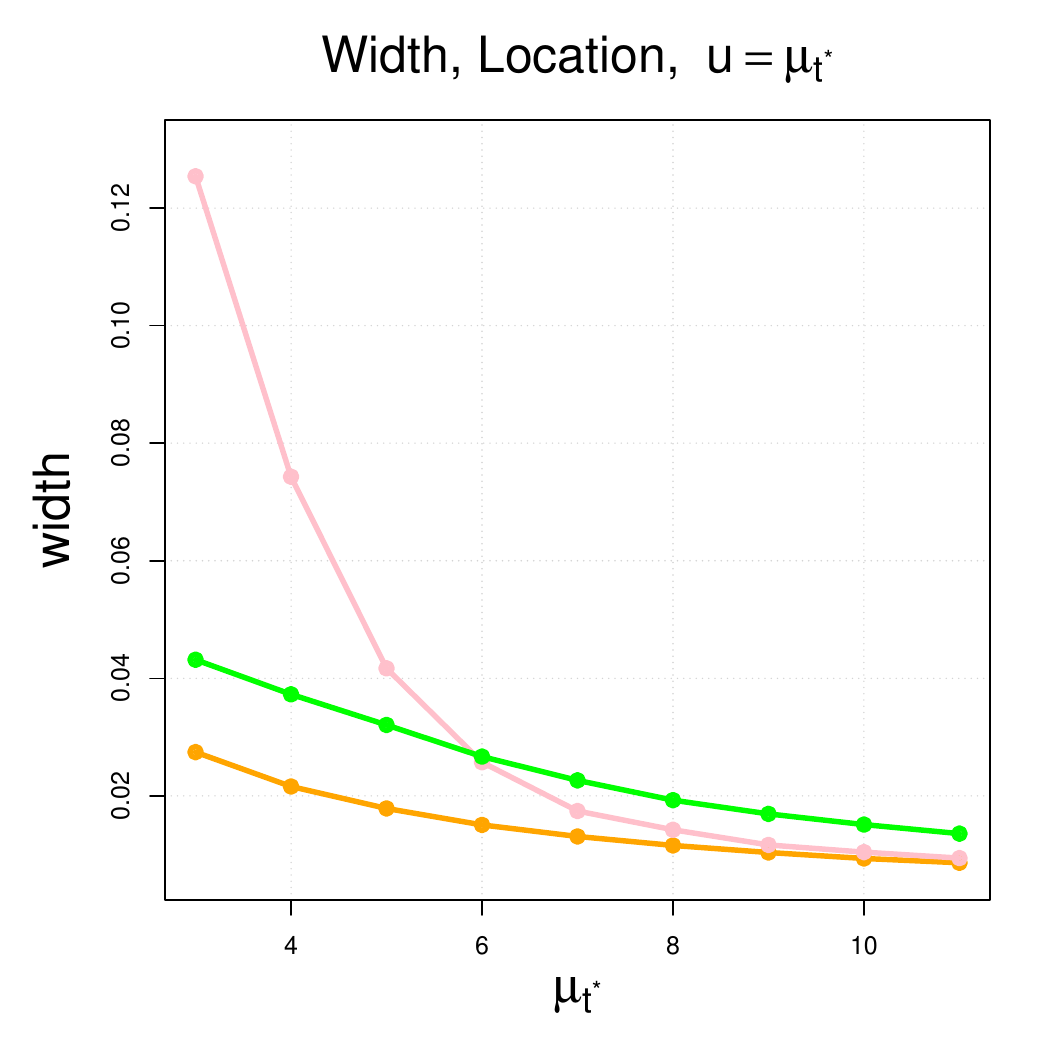}
		\end{subfigure}\hfill
		\begin{subfigure}[t]{0.32\linewidth}
			\centering
			\includegraphics[width=\linewidth]{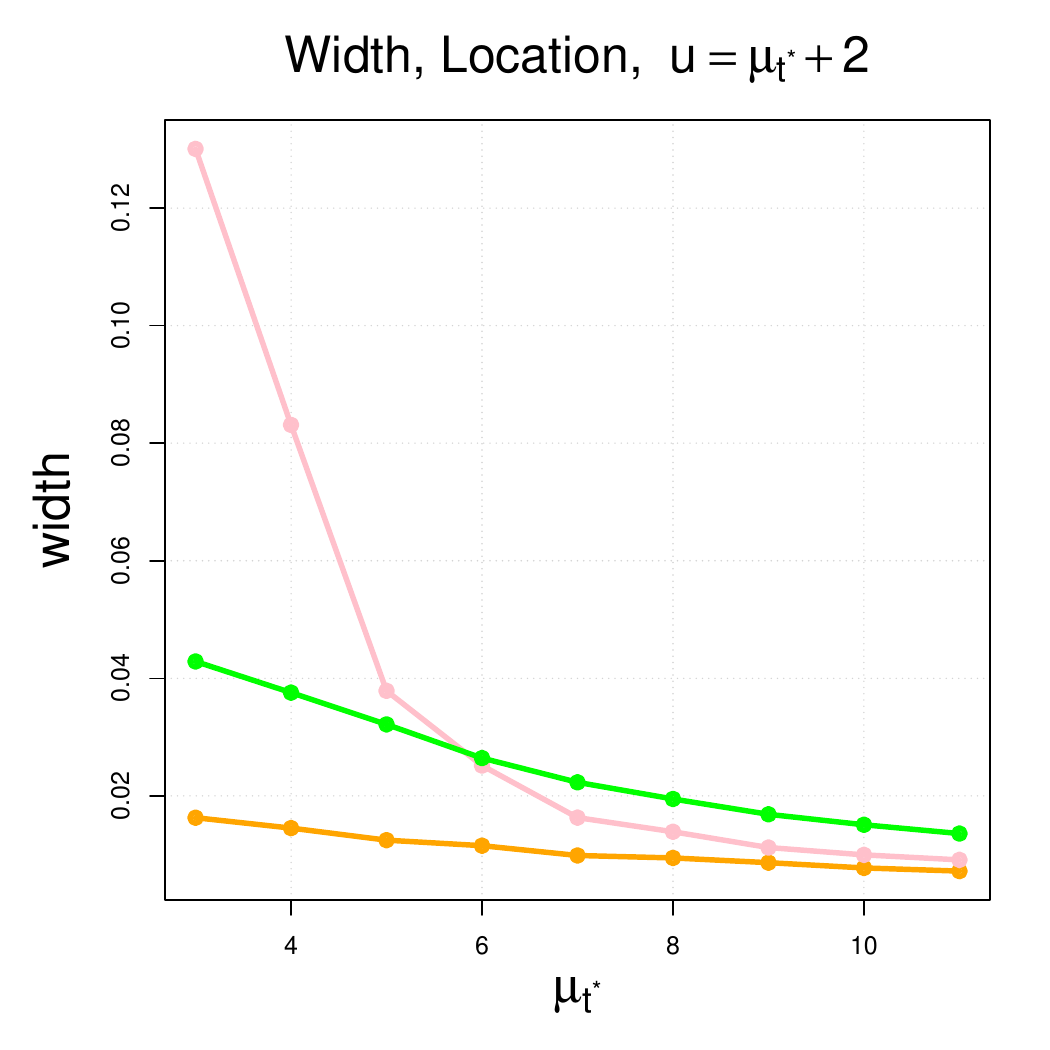}
		\end{subfigure}\hfill
		
		\begin{subfigure}[t]{0.32\linewidth}
			\centering
			\includegraphics[width=\linewidth]{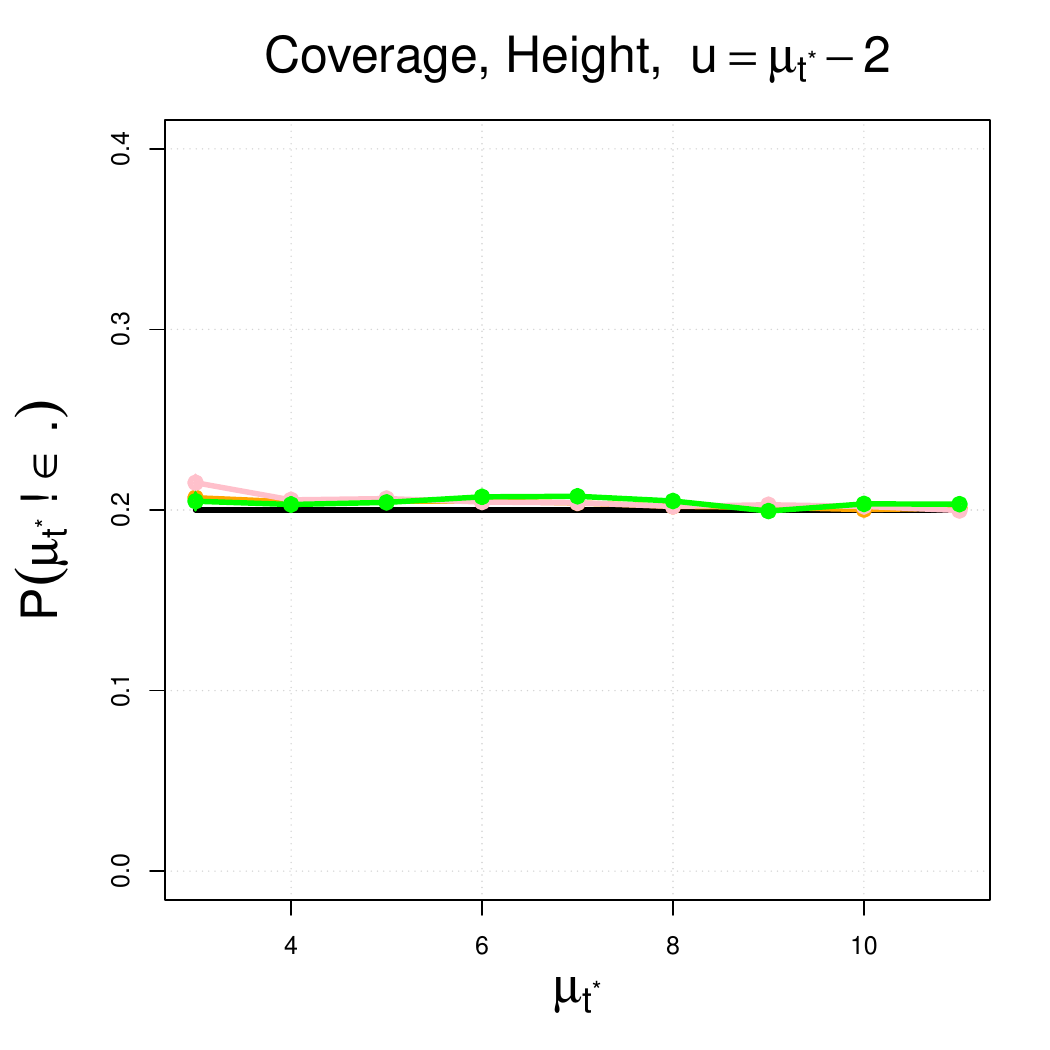}
		\end{subfigure}\hfill
		\begin{subfigure}[t]{0.32\linewidth}
			\centering
			\includegraphics[width=\linewidth]{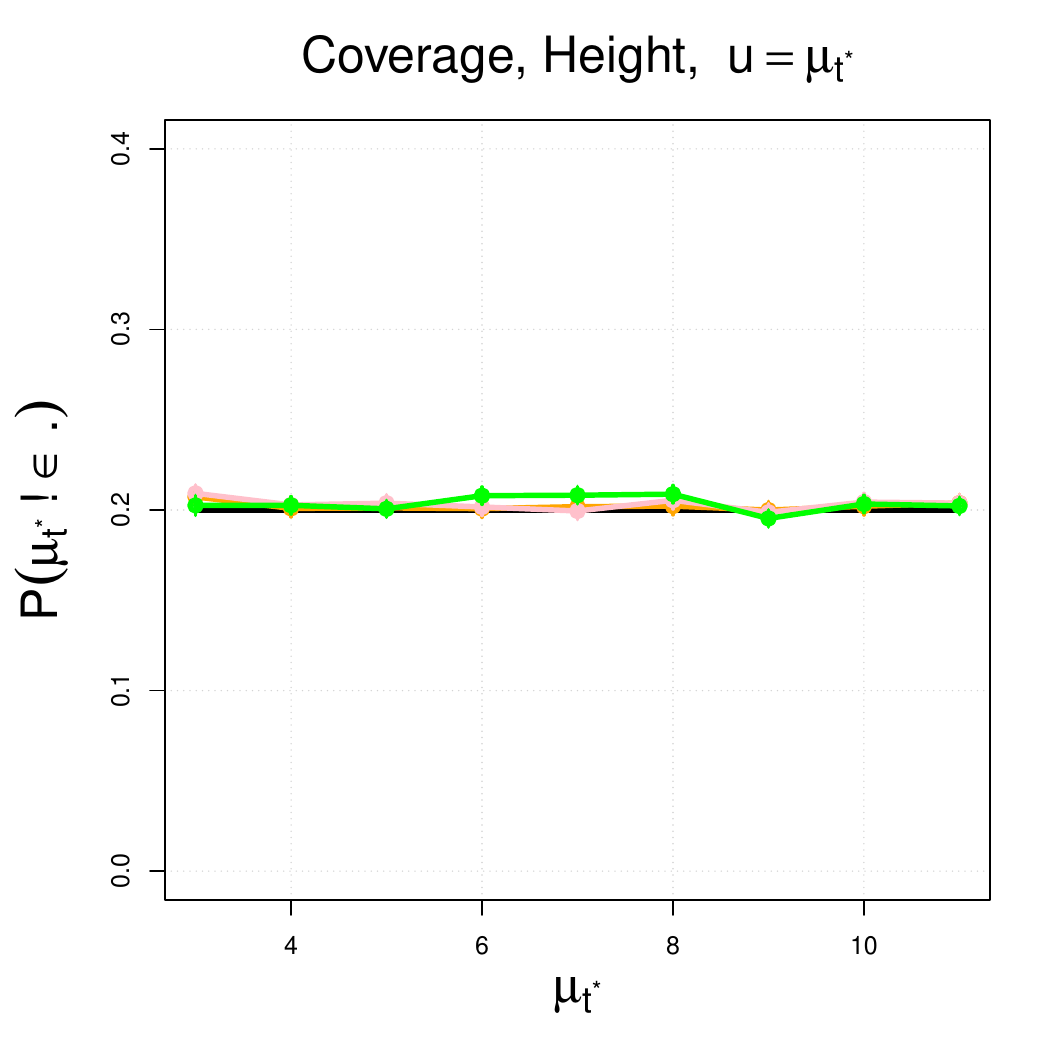}
		\end{subfigure}\hfill
		\begin{subfigure}[t]{0.32\linewidth}
			\centering
			\includegraphics[width=\linewidth]{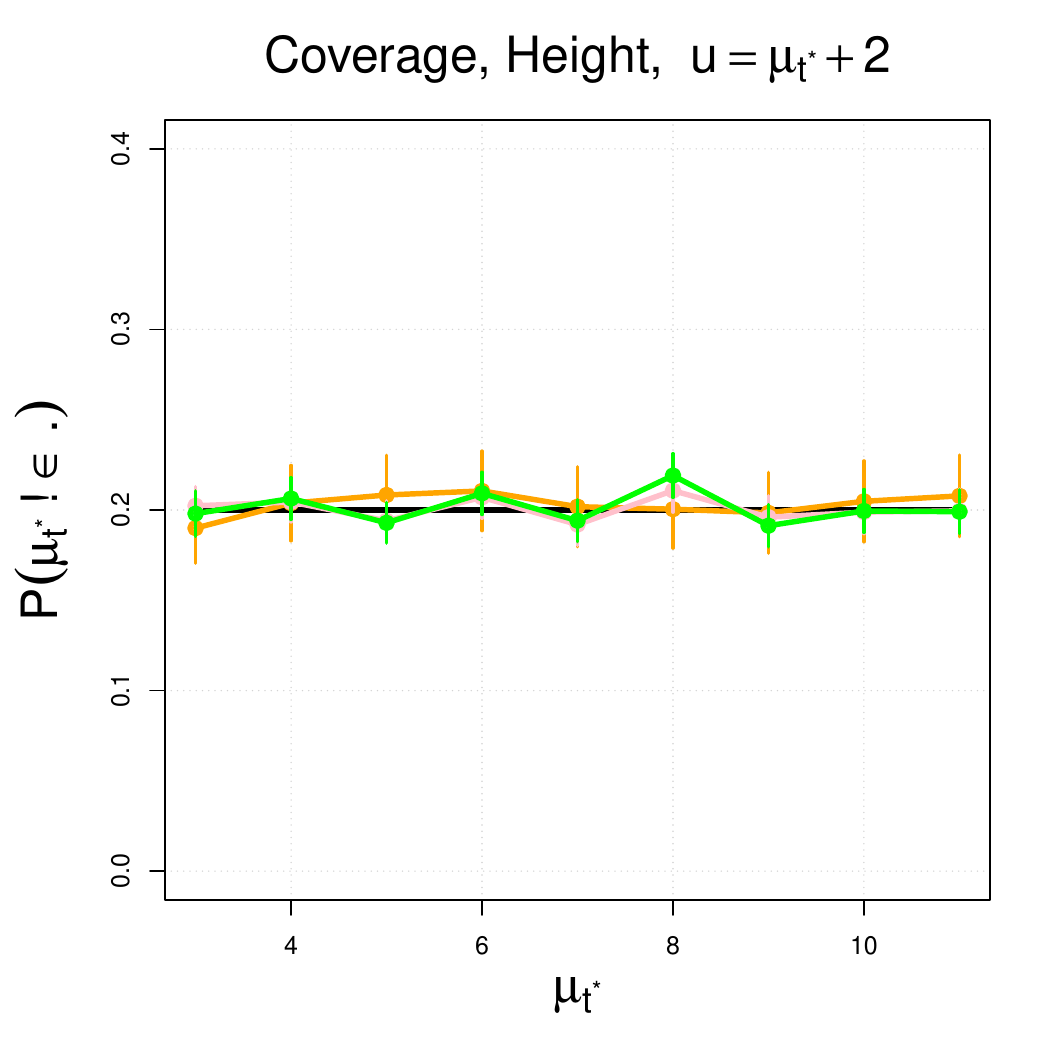}
		\end{subfigure}\hfill
		
		\begin{subfigure}[t]{0.32\linewidth}
			\centering
			\includegraphics[width=\linewidth]{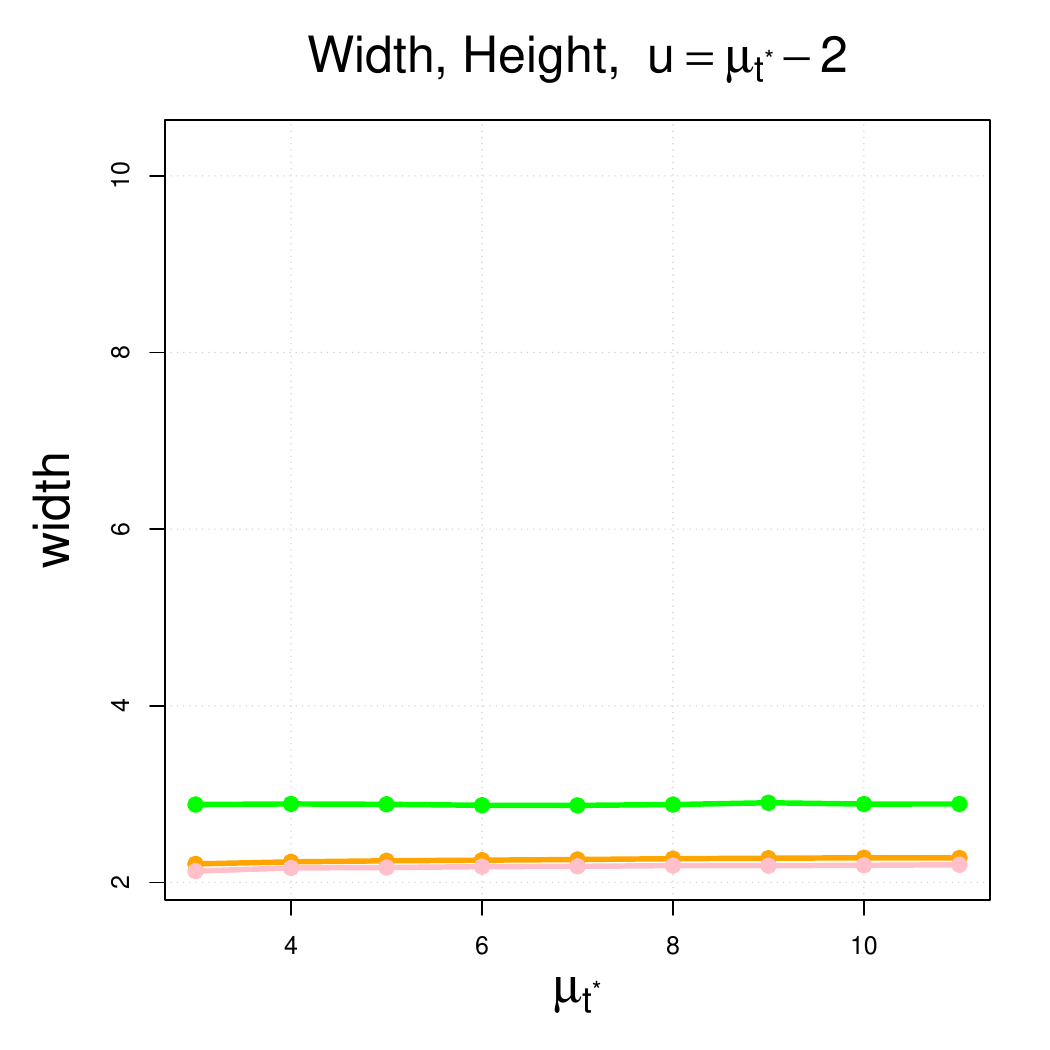}
		\end{subfigure}\hfill
		\begin{subfigure}[t]{0.32\linewidth}
			\centering
			\includegraphics[width=\linewidth]{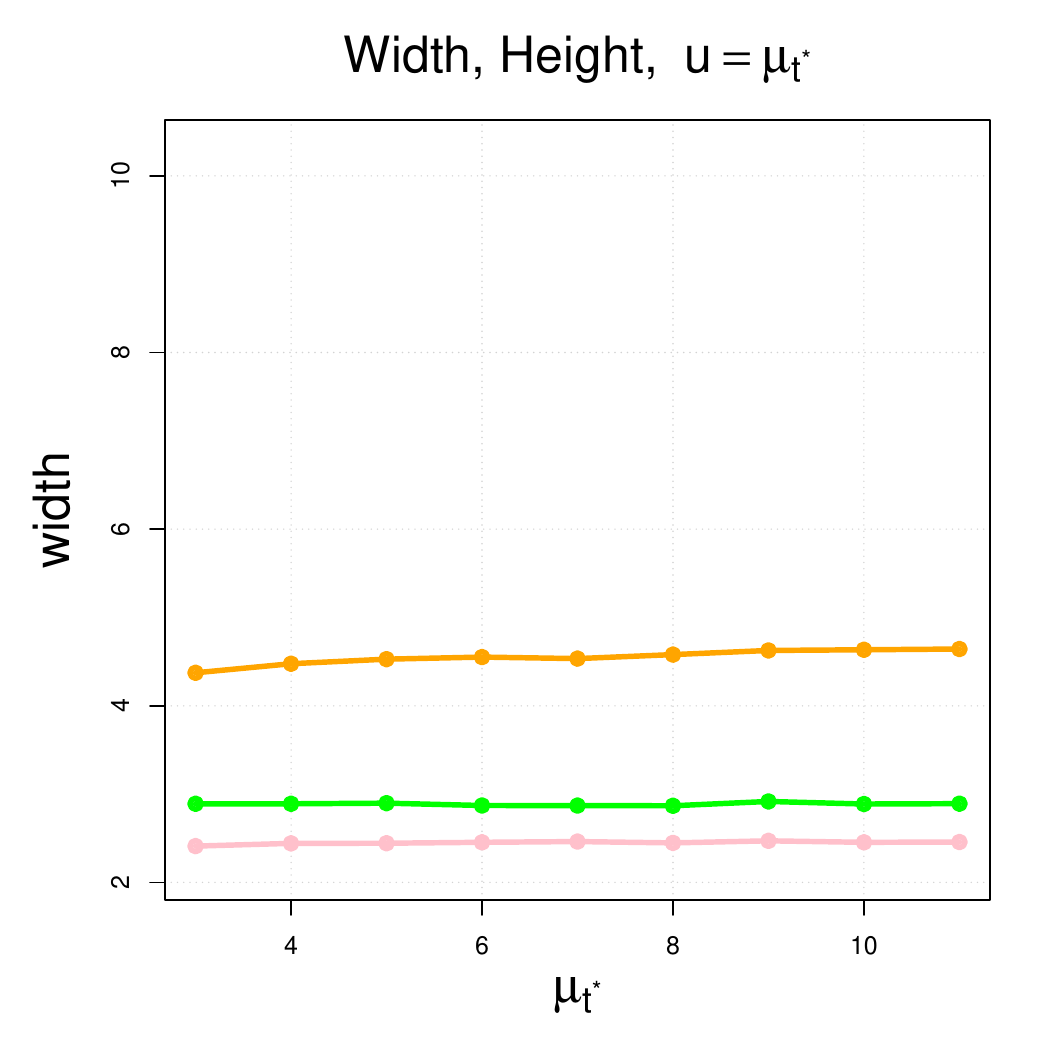}
		\end{subfigure}\hfill
		\begin{subfigure}[t]{0.32\linewidth}
			\centering
			\includegraphics[width=\linewidth]{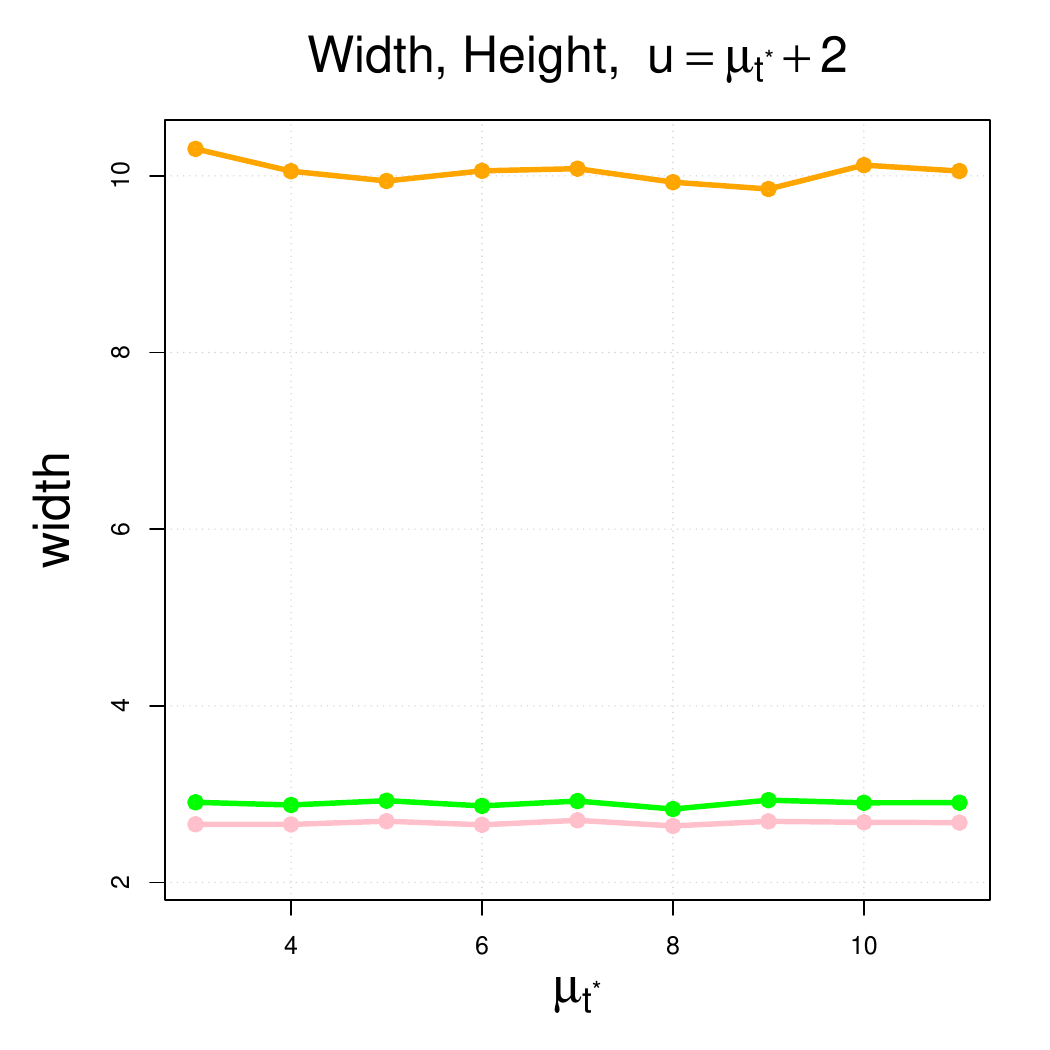}
		\end{subfigure}\hfill
		
		\caption{Comparing miscoverage and width of non-randomized, carve, and split methods for peak inference, in the one-dimensional experiment described in Section~\ref{subsec:experiment-2-additional}. Top two rows correspond to inference for location, bottom two rows to inference for height.}
		\label{fig:experiment-2-1d}
	\end{figure}

\end{document}